\documentclass[a4paper, twoside,openright]{report}

\usepackage[a4paper,top=3cm,bottom=3cm,left=3cm,right=3cm]{geometry} 
\usepackage[fontsize=11pt]{scrextend}
\usepackage[english]{babel}
\usepackage[fixlanguage]{babelbib}
\usepackage[utf8]{inputenc} 
\usepackage[T1]{fontenc}
\usepackage{rotating}
\usepackage{fancyhdr}               

\usepackage{amssymb}
\usepackage{amsmath, stmaryrd}
\usepackage{amsthm}         
\usepackage{proof}
\usepackage{graphicx}
\usepackage{pdfpages}
\definecolor{Black}{rgb}{0.0, 0.0, 0.0}
\definecolor{Green}{rgb}{0.0, 0.6, 0.3}
\definecolor{Blue}{rgb}{0.0, 0.0, 1.0}
\usepackage{listings}          
\usepackage{hyperref}     
\hypersetup{colorlinks,allcolors=Blue}
\usepackage[normalem]{ulem}
\usepackage{tikz}
\usetikzlibrary{shapes, positioning}
\usetikzlibrary{patterns}
\usepackage{extarrows}
\usepackage{mathrsfs} 
\usepackage{enumitem} 
\usepackage{mathtools}
\usepackage{booktabs}
\usepackage{bussproofs}
\usepackage[all,2cell,cmtip]{xy} 
\UseAllTwocells 
\usepackage{relsize}
\usepackage{stmaryrd}
\usepackage{longtable}
\usepackage{fancyvrb}
\VerbatimFootnotes

\usepackage{multirow}
\usepackage{microtype}
\usepackage{multicol}
\usepackage{needspace}

\fancypagestyle{plain}{
  \fancyfoot{}
  \fancyhead{}
  
}

\nocite{*}

\hypersetup{
    colorlinks,
    linkcolor=Blue,
    citecolor=Blue
}

\newtheorem{definition}{Definition}[section]
\newtheorem{convention}{Convention}[section]
\newtheorem{theorem}{Theorem}[section]
\newtheorem{lemma}{Lemma}[section]
\newtheorem{fact}{Fact}[section]

\usepackage{chngcntr} 

\newcounter{claimcounter}

\numberwithin{claimcounter}{theorem} 
\numberwithin{claimcounter}{lemma}   
\renewcommand{\theclaimcounter}{\arabic{claimcounter}}
\newenvironment{claim}
{\refstepcounter{claimcounter}%
 \begin{quote}\noindent\textit{Claim \theclaimcounter.} \hspace{0.5em}\ignorespaces}
{\end{quote}}

\newcounter{remarkcounter} 

\newenvironment{remark}[1][]
{\medskip\refstepcounter{remarkcounter}
 \noindent $\triangleright$\ \textit{\textbf{Remark~\theremarkcounter#1.}}\hspace{0.5em}\ignorespaces}

\usepackage{mdframed}
\newmdenv[
    topline=false, 
    bottomline=false, 
    rightline=false, 
    leftline=true, 
    linewidth=0.5pt, 
    linecolor=Blue, 
    skipabove=5pt, 
    skipbelow=5pt, 
]{mydefinition}
\newmdenv[
    topline=false, 
    bottomline=false, 
    rightline=false, 
    leftline=true, 
    linewidth=0.5pt, 
    linecolor=Blue, 
    skipabove=0pt, 
    skipbelow=0pt, 
]{mydefinition2}
\newmdenv[
    topline=false, 
    bottomline=false, 
    rightline=false, 
    leftline=true, 
    linewidth=0.5pt, 
    linecolor=black, 
    skipabove=5pt, 
    skipbelow=5pt, 
]{myproof1}

\newmdenv[
    topline=false, 
    bottomline=false, 
    rightline=false, 
    leftline=true, 
    linewidth=0.5pt, 
    linecolor=Black, 
    skipabove=5pt, 
    skipbelow=5pt, 
]{myproof}
\renewenvironment{proof}[1][\proofname]{%
  \begin{myproof}%
  \textbf{#1.} %
}{%
  \hfill $\Box$%
  \end{myproof}%
}

\newenvironment{claimproof}[1][Proof]
{\begin{myproof} 
    \noindent\textbf{#1.} \ignorespaces} 
{\par\vspace{2mm}  \par\end{myproof}} 

\newcommand{\ExternalLink}{
    \tikz[x=1.2ex, y=1.2ex, baseline=-0.05ex]{
     \color{Blue}
        \begin{scope}[x=1ex, y=1ex]
            \clip (-0.1,-0.1) 
                --++ (-0, 1.2) 
                --++ (0.6, 0) 
                --++ (0, -0.6) 
                --++ (0.6, 0) 
                --++ (0, -1);
            \path[draw, 
                line width = 0.5, 
                rounded corners=0.5] 
                (0,0) rectangle (1,1);
        \end{scope}
        \path[draw, line width = 0.5] (0.5, 0.5) 
            -- (1, 1);
        \path[draw, line width = 0.5] (0.6, 1) 
            -- (1, 1) -- (1, 0.6);
        }
    }

\lstdefinelanguage{OCaml}{
    keywords=[3] {let, in, val, rec, fun, match, with, type, and, open, module, struct, sig, end, if, then, else, begin, end},
    keywordstyle=[3]\color{Green}\bfseries,
    identifierstyle=\color{black},
    commentstyle=\itshape\color{gray},
    stringstyle=\color{Blue},
    morecomment=[s]{(*}{*)},
    morestring=[b]",
    sensitive=true
}

\DeclareSymbolFont{symbolsC}{U}{txsyc}{m}{n}
\DeclareMathSymbol{\strictif}{\mathrel}{symbolsC}{74}
\newcommand{\customref}[1]{1(b)i.}
\newcommand{\customreff}[1]{1(b)ii.}

\usepackage{pifont}
\usepackage{colortbl}

\begin{document}

\includepdf[pages=1]{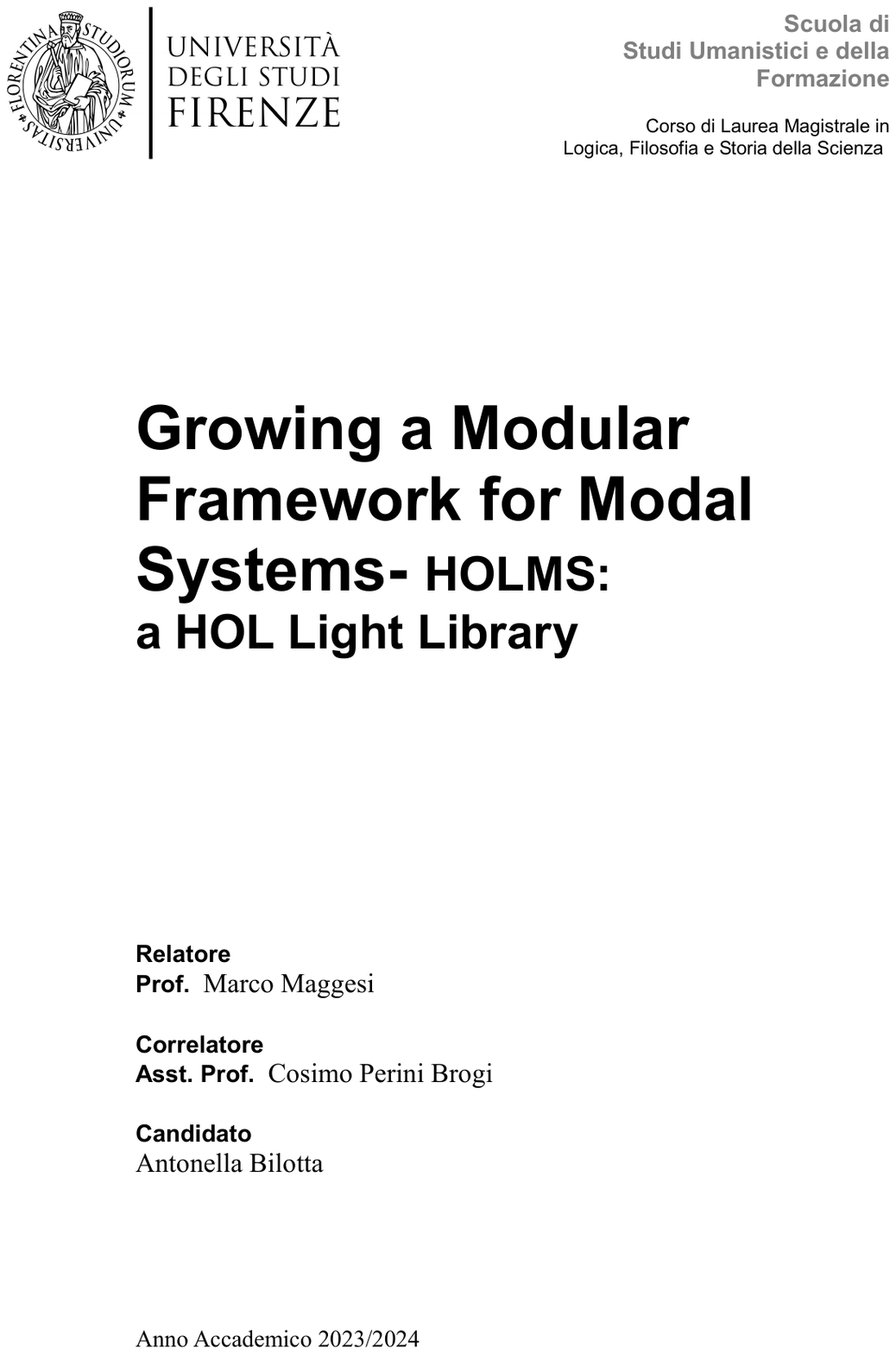} 

\begin{abstract}
    The present dissertation introduces the research project on HOLMS (\textbf{HOL} Light Library for \textbf{M}odal \textbf{S}ystems), a growing modular framework for modal reasoning within the HOL Light proof assistant. 
    To provide an accessible introduction to the library, the fundamentals of modal logic are outlined first, followed by a concise manual for the proof assistant itself.
    The core contribution of this work on HOLMS is the development of a unified and modular strategy for proving adequacy theorems with respect to relational semantics directly within HOL Light for several normal modal systems, currently including K, T, K4, and GL. Adequacy theorems establish a formal connection between syntactic proof systems and their intended relational models, ensuring that derivable statements align with valid ones. This approach extends previous research on Gödel-Löb logic (GL) by two HOLMS developers. It also assesses the generality and compositionality of the completeness proofs in George Boolos' monograph  \textit{The logic of provability}.
    Beyond theoretical contributions, HOLMS incorporates automated decision procedures and a countermodel constructor for K, T, K4, and GL, illustrating how general-purpose proof assistants can be effectively combined with research on labelled sequent calculi and key insights from correspondence and bisimulation theories. The implementation in HOL Light demonstrates the feasibility of mechanising modal reasoning in a flexible and robust manner, paving the way for further developments of the HOLMS framework.
\end{abstract}

\tableofcontents
\addcontentsline{toc}{chapter}{Introduction} 

\chapter*{Introduction}
\markboth{Introduction}{Introductiom}

\textbf{HOLMS Project.} This thesis has been developed in the context of the research project ``HOLMS: \textbf{HOL}-Light Library for \textbf{M}odal \textbf{S}ystems''. HOLMS constitutes a novel framework within the HOL Light proof assistant, designed for automated theorem proving and countermodel construction in modal logics. It extends previous work on Gödel-Löb logic (GL) by two HOLMS developers~\cite{DBLP:journals/jar/MaggesiB23}, offering a more general and modular approach.

The first version of HOLMS was introduced in November 2024 at the \emph{6th International Workshop on Artificial Intelligence and
fOrmal VERification, Logic, Automata, and sYnthesis: OVERLAY 2024}~\cite{DBLP:conf/overlay/BilottaMBQ24}. On that occasion, we generalised the approach to cover a wider range of normal modal systems, starting with the minimal system K. 
That initial version of HOLMS already featured a flexible framework for automating proof search and countermodel generation, leveraging labelled sequent calculi, interactive theorem proving, and formal completeness theorems. As a result, it laid the foundation for a comprehensive tool for modal logic reasoning, combining high levels of confidence and automation.

By employing concepts drawn from the \textit{correspondence theory}, this work extends previous results and creates a \textit{uniform, scalable method} for dealing with several modal systems within HOL light. The integration of two new modal systems--K4 and T-- into HOLMS serves as evidence of the \textit{modularity} of the implemented framework. 

The current version of~\href{https://archive.softwareheritage.org/swh:1:dir:6dbea8a51b27f3b4de7808eb0bc2fa82706962c4;origin=https://github.com/HOLMS-lib/HOLMS;visit=swh:1:snp:2c0efd349323ed6f8067581cf1f6d95816e49841;anchor=swh:1:rev:1caf3be141c6f646f78695c0eb528ce3b753079a}{ HOLMS repository}, which incorporates all the new features introduced in this thesis, has been submitted to an international conference.
\medskip

\noindent \textbf{Source Code.} 
HOLMS consists of just under six thousand lines of OCaml code, organised into a~\href{https://github.com/HOLMS-lib/HOLMS}{\emph{git repository}}.  The project~\href{https://holms-lib.github.io/}{\textit{web page}}~\cite{holms-website} gives a brief overview of the library, its evolution and related publications, while the~\href{https://holms-lib.github.io/}{\textit{readme file}} highlights the specificities of this second version, and provides a usage guide for our library at its current status.
The present version of HOLMS was released on  ~\href{https://archive.softwareheritage.org/swh:1:dir:6dbea8a51b27f3b4de7808eb0bc2fa82706962c4;origin=https://github.com/HOLMS-lib/HOLMS;visit=swh:1:snp:2c0efd349323ed6f8067581cf1f6d95816e49841;anchor=swh:1:rev:1caf3be141c6f646f78695c0eb528ce3b753079a}{\textit{Software Heritage}}, alongside the thesis we will use the symbol \ExternalLink to link the lines of the source code which we are referring to.
\medskip

\noindent \textbf{Authored published work.} During my Master, I authored the following conference and journal papers, on which the present thesis is based: \begin{itemize}
    \item ``Growing HOLMS, a HOL Light Library for Modal `Systems''~\cite{DBLP:conf/overlay/BilottaMBQ24} \\ A. Bilotta, C.Perini Brogi, M.Maggesi and L. Quartini\\ Communication paper presented at the \textit{International workshop OVerLAY 2024}  introducing a first embryo version of HOLMS library;
    \item ``HOLMS framework for modal reasoning in HOL'' \\ A. Bilotta, C.Perini Brogi and M.Maggesi \\
     An extended version of the communication paper is actually under review, providing an overview of the latest developments in the library.
\end{itemize}
\medskip

\noindent \textbf{Contents.} This thesis intends to describe in detail the HOLMS library and its formalisation of \textit{adequacy theorems} and \textit{decision procedures} within the interactive theorem prover HOL Light. Since the author has been involved in the HOLMS project from its inception, this work also provides an opportunity to reconstruct the evolution of the library.

The first chapter introduces the \textbf{fundamentals of modal logic}, offering an \textit{informal} overview of the logical concepts employed in the implementation of HOLMS. This chapter can be read alongside chapter four, both for a better understanding of the formal version of the proofs developed with the aid of the theorem prover and to appreciate the implementation choices made in HOL Light. After outlining the focus and evolution of modal logic, we provide definitions of \textit{syntax} and \textit{semantics} for modal logics, along with an \textit{axiomatic calculus for normal systems}. We then present key results from\textit{ correspondence theory}, followed by a detailed proof of \textit{adequacy theorems} and a brief exploration of \textit{bisimulation theory} and \textit{labelled sequent calculi} for modal logics.

In the second chapter, we present \textbf{HOL Light}, beginning with an explanation of what a theorem prover is and then focusing on the specific features of this proof assistant. We cover its \textit{logic}, including language, metalanguage, syntax, and deductive system, and conclude with a detailed explanation of how \textit{interactive proving} works in HOL Light.

The third chapter is more concise than the previous ones but provides an overview of \textbf{various approaches to formalising modal logic within HOL Light}. Beginning with HOL Light tutorial \textit{deep and shallow embeddings}, it then focuses on earlier work of two HOLMS developers on the \textit{implementation of Gödel-Löb provability logic}.
Finally, the chapter examines the \textbf{HOLMS library}, describing the project's aims, methodology, and evolution. It also includes a comprehensive map of the HOLMS source code and repository.

The fourth and final chapter ``grows'' the current version of our repository and--to describe in depth the implementation of the source code--incorporates several code listings. It provides a detailed account of the \textit{deep embedding} for both syntax and semantics, along with the \textit{formalisation of the deducibility predicate} for normal systems. In addition, the chapter implements key concepts and results from \textit{correspondence} and \textit{bisimultation theories} and presents step-by-step the \textit{proofs of adequacy theorems}. Building on these foundations, we conclude by presenting both a \textit{naive decision procedure} and a \textit{principled decision procedure} based on labelled sequent calculi and adequacy results.
\medskip

\noindent \textbf{Logic Notation.} To clearly differentiate the logical operators of HOL Light from those of our HOLMS modal language, as well as from the \textit{abstract} modal language and the metalanguage, we provide a glossary that clarifies these distinctions in Figure~\ref{fig:glossary}.
\medskip
\medskip

\begin{figure}[hb]
    \centering
        \begin{tabular}{|c|c|c|c|}
            \hline
            \multicolumn{2}{|c|}{\textbf{HOL Light notation}} & \multicolumn{2}{c|}{\textbf{``Informal'' language}} \\
            \midrule
            \textbf{HOL Light} & \textbf{HOLMS} & \textbf{Metalanguage} & \textbf{Object Language} \\
            \midrule
            \verb|F|  & \verb|False| & false & $\bot$ \\
            \verb|T|  & \verb|True|  & true  & $\top$ \\
            \verb|~|  & \verb|Not|   & not   & $\neg$ \\
            \verb|/\| & \verb|&&|   & and   & $\land$ \\
            \verb|\/| & \verb||||   & or    & $\lor$ \\
            \verb|==>|& \verb|-->|  & if...then... ($\implies$)& $\to$ \\
            \verb|<=>|& \verb|<->|  & iff ($\iff$) & $\leftrightarrow$ \\
                      & \verb|Box|  &       & $\Box$ \\
            \verb|!|  &             & for all ($\forall$) &  \\
            \verb|?|  &             & exists ($\exists$) &  \\
            \verb|?!| &             & exists exactly one ($\exists!$) &  \\
        \hline
        \end{tabular}

    \caption{A glossary for the logical notation used in this work}
    \label{fig:glossary}
\end{figure}

\chapter{An ``Informal'' Introduction to Modal Logic}\label{chap:1}

This chapter is designed to provide the reader with all the basic knowledge required to comprehend HOLMS implementation of modal reasoning in HOL Light. 
The entire chapter can be regarded as an \textit{informal/abstract} introduction to modal logic and can be read in parallel with the presentation of HOLMS in Chapter~\ref{chap:4}.

Such a parallel reading of the two chapters allows to appreciate the ideas and choices followed in HOLMS to \textit{formalise}\footnote{This expression here refers to \textit{formalisation} within a proof assistant.} the \textit{abstract} concepts of modal logic presented here. Concurrently, the theorems \textit{informally} proved in this chapter will adhere to the very same strategies employed in HOLMS, allowing the user to leverage these proofs to better understand their implementation in HOL Light. To further facilitate this understanding, links to lines of the source code formalising abstract concepts have been added with the symbol \ExternalLink.

First, this chapter explains what we refer to by using the expression ``modal logic'' and reconstructs a concise \textbf{historical account} of the development of this field. After introducing modal \textbf{syntax} and \textbf{semantics}, an \textbf{axiomatic calculus for normal modal systems} is defined. Thereafter, we show the main results of \textbf{correspondence theory}, in order to modularly demonstrate completeness and soundness.
Finally, we focus on logics within the \textbf{modal cube} and their \textbf{adequacy theorems} with respect to relational semantics are proved.  Furthermore, \textbf{bisimulation theory} is briefly presented and \textbf{labelled sequent calculi} are introduced.

It should be noted that, since HOLMS proof of completeness is a refinement of the demonstrative strategy outlined by Boolos~\cite[\S5]{boolos1995logic}, our presentation of modal reasoning and--in particular of the adequacy theorems--will follow and extend the approach of this classical textbook.

\section{What is Modal Logic?}
Modal logic is the branch of logic that deals with modal propositions. Modal sentences are typical of natural language and are characterised by specific linguistic markers, known as \textbf{modal operators}, that specify the manner or mode in which a proposition holds.

For example, ``How many roads \textit{must} a man walk down/ \textit{before} you call him a man?''\footnote{Bob Dylan, ``Blowin' in the Wind'',  \textit{The Freewheelin' Bob Dylan}, 1963.} and ``I'\textit{d} rather be a forest than a street/ Yes, I \textit{would}/ If I \textit{could}/ I surely \textit{would}''\footnote{Simon and Garfunkel, ``El Condor Pasa (If I could)'', \textit{Bridge Over Troubled Water}, 1970.} are modal propositions by distinguished songwriters and their modal operators are reported in italics. 

While classical logic concerns the truth values of the statements, modal logic extends this perspective by reflecting on \textit{how} statements are affirmed or denied. Modal sentences go beyond the true/false dichotomy and they can concern many linguistic constructs that cannot be reduced to truth-functional logical operators (\textbf{modalities}), such as necessity and possibility, knowledge and belief or obligation and permission. Some classical modalities are listed below (Table~\ref{tab:modalities}), together with their specific modal operators.

\begin{table}[h]
    \centering
    \scriptsize
    \renewcommand{\arraystretch}{1.3} 
    \begin{tabular}{| p{2.5cm} | p{3.5cm} | p{3cm} | p{4cm} |}
        \hline
        \textbf{Modality} & \textbf{Specify the truth} & \textbf{Operators} & \textbf{Expresses ...}  \\
         & \textbf{according to...} & & \\
        \hline
        \multirow[t]{2}{=}{\textbf{Alethic}} & Metaphysic constraints
        & "necessarily"   & What must be the case  \\
        \cline{3-4}
        & & "possibly" &  What may be the case  \\
        \hline
        \textbf{Epistemic} & Epistemic constraints
        & "knows that"   &  What is known to be true  \\
        \hline
        \textbf{Doxastic} & Beliefs constraints
        & "believes that" &  What is believed to be true \\
        \hline
        \multirow[t]{3}{=}{\textbf{Deontic}} & Normative constraints
        & "ought to"  &  What is required or mandatory  \\
        \cline{3-4}
        &  & "is permitted to"  &  What is allowed or permitted \\
        \cline{3-4}
        & & "is forbidden to" & What is forbidden  \\
        \hline
        \multirow[t]{4}{=}{\textbf{Temporal}} & Time constraints
        & "always"    & That something is always true  \\
        \cline{3-4}
        & & "eventually" & That something will be true at some point \\
        \cline{3-4}
        & & "sometimes"  &  That something is sometimes true  \\
        \cline{3-4}
        & & "until"      &  That something holds until another condition is met  \\
        \hline
        \textbf{Dynamic} & Program execution & $[a]P$ &  \( P \) is true after executing $a $ \\
        \hline
    \end{tabular}
    \caption{Modalities and modal operators}
    \label{tab:modalities}
\end{table}

Among these, \textit{alethic modalities}  have ever had a primary role. Modal logic was born precisely to treat them: the language of propositional logic was extended with the two signs box `$\Box$' (\textit{`necessarily'}) and diamond `$\Diamond$' (\textit{`possibly'})\footnote{Note that possibility and necessity are interdefinable over a classical base: a statement is \textit{necessary} if and only if its negation is not \textit{possible}. We can take `$\Box$' as primitive and define `$\Diamond$' as `$\neg \ \Box \ \neg$'.}; and a semantics reflecting our intuitive understanding of necessity and possibility was developed.

The prominence of \textit{altehic modalities} is both historical and theoretical: on the one hand, this was the first modal logic developed by Aristotle and the most investigated during the centuries; on the other hand, the \textit{relational semantics} developed by Kripke for alethic modalities can describe all the others as different interpretations of box and diamond.

In this chapter, we will analyse modal logics independently from their modalities/interpretations. In particular, we will  study the notion of implication between modal sentences by using two logical relationships: 
\begin{itemize}
\item[$\vDash$]: the model-theoretic logical consequence relation (semantical notion)
\item[$\vdash$]: the proof-theoretic deducibility relation (syntactical notion)
\end{itemize}

\subsection{A Brief History of Modal Logic}

In this section, we will briefly illustrate the development of modal logic during the millennia and how a discipline involving metaphysical and uncertain concepts, such as necessity and possibility, has moved from the periphery of logic to a primary role.

\subsubsection{Aristotle and Medieval logicians}
Aristotle profusely reasoned about modal logic in \textit{Prior Analytics}, \textit{Metaphysics} and \textit{De Interpretatione}. This work is sometimes accused of being confused and almost incomprehensible due to its many errors and inconsistencies \cite{boolos1995logic}; still, here Aristotle developed some basic results that will be verified with modern modal instruments.
For example, he developed a theory of syllogism with modal sentences in the premises, and he sketched a modal \textit{square of opposition} to illustrate the relationship between alethic modalities via negation. Aristotle also informally proved some lemmas that hold for each \textit{normal}\footnote{This concept will be introduced in the following sections.} modal logic, such as $\vDash \Box \ (A \to B) \to (\Box A \to \Box B) $.

The reading of Aristotelic work stimulates the medieval disputes about \textit{future contingents}\footnote{The problem of \textit{future contingents} could be roughly summarized as follows: (1) God is omniscient (2) Then he can predict future events (3) So how could the action of a man be contingent? This philosophical discussion is primarily theological and hides the \textit{free will} problem, but it had also an important impact on the study of modal logic.} and the creation of a refined modal doctrine that catches the similarities (\textit{similitudo}) between quantifiers and modal operators ($\forall \sim \Box$ and $\exists \sim \Diamond$). 

From the 15th century, the interest in modal logic tended to decline, and the rebirth of logic as mathematical logic in the second half of the 19th century did not correspond to a renaissance of modal logic. To observe a renewed interest in this topic, we have to wait till the beginning of the 20th century and to Lewis' attempt to characterise a \textit{strict conditional}.

\subsubsection{The Attempt of Characterising Strict Conditional}
Stoics were the first to study deeply logical conditionals since their logic focused on the relationship between sentences. They elaborated at least three types of conditionals, that we reconstruct here--not without philological inaccuracies--following~\cite{MugnaiStoria}.
\begin{enumerate}
    \item \textit{Material Conditional} \\
          Philo the Dialectician elaborated the characterisation used by modern logic: \textit{Material conditional} is true \textit{iff} the premises are false or the consequence is true. 
    \item \textit{Diodorus Conditional }\\
     Diodorus Cronus matched the truth conditions of conditionals with temporal constraints.
     \item \textit{Chrysippus Conditional}\\
     Chrysippus' characterisation required a semantic relation between premises and consequences that truth-functional operators cannot describe. Medieval logicians preferred this implication and wrote: ``a conditional is true iff is \textit{impossible} that premises are true and the consequence is false''. It is unclear if this is a genuine modal sentence, or if ``\textit{impossible}'' here simply stands for \textit{``is not the case that''}.
\end{enumerate}

First Hugh McColl in 1880~\cite{McColl1880-MCCSR} and then, more strenuously, Clarence Irving Lewis in 1912~\cite{10.1093/mind/XXI.84.522}, rediscovered and theorized a not-material conditional, which:

\begin{enumerate}
    \item Cannot be defined with a truth table and then is not a truth-functional operator
    \item Is logical equivalent to the modal sentence \textit{it is impossible the conjunction of the antecedent and the negation of consequent} ($ A \strictif B \iff \neg \Diamond (A \land \neg B) $) \footnote{Material conditional can similarly be characterized as $A \to B \iff \neg  (A \land \neg B) $}
\end{enumerate}

Lewis used the sign `$ \strictif $' to denote \textit{strict implication} and defined `$A \strictif B$' as `$\Box (A \to B)$'~\footnote{Analogously you can take `$\strictif$' as primitive and define `$\Box A$' as `$(A \to A) \strictif A$'.}. He, hence, reflected on a \textit{conditional that is necessarily true}, and which formalises the \textit{deduction} of a sentence from another~\footnote{What Lewis meant with the word \textit{deducible} is uncertain: some passages of ``Symbolic Logic'' suggest he was pointing to deduction in formal system, while other lines seem to treat logical/metaphysical necessity.}. To characterise strict implication, Lewis developed several axioms systems of modal logic~\footnote{In these systems, Lewis took $\Diamond$ as primitive.}, presented in his 1932 ``Symbolic Logic''~\cite{Lewis1932-LEWSLA} as $\mathbb{S}\mathbf{1}, \mathbb{S}\mathbf{2}, \mathbb{S}\mathbf{3}, \mathbb{S}\mathbf{4},$ and $\mathbb{S}\mathbf{5}$. Thanks to Lewis' aim of characterising a non-philonian conditional, and of developing a formal system alternative to the one presented in Russel and Whithead's ``Principia Mathematica'' (1910-1913)~\cite{Russel1910Principia}; the modern study of modal logic began and two normal systems still used today ($\mathbb{S}\mathbf{4}$ and $\mathbb{S}\mathbf{5}$) were developed. Moreover, Lewis' uncertain idea that deducibility can be studied with the aid of modal system turns out to be correct with provability logic $\mathbb{GL}$.

\subsubsection{Other Syntactical Developments of Modal Logic}

In the same years, the Polish logician Jan Łukasiewicz laid the foundations of \textit{many-valued logics}, by defining in 1918~\cite{ca44bd83-49fd-3b77-987a-071e647d06f0} a three-valued logic where, in addition to true (=1) and false (=0), he added a third truth value (=1/2) interpreted as `possible'. Nevertheless, this system could derive $ \vDash 1/2 A \land 1/2 B \implies 1/2 (A \land B)$ and a simple interpretation of A and B as ``Schrödinger's cat is alive'' and ``Schrödinger's cat is dead'' shows that this is not a good characterisation of possibility. It is possible that Schrödinger's cat is alive and it is also possible that Schrödinger's cat is dead, but--ignoring quantum mechanics--Schrödinger's cat cannot be both dead and alive.

A short after Lewis' ``Symbolic Logic'', Kurt Gödel published ``An Interpretation of the Intuitionistic Propositional Calculus'' (1933)~\cite{Godel-1933}, and he showed that Lewis modal system $\mathbb{S}4$ could be axiomatised as extensions of the propositional calculus. Since this article, every axiomatisation of a modal system will follow Gödel's advice, by clearly distinguishing the classical basis of the system from the modal axioms and rules.

\subsubsection{Semantical Development of Modal Logic}
Besides this syntactical development of modal logic, there was limited semantical progress, and no satisfactory interpretation of the new theorems nor of the differences between modal systems. A theory of reference, in this sense, was \textit{Leibniz's semantics of possible worlds}. Leibniz drew up his reflection in a metaphysical and theological context next to the problem of \textit{future contingents}. He proposed an innovative solution that connects the evaluation of modal sentences with the metaphysics of possible worlds, by defining a \textit{necessary sentence} as a \textit{sentence that is true in all possible worlds}.

Despite Leibniz not employing consistently his modal intuitions, in 1942 Carnap~\cite{Carnap1942} developed the first modern model-theoretic study of modal logic by formalising \textit{semantics of possible worlds} in his \textit{semantics of state descriptions}. Carnap semantics interpret the box as follows, given M a collection of \textit{state descriptions}: 

\begin{quote}
    `$\Box A$' is true in $S$ \textbf{iff} for every state description $S'$ in $ M$, `$A$' is true in $S'$.
\end{quote}
Carnap's attempt to establish modal semantics was affected by two problems:
\begin{enumerate}
    \item The definition of the truth of a modal sentence in a state description $S$ does not depend on the state description $S$ at all;
    \item This definition yields the result that iterations of the box have no effect ($\Box A \iff \Box \Box A$), contradicting our informal intuition of necessessity. 
\end{enumerate}

In the second half of the 20th century, Arthur Prior~\cite{Prior1957-PRITAM} founded temporal logics semantics, clarifying the truth conditions of the sentence `\textit{it was once the case that}':

\begin{quote}
     `\textit{It was once the case that $A$}' is true at time $t$ \textbf{iff} `$A$' is true at some time $t'$ \textit{earlier} than $t$.
\end{quote}
Notice that this non-vacuous definition mentions a relation of temporal precedence. The introduction of such a \textit{relation} in modal semantics, gave Prior flexibility to define various tense operators~\cite[\S1.1]{ZaltaManuscript-ZALBCI}. 

Saul Kripke~\cite{Kripke1959-KRIACT, Kripke1963-KRISCO} developed \textit{relational} semantics to such an extent that it is well-known as Kripke semantics.
In particular, he  introduced a universe of possible worlds, but he did not follow Carnap and Leibniz in defining the interpretation of the box as a universal quantification over possible worlds:
\begin{quote}
    `$\Box A$' is true at a world $w$ iff `$A$' is true \textit{in every possible world}.
\end{quote}
Following Prior instead, Kripke extended Carnap's semantics by introducing an \textit{accessibility relation} on the possible worlds such that:
\begin{quote}
`$\Box A$' is true at a world $w$ iff `$A$' is true \textit{at every possible world} $w'$ \textit{accessible from} $w$.    
\end{quote}

By exploiting the innovative idea that not every model is accessible from a given world, Kripke's semantics allow us to give an account of the differences between modal systems by analysing their different accessibility relation~\footnote{We will introduce Kripke's semantics and its features in detail in Section~\ref{sec:informal-semnatics}.}. 

From this point on, modal logic has all modern logical instruments at its disposal to become a coherent branch of logic.

\subsubsection{Contemporary Applications of Modal Logic}

Since 1970s, this field has grown to the point that modal logic is a central tool in several scientific and technical disciplines, including knowledge representation, formalisation of reasoning under uncertainty, the study of multi-agent systems, the analysis of computational processes, the verification of consistency of normative corpora, and the modelling of decision-making processes~\cite{DBLP:books/el/07/BBW2007,DBLP:books/cu/BlackburnRV01}.

Various modal operators has turned out to capture interesting abstract properties.
In particular, temporal logics ($\mathbb{TL}$, $\mathbb{LTL}$, $\mathbb{CTL}$, etc.) have found an important application in computer science, and more specifically in formal verification and model checking, where they are used to specify and prove compliance of requirements of hardware or software systems~\cite{clarke1999model,stirling1991modal}.

Furthermore, Provability logic ($\mathbb{GL}$) has succeeded in  properly describing formal provability in mathematical theories~\cite{sep-logic-provability}, as observed in~\cite{boolos1995logic}: 
\begin{quote}
 "The symbolism of modal logic turns out to be an exceeding useful notation for representing the forms of sentences of formal theories [...], and techniques originally devised to study system of modal logic disclose facts of great interest about these notions and their strange property."   
\end{quote}

\newpage 

\section{Modal Logic Syntax}\label{sec:informal-syntax}
We introduce a general syntax that allows us to treat every modal system.
\begin{definition}[\textbf{Alphabet}]

    $\mathcal{L}_{\Box} = \{ \bot; \ \to; \ \Box; \ (; \ )   \} \cup \Phi$ 
\begin{mydefinition}
    \begin{mydefinition2}
        The \textbf{Alphabet} of propositional modal language $\mathcal{L}_{\Box}$ consists of:
\begin{itemize}
    \item an \textbf{infinite denumerable} set $\Phi$ of \textbf{propositional atoms}: $p_0, p_1, p_2, \dots$;
    \item \textbf{logical constants}:
    \begin{itemize}
        \item contradiction symbol: $\bot$;
        \item a classical truth-functional operator for implication: $\rightarrow$;
        \item modal operator: $\Box$;
    \end{itemize}
    \item \textbf{auxiliary symbols}: $(, )$.
\end{itemize}
    \end{mydefinition2}
\end{mydefinition}        
\end{definition}

\begin{convention}[\textbf{Metavariables for formulas and propositional formulas}]
    We will use lowercase letters $p, q, r$ as metavariables for \textbf{propositional atoms} and capital letters $A, B, C$ as metavariables for \textbf{formulae}.
\end{convention}

\begin{definition}[\textbf{Modal Formulas}~\href{https://archive.softwareheritage.org/swh:1:cnt:4a011f69b1180019be5d03c6d5f1cec4550055e8;origin=https://github.com/HOLMS-lib/HOLMS;visit=swh:1:snp:2c0efd349323ed6f8067581cf1f6d95816e49841;anchor=swh:1:rev:1caf3be141c6f646f78695c0eb528ce3b753079a;path=/modal.ml;lines=25-34}{\ExternalLink}]
 $A \in \mathbf{Form}_{\Box} \Coloneqq \bot \ | \ p_i \ | \ A \to B \ | \ \Box A $ 
 \begin{mydefinition}
\begin{mydefinition2}
    The set $\mathbf{Form}_{\Box}$ of modal formulas on $\mathcal{L}_{\Box}$ is defined by induction as follows:
\begin{itemize}
    \item if $p_i \in \Phi$ then $p_i \in \mathbf{Form}_{\Box}$;
    \item $\bot \in \mathbf{Form}_{\Box}$;
    \item if $A \in \mathbf{Form}_{\Box}$ and $B \in \mathbf{Form}_{\Box}$ then $(A \to B) \in \mathbf{Form}_{\Box}$;
\end{itemize}
\end{mydefinition2}
\begin{mydefinition2}
\begin{itemize}
    \item if $A \in \mathbf{Form}_{\Box}$ then $\Box (A)  \in \mathbf{Form}_{\Box}$;
    \item Nothing else is in $\mathbf{Form}_{\Box}$.
\end{itemize}
\end{mydefinition2}
\end{mydefinition}
\end{definition}

Even if our language does not provide all the canonical logical connectives, it allows us to define them as metalinguistic abbreviations.

\begin{convention}[\textbf{Metalinguistic Abbreviations}]\label{cnv:metaling_abbr}\phantom{Allyouneedislove}
    \begin{multicols}{2}
        \begin{enumerate}
            \item $\top \coloneq (\bot \to \bot)$;
            \item $\neg A \coloneq (A \to \bot)$;
            \item $ A \lor B \coloneq (\neg A \to B)$;
            \item $A \land B \coloneq \neg (\neg A \lor \neg B)$;
            \item $A \longleftrightarrow B \coloneq ((A \to B) \land (B \to A)) $;
            \item $\Diamond A \coloneq \neg \ \Box (\neg A)$.
        \end{enumerate}
    \end{multicols}
\end{convention}

We have defined a ``minimal'' modal alphabet in order to ease the comprehension of proofs and definitions. Our alphabet is minimal in the sense that our language takes as primitive a \textit{functionally complete}\footnote{A set of Boolean operators is functionally complete iff it can express all possible truth tables by combining its members into a Boolean expression.} set of only two Boolean operators $\{ \bot; \to \}$ and the modal sign `$\Box$'. Such a choice has some advantages :
\begin{enumerate}
    \item A quick-to-check \textit{induction on the complexity of formulas}. 
    \smallskip
    
    If we want to prove a theorem that is true for all modal formulas, we simply have to show that the theorem holds for `$\bot$' and for all the propositional atoms (base cases), and that the two inductive steps for `$\to$' and `$\Box$' hold. \smallskip
    
    For example, this will allow us to quickly prove Lemma~\ref{lem:truth}, while its proof in HOLMS~\href{https://archive.softwareheritage.org/swh:1:cnt:d364de3158b6405c920115aa0c801af9af49e302;origin=https://github.com/HOLMS-lib/HOLMS;visit=swh:1:snp:2c0efd349323ed6f8067581cf1f6d95816e49841;anchor=swh:1:rev:1caf3be141c6f646f78695c0eb528ce3b753079a;path=/gen_completeness.ml;lines=189-350}{\ExternalLink}--presented in Chaper~\ref{chap:4} (Listing \ref{lst:GEN_TRUTH_LEMMA})--is much longer\footnote{For practical reasons, HOLMS takes every boolean operator as primitive.}.
    
    \item A faster way to \textit{define} properties of modal formulas \textit{by induction}. 
    \smallskip
    
    An inductive definition in a language that takes as primitive all modal operators has three base cases and six inductive steps, while our definitions only have two base cases and two inductive steps. 
    \smallskip
    
    For example, our Definition~\ref{def:truth} of the truth of a sentence is much shorter than HOLMS one (\verb|holds|~\href{https://archive.softwareheritage.org/swh:1:cnt:4a011f69b1180019be5d03c6d5f1cec4550055e8;origin=https://github.com/HOLMS-lib/HOLMS;visit=swh:1:snp:2c0efd349323ed6f8067581cf1f6d95816e49841;anchor=swh:1:rev:1caf3be141c6f646f78695c0eb528ce3b753079a;path=/modal.ml;lines=44-59}{\ExternalLink}). Similarly,  Definition~\ref{def:subformula} of subformula is very brief.
\end{enumerate}

Below, we introduce some conventions about the use of parentheses and associativity of propositional and modal connectives. 
    
\begin{convention}[\textbf{Parentheses conventions and associations}]\phantom{Allyouneedislove}
        \begin{itemize}
            \item The outermost parentheses are omitted;
            \item Parentheses are omitted according to the following order of precedence on connectives:
             $\neg$; $\Box, \Diamond$;  $\land, \lor$; $\to$; $\leftrightarrow$;
            \item Conjunction and disjunction associate to the left.
        \end{itemize}
\end{convention}

Moreover, we formally define the well-known notions of \textit{subformula}, \textit{occurrence} and \textit{subsentence}.  

\begin{definition}[\textbf{Subformula}~\href{https://archive.softwareheritage.org/swh:1:cnt:4a011f69b1180019be5d03c6d5f1cec4550055e8;origin=https://github.com/HOLMS-lib/HOLMS;visit=swh:1:snp:2c0efd349323ed6f8067581cf1f6d95816e49841;anchor=swh:1:rev:1caf3be141c6f646f78695c0eb528ce3b753079a;path=/modal.ml;lines=91-124}{\ExternalLink}]\label{def:subformula}\phantom{Allyounedisle}
    \begin{mydefinition}
        \begin{mydefinition2}
            The predicate of subformula is defined by induction on complexity:
            \begin{itemize}
                \item $A \  \mathbf{sub}  A$;
                \item if $B \to C \ \mathbf{sub}  A \ $ then  $ \ B \ \mathbf{sub} A$ and $C \ \mathbf{sub} A$;
                \item  if $\Box B \ \mathbf{sub} A$ then $B \ \mathbf{sub} A$.
            \end{itemize}
        \end{mydefinition2}
    \end{mydefinition}
\end{definition}

As anticipated, this definition by induction is brief and it guarantees the usual results on subformulas, e.g. \textit{if $\neg B \ \mathbf{sub} A$ then $B \ \mathbf{sub} A$.}

The following fact, discussed in classical textbook such as~\cite[\S 5]{boolos1995logic}, states the intuitive idea  that there is a finite number of subformulas of a given formula.

\begin{fact}[\textbf{Finite number of formulas}~\href{https://archive.softwareheritage.org/swh:1:cnt:4a011f69b1180019be5d03c6d5f1cec4550055e8;origin=https://github.com/HOLMS-lib/HOLMS;visit=swh:1:snp:2c0efd349323ed6f8067581cf1f6d95816e49841;anchor=swh:1:rev:1caf3be141c6f646f78695c0eb528ce3b753079a;path=/modal.ml;lines=138-148}{\ExternalLink},~\href{https://archive.softwareheritage.org/swh:1:cnt:4a011f69b1180019be5d03c6d5f1cec4550055e8;origin=https://github.com/HOLMS-lib/HOLMS;visit=swh:1:snp:2c0efd349323ed6f8067581cf1f6d95816e49841;anchor=swh:1:rev:1caf3be141c6f646f78695c0eb528ce3b753079a;path=/modal.ml;lines=150-157}{\ExternalLink}]\label{fct:finite_formulas}\phantom{Allyouneedislovelove}
    \begin{mydefinition}
        Let $A$ be a modal formula. 
        Then there are only \textbf{finitely} many \textbf{subformula} of $A$.
        \\ Therefore there are only finitely many \textbf{formulas made of subformula} of $A$, and therefore only finitely  many \textbf{sets of formulas made of subformulas} of $A$.
    \end{mydefinition}
\end{fact}

\newpage

\begin{definition}[\textbf{Occurrence of propositional atoms}]\phantom{Allyouneed}
    \begin{mydefinition}
        \begin{mydefinition2}
            A propositional atom $p$ \textbf{occurs} in a modal formula $A$ $(p \in A)$ iff
            $p \ \mathbf{sub} A$.
        \end{mydefinition2}
    \end{mydefinition}
\end{definition}

\begin{definition}[\textbf{Subsentence}~\href{https://archive.softwareheritage.org/swh:1:cnt:e261ca08330d5ef66376347407c27fcece1e9f2e;origin=https://github.com/HOLMS-lib/HOLMS;visit=swh:1:snp:2c0efd349323ed6f8067581cf1f6d95816e49841;anchor=swh:1:rev:1caf3be141c6f646f78695c0eb528ce3b753079a;path=/consistent.ml;lines=141-158}{\ExternalLink}]\phantom{Allyouneedislovelove}
    \begin{mydefinition}
        \begin{mydefinition2}
            A modal formula $B$ is a subsentence of a modal formula A $(B \ \mathbf{subs} \ A)$ iff
            $B \ \mathbf{sub} A$ or $ \neg B \ \mathbf{sub} A$.
        \end{mydefinition2}
    \end{mydefinition}
\end{definition}

\section{Modal Logic Semantics}\label{sec:informal-semnatics}

Here we introduce \textit{relational semantics} à la Kripke, which, as anticipated, formalises the concept of truth of a modal formula and permits us to appreciate the various interpretations of the `box' in different modal systems. 

\begin{definition}[\textbf{Kripke's frame}~\href{https://archive.softwareheritage.org/swh:1:cnt:d364de3158b6405c920115aa0c801af9af49e302;origin=https://github.com/HOLMS-lib/HOLMS;visit=swh:1:snp:2c0efd349323ed6f8067581cf1f6d95816e49841;anchor=swh:1:rev:1caf3be141c6f646f78695c0eb528ce3b753079a;path=/gen_completeness.ml;lines=15-16}{\ExternalLink}]\phantom{Allyouneedislovelove}
    \begin{mydefinition}
        \begin{mydefinition2}
        A \textbf{frame} $ \mathcal{F}$ is an ordered pair $\mathcal{F} = \langle W,R \rangle$ where:
        \begin{itemize}
            \item $W$ is a nonempty set of ``possible worlds'', called \textbf{universe}
            \item $R$ is a binary relation over worlds $R \subseteq W \times W$, called \textbf{accessibility relation}. We write $w$ \textbf{sees} $v$ iff $wRv$.
        \end{itemize}
        \end{mydefinition2}
    \end{mydefinition}
\end{definition}

\begin{convention}[\textbf{Metavariables for possible worlds}]
    We will use $w, v, t, u...$ as metavariables on ”possible worlds”.
\end{convention}

\begin{convention}[\textbf{Properties and Finiteness of Frames}]
 A frame $\mathcal{F}$ is said to have some properties of binary relation iff its accessibility relation has those properties.
A frame $\mathcal{F}$ is said to be finite iff its universe is finite.
\end{convention}

Due to the importance of the \textit{accessibility relation} in characterising different classes of frames and--as we will see--modal systems, we report some important properties of binary relations in the Appedix~\ref{app:binary_relations}. 

\begin{definition}[\textbf{Valuation}]\phantom{Allyouneedisloveloveisallyouneed}
    \begin{mydefinition}
        \begin{mydefinition2}
        A \textbf{valuation} $ V$ on $W$ is a binary relation  $V \subseteq \Phi \times W$. \\ We write that $w$ \textbf{verifies} $p$ iff $pVw$.

        \end{mydefinition2}
    \end{mydefinition}
\end{definition}

\begin{remark}
    Sometimes valuations are presented as functions $V:  \Phi \to \mathcal{P}(W)$ that, intuitively, associates each propositional atom with all the worlds in which this atom is true. Here, similarly, we will provide a definition of forcing relation~(\ref{def:truth}) such that a world \textit{verifies} a modal atom iff it is true in that world.
\end{remark}

Starting from the concepts of \textit{relational frame} and \textit{valuation}, we can define a fundamental notion in a model-theoretic perspective:
\newpage
\begin{definition}[\textbf{Kripke's Model}]\phantom{Allyouneedislove}
    \begin{mydefinition}
        \begin{mydefinition2}
          A \textbf{model} is a triple $\langle W,R, V \rangle$, where: 
          \begin{itemize}
              \item $\langle W,R \rangle$ is a frame;
              \item $V$ is a valutation on $W$.
          \end{itemize}
          $\langle W, R, V \rangle$ is said to be \textbf{based} on the frame $\langle W, R \rangle$ and we write $\langle W, R \rangle_*$.
        \end{mydefinition2}
    \end{mydefinition}
\end{definition}

\begin{convention}[\textbf{Denoting Frames and Models}]
            Unless there are clear indications of what $\mathcal{F}$ denotes, $\mathcal{F} \coloneq \langle W, R \rangle$. 
            Unless there are clear indications of what $\mathcal{M}$ denotes, $\mathcal{M} \coloneq \langle W, R, V \rangle$.
\end{convention}


We introduce \textit{forcing relation}, which computes the truth of modal formulas in a given world of a model. This relation guarantees that \textit{a necessary sentence is true in a certain world} iff \textit{it is true in every possible world accessible from it}.

\begin{definition}[\textbf{Truth in a world of a model}~\href{https://archive.softwareheritage.org/swh:1:cnt:4a011f69b1180019be5d03c6d5f1cec4550055e8;origin=https://github.com/HOLMS-lib/HOLMS;visit=swh:1:snp:2c0efd349323ed6f8067581cf1f6d95816e49841;anchor=swh:1:rev:1caf3be141c6f646f78695c0eb528ce3b753079a;path=/modal.ml;lines=44-59}{\ExternalLink}]\phantom{Allyouneedislove}\label{def:truth}
    \begin{mydefinition}
        \begin{mydefinition2}
         The \textbf{truth of a modal formula $A$ in a world $w$ of a model $\mathcal{M}$} \\ $\mathcal{M}, w \vDash A$ is inductively defined on the structure of $A$ as follows:
         \begin{itemize}
             \item if $A = p_i \in \Phi$ then  $\mathcal{M}, w \vDash p_i \iff p_i V w$;
             \item if $A = \bot$ then  $\mathcal{M}, w \not \vDash \bot$ for any $w, \mathcal{M}$;
             \item if $A= B \to C$ then  $\mathcal{M}, w \vDash B \to C \iff \mathcal{M}, w \not \vDash B \ or \ \mathcal{M}, w \vDash C$;
             \item if $A=\Box B$ then $\mathcal{M}, w \vDash \Box B \iff \forall x \in W(wRx \implies \mathcal{M}, x \vDash B)$.
             \end{itemize}          
             In words we also say that ``$A$ is forced by $w$ in $\mathcal{M}$'' whenever $\mathcal{M}, w \vDash A$.
        \end{mydefinition2}
    \end{mydefinition}
\end{definition}  

From this definition of the forcing relation and Convention~\ref{cnv:metaling_abbr}, some lemmas defining forcing of modal formulas with defined connectives immediately follow: 

\begin{lemma}[\textbf{Forcing relation and defined connectives}]\label{lem:forc_conn}\phantom{A}
\begin{mydefinition}
 \begin{enumerate}
             \item $\mathcal{M},w \vDash \neg B \iff \mathcal{M},w \not \vDash B$;
             \item $\mathcal{M},w \vDash B \land C \iff \mathcal{M},w \vDash B \ and \ \mathcal{M},w \vDash C$;
             \item $\mathcal{M},w \vDash B \lor C \iff \mathcal{M},w \vDash B \ or \ \mathcal{M},w \vDash C$;
             \item $\mathcal{M},w \vDash B \leftrightarrow C \iff \mathcal{M},w \vDash B $ \textit{iff} $ \mathcal{M},w \vDash C$;
             \item $\mathcal{M},w \vDash \Diamond B \iff \exists x \in W(wRx \land \mathcal{M},x \vDash B)$
             \end{enumerate}      
\end{mydefinition}
\end{lemma}

In addition to forcing, Kripke semantics develops four notions to investigate the validity of modal formulas, ordered in increasing level of generality: 
\newpage
\begin{definition}[\textbf{Validity in a model}] \phantom{Allyoune}
    \begin{mydefinition}
        \begin{mydefinition2}
        A modal formula $A$ is \textbf{valid in a model} $\mathcal{M}$ iff it is true in all worlds of $\mathcal{M}$.
        \[\mathcal{M} \vDash A \ \ \text{   iff   } \ \ \forall w \in W(\mathcal{M},w \vDash A )\]
        \end{mydefinition2}
    \end{mydefinition}
\end{definition}  

\begin{definition}[\textbf{Validity in a frame}~\href{https://archive.softwareheritage.org/swh:1:cnt:4a011f69b1180019be5d03c6d5f1cec4550055e8;origin=https://github.com/HOLMS-lib/HOLMS;visit=swh:1:snp:2c0efd349323ed6f8067581cf1f6d95816e49841;anchor=swh:1:rev:1caf3be141c6f646f78695c0eb528ce3b753079a;path=/modal.ml;lines=61-62}{\ExternalLink}]\phantom{Allyouneed}
    \begin{mydefinition}
        \begin{mydefinition2}
         A modal formula $A$ is \textbf{valid in a frame} $\mathcal{F}$ iff it is valid in all models based on $\mathcal{F}$.
        \[\mathcal{F} \vDash A \ \ \text{   iff   } \ \ \forall \mathcal{M} \ \mathcal{F}_*  (\mathcal{M} \vDash A )\]
        \end{mydefinition2}
    \end{mydefinition}
\end{definition}  

\begin{definition}[\textbf{Validity in a class of frames}~\href{https://archive.softwareheritage.org/swh:1:cnt:4a011f69b1180019be5d03c6d5f1cec4550055e8;origin=https://github.com/HOLMS-lib/HOLMS;visit=swh:1:snp:2c0efd349323ed6f8067581cf1f6d95816e49841;anchor=swh:1:rev:1caf3be141c6f646f78695c0eb528ce3b753079a;path=/modal.ml;lines=66-67}{\ExternalLink}]\phantom{Allyouneed}
    \begin{mydefinition}
        \begin{mydefinition2}
        A formula modal $A$ is \textbf{valid in a class of frames }$\mathfrak{S}$ iff it is valid in any frame in $\mathfrak{S}$.
        \[\mathfrak{S} \vDash A \ \ \text{   iff   } \ \ \forall \mathcal{F} \in \mathfrak{S}  ( \mathcal{F} \vDash A) )\]
        \end{mydefinition2}
    \end{mydefinition}
\end{definition}  

\begin{definition}[\textbf{Validity tout-court or Modal tautology}]\phantom{All}
    \begin{mydefinition}
        \begin{mydefinition2}
        A modal formula $A$ is \textbf{valid tout-court} (is a \textbf{modal tautology}) iff it is valid in the class of all frames.
        \[ \vDash A \ \ \text{   iff   } \ \ \forall \mathcal{F} ( \mathcal{F} \vDash A) )\]
        \end{mydefinition2}
    \end{mydefinition}
\end{definition}  

At this stage, all the tools required to treat the model-theoretic relation of \textit{logical consequence} are now at our disposal.

\begin{definition}[\textbf{Logical Consequence}]\phantom{Allyouneedislovelove}
    \begin{mydefinition}
        \begin{mydefinition2}
        Let $\mathfrak{S}$ be a class of frames. Let $A$ a modal formula and $\mathcal{H}$ a set of formulae. $A$ is a \textbf{logical consequence} of $\mathcal{H}$ w.r.t. $\mathfrak{S}$ iff for every frame $\mathcal{F}$ in $\mathfrak{S}$ and for every world $w$ of a model based on $\mathcal{F}$, if all the formulae in $\mathcal{H}$ are forced by $w$ then $A$ is forced by $w$. 
        \begin{center}
           $ \mathcal{H} \vDash_{\mathfrak{S}} A$ \\ \text{   iff   } \\ $\forall \mathcal{F} \in \mathfrak{S} (\forall \mathcal{M} \ \mathcal{F}_*( \forall w \in W((\forall B \in \mathcal{H} ( \mathcal{M},w \vDash B)) \implies \mathcal{M},w \vDash A)))$
        \end{center}      
        \end{mydefinition2}
    \end{mydefinition}
\end{definition}  

\section{Deducibility in Modal Logics}\label{sec:axiomatic-calulus}

Parallel to this semantic relation, we want to implement a proof-theoretic relation of \textit{derivability} to study implications between modal sentences from a purely syntactical point of view. There are many paths one can follow to formalise a \textit{modal deductive system}, like \textit{ sequent calculi}, \textit{natural deduction} or \textit{axiomatic calculi}.
Nevertheless--as two authors of HOLMS pointed out in a previous work on GL~\cite{DBLP:journals/jar/MaggesiB23}--``when dealing with the very notion of tautology--or \textit{theoremhood}, discarding the complexity or structural aspects of \textit{derivability} in a formal system--it is convenient to focus on axiomatic calculi''.
\medskip

\begin{remark}
    From this point on, we use several different fonts to distinguish between \textit{logics}, \textit{axiomatic calculi}, \textit{classes of frames}, \textit{axioms}, and \textit{sets of formulas}. In Table~\ref{tab:modal_logics} below we clarify in which fonts we type each of them. Together with the typefaces, we report three examples: the parameter S and the two letters K and GL. In particular, for these two last examples we list logics $\mathbb{K}$ and $\mathbb{GL}$, their calculi, their appropriate\footnote{This is a fundamental concept from the next section, defined in~\ref{def:appropriate_class},} classes of frames, their specific axiom schemata and the sets of their specific axioms. 
\end{remark}
\medskip \medskip

\begin{table}[hb]
    \centering
    \begin{tabular}{l|c|c|c|c}
        \toprule
        & \textbf{Font} & \textbf{S} & \textbf{K} & \textbf{GL} \\
        \midrule
        \textbf{Logic} & ``double-struck'' & \( \mathbb{S} \) & \( \mathbb{K} \) & \( \mathbb{GL} \) \\
        \textbf{Calculus} & ``sans-serif'' & \( \mathsf{S} \) & \( \mathsf{K} \) & \( \mathsf{GL} \) \\
        \textbf{Class of frames} & ``gothic'' & \( \mathfrak{S} \) & \( \mathfrak{F} \) & \( \mathfrak{ITF} \) \\
        \textbf{Axiom Schema} & ``math bold-face'' & $\mathbf{S}$ & $\mathbf{K}$ & $\mathbf{GL}$ \\
        \textbf{Set of Formulas} & ``calligraphic''  & \( \mathcal{S} \) & \( \mathcal{K} \) & \( \mathcal{GL} \) \\
        \bottomrule
    \end{tabular}
    \caption{Fonts employed to distinguish different logical notions}
    \label{tab:modal_logics}
\end{table}

\subsubsection{Deductive Systems}\label{sub:deductive system}
Before introducing axiomatic calculi for modal logics, we briefly present the fundamental concept of \textit{deductive system}. Abstractly, a deductive system consists of a set of \textit{initial formal expressions} and a set of \textit{inference rules}, whose purpose is to establish proofs of valid expressions with respect to (w.r.t) a given logic.

A \textit{proof} or \textit{derivation} in such a system is obtained by applying inference rules to the initial expressions and iteratively applying them to the conclusions derived in each step.
A \textit{theorem} or a \textit{lemma} in this context is a formal expression obtained after a finite sequence of proof steps following the procedure described above.

Both \textbf{axiomatic calculi} for modal logics--such as those presented below--and \textbf{labelled sequent calculi} for modal logics--introduced in Section~\ref{sec:labelled-sequent-calculi}--are examples of deductive systems for modal logics. In the first case, the initial set of formal expressions is composed of \textit{axiom schemata}, whereas in the second case, of \textit{initial sequents}.
Axiomatic and sequent calculi also differ in another fundamental aspect: the former are \textit{synthetic}, while the latter \textit{analytic}.

\begin{itemize}
    \item \textit{Synthetic Calculi} 
    \begin{itemize}
     \item Synthetic proofs proceed \textit{top-down}, i.e., from axioms to the formula to be proved;
     \item Conclusions of synthetic rules may lose information from their premises;
     \item Proof search is not guided by the structure of the formula to be proved: proving a theorem requires guessing both the correct instances of axiom schemata and the proper order of inference rule applications.
    \end{itemize}
          
      \item \textit{Analytic Calculi} 
      \begin{itemize}
      \item Analytic proofs proceed \textit{bottom-up}; 
      \item An analytic rule does not lose information: e.g. all formulas in its premises are subformulas of those in the conclusion;
    \item Proof search is guided by the structure of the formula to be proved.  
    \end{itemize}

\end{itemize}

\subsection{Axiomatic Calculi for Modal 
Logics}

Given the alphabet $\mathcal{L}_{\Box}$ and the grammar for modal formulas $\mathbf{Form}_{\Box}$, an \textbf{axiomatic calculus for modal systems} is characterized by:
\begin{itemize}
    \item  A set of modal sentences (\textit{\textbf{Axioms Set}});
    \item A  set of relations under which the axioms are \textit{closed} (\textit{\textbf{Inference Rules}}).
    \begin{mydefinition}
    \begin{mydefinition2}
    A set of modal sentences $\mathcal{S}$ is \textit{\textbf{closed under a rule of inference}} iff \\ it contains all formulas deducible by that rule from the members of $\mathcal{S}$.
    \end{mydefinition2}
\end{mydefinition}
\end{itemize}

Therefore, an axiomatic calculus is a set of \textit{axioms} and \textit{inference rules} that allows to formally define a modal system (logic) and to determine which formulas can be \textit{formally proved} in this deductive system.

\begin{definition}[\textbf{Formal proof}]\phantom{Allyouneedislove}
\begin{mydefinition}
    \begin{mydefinition2}
        A \textit{\textbf{proof}} of a modal sentence $A$ from a set of modal sentences $\mathcal{H}$ in an axiomatic calculus $\mathsf{S}$  is a finite sequence of modal formulas, such that:
        \begin{itemize}
            \item Terminates with $A$;
            \item Each formula in the sequence is:
            \begin{enumerate}
                \item  An \textit{axiom} of $\ \mathsf{S}$; 
                \item An \textit{hypothesis} in $\mathcal{H}$;
                \item Is deducible from earlier formulas \textit{by applying a rule} of $\ \mathsf{S}$.
            \end{enumerate}
        \end{itemize}
    \end{mydefinition2}
\end{mydefinition}
\end{definition}

\begin{definition}[\textbf{Theorem}]\phantom{Allyouneedislove}   
\begin{mydefinition}
    \begin{mydefinition2}
    A modal sentence $A$ is a \textit{\textbf{theorem}} in an axiomatic calculus $\mathsf{S}$ from a set of hypotheses $\mathcal{H}$ $ (\mathcal{H} \vdash_{\mathsf{S}} A)$ iff there is a proof of $A$ from $\mathcal{H}$ in $\mathsf{S}$ .
    \end{mydefinition2}
\end{mydefinition}
\end{definition}

\subsection{A Parametric Calculus for Normal Modal Systems}\label{sub:parametric-calculus}

More specifically, in this introduction to the current version of HOLMS, we focus on \textit{deductive normal systems of propositional modal logic}. Restricting our analysis on these logics has the advantage of not developing many separate axiomatic calculi for all the systems under consideration. Instead, we define a single \textit{parametric}\footnote{With the word \textit{parametric} we refer to results which remain fully independent from the concrete
instantiations of the parameters. This notion will be intrduced in details in Chapter~\ref{chap:3} and is used here to describe our results\textit{as general as possible}.} calculus based on a minimal axiomatic system. 
To understand why such an approach works, we need to clarify what \textit{deductive normal systems of propositional modal logic}--normal modal logics-- are.

\begin{definition}[\textbf{Normal Modal Logic}]\phantom{Allyouneedislove}
    \begin{mydefinition}
        \begin{mydefinition2}
        Let $\mathbb{S}$ be a modal logic. Let $\mathsf{S}$ be an axiomatic calculus that formalises $\mathbb{S}$. Let $\mathcal{S}$ the set of axiom of $\mathsf{S}$.
        $\mathbb{S}$ is a $\textbf{normal modal logic}$ iff:
        \end{mydefinition2}
            \end{mydefinition}
            
            \begin{mydefinition}
            \begin{mydefinition2}
            \begin{enumerate}
                \item $\mathbf{Axioms}:  \ \{\mathbf{taut}; \mathbf{K} \} \subseteq \mathcal{S}$, i.e. $\mathcal{S}$ contains:
                \begin{itemize}
                    \item $\mathbf{taut}$ \ \  complete set of  schemata for classical propositional logic;
                    \item $\mathbf{K}$ \ \ \ \ \  schema of distribution  $\mathbf{K} \coloneq\Box (A \to B) \to (\Box A \to \Box B)$.
                \end{itemize}      
                \item $\mathbf{Inference \ Rules}: \ \ \mathcal{S}$ is closed under:
                \begin{itemize}
                    \item \textbf{Modus \ Ponens} \ \ \ \  \ \ \ $\displaystyle \frac{A \to B \quad A}{B}  \texttt{MP}$
                    \item \textbf{Necessitation \ Rule}\ \ \ $\displaystyle \frac{A}{\Box A} \texttt{RN}$
                    \item \textbf{Substitution \ Rule}\ \ \ $\displaystyle \frac{F}{F_p(A)} \texttt{SUB}$
                \end{itemize}
            \end{enumerate}
        \end{mydefinition2}
    \end{mydefinition}
\end{definition}

\begin{definition}[\textbf{Substitution}]\label{def:substitution}\phantom{Allyouneedislove}
    \begin{mydefinition}
        \begin{mydefinition2}
        Let $F$ be a modal formula. $F_p(A)$, the \textbf{substitution of every occurrence of $p$ with $A$}, is inductively defined on the complexity of the modal formula $A$ as follows:
        \begin{itemize}
            \item if $F= p \in \Phi$ then $F_p(A)= A$;
            \item if $F= q \in \Phi \ and \ q \not =p$ then $F_p(A)= q$;
            \item if $F= \bot$ then $F_p(A)= \bot$;
            \item $(G \to H)_p(A)= (G_p(A) \to H_p(A))$;
            \item $(\Box G )_p(A)= \Box (G_p(A) \ )$;
        \end{itemize}
    \end{mydefinition2}
    \end{mydefinition}
\end{definition}

Per the foregoing definition of a normal system, it is natural to identify a minimal calulus $\mathsf{K}$, which is obtained from any classical propositional calculus by adding the distribution schema $\mathbf{K}$ and the necessitation rule \texttt{RN}. It serves as the minimal logical engine for any axiomatic calculus $\mathsf{S}$ for normal systems. 
Below we report the standard axiomatisation of $\mathsf{K}$~\href{https://archive.softwareheritage.org/swh:1:cnt:d42782345008434be8b43de7feb014732f8821ea;origin=https://github.com/HOLMS-lib/HOLMS;visit=swh:1:snp:2c0efd349323ed6f8067581cf1f6d95816e49841;anchor=swh:1:rev:1caf3be141c6f646f78695c0eb528ce3b753079a;path=/calculus.ml;lines=57-68}{\ExternalLink}:
\[\mathsf{K} = \langle \textnormal{AX}_{\mathsf{K}}= \mathbf{taut} \cup \{ \mathbf{K} \} ; \ \textnormal{RULES}_{\mathsf{K}}= \{\texttt{MP;SUB;RN}\} \rangle.\]

On the basis of this argument, we conceptualise derivability in an axiomatic calculus $\mathsf{S}$ ($\mathcal{H} \vdash_{\mathsf{S}} A$) as derivability in the \textit{minimal axiomatic system} $\mathsf{K}$ that is \textit{modularly extended} by additional axiom schemas in $\mathcal{S}$ ($\mathcal{S.H} \vdash A$). 
Note that such a formal predicate conceptualises an abstract notion of deducibility close to the one discussed in~\cite{fitting2013proof}.

\newpage

\begin{definition}[\textbf{Deducibility from hypotheses in normal systems}~\href{https://archive.softwareheritage.org/swh:1:cnt:d42782345008434be8b43de7feb014732f8821ea;origin=https://github.com/HOLMS-lib/HOLMS;visit=swh:1:snp:2c0efd349323ed6f8067581cf1f6d95816e49841;anchor=swh:1:rev:1caf3be141c6f646f78695c0eb528ce3b753079a;path=/calculus.ml;lines=74-80}{\ExternalLink}]\label{def:deducibility}\phantom{}
    \begin{mydefinition}
        \begin{mydefinition2}
           The ternary predicate $\mathcal{S}.\mathcal{H} \vdash A$, denoting the \textbf{deducibility of a formula} $A$ from a set of hypotheses $\mathcal{H}$ in an axiomatic extension of the logic $\mathbb{K}$ via schemas in the set $\mathcal{S}$, is inductively defined by the following conditions:
\begin{itemize}
    \item For every instance $A$ of axiom schemas for the calculus $\mathsf{K}$, $\mathcal{S}.\mathcal{H} \vdash A$;
    \item For every instance $A$ of schemas in $\mathcal{S}$, $\mathcal{S}.\mathcal{H} \vdash A$;
    \item For every $A \in \mathcal{H}$, $\mathcal{S}.\mathcal{H} \vdash A$;
    \item If $\mathcal{S}.\mathcal{H} \vdash B \rightarrow A$ and $\mathcal{S}.\mathcal{H} \vdash B$, then $\mathcal{S}.\mathcal{H} \vdash A$;
    \item If $\mathcal{S}.\varnothing \vdash A$, then $\mathcal{S}.\mathcal{H} \vdash \Box A$ for any set of formulas $\mathcal{H}$.
\end{itemize}
        \end{mydefinition2}
    \end{mydefinition}
\end{definition}

\begin{convention}
    From this point forward, any mention of a modal system will refer specifically to a deductive normal system of propositional modal logic.
\end{convention}

\subsubsection{Deduction Theorem}\label{sub:deduction-theorem}

Our notion of deducibility $\mathcal{S.H} \vdash A$ slightly differs from the derivability relation presented in classical textbooks $\mathcal{S} \vdash A$. Nevertheless, the two concepts might be assimilated by proving a deduction theorem. 
The proof below follows the same strategy of HOLMS one, and is loosely inspired by the analogous theorem in~\cite[\S 4]{fitting2013proof}.

\begin{theorem}[\textbf{Deduction theorem}~\href{https://archive.softwareheritage.org/swh:1:cnt:d42782345008434be8b43de7feb014732f8821ea;origin=https://github.com/HOLMS-lib/HOLMS;visit=swh:1:snp:2c0efd349323ed6f8067581cf1f6d95816e49841;anchor=swh:1:rev:1caf3be141c6f646f78695c0eb528ce3b753079a;path=/calculus.ml;lines=1015-1028}{\ExternalLink}]\label{lem:deduction-theorem}\phantom{Allyouneedsloveloveisallyouneed}
    \begin{mydefinition}
        For any modal formulas $A,B$ and any sets of formulas $\mathcal{S}, \mathcal{H}$, the following equivalence holds: 
$\mathcal{S}.\mathcal{H} \cup \{ B \} \vdash A \ $ \textnormal{iff}  $ \ \mathcal{S}.\mathcal{H} \vdash B \rightarrow A$.  
    \end{mydefinition}
\end{theorem}
\begin{proof}
   We shall prove the two sides of bi-implication:
   
\begin{itemize}
    \item[$\implies$] $[\textnormal{if } \ \mathcal{S}.\mathcal{H} \cup B \vdash A \ $ \textnormal{then}  $ \ \mathcal{S}.\mathcal{H} \vdash B \rightarrow A]$. \\ Taken a certain $\mathcal{S}$ and a generic $B$, we shall distinguish between two cases:
    \begin{enumerate}
        \item $B \in \mathcal{H}$ \\
        If $B \in \mathcal{H}$ then $\mathcal{H} = \mathcal{H} \cup B$ and so we have  $\mathcal{S}.\mathcal{H} \vdash A$. 
        The thesis follows by \texttt{MP} on $\mathcal{S.H} \vdash A$ and the instance $\mathcal{S.H} \vdash B \to A$ of Claim~\ref{clm:MLK_axiom_addimp}.
        \begin{claim}\label{clm:MLK_axiom_addimp}
        $\forall  \mathcal{S,H},C,D$ if $\mathcal{S}. \mathcal{H} \vdash C $ then $  \mathcal{S}. \mathcal{H} \vdash D \to C$    
        \begin{claimproof}
            Because $C \to (D \to C)$ is a propositional tautology and then an axiom of $\mathsf{K}$, then by the inductive definition of derivability follows $\mathcal{S.H} \vdash C \to (D \to C)$.
            Then the thesis follows by applying \texttt{MP} to this result and the hypothesis.  \hfill $\lhd$
        \end{claimproof}
        \end{claim}

        \item $B \not \in \mathcal{H}$ 
        \smallskip
        
        \begin{remark}
            $\mathcal{H} = (\mathcal{H} \cup B) \setminus B$ 
        \end{remark}
        \begin{claim}\label{clm:MODPROVES_DEDUCTION_DELETE}
        For any $\mathcal{S}, \mathcal{H}, A, B$ holds that \\
        if $\mathcal{S}. \mathcal{H} \vdash A$ and $B \in \mathcal{H} $ then $\mathcal{S}. \mathcal{H}\setminus B \ \vdash B \to A$.
        \begin{claimproof}
             We want to prove this claim for each A modal formula such that $\mathcal{S.H} \vdash A$. Consequently, we develop a proof by induction on the definition of derivability.
            \begin{enumerate}
                 \item If $A$ is an axiom of $\mathsf{K}$  \\ then, by definition of derivability, $A$ is provable in every system from each set of hypotheses and, in particular, $\mathcal{S.H}\setminus B \vdash A$.  
                Then, the thesis follows by \texttt{MP} and Claim~\ref{clm:MLK_axiom_addimp};
                \item  If $A$ is an axiom schema in $\mathcal{S}$ \\
                then, analogously, for every set of hypothesis $A$ is provable and $\mathcal{S.H}\setminus B \vdash A$.
                 Then, the thesis follows by \texttt{MP} and Claim~\ref{clm:MLK_axiom_addimp};
                  \item If $A$ is an hypothesis in $\mathcal{H}$ we shall distinguish two cases:
                  \begin{enumerate}
                      \item $A=B$ \\
                    Then $A \to A$ is a propositional tautology and then an axiom of $\mathsf{K}$. Then $\mathcal{S.H}\setminus B \vdash A \to A$.
                    \item $A \not = B$ \\
                    Then $A \in \mathcal{H}\setminus B$. \\
                    But then, by definition of derivability, $\mathcal{S.H}\setminus B \vdash A$. \\
                    Then, the thesis follows by \texttt{MP} and Claim~\ref{clm:MLK_axiom_addimp};
                  \end{enumerate}

                   \item If $\mathcal{S.H \vdash A}$ follows from $\mathcal{S.H \vdash C \to A}$ and $\mathcal{S.H \vdash C}$ for a certain C \\
                Then the induction hypothesis for $C \to A$ and $C$ hold: 
                \begin{itemize}
                    \item $\mathcal{S.H}\setminus B \vdash B \to (C \to A)$;
                    \item $\mathcal{S.H} \setminus B \vdash \ B \to C$. 
                \end{itemize} 
                By \texttt{MP} and axioms of $\mathsf{K}$, it follows: \\ $\mathcal{S.H}\setminus B \vdash( B \to (C \to A)) \to (B \to C)$. \\
                Furthermore observe that \\ $(D \to E \to F) \to (D \to E) \to (D \to F) \ $ is a propositional tautology and then an axiom of $\mathsf{K}$ such that $\mathcal{S.H}\setminus B \vdash (B \to C \to A) \to (B \to C) \to (B \to A)$.
                Then $\mathcal{S.H}\setminus B \vdash B \to A$ follows by \texttt{MP}.
                \item[(e)] If $A=\Box C$ and $\mathcal{S.H}\vdash A$ follows from $\mathcal{S.\varnothing \vdash}C$ \\
                Then the induction hypothesis for $C$ holds: \\ $\mathcal{S.\varnothing}\setminus B \vdash B \to C$ and $B \in \varnothing$.\\
                But for no $B$, $B \in \varnothing$ and then simply $\mathcal{S.\varnothing \vdash }C$. \\
                Then, by \texttt{RN}, $\mathcal{S}. \varnothing \vdash \Box C$ and then, by definition of derivability, $\mathcal{S. H} \vdash \Box C.$ 
             \end{enumerate} \hfill $\lhd$
        \end{claimproof}
        \end{claim}
    Then given $\mathcal{H}'=\mathcal{H} \cup B$, we have that:
    $\mathcal{S.H} \cup B \vdash A \ $ and $\ B \in (H \cup B)$. 
    Then it follows from Claim~\ref{clm:MODPROVES_DEDUCTION_DELETE}: \\ $\mathcal{S.(H} \cup B) \setminus B \ \vdash \ B \to A$ that is equal to $\mathcal{S. H} \vdash B \to A$.
    \end{enumerate}  
    
    \item[ $\impliedby $] $[$ if $ \mathcal{S}.\mathcal{H} \vdash B \rightarrow A \ $ \textnormal{then} $\ \mathcal{S}.\mathcal{H} \cup B \vdash A]$. \\ We fix some generic $\mathcal{S,H}, B$.
    
    \begin{remark}\label{rmk:ded_4}
    $\mathcal{S.H} \cup B \vdash B$ holds for the definition of derivability.
    \end{remark}

    \begin{claim}\label{clm:MODPROVES_MONO2}
    For any $\mathcal{S}, \mathcal{H, H}', A$ holds that \\
    if $\mathcal{S}. \mathcal{H} \vdash A$ and $\mathcal{H} \subseteq \mathcal{H}' $ then $\mathcal{S. H}' \vdash  A$.   
    \begin{claimproof}
    We want to prove the \textbf{monotonicity} of derivability relation for each A modal formula such that $\mathcal{S.H} \vdash A$. Consequently, we develop a proof by induction on the definition of derivability.
    \begin{enumerate}
        \item If $A$ is an axiom of $\mathsf{K}$ \\ then, by definition of derivability, $A$ is provable in every system from each set of hypotheses and, in particular, $\mathcal{S.H}'\vdash A$;
        
         \item If $A$ is an axiom schema in $\mathcal{S}$ \\
        then, analogously, for every set of hypothesis $A$ is provable and $\mathcal{S.H}' \vdash A$;
        
         \item If $A$ is an hypothesis in $\mathcal{H}$ \\then, because $\mathcal{H} \subseteq \mathcal{H}'$, $A \in \mathcal{H}'$ and then, by derivability relation,  $\mathcal{S.H}' \vdash A$;

         \item If $\mathcal{S.H \vdash}A$ follows from $\mathcal{S.H \vdash C \to}A$ and $\mathcal{S.H \vdash C}$ for a certain C\\
                then the induction hypothesis for $C \to A$ and $C$ hold: 
                \begin{itemize}
                    \item $\mathcal{S.H'}\vdash C \to A$;
                    \item $\mathcal{S.H'} \vdash C$. 
                \end{itemize} 
        Then, by \texttt{MP}, $\mathcal{S.H'}\vdash A;$

        \item If $A=\Box C$ and $\mathcal{S.H \vdash} \Box C$ follows from $\mathcal{S.\varnothing \vdash}C$ \\
                then by necessitation rule \texttt{RN} $\mathcal{S.H}' \vdash \Box C$.
    \end{enumerate}\hfill $\lhd$
    \end{claimproof}
\end{claim}

Given the hypothesis $\mathcal{S.H} \vdash B \to A$ and $\mathcal{H} \subseteq \mathcal{H} \cup B$, we derive by monotonicity (Claim~\ref{clm:MODPROVES_MONO2}) that $\mathcal{S.H}\cup B \vdash B \to A$.

\end{itemize}
Then by applying \texttt{MP} on this result and on Remark~\ref{rmk:ded_4},  we obtain the thesis $\mathcal{S.H}\cup B \vdash A$. \\ \phantom{A} 
\end{proof}

\smallskip

\begin{remark}
  Thanks to the \textit{deduction theorem}, we can ignore the set of hypotheses by working on $\mathcal{S}. \varnothing \vdash \bigwedge \mathcal{H} \to A$ instead than on $\mathcal{S}. \mathcal{H} \vdash A$. To lighten the script, we will use this trick and we will write $\mathcal{S} \vdash A$ instead of $\mathcal{S}. \varnothing \vdash A$.  
\end{remark}

\begin{convention}
    $\mathcal{S} \vdash A \coloneq \mathcal{S}. \varnothing \vdash A$
\end{convention}

\subsubsection{Soundness of K w.r.t. Relational Semantics}

At this stage, we both have at our disposal the formal notion of \textit{logical consequence} and the one of \textit{syntactical derivability in normal systems}. 

Given $\mathbb{K}$ the minimal normal system with no other axiom than $\mathbf{taut}$ and $\mathbf{K}$, we can observe an important relation between \textit{derivability in} $\mathbb{K}$  and \textit{logical consequence}.

\begin{remark}
    Given our definition of deducibility in a normal system, we will write $\varnothing \vdash A$ to describe deducibility in the system $\mathbb{K}$. In fact, any axiomatic calculus for $\mathbb{K}$ has no additional axiom to $\mathsf{K}$. 
\end{remark}
\newpage
\begin{theorem}[\textbf{Soundness of $\mathbf{K}$ with respect to all frames}]\phantom{}\label{thm:K_soundness}
  \begin{mydefinition}
    Given $A$ a modal formula, \\
    $\varnothing \vdash A \implies \vDash A$ 
  \end{mydefinition}
\end{theorem}
  \begin{proof}   
  If something is a theorem of $\mathbb{K}$ ($\varnothing \vdash A$), it is either an axiom of $\mathsf{K}$ or it deduced by an inference rule of $\mathsf{K}$:
  \begin{enumerate}
      \item All the axioms of $\mathsf{K}$ are $\mathbf{modal\  tautologies}$:
      \begin{itemize}
          \item $\mathbf{taut}$: complete set of schemata for propositional tautologies \\
          Because the definition of $\mathbf{truth \ in \ a \ world \ of \ a \ model}$ for propositional (without the `$\Box$') formulas is the same as in classical propositional semantics, $\mathbf{taut}$ are also $\mathbf{modal\  tautologies}$.
          
          \item $\mathbf{K}$: schema of all distribution axioms \\
          RAA: $\exists A, B \subseteq \mathbf{Form}_{\Box} \  \not \vDash \Box ( A \to B) \to (\Box A \to \Box B) $ \\ 
          that is to say, following classical logic Kripke's semantics: \\
          $\exists \mathcal{M} \exists w \in W ( \mathcal{M},w \vDash \Box ( A \to B) \ and \ \mathcal{M},w \vDash \Box A \ and $ \underline{ $ \mathcal{M},w \not \vDash \Box B$}$)$ \\
          But, then it follows: 
          \begin{itemize}
              \item $\mathcal{M},w \vDash \Box ( A \to B) $ implies $\forall x \in W(wRx \implies \mathcal{M}, x \vDash A \to B )$
              \item $\mathcal{M},w \vDash \Box A$ implies  $\forall x \in W(wRx \implies \mathcal{M}, x \vDash A )$
          \end{itemize}
          Then, from the inductive step for `$\to$' of forcing relation, follows: \\ $\forall x \in W(wRx \implies \mathcal{M}, x \vdash B)$ then \underline{${\mathcal{M}, w \vdash \Box B}$}. A contradiction.  $\lightning$ \\
          Then all distribution axioms are $\mathbf{modal\  tautologies}$.
      \end{itemize}
      
      \item All the inference rules of $\mathsf{K}$ preserve $\mathbf{modal\  tautologies}$:
      \begin{itemize}
          \item \texttt{MP} preserves $\mathbf{truth \ in \ a \ world \ of \ a \ model}$ and then $\mathbf{modal \ tautologies}$ \\
          if $\mathcal{M},w \vDash A \to B$ and $\mathcal{M},w \vDash A$ \\
          then because of the inductive definition of $\mathbf{truth \ in \ a \ world \ of \ a \ model}$
          \rule{13,5cm}{0.2pt} 
          $\mathcal{M},w \vDash B$
          
          \item \texttt{RN} preserves $\mathbf{validity \ in \ a \ model}$ and then $\mathbf{modal \ tautologies}$\\
          if $\mathcal{M}\vDash A$ \\
          then due to the definition of $\mathbf{validity \ in \ a \ model}$  $\forall w \in W (\mathcal{M},w \vDash A)$\\ and also $\forall w \in W (\forall x \in W (wRx \implies\mathcal{M},x \vDash A))$
          \\ then $\forall w \in W (\mathcal{M},w \vDash \Box A)$ \\
          \rule{13,5cm}{0.2pt} 
          $\mathcal{M} \vDash \Box A$
          \\ \\
          Observe that \texttt{RN} does not preserve $\mathbf{truth \ in \ a \ world \ of \ a \ model}$ \\
          We can take as counterexample a model $\langle\{ w;x\}; \{ (w,x)\}\rangle$ with a valuation $V= \{(w;p) \}$ that is represented in the following graph:

          \begin{center}
          \begin{tikzpicture}
          \node (w) at (0,0) [circle,fill,inner sep=1.5pt,label=below:$w$] {};
          \node (x) at (4,0) [circle,fill,inner sep=1.5pt,label=below:$x$] {};
          \node at (-0.5,0.5) {\small$\mathcal{M}, w \vDash p$};
          \draw[->] (w) -- node[above] {$R$} (x);
          \end{tikzpicture}
          \end{center}
          In such a model, we have: $\mathcal{M},w \vDash p$ but $\mathcal{M},x \not \vDash p$ thus $\mathcal{M},w \not \vDash \Box p$. 

          \item \label{itm:SUB} \texttt{RN} preserves $\mathbf{validity \ in \ a \ frame}$ and then $\mathbf{modal \ tautologies}$\\
          if $\mathcal{F}\vDash F$ \\
          then for an arbitrary valuation $V$ it holds $\langle W,R,V \rangle \vDash F$.\\
          I consider a valuation $V^*$ on $W$ such that:
          \[V^* =
          \begin{cases} 
          wV^*p & \text{iff } \mathcal{M},w \vDash A \\
          wV^*q & \text{iff } wVq \ with \ q \not = p 
          \end{cases}
          \]
          Then we will define $\mathcal{M}^*=\langle W,R,V^* \rangle$ and then we prove by induction the complexity of $F$ that: 
          \begin{claim}       
          $\forall F \in \mathbf{Form}_{\Box} (\mathcal{M}^*,w \vDash F \iff \mathcal{M},w \vDash F_p(A))$
          
          \begin{claimproof}
              \begin{itemize}
                  \item $F= \bot$ \\
                  $F_p(A)= \bot $ and 
                  $\mathcal{M}^*,w \vDash \bot \iff \mathcal{M},w \vDash \bot$         
                  
                  \item $F \in \Phi$ 
                  \begin{itemize}
                      \item if $F=p$  \\ then $ F_p(A)=A$  and $\mathcal{M}^*,w \vDash A \iff \mathcal{M},w \vDash A$; 
                      \item if $F=q \not = p$ \\ then $ F_p(A)=q$  and $\mathcal{M}^*,w \vDash q \iff wV^*q \iff wVq\iff \mathcal{M},w \vDash q$;
                  \end{itemize}
                                   
                  \item $F= B \to C$ \\
                  $F_p(A)= B_p(A) \to C_p(A) $ and \\
                  $\mathcal{M}^*,w \vDash B \to C \iff \mathcal{M}^*,w \not \vDash B  \ or \ \mathcal{M}^*,w \vDash C$ $\xLeftrightarrow{\text{ind hp}}$ $\mathcal{M},w \not \vDash B_p(A) \ or \ \mathcal{M},w \vDash C_p(A)$ $  \iff \mathcal{M},w \vDash B_p(A) \to C_p(A)$;
                  \end{itemize}
              \begin{itemize}                
                   \item $F= \Box B$ \\
                  $F_p(A)= \Box B_p(A)$ and \\
                  $\mathcal{M}^*,w \vDash \Box B \iff \forall x \in W (\mathcal{M}^*,x \vDash B)$ $\xLeftrightarrow{\text{ind hp}}$ $\forall x \in W (\mathcal{M},x \vDash B_p(A)$ $ \iff \mathcal{M},w \vDash \Box B_p(A)$,
              \end{itemize}\hfill $\lhd$
            \end{claimproof}
            \end{claim}
          Moreover, from $\mathcal{F}\vDash F$ follows $\mathcal{M}^*\vDash F$ and then for every $w \in W$ holds the lemma and $\forall w \in W(\mathcal{M},w \vDash F_p(A))$.       
          
          \rule{13,5cm}{0.2pt} 
          $\mathcal{F} \vDash F_p(A)$
          \\ \\
          Observe that \texttt{SUB} does not preserve $\mathbf{validity \ in \ a \ model}$ \\
          We can take as counterexample a model such that $wVp$ and $not \ wVq$ then  $\mathcal{M} \vDash p$ but $\mathcal{M} \not \vDash p_p(q)=q$.
      \end{itemize}
  \end{enumerate}
\end{proof}

\subsubsection{Examples of Normal Tautologies}
To show some examples of modal tautologies, we prove here two lemmas that hold for each normal system and which will be useful in the following sections. Let us call them with their HOLMS names.

\newpage

\begin{theorem}[\texttt{MLK\textunderscore imp\textunderscore box}~\href{https://archive.softwareheritage.org/swh:1:cnt:d42782345008434be8b43de7feb014732f8821ea;origin=https://github.com/HOLMS-lib/HOLMS;visit=swh:1:snp:2c0efd349323ed6f8067581cf1f6d95816e49841;anchor=swh:1:rev:1caf3be141c6f646f78695c0eb528ce3b753079a;path=/calculus.ml;lines=384-387}{\ExternalLink}]\label{lem:MLK_imp_box}\phantom{A}
    \begin{mydefinition}
        Let $\mathcal{S}$ a set of axiom schemata.
        $\mathcal{S} \vdash A \to B \ \ \implies \ \ \mathcal{S} \vdash \Box A \to \Box B$
    \end{mydefinition}
\end{theorem}
\begin{proof} 
    \begin{itemize}
        \item $\mathcal{S} \vdash A \to B \ \ \overset{\texttt{RN}}{\implies} \mathcal{S} \vdash \Box(A \to B)$ 
        \item  $\overset{\mathbf{K}}{\implies} \mathcal{S} \vdash \Box(A \to B) \to (\Box A \to \Box B)$
    \end{itemize}
        Then, by \texttt{MP}, it follows: $\mathcal{S} \vdash \Box A \to \Box B$
\end{proof}

\begin{lemma}[\texttt{MLK\textunderscore box\textunderscore and}~\href{https://archive.softwareheritage.org/swh:1:cnt:d42782345008434be8b43de7feb014732f8821ea;origin=https://github.com/HOLMS-lib/HOLMS;visit=swh:1:snp:2c0efd349323ed6f8067581cf1f6d95816e49841;anchor=swh:1:rev:1caf3be141c6f646f78695c0eb528ce3b753079a;path=/calculus.ml;lines=393-396}{\ExternalLink} and \texttt{MLK\textunderscore box\textunderscore and\textunderscore inv}~\href{https://archive.softwareheritage.org/swh:1:cnt:d42782345008434be8b43de7feb014732f8821ea;origin=https://github.com/HOLMS-lib/HOLMS;visit=swh:1:snp:2c0efd349323ed6f8067581cf1f6d95816e49841;anchor=swh:1:rev:1caf3be141c6f646f78695c0eb528ce3b753079a;path=/calculus.ml;lines=398-401}{\ExternalLink}]\label{lem:MLK_box}\phantom{A}
    \begin{mydefinition}
        Let $\mathcal{S}$ be a set of axiom schemata. 
        $\mathcal{S} \vdash \Box (A \land B) \leftrightarrow (\Box A \land \Box B)$
    \end{mydefinition}
\end{lemma}
\begin{proof} 
Observe that using $\mathbf{taut}$ and \texttt{MP} we can prove: \\
    (0) $S \vdash A \land B \iff S \vdash A \ and \ S \vdash B $ \\
    Then it is necessary to prove the two verses of the implication:
    \begin{itemize}
      \item[\(\to\)]  $\mathcal{S} \vdash \Box (A \land B)$\\
      Moreover, we know that: \\
      (1) $\overset{\text{taut}}{\implies} \mathcal{S} \vdash (A \land B) \to A$  $\overset{\texttt{MLK-imp-box}}{\implies}$ 
      $\vdash \Box (A \land B) \to \Box A$ \\
      (2) $\overset{\text{taut}}{\implies} \mathcal{S} \vdash (A \land B) \to B$  $\overset{\texttt{MLK-imp-box}}{\implies}$ 
      $\vdash \Box (A \land B) \to \Box B$ \\
      $\overset{\text{MP(1,HP)}}{\implies} \mathcal{S} \vdash \Box A$ and also $\overset{\text{MP(2,HP)}}{\implies} \mathcal{S} \vdash \Box B$ then the thesis follows by (0)
      
      \item[\(\leftarrow\)] $S \vdash \Box A \land \Box B$ \\
      (3)  $\overset{\text{taut}}{\implies} \mathcal{S} \vdash (\Box A \land \Box B) \to \Box A $ $\overset{\text{MP(3,HP)}}{\implies} \mathcal{S} \vdash \Box A $ \\
      (4) $\overset{\text{taut}}{\implies} \mathcal{S} \vdash (\Box A \land \Box B) \to \Box B $ $\overset{\text{MP(4,HP)}}{\implies} \mathcal{S} \vdash \Box B $ \\
      (5) $\overset{\text{taut}}{\implies} \mathcal{S} \vdash A \to (B \to (A \land B)) \ \overset{\texttt{MLK-imp-box}}{\implies}$ $\vdash \Box A \to  \Box (B \to (A \land B))  $
      \\(6) $\overset{\mathbf{K}}{\implies} \mathcal{S} \vdash \Box (B \to (A \land B)) \to (\Box B \to \Box (A \land B))$\ \\
      Then we apply repeatedly modus ponens: \\
      (7) $\overset{\text{MP(5,3)}}{\implies} \mathcal{S} \vdash \Box (B \to (A \land B))$ \\
      (8) $\overset{\text{MP(6,7)}}{\implies} \mathcal{S} \vdash  \Box B \to \Box (A \land B)$ \\
      (8)$\overset{\text{MP(8,4)}}{\implies} \mathcal{S} \vdash  \Box (A \land B)$
    \end{itemize}
\end{proof}

\newpage

\section{Normal Systems within the Modal Cube}\label{sec:normal-systems}

In this work, as in HOLMS, we will focus on normal systems and, in particular, on some of the systems within the \textit{modal cube}\footnote{This notion is discussed at the end of the section, in subsection~\ref{sub:modal-cube}.}. We present here the standard axiomatisation of eight well-known modal logics.

\begin{enumerate}
    \item$\mathbb{K}$: is the minimal normal system that has not additional axiom schemata;
     \item $\mathbb{D}$: extends $\mathbb{K}$ with the set of schemata $\mathcal{D}= \{\mathbf{D} \coloneq \Box A \to \Diamond A\}$;
    \item $\mathbb{T}$: extends $\mathbb{K}$ with the set of schemata $\mathcal{T}= \{\mathbf{T}\coloneq \Box A \to A\}$~\href{https://archive.softwareheritage.org/swh:1:cnt:1352cb294724b2251b9346657432700ecbed10d2;origin=https://github.com/HOLMS-lib/HOLMS;visit=swh:1:snp:2c0efd349323ed6f8067581cf1f6d95816e49841;anchor=swh:1:rev:1caf3be141c6f646f78695c0eb528ce3b753079a;path=/t_completeness.ml;lines=8-9}{~\ExternalLink};
    \item $\mathbb{K}\mathbf{4}$: extends $\mathbb{K}$ with the set of schemata $\mathcal{K}4= \{\mathbf{4}\coloneq \Box A \to \Box \Box A\}$~\href{https://archive.softwareheritage.org/swh:1:cnt:32869a1df338d91c4cd4a0cb9e73fb4f1be29991;origin=https://github.com/HOLMS-lib/HOLMS;visit=swh:1:snp:2c0efd349323ed6f8067581cf1f6d95816e49841;anchor=swh:1:rev:1caf3be141c6f646f78695c0eb528ce3b753079a;path=/k4_completeness.ml;lines=8-9}{~\ExternalLink};
    \item $\mathbb{S}\mathbf{4}$: extends $\mathbb{K}$ with the set of schemata $\mathcal{S}4= \{\mathbf{T}; \ \mathbf{4}\}$;
    \item $\mathbb{B}$: extends $\mathbb{K}$with the set of schemata $\mathcal{B}= \{\mathbf{T}; \ \mathbf{B}\coloneq A \to \Box \Diamond A\}$;
    \item $\mathbb{S}\mathbf{5}$: extends $\mathbb{K}$ with the set of schemata $\mathcal{S}5= \{\mathbf{T}; \ \mathbf{5}\coloneq \Diamond A \to \Box \Diamond A\}$;
    \item $\mathbb{GL}$: extends $\mathbb{K}$ with the set of schemata $\mathcal{GL} =\{\mathbf{GL}\coloneq \Box (\Box A \to A) \to \Box A\}$~\href{https://archive.softwareheritage.org/swh:1:cnt:52c9de12454731a9a291b06bc750c0d2d14e1fb4;origin=https://github.com/HOLMS-lib/HOLMS;visit=swh:1:snp:2c0efd349323ed6f8067581cf1f6d95816e49841;anchor=swh:1:rev:1caf3be141c6f646f78695c0eb528ce3b753079a;path=/gl_completeness.ml;lines=11-12}{~\ExternalLink};
\end{enumerate}

It is interesting to examine the
relationships between these systems focusing on which system \textit{extends} another.

\begin{definition}[\textbf{Inclusion between modal systems}]\phantom{} 
\begin{mydefinition}
    \begin{mydefinition2}
        A system $\mathbb{S}'$ \textbf{extends }a system $\mathbb{S}$ if every theorem of $\mathbb{S}$ is a theorem of $\mathbb{S}'$. Formally, we write $\mathbb{S} \subseteq \mathbb{S}'$ if, for any modal formula $A$, $\mathcal{S} \vdash A$ implies $\mathcal{S}' \vdash A$. We write $\mathbb{S} \longrightarrow \mathbb{S}' $ when $\mathbb{S}'$ is a \textbf{proper  extension }of $\mathbb{S}$, i.e., $\mathcal{S} \subseteq \mathcal{S}'$ and $\mathbb{S}'  \subsetneq \mathbb{S}$. 
    \end{mydefinition2}
\end{mydefinition}
\end{definition}

From the definition of normal systems, it follows immediately that every normal logic extends $\mathbb{K}$. Similarly, it is straightforward to draw the arrows from $\mathbb{K}\mathbf{4}$ to $\mathbb{S}\mathbf{4}$ and from $\mathbb{T}$ to $\mathbb{S}\mathbf{4}$, $\mathbb{S}\mathbf{5}$ and $\mathbb{B}$. The diagram in Figure~\ref{fig:firstinclusion} sketch these first inclusions.

 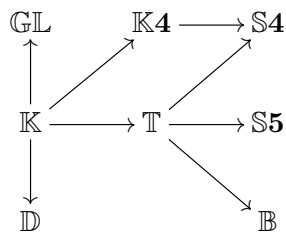
\begin{figure}[h]
     $$\xymatrix{
      \mathbb{GL}  &  \mathbb{K}\mathbf{4}  \ar@{->}[r] & \mathbb{S}\mathbf{4} \\
     \mathbb{K} \ar@{->}[d] \ar@{->}[u] \ar@{->}[ru] \ar@{->}[r] & \mathbb{T} \ar@{->}[rd] \ar@{->}[r] \ar@{->}[ru] & \mathbb{S}\mathbf{5} \\
     \mathbb{D} & & \mathbb{B}
 }
 $$
 
      \caption{An intuitive diagram of the inclusions within the modal cube}
      \label{fig:firstinclusion}
\end{figure}

We provide an improvement of this first sketch with the following Figure~\ref{fig:properinclusions}, which renders count of the well-known proper inclusions between the normal modal logics under analysis.
\begin{figure}[h]
     $$\xymatrix{
     & & \mathbb{GL}  \\
     &  \mathbb{K}\mathbf{4} \ar@{->}[ru] \ar@{->}[r] & \mathbb{S}\mathbf{4} \ar@{->}[r] & \mathbb{S}\mathbf{5} \\
     \mathbb{K} \ar@{->}[ru] \ar@{->}[r] & \mathbb{D} \ar@{->}[r] & \mathbb{T} \ar@{->}[u] \ar@{->}[r] & \mathbb{B} \ar@{->}[u] 
 }
 $$
 \caption{A diagram of the proper inclusion within the cube}
\label{fig:properinclusions}
\end{figure}
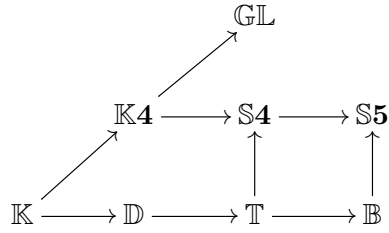
To exemplify the derived inclusions, we prove that $\mathbb{GL}$ extends $\mathbb{K}\mathbf{4}$, i.e. $\mathcal{GL} \vdash \mathbf{4}$, that $\mathbb{T}$ extends $\mathbb{D}$, i.e. $\mathcal{T} \vdash \mathbf{D}$,  and that $\mathbb{S}\mathbf{5}$ extends $\mathbb{B}$, i.e. $\mathcal{S}5 \vdash \mathbf{B}$. To prove that the inclusions are proper we will exploit some results of correspondence theory (see Remark~\ref{rmk:proper_inclusions}).
\newpage

\begin{lemma}[\textbf{$\mathbb{GL}$ extends $\mathbb{K}$4}~\href{https://archive.softwareheritage.org/swh:1:cnt:52c9de12454731a9a291b06bc750c0d2d14e1fb4;origin=https://github.com/HOLMS-lib/HOLMS;visit=swh:1:snp:2c0efd349323ed6f8067581cf1f6d95816e49841;anchor=swh:1:rev:1caf3be141c6f646f78695c0eb528ce3b753079a;path=/gl_completeness.ml;lines=228-236}{\ExternalLink}] \label{lem:GLextendsK4}
\phantom{Allyouneedislove}
\begin{mydefinition}
$\mathbb{K}\mathbf{4} \subseteq \mathbb{GL}$    
\end{mydefinition}
\end{lemma}
\begin{proof}
We want to prove that for any $B \in \mathbf{Form}_{\Box}$, if $\mathbb{K}\mathbf{4} \vdash B$ then
$\mathcal{GL} \vdash B$. \\
By the axiomatisation of $\mathbb{K}\mathbf{4}$ we know that $B$ is either a theorem of $\mathbb{K}$ or it follows from $\mathbf{4}$. Then all we need to prove is $\mathcal{GL} \vdash \Box A \to \Box \Box A$. 
\begin{enumerate}
    \item $\mathcal{GL}$ $ \vdash A \to \underline{((\Box \Box A \land \Box A)} \to (\Box A \land A))$ \hfill \textit{classical logic} (tautology)
    \item $\mathcal{GL}$ $ \vdash (\underline{\Box \Box A \land \Box A)} \to \Box (\Box A \land A))$ \hfill \textit{normality} (Lemma \ref{lem:MLK_box} on 1)
    \item $\mathcal{GL}$ $ \vdash A \to (  \Box (\Box A \land A) \to (\Box A \land A))$ \hfill \textit{classical logic} on 1,2
    \item $\mathcal{GL}$ $ \vdash \Box A \to \underline{\Box(  \Box (\Box A \land A) \to (\Box A \land A))}$ \hfill  \textit{normality} (Lemma~\ref{lem:MLK_imp_box} on 3) 
    \item $\mathcal{GL}$ $ \vdash \Box \underline{(  \Box (\Box A \land A) \to (\Box A \land A))} \to  \Box (\Box A \land A)$ \hfill $\mathbb{GL}$ \textit{axiom} 
    \item $\mathcal{GL}$ $ \vdash  \Box A \to  \underline {\Box (\Box A \land A)}$ \hfill \textit{classical logic} on 4,5 
    \item $\mathcal{GL}$ $ \vdash  (\Box A \land A) \to \Box A$ \hfill \textit{classical logic} (tautology)
     \item $\mathcal{GL}$ $ \vdash  \underline {\Box (\Box A \land A)} \to \Box \Box A$ \hfill \textit{normality} (Lemma~\ref{lem:MLK_imp_box} on 7)
     \item $\mathcal{GL}$ $ \vdash  \Box A\to \Box \Box A$ \hfill \textit{classical logic} on 6,8
\end{enumerate}
\end{proof}

\begin{lemma}[\textbf{$\mathbb{T}$ extends $\mathbb{D}$}]
\phantom{Allyouneedislove}
\begin{mydefinition}
$\mathbb{D} \subseteq \mathbb{T}$    
\end{mydefinition}
\end{lemma}
\begin{proof}
Similarly, we must prove $\mathcal{T} \vdash \Box A \to \Diamond A$. Note that $\Diamond A \coloneq \neg \Box \neg A$.
\begin{enumerate}
    \item $\mathcal{T}$ $ \vdash \Box A \to \underline{A}$ \hfill $\mathbb{T} $\textit{ axiom}
    \item $\mathcal{T}$ $ \vdash \Box \neg A \to \neg A$ \hfill $\mathbb{T} $\textit{ axiom}
    \item $\mathcal{T}$ $ \vdash \underline{A} \to \neg \Box \neg A$ \hfill \textit{classical logic} on 2
    \item $\mathcal{T}$ $ \vdash \Box A \to \neg \Box \neg A$ \hfill \textit{classical logic} on 1, 3
    
\end{enumerate}
\end{proof}
\newpage
\begin{lemma}[\textbf{$\mathbb{S}\mathbf{5}$ extends $\mathbb{B}$}]
\phantom{Allyouneedislove}
\begin{mydefinition}
$\mathbb{S}\mathbf{5} \subseteq \mathbb{B}$    
\end{mydefinition}
\end{lemma}
\begin{proof}
We want to prove that for any $C \in \mathbf{Form}_{\Box}$, if $\mathcal{B} \vdash C$ then
$\mathcal{S}5 \vdash C$. \\
By the axiomatisation of $\mathbb{B}$, we know that $C$ is either a theorem of $\mathbb{K}$ or it follows from $\mathbf{T}$ or $\mathbf{B}$. Because $\mathbf{T} \in \mathcal{S}5$, all we need to prove is $\mathcal{S}5 \vdash A \to \Box \Diamond A$. 

\begin{enumerate}
    \item $\mathcal{S}5$ $ \vdash \Box \neg A \to \neg A$ \hfill $\mathbb{T} $\textit{ axiom}
    \item $\mathcal{S}5$ $ \vdash A \to \neg \Box \neg A$ \hfill \textit{classical logic} on 2.
    \item $\mathcal{S}5$ $ \vdash A \to \underline{\Diamond A}$ \hfill  $\Diamond A \coloneq \neg \Box \neg A$
    \item $\mathcal{S}5$ $ \vdash \underline{\Diamond A} \to \Box \Diamond A$ \hfill $\mathbb{S}\mathbf{5} $\textit{ axiom} ($\mathsf{5}$)
    \item $\mathcal{S}5$ $ \vdash A \to \Box \Diamond A$ \hfill \textit{classical logic} on 1. and 2. 
    
\end{enumerate}
\end{proof}

\smallskip

\subsection{Modal Cube}\label{sub:modal-cube}
    The expression \textit{modal cube} commonly refers to some normal systems that form a nested hierarchy.
    This grouping is well-described in Figure~\ref{fig:modal_cube} below, drawn from~\cite[\S 8]{sep-logic-modal}. This diagram expands our understanding of the relationship between the normal systems depicted in Figure~\ref{fig:properinclusions}.

\begin{figure}[ht]
    \centering
    \includegraphics[width= 0.6\textwidth]{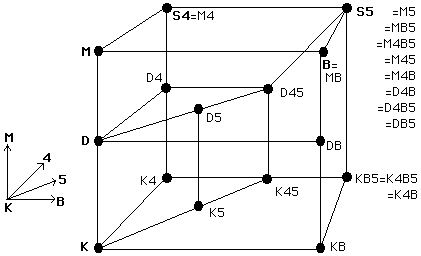}\caption{Map of the proper inclusions between modal logics. Here $\mathbf{M} \coloneq \mathbf{T}$}
    \label{fig:modal_cube}
\end{figure}

The minimal system $\mathbb{K}$ is represented as the core vertex of a cube. At the opposite vertex, lays $\mathbb{S}5$, the strongest theory within the cube. In the middle, there are other interesting extensions of $\mathbb{K}$, such as $\mathbb{D,\ T,\ B,\ K}4$ and $\mathbb{S}4$ . 
    
    Since $\mathbb{S}5$ does not describe every modality of interest, other normal systems outside the cube (which properly extend or are not extended by $\mathbb{S}5$) have been formulated and studied, such as $\mathbb{GL}$. Note that all the systems under our analysis occupy important positions within the cube, except from $\mathbb{GL}$ which is interesting due to its capabilities of describing provability.

\section{Correspondence Theory}\label{sec:correspondence-theory}

Correspondence theory is an important and classical area of modal logic. It analyses and searches for connections between properties of relational frames and modal formulas that hold in exactly the relational frames with those properties \cite{VanBenthem2001}.
By studying these correspondences, we can identify the class of frames that guarantees soundness and completeness w.r.t. relational semantics for each normal logic. 

\subsection{Correspondence Results}

We prove below some important results of the correspondence theory.

\begin{lemma}[\textbf{Correspondence lemmata}]\label{lem:Correspondence}
\phantom{Allyouneedislove}
\begin{mydefinition}
   \begin{enumerate}
      \item $\langle W,R\rangle \vDash \mathbf{D} \coloneq \Box A \to \Diamond  A$ \textnormal{iff} R is serial. 
      \item\label{itm:corr_2} $\langle W,R\rangle \vDash \mathbf{T} \coloneq \Box A \to A$ \textnormal{iff} R is reflexive.~\href{https://archive.softwareheritage.org/swh:1:cnt:1352cb294724b2251b9346657432700ecbed10d2;origin=https://github.com/HOLMS-lib/HOLMS;visit=swh:1:snp:2c0efd349323ed6f8067581cf1f6d95816e49841;anchor=swh:1:rev:1caf3be141c6f646f78695c0eb528ce3b753079a;path=/t_completeness.ml;lines=23-27}{\ExternalLink} 
      \item $\langle W,R\rangle \vDash \mathbf{4} \coloneq \Box A \to \Box \Box A$ \textnormal{iff} R is transitive.~\href{https://archive.softwareheritage.org/swh:1:cnt:32869a1df338d91c4cd4a0cb9e73fb4f1be29991;origin=https://github.com/HOLMS-lib/HOLMS;visit=swh:1:snp:2c0efd349323ed6f8067581cf1f6d95816e49841;anchor=swh:1:rev:1caf3be141c6f646f78695c0eb528ce3b753079a;path=/k4_completeness.ml;lines=27-33}{\ExternalLink} \label{itm:4correspondence} 
      \item $\langle W,R\rangle \vDash \mathbf{5} \coloneq \Diamond A \to \Box \Diamond  A$ \textnormal{iff} R is euclidean.
      \item $\langle W,R\rangle \vDash \mathbf{B} \coloneq  A \to \Box \Diamond A$ \textnormal{iff} R is symmetric.
      \item $\langle W,R\rangle \vDash \mathbf{GL} \coloneq  \Box (\Box A \to A) \to \Box A$ \textnormal{iff} R is transitive and converse well-founded.~\href{https://archive.softwareheritage.org/swh:1:cnt:52c9de12454731a9a291b06bc750c0d2d14e1fb4;origin=https://github.com/HOLMS-lib/HOLMS;visit=swh:1:snp:2c0efd349323ed6f8067581cf1f6d95816e49841;anchor=swh:1:rev:1caf3be141c6f646f78695c0eb528ce3b753079a;path=/gl_completeness.ml;lines=115-129}{\ExternalLink}
    \end{enumerate}  
\end{mydefinition}
\end{lemma}
\begin{proof} 
For each schema, we will prove the first implication by contraposition and construction of the countermodel, and the second by \textit{reductio ad absurdum}. 
\medskip

\noindent \textbf{Observation:} To prove our claims, e.g. $\mathcal{F} \vDash \Box A \to \Diamond A $ \textit{iff R is serial}, it is sufficient to prove something easier, e.g. $\mathcal{F} \vDash \Box p \to \Diamond p $ \textit{iff R is serial}  where $p$ is a propositional atom. Indeed, the preservation of validity in modal frames by substitution (\ref{itm:SUB} in Lemma~\ref{thm:K_soundness}) guarantees the equivalence of these two statements, and e.g. $\mathcal{F} \vDash \Box A \to \Diamond A \iff \mathcal{F} \vDash \Box p \to \Diamond p$.

\begin{enumerate}
    \item $[\mathcal{F} \vDash \mathbf{D} \iff \mathcal{F}$ is serial$]$, i.e. \\ $[\mathcal{F} \vDash \Box p \to \Diamond p \iff \forall x \in W \ \exists y \in W (xRy)]$
      
      \begin{itemize}
          \item[$\implies$][$R$ \textit{not serial} $\implies \exists \mathcal{M} \ \mathcal{F}_* \ \exists w \in W(\mathcal{M},w \vDash \Box p \ and \ \mathcal{M},w \not \vDash \Diamond p) $]
          
          Because $R$ is not serial, there exists a world $x$ that cannot see any world. In literature, a world $x \in W$ such that $\neg \exists y \in W (xRy)$ is called ``dead-end''. By the inductive definition of forcing, for any dead-end $x$ we have $\mathcal{M}, x \vDash \Box p$ and $\mathcal{M}, x \not \vDash \Diamond p$.  This countermodel is represented in the following graph\footnote{In the following diagrams, blue arrows highlight the absence of a relationship between two worlds. Such a choice is intended to visualise that the relationship under analysis does not have a given property.}:

          \begin{center}
          \begin{tikzpicture}
          \node (x) at (0,0) [circle,fill,inner sep=1.5pt,label=below:$x$] {};
          \node (y) at (4,0) [circle,fill,inner sep=1.5pt,label=below:$\dots$] {};
          \node at (-0.5,0.7) {\small$\mathcal{M},x \vDash \Box p$};
          \node at (-0.5,-0.7) {\small$\mathcal{M},x \not \vDash \Diamond p$};
          \draw[->, blue] (x) -- node[above] {} (y);
          \draw[->, blue] (x) to[out=135, in=225, looseness=30] (x);
          \end{tikzpicture}
          \end{center}
          
          \item[$\impliedby$] [$R$ \textit{is serial} $\implies \langle W,R \rangle  \vDash \Box p \to \Diamond p   ) $] \\
          RAA: We assume that $\exists \mathcal{M} \ \mathcal{F}_* \ \exists w \in W(\mathcal{M},w \vDash \Box p \ and \ $ \underline{$ \mathcal{M},w \not \vDash \Diamond p)$}. \\
          By seriality, there exists a world $j$ such that $wRj$. Given this world $j$ and the ``absurde'' hypotesis  $\mathcal{M}, w\vDash \Box p$, then $\mathcal{M}, j \vDash p$. But, if $\mathcal{M}, j \vDash p$ and $wRj$ then \underline{$w \vDash\Diamond p$}. A contradiction. $\lightning$
      \end{itemize}
       \hfill $\lhd$
    \item $[\mathcal{F} \vDash \mathbf{T} \iff \mathcal{F}$ is reflexive$]$, i.e. $[\mathcal{F} \vDash \Box p \to  p \iff \forall x \in W (xRx)]$
 \begin{itemize}
          \item[$\implies$] [$R$ \textit{not reflexive} $\implies \exists \mathcal{M} \ \mathcal{F}_* \ \exists w \in W(\mathcal{M},w \vDash \Box p \ and \ \mathcal{M},w \not \vDash p   ) $]
          
          Because $R$ is not reflexive, then there exists a world $x$ that cannot see itself.
          Let $V$ be a valuation on $W$ such that $\forall y \in W (pVy \iff y \not = x)$ and $\mathcal{M}= \langle W,R,V \rangle$.
          For such $x$ and $V$, we have:
          \begin{itemize}
              \item  $\mathcal{M}, x \not \vDash p$ because  $\neg pVx$.
              \item $\mathcal{M},x \vDash \Box p$. Indeed, for any world $y$, if $xRy$ then $x \not = y$; thus $pVy$ and $\mathcal{M}, y \vDash p$. 
              
          \end{itemize}
         This countermodel is represented in the following graph:

          \begin{center}
          \begin{tikzpicture}
          \node (x) at (0,0) [circle,fill,inner sep=1.5pt,label=below:$x$] {};
          \node (y) at (4,1) [circle,fill,inner sep=1.5pt,label=below:$y$] {};
          \node (k) at (4,-1) [circle,fill,inner sep=1.5pt,label=below:$\dots$] {};
          \node at (0,-0.7) {\small$\mathcal{M},x \not \vDash p$};
          \node at (5,1) {\small$\mathcal{M},y  \vDash p$};
          \node at (5.0,-1) {\small$\mathcal{M}, \dots  \vDash p$};
          \draw[->] (x) -- node[above] {} (y);
          \draw[->] (x) -- node[above] {} (k);
          \draw[->, blue] (x) to[out=135, in=225, looseness=30] (x);
          \end{tikzpicture}
          \end{center}

          \item[$\impliedby$] [$R$ \textit{is reflexive} $\implies \langle W,R \rangle  \vDash \Box p \to p   ) $] \\  \textbf{RAA}: $\exists \mathcal{M} \ \mathcal{F}_* \ \exists w \in W(\mathcal{M},w \vDash \Box p \ and \ $ \underline{$ \mathcal{M},w \not \vDash p)$}. \\
          Then by reflexivity, $wRw$. Given the ``absurde'' hypothesis $\mathcal{M}, w\vDash \Box p$, then \underline{$\mathcal{M}, w \vDash p$}. A contradiction. $\lightning$
      \end{itemize}
       \hfill $\lhd$
     \item $[\mathcal{F} \vDash \mathbf{4} \iff \mathcal{F}$ is transitive$]$, i.e. \\ $[\mathcal{F} \vDash \Box p \to \Box \Box p \iff \forall x,y,z \in W (xRy \ and \ yRz \implies xRz)]$
     \begin{itemize}
          \item[$\implies$] [$R$ \textit{not transitive} $\implies \exists \mathcal{M} \mathcal{F}_* \ \exists w \in W(\mathcal{M},w \vDash \Box p \ and \ \mathcal{M},w \not \vDash \Box \Box p   ) $]
          
          Because $R$ is not transitive, there exist three worlds $x,y,z$  such that $x R y$ and $yRz$ but $\neg (xRz)$.
          Let $V$ be a valuation on $W$ such that $\forall w \in W (pVw \iff xRw)$ and $\mathcal{M}= \langle W,R,V \rangle$.
          For such $x$ and $V$, we have:
          \begin{itemize}
              \item $\mathcal{M},x \vDash \Box p$. Indeed, for any world $y$, if $xRy$ then $pVy$ and thus $\mathcal{M},y \vDash p$. 
              \item  $\mathcal{M}, x \not \vDash \Box \Box p$. Indeed, $\neg (xRz)$ then $\neg (pVz)$; thus $\mathcal{M}, z \not \vDash p$ and $\mathcal{M}, y \not \vDash \Box p$.
          \end{itemize}
           This countermodel is represented in the following graph:

          \begin{center}
          \begin{tikzpicture}
          \node (x) at (0,0) [circle,fill,inner sep=1.5pt,label=below:$x$] {};
          \node (y) at (4,0) [circle,fill,inner sep=1.5pt,label=below:$y$] {};
          \node (k) at (4,1) [circle,fill,inner sep=1.5pt,label=below:$\dots$] {};
          \node (z) at (4,-1) [circle,fill,inner sep=1.5pt,label=below:$z$] {};
          \node at (5,1) {\small$\mathcal{M}, \dots \vDash p$};
          \node at (5,0) {\small$\mathcal{M},y \vDash p$};
          \node at (5.0,-1) {\small$\mathcal{M}, z \not  \vDash p$};
          \draw[->] (x) -- node[above] {$R$} (k);
          \draw[->] (x) -- node[above] {$R$} (y);
          \draw[->] (y) -- node[right] {$R$} (z);
          \draw[->,Blue] (x) -- node[above] {$R$} (z);
          \end{tikzpicture}
          \end{center}

          \item[$\impliedby$] [$R$ \textit{is transitive} $\implies \langle W,R \rangle  \vDash \Box p \to \Box \Box p   ) $] \\  \textbf{RAA}: $\exists \mathcal{M} \ \mathcal{F}_* \ \exists w \in W(\mathcal{M},w \vDash \Box p \ and \ $ \underline{$ \mathcal{M},w \not \vDash \Box \Box p)$}. \\
          From $\mathcal{M},w \vDash \Box p$, we know that 
          for each world $ j$, if $wRj$ then $\mathcal{M},j \vDash p$. By transitivity, we know that each world $k$ that is seen by $j$ ($jRk$) is also seen by $w$ ($wRk$). Then for such a $k$, $\mathcal{M},k \vDash p$ and thus \underline{$\mathcal{M}, w \vDash \Box \Box p$}. A contradiction. $\lightning$
      \end{itemize}
     
     \hfill $\lhd$
      
    \item $[\mathcal{F} \vDash \mathbf{5} \iff \mathcal{F}$ is euclidean$]$, i.e. \\ $[\mathcal{F} \vDash \Diamond p \to \Box \Diamond p \iff \forall x,y,z \in W (xRy \ and \ xRz \implies yRz)]$
  \begin{itemize}
          \item[$\implies$] [$R$ \textit{not euclidean} $\implies \exists \mathcal{M} \mathcal{F}_* \ \exists w \in W(\mathcal{M},w \vDash \Diamond p \ and \  \mathcal{M},w \not \vDash \Box \Diamond p   ) $]
          
          Because $R$ is not euclidean, there are three worlds $x,y,z$  such that $x R y$ and $xRz$ but $\neg (yRz)$.
          Let $V$ be a valuation on $W$ such that $\forall w \in W(pVw \iff w=y)$ and $\mathcal{M}= \langle W,R,V \rangle$. For such $x$ and $V$, we have:
          \begin{itemize}
              \item  $\mathcal{M},x \vDash \Diamond p$ because  $ pVy$ and thus $\mathcal{M}, y \vDash p$.
              \item $\mathcal{M},x \not \vDash \Box \Diamond p$. Indeed, any world $k$ that is seen by $z$ ($zRk$)  is different from $y$ ($k \not = y$); thus $\neg (pVk)$ and $\mathcal{M},k \not \vDash p$. By the inductive definition of truth, $\mathcal{M},z \not \vDash \Diamond p$.
          \end{itemize}
         This countermodel is represented in the following graph:

          \begin{center}
          \begin{tikzpicture}
          \node (x) at (0,0) [circle,fill,inner sep=1.5pt,label=below:$x$] {};
          \node (y) at (4,1) [circle,fill,inner sep=1.5pt,label=below:$y$] {};
          \node (z) at (4,-1) [circle,fill,inner sep=1.5pt,label=below:$z$] {};
          \node (k) at (8,0) [circle,fill,inner sep=1.5pt,label=below:$k$] {};
          \node (j) at (8,-1) [circle,fill,inner sep=1.5pt,label=below:$\dots$] {};
          \node at (5.0,1) {\small$\mathcal{M}, y   \vDash  p$};
          \node at (9.5,0) {\small$\mathcal{M}, k \not  \vDash p$};
          \node at (9.5,-1) {\small$\mathcal{M}, \dots \not  \vDash p$};
          \draw[->] (x) -- node[above] {} (y);
          \draw[->, Blue] (y) -- node[above] {} (z);
          \draw[->] (x) -- node[right] {} (z);
          \draw[->] (z) -- node[above] {} (k);
          \draw[->] (z) -- node[above] {} (j);
          \end{tikzpicture}
          \end{center}

          \item[$\impliedby$] [$R$ \textit{is euclidean} $\implies \langle W,R \rangle  \vDash \Diamond p \to \Box \Diamond p   ) $] \\  \textbf{RAA}: $\exists \mathcal{M} \ \mathcal{F}_* \ \exists w \in W(\mathcal{M},w \vDash \Diamond p \ and \ $ \underline{$ \mathcal{M},w \not \vDash \Box \Diamond p)$}. \\
          From $\mathcal{M},w \vDash \Diamond p$, we know that there exists a world $ j$, such that $wRj$ and $\mathcal{M},j \vDash p$. 
          Because $\mathcal{F}$ is euclidean, we know that each world $k$ that is seen by $w$ ($wRk$) is also seen by $j$ ($jRk$). Then for each $k$, $\mathcal{M},k \vDash \Diamond p$ and thus \underline{$\mathcal{M}, w \vDash \Box \Diamond p$}. A contradiction. $\lightning$
      \end{itemize}
      \hfill $\lhd$
      
    \item$[\mathcal{F} \vDash \mathbf{B} \iff \mathcal{F}$ is symmetric$]$, i.e. \\ $[\mathcal{F} \vDash  p \to \Box \Diamond p \iff \forall x,y \in W (xRy \implies yRx)]$
    \begin{itemize}
          \item[$\implies$] [$R$ \textit{not symmetric} $\implies \exists \mathcal{M} \mathcal{F}_* \ \exists w \in W(\mathcal{M},w \vDash p \ and \  \mathcal{M},w \not \vDash \Box \Diamond p   ) $]
          
          Because $R$ is not symmetric, there exists a couple of worlds $x,y$  such that $x R y$ but $\neg (yRx)$.
          Let $V$ be a valuation such that  $\forall w \in W (pVw \iff w=x)$ and $\mathcal{M}= \langle W,R,V \rangle$. For such $x$ and $V$, we have:

          \begin{itemize}
              \item $\mathcal{M},x \vDash p$ because $pVx$.
              \item $\mathcal{M},x \not \vDash \Box \Diamond p$. Indeed  each world $k$ that is seen by $y$  ($yRk$) is different from $x$ ($k \not = x)$; thus $\neg (pVk)$ and $\mathcal{M},k \not \vDash p$ from which $\mathcal{M},y \not \vDash \Diamond p$.
          \end{itemize}
         This countermodel is represented in the following graph:

          \begin{center}
          \begin{tikzpicture}
          \node (y) at (0,0) [circle,fill,inner sep=1.5pt,label=below:$y$] {};
          \node (x) at (4,1) [circle,fill,inner sep=1.5pt,label=below:$x$] {};
          \node (k) at (4,0) [circle,fill,inner sep=1.5pt,label=below:$k$] {};
          \node (z) at (4,-1) [circle,fill,inner sep=1.5pt,label=below:$\dots$] {};
          \node at (0,-1) {\small$\mathcal{M},y \not \vDash \Diamond p$};
          \node at (5.0,0) {\small$\mathcal{M}, z \not  \vDash p$};
          \node at (5.0,-1) {\small$\mathcal{M}, \dots \not  \vDash p$};
          \node at (5,1) {\small$\mathcal{M}, x  \vDash p$};
          \draw[->] (y) -- node[above] {} (z);
          \draw[->] (y) -- node[above] {} (k);
         \draw[->, Blue, bend left=30] (y) to node[above] {} (x);
          \draw[->] (x) -- node[above] {} (y);
          \end{tikzpicture}
          \end{center}
          
     \item[$\impliedby$] [$R$ \textit{is symmetric} $\implies \langle W,R \rangle  \vDash p \to \Box \Diamond p   ) $] \\  \textbf{RAA}: $\exists \mathcal{M} \ \mathcal{F}_* \ \exists w \in W(\mathcal{M},w \vDash p \ and \ $ $ \mathcal{M},w \not \vDash \Box \Diamond p)$. \\
          From $\mathcal{M},w  \not \vDash \Box \Diamond p$, we know that there exists a world $ j$, such that $wRj$ and \underline{$\mathcal{M},j \not \vDash \Diamond p$}. 
          By symmetry, also $jRw$ and from $\mathcal{M},w \vDash p$, i follows \underline{$\mathcal{M},j \vDash \Diamond p$}.
          
    \end{itemize}
  
    \item$[\mathcal{F} \vDash \mathbf{GL} \iff \mathcal{F}$ is transitive and converse well-founded$]$, i.e. \\ $[\mathcal{F} \vDash   \Box (\Box p \to p) \to \Box p \iff $ $\mathcal{F}$ is transitive  and $ \forall X \subseteq W(X \not = \varnothing \implies \exists x \in X(\neg \exists y \in X (xRy))) ]$

     \begin{itemize}
          \item[$\implies$] [$\mathcal{F} \vDash \Box (\Box p \to p) \to \Box p \implies \mathcal{F}$ \textit{ transitive and converse well-founded}]
          \begin{enumerate}
              \item $[\mathcal{F}$ \textit{is transitive}$]$ \\ 
              By preservation of validity in modal frames by substitution, from $\mathcal{F} \vDash \Box (\Box p \to p) \to \Box p$  follows that $\mathcal{F} \vDash \Box (\Box A \to A) \to \Box A$. Then, by the definition of derivability and Lemma~\ref{thm:K_soundness} (Soundness of $\mathbb{K}$), all theorems of $\mathbb{GL}$ are valid in $\mathcal{F}$. From Lemma~\ref{lem:GLextendsK4}, $\mathbf{4}$ is a theorem of $\mathbb{GL}$ and then $\mathcal{F} \vDash \mathbf{4}$. Consequently, from the \ref{itm:4correspondence}$^{rd}$ correspondence of this Lemma~\ref{lem:Correspondence}, $\mathcal{F}$  is transitive.
              \item $[\mathcal{F} $ \textit{is converse well-founded}$]$ By contraposition:
              $[R$  not converse WF $\implies \exists \mathcal{M} \mathcal{F}_* \ \exists w \in W(\mathcal{M},w \vDash \Box(\Box p \to p) \ and \  \mathcal{M},w \not \vDash \Box p   ) ]$
              Because $R$ is not converse well-founded, there is a nonempty set $X \subseteq W$ with no R-greatest element ($ \forall x \in X (\exists y \in X (xRy))$). Let $V$ be the valuation on $W$ such that $\forall w \in W(wVp $ iff $w \not \in X)$. Let $\mathcal{M}=\langle W,R,V \rangle$.
              Because $X$ has no R-greatest element we can consider a generic $x \in X$ and a world $y \in W$ such that $xRy$, and then we have:
              \begin{itemize}
                  \item $\mathcal{M},x \not \vDash \Box p$. Indeed $x \in X$ has an R-greater world $k$ such that $k \in X \ and \ xRk$. For such a world $k$: $\neg(pVk)$; thus $\mathcal{M},k \not \vDash p$ and $xRk$.
                  \item $\mathcal{M},w \vDash \Box(\Box p \to p)$, i.e. for any world $y$ if $\mathcal{M}, y \vDash \Box p \to p$. This step is proved by contraposition $\mathcal{M}, y \not \vDash p$ implies $\mathcal{M}, y \not \vDash \Box p$. 
                  If $\mathcal{M}, y \not \vDash p$ then $\neg(pVy)$; thus $y \in X$. If $y$ is a member of $X$ then there exists a world $z \in X$ that is R-greater than $y$ ($z \in X \ and \ yRz$). Then for such a world $z$: $\neg(pVz)$; thus $\mathcal{M},z \not \vDash p$ and $yRz$. Finally, $\mathcal{M},y \not \vDash \Box p$
              \end{itemize}
              
    \item[$\impliedby$] [$R$ \textit{is transitive and conv WF} $\implies \langle W,R \rangle  \vDash \Box (\Box p \to p)\to \Box p   ) $] \\  \textbf{RAA}: $\exists \mathcal{M} \ \mathcal{F}_* \ \exists w \in W($ \underline{$\mathcal{M},w \vDash \Box (\Box p \to p) $} and  $\mathcal{M},w \not \vDash \Box p)$. \\
    We consider a set $X= \{x \in W | wRx \ and \ \mathcal{M},x \not \vDash p\}$. 
    \begin{itemize}
        \item We prove that $X$ is nonempty: Since $\mathcal{M},w \not \vDash \Box p $ there exists a world $j \in W$ such that $wRj$ and $\mathcal{M},j \not \vDash p$. This $j$ is in $X$.
        \item Because $R$ is converse well-founded and $X$ is nonempty, $X$ has an $R$-greatest element $x$ such that $x \in X$ and $\forall y \in X \neg (xRy)$.
        \item We want to prove that \underline{$\mathcal{M},w \not \vDash \Box (\Box p \to p)$}, in particular $\mathcal{M},x \not \vDash \Box p \to p$. Because $x \in X$ then $\mathcal{M},w \not \vDash p$. Moreover, we prove that $\mathcal{M},x \vDash \Box p$: given $k \in W$ such that $xRk$, we have that $k$ is $R$-greater than $x$ and then $k \not \in X$; then, by transitivity, $wRk$ and then $\mathcal{M}, k \vDash p$; thus $\mathcal{M}, x \vDash \Box p$.
    \end{itemize}

          \end{enumerate}
               
    \end{itemize}
  
\end{enumerate} 
\end{proof}

\subsection{Characteristic Frames}

The following Table~\ref{tab:char-prop} summarizes the correspondences between the axiom schemata under analysis and the property of the classes of frames in which the schemata hold.

\begin{table}[h]\label{tab:characteristic}
    \centering
    \begin{tabular}{l | l}
        \toprule
        \textbf{Schema} & \textbf{Characteristic Property} \\
        \midrule
        $\mathbf{D}\coloneq \Box A \to \Diamond A$ & \textbf{serial}: $\forall w \in W \ \exists y \in W (wRy)$\\
        \hline
        $\mathbf{T}\coloneq \Box A \to A$   & \textbf{reflexive}: $\forall w \in W (wRw)$  \\
        \hline
        $\mathbf{4}\coloneq \Box A \to \Box \Box A$ &  \textbf{transitive}: $\forall w,x,z \in W\,(wRx \ and \ xRz \implies wRz)$\\
        \hline 
        $\mathbf{5}\coloneq \Diamond A \to \Box \Diamond A$ & \textbf{euclidean}: $\forall w,x,y \in W (wRx \ and \ wRy \implies yRx)$\\
        \hline
        $\mathbf{B}\coloneq A \to \Box \Diamond A$ & \textbf{symmetric}: $\forall w,x \in W (wRx \implies xRw)$\\
         \hline
        $\mathbf{GL}\coloneq \Box (\Box A \to A) \to \Box A$ & \textbf{transitive}: $\forall w, x, y \in W (wRx \ and \ xRy \implies wRy)$\\
        &  \textbf{CWF}: $\forall X \subseteq W\,(X \not = \varnothing \implies \exists w \in X(\neg \exists x \in X(wRx)))$\\
         \bottomrule
    \end{tabular}
    \caption{Properties characteristic to axiom schemata defining the cube}\label{tab:char-prop}
\end{table}

This figure also suggests rephrasing these specific results in a general setting. To do this, we define a predicate representing the class of frames \textit{characteristic} to a given modal system.

\begin{definition}[\textbf{Characteristic Property}]\phantom{Allyouneed}
\begin{mydefinition}
    \begin{mydefinition}
        A frame property $P$ is \textbf{characteristic} to a modal system $\mathbb{S}$ when the following equivalence holds: every axiom schema in $\mathcal{S}$ is valid in a frame if and only if the frame has the given property $P$.
    \end{mydefinition}
\end{mydefinition}
\end{definition}

\begin{definition}[\textbf{Characteristic Class of frames}~\href{https://archive.softwareheritage.org/swh:1:cnt:d364de3158b6405c920115aa0c801af9af49e302;origin=https://github.com/HOLMS-lib/HOLMS;visit=swh:1:snp:2c0efd349323ed6f8067581cf1f6d95816e49841;anchor=swh:1:rev:1caf3be141c6f646f78695c0eb528ce3b753079a;path=/gen_completeness.ml;lines=44-47}{\ExternalLink}]\label{def:characteristic}\phantom{Allyouneed}
\begin{mydefinition}
    \begin{mydefinition}
     The \textbf{characteristic class} for the system $\mathbb{S}$ consists exactly of the frames satisfying the characteristic property to $\mathbb{S}$.
    \end{mydefinition}
\end{mydefinition}
\end{definition}

As the set of formulas valid in a frame is closed under our inference rules  (Theorem \ref{thm:K_soundness} on soundness of $\mathbb{K}$ w.r.t. every class of frames), it is relatively easy to provide an interesting \textit{characterisation} 
of the notion of characteristic class of frames.

\begin{lemma}[\textbf{Characterisation of the concept of characteristic}~\href{https://archive.softwareheritage.org/swh:1:cnt:d364de3158b6405c920115aa0c801af9af49e302;origin=https://github.com/HOLMS-lib/HOLMS;visit=swh:1:snp:2c0efd349323ed6f8067581cf1f6d95816e49841;anchor=swh:1:rev:1caf3be141c6f646f78695c0eb528ce3b753079a;path=/gen_completeness.ml;lines=82-105}{\ExternalLink}]\label{lem:CHAR_CAR} \phantom{All}
    \begin{mydefinition}
     Let $\mathbb{S}$  be a modal system.  Let $\mathfrak{S}$ be a class of frames. \\ $\mathfrak{S}$ is the \textnormal{\textbf{characteristic  class for}} $\mathbb{S}$ if and only if  $\mathfrak{S}$ consists exactly of the frames in which every theorem of $\mathbb{S}$ is valid.
    \end{mydefinition}
\end{lemma}
\begin{proof}
    Let $\mathcal{F}$ a frame. We want to prove that $\mathcal{F}$ belongs to the class of characteristic frames \textbf{iff} every theorem of $\mathcal{S}$  is valid in $\mathcal{F}$.
    \[ [\forall A \in \mathbf{Form}_{\Box} (A \in \mathcal{S} \implies \mathcal{F} \vDash A)\  \mathbf{iff} \ \forall A \in \mathbf{Form}_{\Box} (\mathcal{S} \vdash A \implies \mathcal{F} \vDash A)]\]
         
    \begin{itemize}
        \item[$\implies$]  This verse is proved by induction on the recursive definition of $\mathcal{S} \vdash A$:
        \begin{itemize}
            \item If $A$ is either an instance of an axiom schema for $\mathsf{K}$ or $A$ follows from an inference rule (\texttt{MP}, \texttt{RN} and \texttt{SUB}), Theorem \ref{thm:K_soundness} proves the claim.
            \item  If, instead, $A$ is an instance of an axiom schema in $\mathcal{S}$, the antecedent of the implication proves the claim.
        \end{itemize}
        \item [$\impliedby$] This verse obviously follows. If every theorem of $\mathcal{S}$ is valid in a certain frame, then its axioms- which are also theorems- are valid in that frame. 
    \end{itemize}
\end{proof}

\subsection{Appropriate Frames}

As announced, correspondence theory plays an important role in proving the adequacy theorem. If such a characterisation immediately guarantees the soundness of normal logics w.r.t. characteristic frames, we will focus instead on \textit{finite characteristic} frames to prove completeness. To better handle these frames, we introduce the concept of \textit{appropriate class of frames}.

\begin{definition}[\textbf{Appropriate class of frames}~\href{https://archive.softwareheritage.org/swh:1:cnt:d364de3158b6405c920115aa0c801af9af49e302;origin=https://github.com/HOLMS-lib/HOLMS;visit=swh:1:snp:2c0efd349323ed6f8067581cf1f6d95816e49841;anchor=swh:1:rev:1caf3be141c6f646f78695c0eb528ce3b753079a;path=/gen_completeness.ml;lines=111-114}{\ExternalLink}]\label{def:appropriate_class}\phantom{Allyouneed}
\begin{mydefinition}
\begin{mydefinition2}
    A class of \textbf{finite} frame $\mathfrak{S}$ is said to be \textbf{appropriate} to a system $\mathbb{S}$ iff \\ $\mathfrak{S}$ consists exactly of the frames in which any theorem of $\mathbb{S}$ is valid.
\end{mydefinition2}
\end{mydefinition}
\end{definition}

By our characterisation in Lemma~\ref{lem:CHAR_CAR}, it is immediate to prove that the class of frames appropriate to $\mathbb{S}$ coincides with its characteristic class \textit{restricted to finite frames}. 

\begin{lemma}[\textbf{Characterisation of the concept of appropriate}~\href{https://archive.softwareheritage.org/swh:1:cnt:d364de3158b6405c920115aa0c801af9af49e302;origin=https://github.com/HOLMS-lib/HOLMS;visit=swh:1:snp:2c0efd349323ed6f8067581cf1f6d95816e49841;anchor=swh:1:rev:1caf3be141c6f646f78695c0eb528ce3b753079a;path=/gen_completeness.ml;lines=130-135}{\ExternalLink}]\label{lem:APPR_CAR} \phantom{All}
    \begin{mydefinition}
      Let $\mathfrak{S}$ be a class of frames. Let $\mathcal{S}$  be a modal system. $\mathfrak{S}$ is the \textnormal{\textbf{appropriate  class for}} $\mathcal{S}$ if and only if $\mathfrak{S}$ is the characteristic class for $\mathcal{S}$ restricted to finite frames.
    \end{mydefinition}
\end{lemma}

Correspondence lemmata~(\ref{lem:Correspondence}), together with the soundness of $\mathbb{K}$ w.r.t every class of frames~(\ref{thm:K_soundness}), and the characterisations of characteristic frames~(\ref{lem:CHAR_CAR}) and appropriate frames~(\ref{lem:APPR_CAR}) allow us to arrange all the correspondence results in the Table~\ref{tab:char-appr} below. In particular, this table will stress the characteristic properties and the characteristic/appropriate class for each normal logic under analysis. As usual, we will use \textit{gothic font} for the classes of appropriate frames, e.g. $\mathfrak{RF}$ is the intersection between the characteristic class of frames $\mathfrak{R}$ and the class of finite frames $\mathfrak{F}$.

\begin{table}[h]
    \centering
    \begin{tabular}{l | l | l | l | l}
        \toprule
        \textbf{Logic} & \textbf{Specific axioms}& \textbf{Char. Properties} & \textbf{Char. Cl.} & \textbf{App. Cl.} \\
        \midrule
        $\mathbb{K}$  & $\varnothing$  & None & $\mathfrak{All}$~\href{https://archive.softwareheritage.org/swh:1:cnt:d364de3158b6405c920115aa0c801af9af49e302;origin=https://github.com/HOLMS-lib/HOLMS;visit=swh:1:snp:2c0efd349323ed6f8067581cf1f6d95816e49841;anchor=swh:1:rev:1caf3be141c6f646f78695c0eb528ce3b753079a;path=/gen_completeness.ml;lines=15-16}{\ExternalLink} & $\mathfrak{F}$~\href{https://archive.softwareheritage.org/swh:1:cnt:d364de3158b6405c920115aa0c801af9af49e302;origin=https://github.com/HOLMS-lib/HOLMS;visit=swh:1:snp:2c0efd349323ed6f8067581cf1f6d95816e49841;anchor=swh:1:rev:1caf3be141c6f646f78695c0eb528ce3b753079a;path=/gen_completeness.ml;lines=24-25}{\ExternalLink}\\
        \hline
        $\mathbb{D}$ & $\mathcal{D} = \{ \mathbf{D} \}$  & Seriality & $\mathfrak{Ser}$ & $\mathfrak{SerF}$\\
        \hline
        $\mathbb{T}$  &  $\mathcal{T} = \{ \mathbf{T} \}$ & Reflexivity & $\mathfrak{R}$~\href{https://archive.softwareheritage.org/swh:1:cnt:1352cb294724b2251b9346657432700ecbed10d2;origin=https://github.com/HOLMS-lib/HOLMS;visit=swh:1:snp:2c0efd349323ed6f8067581cf1f6d95816e49841;anchor=swh:1:rev:1caf3be141c6f646f78695c0eb528ce3b753079a;path=/t_completeness.ml;lines=29-34}{\ExternalLink} & $\mathfrak{RF}$~\href{https://archive.softwareheritage.org/swh:1:cnt:1352cb294724b2251b9346657432700ecbed10d2;origin=https://github.com/HOLMS-lib/HOLMS;visit=swh:1:snp:2c0efd349323ed6f8067581cf1f6d95816e49841;anchor=swh:1:rev:1caf3be141c6f646f78695c0eb528ce3b753079a;path=/t_completeness.ml;lines=88-94}{\ExternalLink} \\
        \hline
        $\mathbb{K}\mathbf{4}$ & $\mathcal{K}4 = \{ \mathbf{4} \}$ & Transitivity & $\mathfrak{T}$~\href{https://archive.softwareheritage.org/swh:1:cnt:32869a1df338d91c4cd4a0cb9e73fb4f1be29991;origin=https://github.com/HOLMS-lib/HOLMS;visit=swh:1:snp:2c0efd349323ed6f8067581cf1f6d95816e49841;anchor=swh:1:rev:1caf3be141c6f646f78695c0eb528ce3b753079a;path=/k4_completeness.ml;lines=35-40}{\ExternalLink} & $\mathfrak{TF}$~\href{https://archive.softwareheritage.org/swh:1:cnt:32869a1df338d91c4cd4a0cb9e73fb4f1be29991;origin=https://github.com/HOLMS-lib/HOLMS;visit=swh:1:snp:2c0efd349323ed6f8067581cf1f6d95816e49841;anchor=swh:1:rev:1caf3be141c6f646f78695c0eb528ce3b753079a;path=/k4_completeness.ml;lines=96-102}{\ExternalLink}  \\
        \hline 
         $\mathbb{S}\mathbf{4}$ & $\mathcal{S}4 = \{ \mathbf{T}; \mathbf{4}\}$ & Reflexivity-Transitivity &$\mathfrak{R} \cap \mathfrak{T} $ & $\mathfrak{RF} \cap \mathfrak{TF}$\\
         \hline 
         $\mathbb{B}$ & $\mathcal{B} = \{ \mathbf{T}; \mathbf{B}\}$ &  Reflexivity-Symmetry & $\mathfrak{R} \cap \mathfrak{Sym} $ & $\mathfrak{RF} \cap \mathfrak{SymF}$  \\
         \hline
         $\mathbb{S}\mathbf{5}$ & $\mathcal{S}5 = \{ \mathbf{T}; \mathbf{5}\}$ & Reflexivity-Euclideanity & $\mathfrak{R} \cap \mathfrak{E} $ & $\mathfrak{RF} \cap \mathfrak{EF}$\\
        \hline
         $\mathbb{GL}$ &  $\mathcal{GL} = \{ \mathbf{GL}\}$ & Transitivty- & $\mathfrak{T} \cap \mathfrak{NT} $~\href{https://archive.softwareheritage.org/swh:1:cnt:52c9de12454731a9a291b06bc750c0d2d14e1fb4;origin=https://github.com/HOLMS-lib/HOLMS;visit=swh:1:snp:2c0efd349323ed6f8067581cf1f6d95816e49841;anchor=swh:1:rev:1caf3be141c6f646f78695c0eb528ce3b753079a;path=/gl_completeness.ml;lines=26-32}{\ExternalLink} & $\mathfrak{ITF}$~\href{https://archive.softwareheritage.org/swh:1:cnt:52c9de12454731a9a291b06bc750c0d2d14e1fb4;origin=https://github.com/HOLMS-lib/HOLMS;visit=swh:1:snp:2c0efd349323ed6f8067581cf1f6d95816e49841;anchor=swh:1:rev:1caf3be141c6f646f78695c0eb528ce3b753079a;path=/gl_completeness.ml;lines=141-148}{\ExternalLink} \\
         & & Converse well-foundness & &   \\
         \bottomrule
    \end{tabular}\\
    \footnotesize{The class of reflexive and euclidean frames coincides with the class of equivalent frames, by~\ref{itm:binary_4}$^{th}$  point of Lemma~\ref{lem:binary}.}\caption{Classes characteristic and appropriate to the modal cube}\label{tab:char-appr}
\end{table}

After proving that $\mathfrak{ITF}$ is appropriate to $\mathbb{GL}$, we should remark three important facts that are enlightened by the results displayed in this table. 

\newpage

\begin{remark}[: $\mathfrak{ITF}$ is appropriate to $\mathbb{GL}$]
While for any the other modal systems we report in the column ``Appropriate class'' its characteristic class restricted to finite frames, for $\mathbb{GL}$ we report $\mathfrak{ITF}$--the class of finite transitive irreflexive frames.
The characteristic class to $\mathbb{GL}$ is constituted instead of transitive converse well-founded frames. We should prove that the restriction of this last class to finite frames coincides with $\mathfrak{ITF}$.
\end{remark} 

\begin{lemma}[\textbf{Irreflexive and converse well-founded frames}~\href{https://archive.softwareheritage.org/swh:1:cnt:52c9de12454731a9a291b06bc750c0d2d14e1fb4;origin=https://github.com/HOLMS-lib/HOLMS;visit=swh:1:snp:2c0efd349323ed6f8067581cf1f6d95816e49841;anchor=swh:1:rev:1caf3be141c6f646f78695c0eb528ce3b753079a;path=/gl_completeness.ml;lines=169-178}{\ExternalLink}]\label{lem:ITF-transnt}\phantom{Allyouneed}
\begin{mydefinition}
    Let $\langle W,R\rangle$ be a \textbf{finite  transitive} frame ($\langle W,R\rangle \in \mathfrak{TF}$). \\ Then $\langle W,R\rangle$ is \textbf{converse  well-founded} iff $\langle W,R\rangle$ is \textbf{irreflexive}.
\end{mydefinition}
\end{lemma}

\begin{proof}
    \begin{itemize}
    \item[$\implies$] Follows from the \ref{itm:CWF_IR}$^{th}$ point of Theorem~\ref{lem:binary}, which shows that any converse well-founded relation is irreflexive.
    
        \item[$\impliedby$] Suppose that $R$ is irreflexive. Let $x_1, \dots, x_n$ a sequence of elemets of $W$ such that for all $i < n$, then $x_iRx_{i+1}$. 
        \begin{claim}
            If $i < j$ then $x_i \not = x_j$
            \begin{claimproof}
                NAA: Suppose that $x_i = x_j$. \\
                By $i < j$  and transitivity,$x_iRx_j$ contra irreflexivity. $\lightning$ \hfill $\lhd$
            \end{claimproof}
        \end{claim}
        RAA: Suppose that $R$ is not converse well-founded, i.e. there exists $X$ a nonempty subset of $W$ such that $\forall x \in X (\exists y \in X (xRy))$. Then it is easy to prove by induction that for each natural number $n$ there exists a sequence $x_1, \dots, x_n$ such that for all $i < n$, $x_iRx_{i+1}$. Since $\forall i < n (i<i+1)$ then all the $x_i$ in the sequence are distinct.
        Therefore for each $n\in \mathbb{N}$, we prove that $X$ has at least $n$ elements. But $X \subseteq W$, thus \underline{W is infinite}. A contradiction. $\lightning$
    \end{itemize}
\end{proof}
\begin{remark}[: Strengthen the theory, restricting the class of appropriate frames]
From the addition of new axioms, on the one hand,
we strengthen the theory under analysis, which proves more theorems; on the other hand, this theory will be sound for a more restricted class of frames. In other words, while the set of provable theorems expands, the range of models that validate our axioms becomes narrower.    
\end{remark}

\begin{remark}[: Expressivity Power]\label{rmk:expressivity_power}
As shown in Table~\ref{tab:characteristic}, all the axiom schemata under analysis can be characterised by a property of accessibility relations.
All these characteristic properties are expressible in \textit{first-order language}, except converse well-foundness which requires \textit{second-order language}.

Then, identifying these correspondences allows us to reason about modal theoremhood in many normal systems by analysing validity in frames characterised by classical first-order formulas, which we are more used to dealing with.

\end{remark}
By contrast, \textit{converse well-foundness} clarifies that the modal language has a major definability power than the first-order language. Transitive converse well-foundness is, indeed, a property expressible in modal language ($\exists A \in \mathbf{Form}_{\Box}(\mathcal{F}\vDash A \iff \mathcal{F}$ is transitive and converse well-founded). e.g.  $A= \mathbf{GL}$) and not in first-order language ($\neg \exists A \in \mathbf{FOL}(\mathcal{F}\vDash A \iff \mathcal{F}$ is transitive and converse well-founded)). It is also important to notice that modal logic is less expressive than second-order logic, it is indeed provable ~\cite[\S 4]{boolos1995logic} that both first-order language and modal language cannot express converse well-foundness.

\begin{remark}[: Proper inclusions in the modal cube]\label{rmk:proper_inclusions}
    Thanks to correspondence theory, it is easier to prove that the inclusions between modal systems in the diagram~\ref{fig:properinclusions} are proper. As observed in the previous Remark~\ref{rmk:expressivity_power}, dealing with inclusions of classes characterised by first-order properties is simpler than reasoning with classes of frames in which some modal schemata are valid.
    
    
    For example, it is easy to build up a reflexive, symmetric, non-euclidean frame (e.g. $\langle W = \{w,r,y\}; R= \{(w,w);(y,y);(x,x);(w,y);(y,w);(w,z);(z,w) \} \rangle$). Such a ``counterframe'' shows that from \underline{$\mathfrak{RF}\cap\mathfrak{SymF} \not \subseteq \mathfrak{RF}\cap\mathfrak{EF}$} follows that inclusion between the two systems $\mathbb{B} \subseteq \mathbb{S}\mathbf{5}$ , proved axiomatically, is proper, $\mathbb{S}\mathbf{5} \not \subseteq \mathbb{B}$. 

    By definition, the class appropriate to $\mathbb{S}$
    consists exactly of the finite frames in which the theorems of $\mathbb{S}$ are valid.

    By \textit{negatio ad absurdum}, we assume that $\mathbb{S}\mathbf{5} \subseteq \mathbb{B}$. But then, by the proved inclusion $ \mathbb{B} \subseteq \mathbb{S}\mathbf{5}$, the theorem proved by the two systems are the same, and then the finite frames in which these theorems are valid coincide: 
    \underline{$\mathfrak{RF}\cap\mathfrak{SymF} \not = \mathfrak{RF}\cap\mathfrak{EF}$}. A contradiction. $\lightning$
   
\end{remark}

\section{Adequacy Theorems}
Logical deductive systems are designed to establish which valid formulas can be stated within a given language.  Once logicians have developed such a deductive system, they must verify that it checks two fundamental requirements with respect to a given semantics:
\begin{itemize}
    \item \textbf{Soundness}, ensures that every theorem derived within the system is valid;
    \item \textbf{Completeness},  every valid formula can be proved within the system.
\end{itemize}
Demonstrating \textbf{adequacy} of formal systems-their soundness and completeness-is a central task in logical inquiry. In fact, the adequacy theorems establish a correspondence between the syntactic characterisation of logical consequence via $\vdash$ and its semantic characterisations via $\vDash$.
Soundness and completeness for modal systems w.r.t~relational frames have the following aspects:
\newpage
\begin{definition}[\textbf{Soundness w.r.t relational frames}]\phantom{Allyouneedisove}
    \begin{mydefinition}
    \begin{mydefinition2}
    A modal system $\mathbb{S}$ is \textbf{sound} with respect to a class of frames $\mathfrak{S}$ iff \\ for all $A \in \mathbf{Form}_{\Box}$,
    $\vdash_{\mathbb{S}} A \ \implies \ \mathfrak{S} \vDash A$
    \end{mydefinition2}
    \end{mydefinition}
\end{definition}

\begin{definition}[\textbf{Completeness w.r.t relational frames}]\phantom{Allyouneedisove}
    \begin{mydefinition}
    \begin{mydefinition2}
    A modal system $\mathbb{S}$ is \textbf{complete} with respect to a class of frames $\mathfrak{S}$ iff \\ for all $A \in \mathbf{Form}_{\Box}$,
    $\mathfrak{S} \vDash A \ \implies \ \vdash_{\mathbb{S}} A$
    \end{mydefinition2}
    \end{mydefinition}
\end{definition}

In particular, we want to prove the adequacy theorems for normal systems w.r.t. their appropriate classes of frames.

\subsection{Soundness}\label{sub:soundness}

We should observe that the identification of appropriate classes for our normal logics provides per sè a proof of soundness for each axiomatic calculus w.r.t~a characteristic class of finite relational frames. 

Similarly, it is possible to achieve a similar result for frames that are not necessarily finite and prove soundness w.r.t.~\textit{characteristic} frames. This proof is obtained by considering the fact that the characteristic classes are closed under the deduction rules of our axiomatic proof systems as stated by the characterisation of the characteristic predicate in Lemma~\ref{lem:CHAR_CAR}.

Soundness can then be stated and proven parametrically as follows:

\begin{theorem}[\textbf{Parametric soundness}~\href{https://archive.softwareheritage.org/swh:1:cnt:d364de3158b6405c920115aa0c801af9af49e302;origin=https://github.com/HOLMS-lib/HOLMS;visit=swh:1:snp:2c0efd349323ed6f8067581cf1f6d95816e49841;anchor=swh:1:rev:1caf3be141c6f646f78695c0eb528ce3b753079a;path=/gen_completeness.ml;lines=71-80}{\ExternalLink}]\label{thm:parametric_soundness}\phantom{Allyouneedisloveloveisallyouneed}
\begin{mydefinition}
     Let $\mathbb{S}$ be a normal modal system. Let $\mathcal{S}$ be the set of specific axiom schemata for  $\mathbb{S}$.  Let $\mathfrak{S}$ be the characteristic class of frames for $\mathbb{S}$. \\
     For every $A \in \mathbf{Form}_{\Box}$, if  $\mathcal{S} \vdash A $ then  $\ \mathfrak{S} \vDash A$.
\end{mydefinition}
\end{theorem}
\begin{proof}
    We assume that $\mathcal{S} \vdash A $.
    $\mathfrak{S}$ is the characteristic  class for $\mathbb{S}$, then by Lemma~\ref{lem:CHAR_CAR} $\mathfrak{S}$ consists exactly of the frames in which every theorem of $\mathbb{S}$ is valid. Then every frame $\mathcal{F}$ in $\mathfrak{S}$ proves $A$ and thus $\ \mathfrak{S} \vDash A$.
\end{proof}

From this general result, we obtain as corollaries the proofs of soundness for all the axiom systems formalisable by our deducibility relation. 

\begin{theorem}[\textbf{Soundness within the modal cube}]\phantom{All}
    \begin{mydefinition}
        \begin{enumerate}
            \item For every $A \in \mathbf{Form}_{\Box}$, if  $\mathcal{\varnothing} \vdash A $ then  $\ \mathfrak{All} \vDash A$.~\href{https://archive.softwareheritage.org/swh:1:cnt:ae230138ff15476c8dab9e32606bceca7168285b;origin=https://github.com/HOLMS-lib/HOLMS;visit=swh:1:snp:2c0efd349323ed6f8067581cf1f6d95816e49841;anchor=swh:1:rev:1caf3be141c6f646f78695c0eb528ce3b753079a;path=/k_completeness.ml;lines=21-25}{\ExternalLink}
            \item For every $A \in \mathbf{Form}_{\Box}$, if  $\mathcal{D} \vdash A $ then  $\ \mathfrak{Ser} \vDash A$.
            \item For every $A \in \mathbf{Form}_{\Box}$, if  $\mathcal{T} \vdash A $ then  $\ \mathfrak{R} \vDash A$.~\href{https://archive.softwareheritage.org/swh:1:cnt:1352cb294724b2251b9346657432700ecbed10d2;origin=https://github.com/HOLMS-lib/HOLMS;visit=swh:1:snp:2c0efd349323ed6f8067581cf1f6d95816e49841;anchor=swh:1:rev:1caf3be141c6f646f78695c0eb528ce3b753079a;path=/t_completeness.ml;lines=78-82}{\ExternalLink}
            \item For every $A \in \mathbf{Form}_{\Box}$, if  $\mathcal{K}4 \vdash A $ then  $\ \mathfrak{T} \vDash A$.~\href{https://archive.softwareheritage.org/swh:1:cnt:32869a1df338d91c4cd4a0cb9e73fb4f1be29991;origin=https://github.com/HOLMS-lib/HOLMS;visit=swh:1:snp:2c0efd349323ed6f8067581cf1f6d95816e49841;anchor=swh:1:rev:1caf3be141c6f646f78695c0eb528ce3b753079a;path=/k4_completeness.ml;lines=86-90}{\ExternalLink}
            \item For every $A \in \mathbf{Form}_{\Box}$, if  $\mathcal{S}4 \vdash A $ then  $\ \mathfrak{R}\cap\mathfrak{T} \vDash A$.
            \item For every $A \in \mathbf{Form}_{\Box}$, if  $\mathcal{B} \vdash A $ then  $\ \mathfrak{R}\cap\mathfrak{Sym} \vDash A$.
            \item For every $A \in \mathbf{Form}_{\Box}$, if  $\mathcal{S}5 \vdash A $ then  $\ \mathfrak{R}\cap\mathfrak{E} \vDash A$.
            \item For every $A \in \mathbf{Form}_{\Box}$, if  $\mathcal{GL} \vdash A $ then  $\ \mathfrak{NT} \vDash A$.~\href{https://archive.softwareheritage.org/swh:1:cnt:52c9de12454731a9a291b06bc750c0d2d14e1fb4;origin=https://github.com/HOLMS-lib/HOLMS;visit=swh:1:snp:2c0efd349323ed6f8067581cf1f6d95816e49841;anchor=swh:1:rev:1caf3be141c6f646f78695c0eb528ce3b753079a;path=/gl_completeness.ml;lines=131-135}{\ExternalLink}
        \end{enumerate}
    \end{mydefinition}
\end{theorem}
\begin{proof}
    The first point is a direct consequence of Lemma~\ref{thm:K_soundness}. The other theorems, instead, are obtained as corollaries of Theorem~\ref{thm:parametric_soundness}. by instantiating the parameters in $\mathcal{S}$ and $\mathfrak{S}$ following Table~\ref{tab:characteristic}, which states the characteristic class of each system. \\ \phantom{A}
\end{proof}

\subsection{Consistency}\label{sub:consisency}

Having proved that every system $\mathbb{S}$ in the cube is sound w.r.t.~a non-empty class of frames $\mathfrak{S}$ \footnote{In particular, we have proved these results for characteristic classes.}, we can prove consistency of $\mathbb{S}$ with a classical strategy. 

Note that \textit{consistency} is a minimum requirement for any logical deductive system.

\begin{theorem}[\textbf{Consistency within the modal cube}]\phantom{All}
    \begin{mydefinition}
        \begin{multicols}{4}
        \begin{enumerate}
           \item $\mathcal{\varnothing} \not \vdash \bot $~\href{https://archive.softwareheritage.org/swh:1:cnt:ae230138ff15476c8dab9e32606bceca7168285b;origin=https://github.com/HOLMS-lib/HOLMS;visit=swh:1:snp:2c0efd349323ed6f8067581cf1f6d95816e49841;anchor=swh:1:rev:1caf3be141c6f646f78695c0eb528ce3b753079a;path=/k_completeness.ml;lines=48-56}{\ExternalLink} ;
            \item $\mathcal{D} \not \vdash \bot $ ;
            \item $\mathcal{T} \not \vdash \bot $~\href{https://archive.softwareheritage.org/swh:1:cnt:1352cb294724b2251b9346657432700ecbed10d2;origin=https://github.com/HOLMS-lib/HOLMS;visit=swh:1:snp:2c0efd349323ed6f8067581cf1f6d95816e49841;anchor=swh:1:rev:1caf3be141c6f646f78695c0eb528ce3b753079a;path=/t_completeness.ml;lines=139-145}{\ExternalLink} ;
            \item $\mathcal{K}4 \not \vdash \bot $~\href{https://archive.softwareheritage.org/swh:1:cnt:32869a1df338d91c4cd4a0cb9e73fb4f1be29991;origin=https://github.com/HOLMS-lib/HOLMS;visit=swh:1:snp:2c0efd349323ed6f8067581cf1f6d95816e49841;anchor=swh:1:rev:1caf3be141c6f646f78695c0eb528ce3b753079a;path=/k4_completeness.ml;lines=147-153}{\ExternalLink} ;
            \item $\mathcal{S}4 \not \vdash \bot $ ;
            \item $\mathcal{B} \not \vdash \bot $ ;
            \item $\mathcal{S}5 \not \vdash \bot $ ;
            \item $\mathcal{GL} \not \vdash \bot$~\href{https://archive.softwareheritage.org/swh:1:cnt:52c9de12454731a9a291b06bc750c0d2d14e1fb4;origin=https://github.com/HOLMS-lib/HOLMS;visit=swh:1:snp:2c0efd349323ed6f8067581cf1f6d95816e49841;anchor=swh:1:rev:1caf3be141c6f646f78695c0eb528ce3b753079a;path=/gl_completeness.ml;lines=216-222}{\ExternalLink} ;
        \end{enumerate}
    \end{multicols}
    \end{mydefinition}
\end{theorem}
\begin{proof}
   Each point is proved by \textit{negatio ad absurdum}, by following the partially-parametrised strategy below:
   \begin{itemize}
       \item NAA:  We assume that $\mathcal{S} \vdash \bot$;
       \item By soundness w.r.t $\mathfrak{S}$, $\mathfrak{S} \vDash \bot$;
       \item By definitions of \textit{validity in a class of frames} and \textit{in a frame}: \\ $\forall \mathcal{F} \in \mathfrak{S} (  \forall \mathcal{M} \ \ \mathcal{F}_* \ ( \forall w \in W ( \mathcal{M}, w \vDash \bot)))$; 
       \item \textbf{If $\mathfrak{S}$ is nonempty}, then such a $w$ exists: \underline{$\mathcal{M}, w \vDash \bot$};
       \item By definition of forcing, \underline{$\mathcal{M}, w \not \vDash \bot$ for any model $\mathcal{M} $ and world $w$};
       \item  A contradiction. $\lightning$
   \end{itemize}
   To employ this strategy we need to check that every class of characteristic frame is not empty, i.e. to provide an example of a  frame with the characteristic property under analysis,
   We consider the following two frames:
   \begin{center}
    \begin{tikzpicture}
        \draw (0,0) rectangle (2,2);
        \node at (1,2.5) {\(\mathcal{F}_1\)};
        \filldraw (1,1) circle (0.05) node[below] {\( w \)};
        
        \draw (3,0) rectangle (5,2);
        \node at (4,2.5) {\(\mathcal{F}_2\)};
        \filldraw (4,1) circle (0.07) node[below] {\( w \)};
        
        \draw[->] (4,1) .. controls (4.5,1.5) and (3.5,1.5) .. (4,1);
    \end{tikzpicture}
\end{center}
   
   \begin{enumerate}
       \item $\mathfrak{All} \not = \varnothing$ because $\mathcal{F}_1 $ is a frame;
       \item $\mathfrak{Ser} \not = \varnothing$ because $\mathcal{F}_2 $ is serial frame;
       \item $\mathfrak{R} \not = \varnothing$ because $\mathcal{F}_2 $ is a reflexive frame;
       \item $\mathfrak{T} \not = \varnothing$ because $\mathcal{F}_2 $ is a transitive frame;
       \item $\mathfrak{R \cap T} \not = \varnothing$ because $\mathcal{F}_2 $ is a reflexive-transitive frame;
       \item $\mathfrak{R \cap Sym} \not = \varnothing$ because $\mathcal{F}_2 $ is a reflexive-symmetric frame;
       \item $\mathfrak{R \cap E}\not = \varnothing$ because $\mathcal{F}_2 $ is a reflexive-euclidean frame ;
       \item $\mathfrak{ITF}\not = \varnothing$ because $\mathcal{F}_1 $ is an irreflexive-transitive frame ;
   \end{enumerate}
\end{proof}

\subsection{Completeness}\label{sub:completeness}
It is now left to prove the opposite--and much more complex to demonstrate--direction of soundness, the completeness theorems. 

As previously remarked, we present here the proof strategy used in HOLMS to prove the completeness of the normal systems under analysis w.r.t.~their characteristic classes of finite frames.

We will not present here the well-established canonical model method~\cite[\S~4]{DBLP:books/cu/BlackburnRV01}, even though it proves completeness for logics within the cube.
In fact, in HOLMS, we pursue a proof of completeness \textit{as parametric as possible}, whereas this approach cannot be directly applied to GL due to issues related to (in)compactness~\cite[\S~6]{boolos1995logic}.

As will be discussed in more detail in Chapter~\ref{chap:4}, the proof presented here follows the strategy outlined by George Boloos in his classical textbook~\cite[\S 5]{boolos1995logic}. Our demonstration avoids code duplication as far as possible and succeed in further clarifying the \textit{modularity} of Boolos' strategy and in explicitly parametrising it.

\begin{remark}
    We should note that, unlike what we did for soundness, we cannot prove a \textit{parametric} version of the completeness theorem that holds for every normal system. The reason for this will become clear later, as we present the proof sketch. Instead, we will show that each system under analysis is complete with respect to its appropriate class, following the \textit{most general} strategy possible. Such a results will be called \textit{ad-hoc polymorphic} in Chapter~\ref{chap:4} using a classical terminology.
\end{remark}

\subsubsection{Proof Sketch}\label{sub:completenes-proof-sketch}
Let $\mathbb{S}$ be a certain normal system within the cube, then $\mathcal{S}$ is the set of axiom schemata specific to $\mathbb{S}$. Let $\mathfrak{S}$ the set of frames appropriate to $\mathbb{S}$. Let $A$ be a modal formula.

We aim to prove:
\begin{center}
  \textit{Claim.} \textit{If $\mathfrak{S} \vDash A$ then $\mathcal{S} \vdash A$}
\end{center}

We proceed by contraposition; thus, we claim the following statment: 
\begin{center}
    \textit{Claim.}\label{clm:proof_sketch}
    \textit{If $\mathcal{S}. \varnothing \not \vdash A$ then $\mathfrak{S} \not \vDash A$}
\end{center}

By the definition of forcing and validity, this is equivalent to prove:
\begin{center}
    \textit{Claim.} \textit{If $\mathcal{S}. \varnothing \not \vdash A \ \ then \ \ \exists \mathcal{F} \in \mathfrak{S} ( \exists \mathcal{M} \ \mathcal{F}_* (  \exists w \in W ( \mathcal{M}, w \not \vDash A)))$}
\end{center}
   
This means that for each set of axioms $\mathcal{S}$ and for each modal formula $A$, we have to find a countermodel $\mathcal{M}^{\mathcal{S}}_{A}$ inhabiting $\mathfrak{S}$, and a `counterworld' $X^{\mathcal{S}}_{A}$ inhabiting $\mathcal{M}^{\mathcal{S}}_{A}$ such that $\mathcal{M}^{\mathcal{S}}_{A},X^{\mathcal{S}}_{A}\not \vDash A$. 
To do so, we formalise the argument in \cite[\S 5]{boolos1995logic} and implement the following key-strategy:
\begin{itemize}
    \item \textit{Fully parametrised part of the proof}:
\begin{enumerate}
    \item\label{itm:proof_sketch_1} We \textbf{identify a} parametric notion of \textbf{(counter)model} $\mathcal{M}^{\mathcal{S}}_{A}$ in $\mathfrak{S}$ having: 
    \begin{enumerate}
        \item $W^{\mathcal{S}}_{A}$: maximal consistent sets of modal subsenteces of $A$ as worlds;
        \item $R^{\mathcal{S}}_{A}$:  an accessibility relation that verifies two given constraints;
        \item $V^{\mathcal{S}}_{A}$: a valuation such that $V^{\mathcal{S}}_{A}(q,w)$ iff $q$ is a subformula of $ A$ and $q \in w$.
    \end{enumerate}
    Such a model is designed to verify the~\ref{itm:proof_sketch_2}$^{nd}$ point of this sketch.
    \smallskip
    \item \label{itm:proof_sketch_2} We \textbf{prove a \textit{general truth-lemma}} that holds for every countermodel verifying the requirements in~\ref{itm:proof_sketch_1} and independent from the considered $\mathcal{S}$ and $A$. 
    This step allows the reduction of the model-theoretic notion of \textit{forcing} ($ \mathcal{M}^{\mathcal{S}}_{A}, w \vDash A$) to the set-theoretic notion of \textit{membership} ($A \in w$). 
    \smallskip
     \item \label{itm:proof_sketch_3} We \textbf{identify a `counterworld'} $X^{\mathcal{S}}_{A}$ in $\mathcal{M}^{\mathcal{S}}_{A}$ such that $A \not \in X^{\mathcal{S}}_{A}$; thus $\mathcal{M}^{\mathcal{S}}_{A},X^{\mathcal{S}}_{A}\not \vDash A$  by the Lemma in~\ref{itm:proof_sketch_2}.
\end{enumerate}
\item \textit{Unparametrisable (Ad-hoc polymorphic) part of the proof}:
\begin{enumerate}
    \item[I.] Identification for each system $\mathbb{S}$ of its specific countermodel $\mathcal{M}^{\mathcal{S}}_{A}$, and in particular of its accessibility relation $Rel^{\mathcal{S}}_{A}$;
    \item[II.] Verification that this model satisfies the two requirements for $Rel^{\mathcal{S}}_{A}$.
\end{enumerate}
\end{itemize}

\subsubsection{Maximal Consistent Sets of Formulas}\label{sub:maximal-consistency}
As anticipated in the proof sketch, the universe of our countermodel will be composed of maximal consistent formulas. We are going to introduce the two predicates of \textit{consistency} and \textit{maximal consistency}, and some of their properties.

\begin{definition}[$\mathcal{S}$-\textbf{consistency}~\href{https://archive.softwareheritage.org/swh:1:cnt:e261ca08330d5ef66376347407c27fcece1e9f2e;origin=https://github.com/HOLMS-lib/HOLMS;visit=swh:1:snp:2c0efd349323ed6f8067581cf1f6d95816e49841;anchor=swh:1:rev:1caf3be141c6f646f78695c0eb528ce3b753079a;path=/consistent.ml;lines=11-12}{\ExternalLink}]\phantom{Allyouneedisloveloveisallyou}
\begin{mydefinition}
\begin{mydefinition}
    Let $\mathcal{S}$ be a set of axiom schemata. Let $\mathcal{X}$ be a finite set of modal formulas. \\
    $\mathcal{X}$ is $\mathcal{S}$-\textbf{consistent} iff $\mathcal{S} \not \vdash \neg \bigwedge \mathcal{X}$ \footnotemark.
\end{mydefinition}
\end{mydefinition}
\end{definition}
\footnotetext{Note that `$\bigwedge$' is not a primitive symbol of $\mathcal{L}_{\Box}$ but, as usually, it could be defined as the conjunction of all the formulas in $\mathcal{X}$.}

\begin{lemma}[\textbf{False implies not consistency}~\href{https://archive.softwareheritage.org/swh:1:cnt:e261ca08330d5ef66376347407c27fcece1e9f2e;origin=https://github.com/HOLMS-lib/HOLMS;visit=swh:1:snp:2c0efd349323ed6f8067581cf1f6d95816e49841;anchor=swh:1:rev:1caf3be141c6f646f78695c0eb528ce3b753079a;path=/consistent.ml;lines=73-75}{\ExternalLink}]\label{lem:consistency_lemma}\phantom{Allyouneedisloveloveisall}
\begin{mydefinition}
    Let $\mathcal{S}$ be a set of axiom schemata, Let $\mathcal{X}$ be a set of formulas. \\
    If $\bot \in \mathcal{X}$ then $\mathcal{X}$ is not $\mathcal{S}-$consistent.
\end{mydefinition}
\end{lemma}
\begin{proof}
    We assume that $\bot \in \mathcal{X}$. \\
    Note that, by definition of big conjunction, holds $\bigwedge \mathcal{X}=\bigwedge (\mathcal{X} \setminus \{ \bot \}) \land \bot$. \\
    We claim for $\mathcal{X}$ $\mathcal{S}-$consistency, i.e.  $\mathcal{S} \vdash \neg \bigwedge \mathcal{X}$. This is equivalent to prove $\mathcal{S} \vdash \neg (\bigwedge (\mathcal{X} \setminus \{ \bot \}) \land \bot) $ that--by propositional tautologies-- is equivalent to prove $\mathcal{S}  \vdash \neg\bigwedge (\mathcal{X} \setminus \{ \bot \}) \lor \neg \bot $. But, also by tautologies, $\mathcal{S} \vdash \neg \bot$ and then $\mathcal{S}  \vdash \neg\bigwedge (\mathcal{X} \setminus \{ \bot \}) \lor \neg \bot $, that prove our claim.    
\end{proof}
\newpage

\begin{definition}[$\mathcal{S},A$-\textbf{maximal consistency}~\href{https://archive.softwareheritage.org/swh:1:cnt:e261ca08330d5ef66376347407c27fcece1e9f2e;origin=https://github.com/HOLMS-lib/HOLMS;visit=swh:1:snp:2c0efd349323ed6f8067581cf1f6d95816e49841;anchor=swh:1:rev:1caf3be141c6f646f78695c0eb528ce3b753079a;path=/consistent.ml;lines=82-85}{\ExternalLink}]\phantom{Allyouneedisloveloveisallyou}
\begin{mydefinition}
\begin{mydefinition}
    Let $\mathcal{S}$ be a set of axiom schemata. Let $\mathcal{A}$ be a modal formula. Let $\mathcal{X}$ be a set of modal formulas. 
    $\mathcal{X}$ is $\mathcal{S},A$-\textbf{maximal consistent} iff 
    \begin{itemize}
        \item $\mathcal{X}$ is $\mathcal{S}$-consistent; and
        \item For all modal formulas $B$ that are subformulas of $A$, $B \in \mathcal{X}$ or $\neg B \in \mathcal{X}$.
    \end{itemize}
    We denote by $\textit{MAX}^A_\mathcal{S}$ the set of $\mathcal{S},A$-maximal consistent sets of formulas.
\end{mydefinition}
\end{mydefinition}
\end{definition}

\begin{lemma}[\textbf{Subformulas, negation and maximal consistency}~\href{https://archive.softwareheritage.org/swh:1:cnt:e261ca08330d5ef66376347407c27fcece1e9f2e;origin=https://github.com/HOLMS-lib/HOLMS;visit=swh:1:snp:2c0efd349323ed6f8067581cf1f6d95816e49841;anchor=swh:1:rev:1caf3be141c6f646f78695c0eb528ce3b753079a;path=/consistent.ml;lines=107-110}{\ExternalLink}]\label{lem:maximal_consistent_mem_not}\phantom{A}
\begin{mydefinition}
    Let $\mathcal{S}$ be a set of axiom schemata. Let $A,B$ be a modal formula. \\
    Let $\mathcal{X}$ be a $\mathcal{S},A$-maximal consistent set of formulas. \\
    If $B$ is a subformula of $A$ then $B \in \mathcal{X} \iff \neg B \not \in \mathcal{X}$. 
\end{mydefinition}
\end{lemma}
\begin{proof}
    We assume that $B$ is a subformula of $A$ and we should prove the two verses of the bi-implication.
    \begin{itemize}
        \item[$\implies$] By \textit{negatio ad absurdum}.\\
        We assume that both $B$ and $\neg B$ belong to $\mathcal{X}$. But then $\mathcal{S}$ proves the modal tautology $\neg (B \land \neg B)$; because $\mathcal{S}$ proves the negation of two conjuncts then $\mathcal{S} \vdash \neg \bigwedge\mathcal{X}$ , thus $\mathcal{X}$ is not $\mathcal{S}-$consistent and not $\mathcal{S},A-$maximal consistent. A contradiction. $\lightning$

        \item[$\impliedby$] We assume that \underline{$\neg B \not \in \mathcal{X}$}. \\ Because $B$ is a subformula of $A$ and $\mathcal{X}$ is $\mathcal{S},A-$maximal consistent, then $ B \in X$ or $\neg B \in \mathcal{X}$. \\
        RAA: $B \not \in \mathcal{X}$. Then, by the above line, holds \underline{$\neg B \in \mathcal{X}$}. A contradiction. $\lightning$
    \end{itemize}
\end{proof}
\medskip

\begin{lemma}[\textbf{Maximal consistent extension lemma}~\href{https://archive.softwareheritage.org/swh:1:cnt:e261ca08330d5ef66376347407c27fcece1e9f2e;origin=https://github.com/HOLMS-lib/HOLMS;visit=swh:1:snp:2c0efd349323ed6f8067581cf1f6d95816e49841;anchor=swh:1:rev:1caf3be141c6f646f78695c0eb528ce3b753079a;path=/consistent.ml;lines=166-226}{\ExternalLink}]\label{lem:extension}\phantom{Allyouneedislovelovei}
\begin{mydefinition}
    Let $\mathcal{S}$ be a set of axiom schemata. Let $A$ be a modal formula. \\
    Every $\mathcal{X}$ $\mathcal{S}$-consistent set of \textbf{subsentences} of $A$ can be extended to a $\mathcal{S},A$-maximal consistent set of subsentences of $A$. \\ Formally we write: 
    If $\mathcal{X}$ is $\mathcal{S}$-consistent and $\forall x \in\mathcal{X} (x \ \mathbf{subs}\ A) $ then \\  $ \exists \mathcal{Y} \subseteq \mathbf{Form}_{\Box}(\mathcal{X}\subseteq \mathcal{Y}$ and  $\mathcal{Y} \in \textit{MAX}^A_\mathcal{S} \ and \ \forall y \in \mathcal{Y}(y \ \mathbf{subs} \ A)$. 
\end{mydefinition}
\end{lemma}
\begin{proof} This proof is based on the standard Lindenbaum construction for consistent sets~\cite[pp.~15-16]{Tarski1968-TARUT}.

    Let $\mathcal{X}$ be a \underline{$\mathcal{S}$-consistent} set of \textbf{subsentences} of $A$. 
    Let $\mathcal{U}$ be the set of all the subformulas of $A$. Given the Fact~\ref{fct:finite_formulas} this set is finite and we can enumerate its elements $E_1, \dots, E_n$.

     We define a sequence of sets \(\mathcal{X}_n\) starting from \(\mathcal{X}\):

    \begin{itemize}
        \item Base case: Let \(\mathcal{X}_0 = \mathcal{X}\).
        \newpage
        \item Inductive step:  
        For each $i\le n$:
        \begin{itemize}
            \item If \(\mathcal{X}_i \cup \{E_i\}\) is \(\mathcal{S}\)-consistent, set \(\mathcal{X}_{i+1} = \mathcal{X}_i \cup \{E_i\}\).
            \item Otherwise, set \(\mathcal{X}_{i+1} = \mathcal{X}_i \cup \{\neg E_i\}\).
        \end{itemize}
    \end{itemize}

    We then define the maximal set as $\mathcal{Y} = \bigcup_{i \le n} \mathcal{X}_i=\mathcal{X}_n$.
    By construction:
    \begin{itemize}
    \item $\mathcal{X \subseteq Y}$;
        \item $\mathcal{Y} \in \textit{MAX}_{\mathcal{S}}^A$ beacuse:
        \begin{itemize}
            \item $\mathcal{S}$-consistent ($\mathcal{S} \not \vdash \neg \bigwedge  \mathcal{X}_n$ ), since each extension step preserves consistency;
        
            \item Maximal, because for every subformula \(E\) of \(A\), either \(E \in \mathcal{Y}\) or \(\neg E \in \mathcal{Y}\).
        \end{itemize}
      \item $\forall y \in \mathcal{Y}(y \in \mathcal{X} \ or \ \exists E \in  \mathcal{U}(y=E \ or \  y = \neg E))$  then $\forall y \in \mathcal{Y}(y \ \mathbf{subs} \ A)$  
    \end{itemize}
\end{proof}

\begin{lemma}[\textbf{Nonempty maximal consistent lemma}~\href{https://archive.softwareheritage.org/swh:1:cnt:e261ca08330d5ef66376347407c27fcece1e9f2e;origin=https://github.com/HOLMS-lib/HOLMS;visit=swh:1:snp:2c0efd349323ed6f8067581cf1f6d95816e49841;anchor=swh:1:rev:1caf3be141c6f646f78695c0eb528ce3b753079a;path=/consistent.ml;lines=228-245}{\ExternalLink}]\label{lem:nonempty_max_cons}\phantom{Allyounee}
\begin{mydefinition}
    Let $\mathcal{S}$ be a set of axiom schemata. Let $A$ be a modal formula.
    If $\mathcal{S} \not \vdash A$ then $\{\neg A\}$ is $\mathcal{S}-$consistent and $\exists \mathcal{Y} \subseteq \mathbf{Form}_{\Box}(\neg A \in \mathcal{Y}$ and  $\mathcal{Y} \in \textit{MAX}^A_\mathcal{S} \ and \ \forall y \in \mathcal{Y}(y \ \mathbf{subs} \ A)$. 
\end{mydefinition}
\end{lemma}
\begin{proof}
    We assume that $\mathcal{S} \not \vdash A$.
    \begin{enumerate}
        \item $\{ \neg A\}$ is $\mathcal{S}-$consistent iff $\mathcal{S} \not \vdash \neg \bigwedge \{\neg A\}$  and that hold by hypothesis and modal tautologies $(\mathcal{S} \vdash \neg \bigwedge \{\neg A\} \leftrightarrow A)$ .
        \item  Because  $\{ \neg A\}$ is $\mathcal{S}-$consistent and because $\neg A \ \mathbf{subs} A$, we obtain the thesis by applying Lemma~\ref{lem:extension}.
    \end{enumerate}
\end{proof}

\begin{lemma}[\textbf{Maximal consistent lemma}~\href{https://archive.softwareheritage.org/swh:1:cnt:e261ca08330d5ef66376347407c27fcece1e9f2e;origin=https://github.com/HOLMS-lib/HOLMS;visit=swh:1:snp:2c0efd349323ed6f8067581cf1f6d95816e49841;anchor=swh:1:rev:1caf3be141c6f646f78695c0eb528ce3b753079a;path=/consistent.ml;lines=87-105}{\ExternalLink}]\label{lem:max_cons}\phantom{Allyounee}
\begin{mydefinition}
    Let $\mathcal{S}$ be a set of axiom schemata. Let $A,B$ be modal formulas such that $B \ \mathbf{sub} \ A$. Let $\mathcal{X}$ be a $\mathcal{S},A-$maximal consistent set of formulas such that $A_1, \dots, A_n \in \mathcal{X}$. \\
    If $\mathcal{S} \vdash A_1 \land \dots \land A_n \to B$ then $B \in \mathcal{X}$.
\end{mydefinition}
\end{lemma}
\begin{proof}
   Let us proceed by\textit{ reductio ad absurdum}. RAA: We assume $B \not \in \mathcal{X}.$ \\
   By Lemma~\ref{lem:maximal_consistent_mem_not}, $\neg B \in \mathcal{X}$.
   We also know that $\mathcal{S} \vdash A_1 \land \dots \land A_n \to B$,then  by propositional tautologies follows $\mathcal{S} \vdash \neg (A_1 \land \dots \land A_n \land 
   \neg B)$, this means that $\mathcal{S} \vdash \neg \bigwedge \mathcal{X}$; thus $\mathcal{X}$ is not $\mathcal{S},A-$maximal consistent. A contradiction. $\lightning$
   
\end{proof}

\subsubsection{Fully Parametrised Part of the Proof}\label{sub:parametrised-completeness}
We now have all the necessary tools to illustrate the parametric part of the completeness proof and hence we proceed to identify the central parametric notions we are interested in:
\smallskip
\begin{enumerate}
     \item\label{itm:first_parametric} \textbf{Parametric Countermodel}.~\href{https://archive.softwareheritage.org/swh:1:cnt:d364de3158b6405c920115aa0c801af9af49e302;origin=https://github.com/HOLMS-lib/HOLMS;visit=swh:1:snp:2c0efd349323ed6f8067581cf1f6d95816e49841;anchor=swh:1:rev:1caf3be141c6f646f78695c0eb528ce3b753079a;path=/gen_completeness.ml;lines=180-183}{\ExternalLink} Let $\mathcal{M}_{\mathcal{S}}^{A}:=\langle W^A_\mathcal{S},{{R}_{\mathcal{S}}^{A}},{V_{\mathcal{S}}^{A}}\rangle$, where:
    \begin{enumerate}
        \item $\mathbf{W_{\mathcal{S}}^{A}}$, the \textbf{universe}, is the subset of $\textit{MAX}_{\mathcal{S}}^{A} $ made of subsentences of $A$ \\
        $W_{\mathcal{S}}^{A} \coloneq \{\mathcal{X} \in \textnormal{MAX}_{\mathcal{S}}^A \ | \ \forall B \in \mathbf{Form}_{\Box} (B \in \mathcal{X} \implies B \ \mathbf{subs} \ A) \}$.
        
        \begin{remark}
         To say that $\mathcal{M}_{\mathcal{S}}^A$ is a model, we must prove that its universe $W_{\mathcal{S}}^A$ is not empty. Nonempty maximal consistent lemma~(\ref{lem:nonempty_max_cons}) states that if $\mathcal{S} \not \vdash A$ then $\exists \mathcal{Y} \subseteq \mathbf{Form}_{\Box}(\neg A \in \mathcal{Y}$ and  $\mathcal{Y} \in \textit{MAX}^A_\mathcal{S} \ and \ \forall y \in \mathcal{Y}(y \ \mathbf{subs} \ A)$. It is evident that such an $\mathcal{Y}$ belongs to $W_{\mathcal{S}}^A$.
         \end{remark}

        \item $\mathbf{R}_{\mathcal{S}}^A$ is an \textbf{accessibility relation} such that:   
        \begin{enumerate}
            \item\label{itm:c1} $\langle W^A_\mathcal{S}, R_{\mathcal{S}}^{A} \rangle \in \mathfrak{S}$;
            \item\label{itm:c2} For all $B \in \mathbf{Form}_{\Box}$, if $\Box B \ \mathbf{sub} \ A$ then, for every $w\in W^A_\mathcal{S}$:
            \begin{center}
                $\Box B \in w$ $ \ \iff \ $  $ \forall x\in W^A_\mathcal{S} (wR_{\mathcal{S}}^{A}x \implies B \in x)$.
                \end{center}
        \end{enumerate}
        
        \item\label{itm:fully_par-V} $\mathbf{V_{\mathcal{S}}^{A}}$ is a \textbf{valuation} over $W^A_\mathcal{S}$ such that:
        \begin{center}
                    $aV_{\mathcal{S}}^{A}w \ \iff \ a \ \mathbf{sub} \ A$ and $a \in w$.            
        \end{center}
    \end{enumerate}
    \begin{remark}
        We should note that, even if this notion of countermodel is fully \textit{parametrised}, we do not provide a \textit{parametric} accessibility relation ${R_{\mathcal{S}}^{A}}$, as we do for the universe ${W_{\mathcal{S}}^{A}}$ and the valuation ${V_{\mathcal{S}}^{A}}$. Then for each modal system we must identify a standard relation ${Rel_{\mathcal{S}}^{A}}$ verifying the two contraints~\ref{itm:c1} and~\ref{itm:c2}.
    \end{remark}
    
    \item For such a model $\mathcal{M}_{\mathcal{S}}^{A}$ when can prove that the \textbf{\textit{general truth lemma}} holds for every subformula of $A$. Then we can reduce the concept of \textit{validity} of a subformula of $A$ to the easier-to-handle concept of \textit{membership}.

    \begin{lemma}[\textbf{Parametric truth lemma}~\href{https://archive.softwareheritage.org/swh:1:cnt:d364de3158b6405c920115aa0c801af9af49e302;origin=https://github.com/HOLMS-lib/HOLMS;visit=swh:1:snp:2c0efd349323ed6f8067581cf1f6d95816e49841;anchor=swh:1:rev:1caf3be141c6f646f78695c0eb528ce3b753079a;path=/gen_completeness.ml;lines=189-350}{\ExternalLink}]\label{lem:truth}\phantom{Allyouneedislovelove}
    \begin{mydefinition}
       Let  $\mathcal{M}_{\mathcal{S}}^{A}$ be a relational model verifying all the requests in~\ref{itm:first_parametric}.
        Let $A \in \mathbf{Form}_{\Box}$ such that $\mathcal{S} \not \vdash A$. Then for every $w \in W_{\mathcal{S}}^{A}$ and for every $B$ subformula of $A$ the following equivalence holds: \begin{center}
            $\mathcal{M}_{\mathcal{S}}^{A}, w \vDash B$ iff $B \in w$.            
        \end{center}
    \end{mydefinition}
    \end{lemma}
    \begin{proof}
        This lemma is proved by induction on the complexity of the modal subformula $B$. We take $w$ as a generic world in the universe of the countermodel $\mathcal{M}_{\mathcal{S}}^{A}$.
        \begin{enumerate}
            \item $B= \bot$
            \begin{itemize}
                \item[$\implies$] If $\mathcal{M}_{\mathcal{S}}^{A}, w \vDash \bot$ then $\bot \in w$. \\
                By contraposition: If $\bot \not \in w$ then $\mathcal{M}_{\mathcal{S}}^{A}, w \not \vDash \bot$ \\
                This implication holds because--by definition of forcing--for every model $\mathcal{M}$ and world $w$, $\mathcal{M},w \not \vDash \bot$.

                \item[$\impliedby$] If $\bot \in w $ then $ \mathcal{M}_{\mathcal{S}}^{A}, w \vDash \bot $ \\ By contraposition: If $\mathcal{M}_{\mathcal{S}}^{A}, w \not \vDash \bot $ then $\bot \not \in w $. \\
                NAA: We assume $\bot \in w $ but then--by the Lemma~\ref{lem:consistency_lemma} on consistent sets--$w$ is not $\mathcal{S}-$consistent then \underline{$w \not \in W_{\mathcal{S}}^A$}. A contradiction. $\lightning$
            \end{itemize}
        \hfill $\lhd $   
        \item $B = p  \in \Phi $, $p$ is hence an atom formula such that $p \ \mathbf{sub} A\ $. \\
        We recall here the definition of the valuation relation in our countermodel: $aV_{\mathcal{S}}^{A}w \ \iff \ a \ \mathbf{sub} \ A$ and $a \in w$. \\
        Then follows $p \in w \ $ iff  $ \ pV_{\mathcal{S}}^{A}w \ $ iff $ \ \mathcal{M}_{\mathcal{S}}^{A}, w \vDash p$.
        \hfill $\lhd$    \\

        \item $B = C  \to D$ \\
            Note that $C, D \ \mathbf{sub} \ B$ and then, by transitivity of subformulas, $C, D \ \mathbf{sub} \ A$.
             Then the following inductive hypotheses hold:
            \begin{itemize}
            \item $\forall w \in W_{\mathcal{S}}^A(C \in w \ $ iff  $ \ \mathcal{M}_{\mathcal{S}}^{A}, w \vDash C)$;
            \item $\forall w \in W_{\mathcal{S}}^A(D \in w \ $ iff  $ \ \mathcal{M}_{\mathcal{S}}^{A}, w \vDash D)$.
        \end{itemize}
        We consider here some modal tautologies: 
        \begin{enumerate}
            \item\label{itm:truth_imp_1} $\mathcal{S} \vdash \neg B \to C$;
            \item\label{itm:truth_imp_2} $\mathcal{S} \vdash \neg B \to \neg  D$;
            \item\label{itm:truth_imp_3} $\mathcal{S} \vdash D \to (\neg C \to \neg B$);
        \end{enumerate}
        Then the following equivalences hold: \\
        \underline{$B \in w$} 
        $\xLeftrightarrow{\text{Lemma}~\ref{lem:maximal_consistent_mem_not}}$ 
        $\neg B \in w$
        $\xRightarrow{\text{MaxCons Lemma}~\ref{lem:max_cons} \text{ on taut} \ref{itm:truth_imp_1} \text{ and } \ref{itm:truth_imp_2}}$ \\
        $\xLeftarrow{\text{MaxCons Lemma}~\ref{lem:max_cons} \text{ on taut} \ref{itm:truth_imp_3}}$ 
        $C \in w$ and $\neg D \in w$ $\xLeftrightarrow{\text{Lemma}~\ref{lem:maximal_consistent_mem_not}}$ 
        $C \in w$ and $D \not \in w$ 
        $\xLeftrightarrow{\text{Inductive HP}}$
        $\mathcal{M}_{\mathcal{S}}^{A}, w \vDash C$ and $\mathcal{M}_{\mathcal{S}}^{A}, w \not \vDash D$
        $\xLeftrightarrow{\text{Forcing}}$
        $\mathcal{M}_{\mathcal{S}}^{A}, w \vDash C \to D$
        $\iff$ \underline{$\mathcal{M}_{\mathcal{S}}^{A}, w \vDash B$}
        \hfill $\lhd$
        \\
        \item $B = \Box C$ \\
        Note that $C \ \mathbf{sub} \ B$ and then, by transitivity of subformulas, $C \ \mathbf{sub} \ A$.
        \\ Then the following inductive hypothesis holds:
        \begin{center}
            $\forall w \in W_{\mathcal{S}}^A(C \in w \ $ iff  $ \ \mathcal{M}_{\mathcal{S}}^{A}, w \vDash C)$.
        \end{center}

        We recall here the second request for our accessibility relation $R_{\mathcal{S}}^{A}$~\ref{itm:c2}:
        \begin{center}
        For all $B \in \mathbf{Form}_{\Box}$, if $\Box B \ \mathbf{sub} \ A$ then, for every $w\in W^A_\mathcal{S}$:
            
                $\Box B \in w$ $ \ \iff \ $  $ \forall x\in W^A_\mathcal{S} (xR_{\mathcal{S}}^{A}w \implies B \in x)$.
        \end{center}
        Then the following equivalences hold:\\
        \underline{$B \in w$} $\iff$ $\Box C \in w$ $\xLeftrightarrow{Request~\ref{itm:c2}}$
        $\forall x\in W^A_\mathcal{S} (xR_{\mathcal{S}}^{A}w \implies C \in x)$
        $\xLeftrightarrow{Inductive \ HP}$
         $\forall x\in W^A_\mathcal{S} (xR_{\mathcal{S}}^{A}w \implies \mathcal{M}_{\mathcal{S}}^{A},x \vDash C )$ 
          $\xLeftrightarrow{forcing \ relation}$ 
        $\mathcal{M}_{\mathcal{S}}^{A}, w \vDash \Box C$
        $\iff$ \underline{$\mathcal{M}_{\mathcal{S}}^{A}, w \vDash B$}
        \hfill $\lhd$ \\
\newpage
        \begin{remark}
        \end{remark}
            This is the exact point where the second request for accessibility relation~\ref{itm:c2} proves to be essential.
        \end{enumerate}
    \end{proof}
\medskip \medskip    
\item Thanks to the precedent lemmas, we prove a \textbf{\textit{countermodel lemma}}, which states the last claim in the proof sketch in Section~\ref{clm:proof_sketch}. Furthermore, having found a countermodel as in~\ref{itm:first_parametric} for the system $\mathbb{S}$ under analysis, completeness for $\mathbb{S}$ follows as a corollary of the countermodel lemma.
\newpage
\begin{lemma}[\textbf{Parametric countermodel lemma}~\href{https://archive.softwareheritage.org/swh:1:cnt:d364de3158b6405c920115aa0c801af9af49e302;origin=https://github.com/HOLMS-lib/HOLMS;visit=swh:1:snp:2c0efd349323ed6f8067581cf1f6d95816e49841;anchor=swh:1:rev:1caf3be141c6f646f78695c0eb528ce3b753079a;path=/gen_completeness.ml;lines=574-593}{\ExternalLink}]\label{lem:parametric_countermodel_lemma}\phantom{Allyouneedislove}
\begin{mydefinition}
    Let $W_{\mathcal{S}}^{A}$ be the universe defined in~\ref{itm:first_parametric}. Let the frame $\langle W_{\mathcal{S}}^{A}, {R_{\mathcal{S}}^{A}} \rangle$ satisfy requirements~\ref{itm:c1} and~\ref{itm:c2}.
    Let $A$ a modal formula. \\
    \begin{center}
        If $\mathcal{S} \not \vdash A$ then ${\langle W_{\mathcal{S}}^{A}, R_{\mathcal{S}}^{A} \rangle}\not \vDash A$.
    \end{center}
\end{mydefinition}
\end{lemma}
\begin{proof}
    Given these hypotheses, we can use $\langle W_{\mathcal{S}}^{A}, R_{\mathcal{S}}^{A}, V_{\mathcal{S}}^{A}  \rangle$ as countermodel and the maximal extension of $\{\neg A \}$ as the desired counterworld. Given $\mathcal{S} \not \vdash A$ then $\{\neg A\}$ is $\mathcal{S}$-consistent; thus by Lemma~\ref{lem:extension} we build $X_{\mathcal{S}}^{A}$ such that $X_{\mathcal{S}}^{A} \in W_{\mathcal{S}}^{A}$ and $\neg A \in X_{\mathcal{S}}^{A}$. Then, by $\mathcal{S}$-consistency, $ A \not \in X_{\mathcal{S}}^{A}$. Finally, by applying Lemma~\ref{lem:truth}, the thesis holds: $\langle W_{\mathcal{S}}^{A}, R_{\mathcal{S}}^{A}, V_{\mathcal{S}}^{A}  \rangle , X_{\mathcal{S}}^{A} \not \vDash A$.
\end{proof}    
\end{enumerate}

\subsubsection{Unparametrisable (Ad-hoc polymorphic) Part of the Proof} \label{sub:unparametrisable-completeness}
As previously remarked, we need now to define for each normal logic a specific accessibility relation on maximal consistent sets verifying the requirements of $\mathcal{M}^A_\mathcal{S}$. 

For each modal system $\mathbb{S}$ in the cube, we then proceed as follows:
\begin{enumerate}
    \item[I.] We \textbf{define the standard accessibility relation} $\mathbf{Rel_{\mathcal{S}}^A}$ specific to each $\mathbb{S} \in\{\mathbb{K;\ D; \ T;}$ $\mathbb{ \ K\textnormal{4}; \ S\textnormal{4};\ B; \ S\textnormal{5};\ GL}\}$. These definitions are listed in Table~\ref{tab:standard_accessibility_reation}.

    \begin{table}[h]
    \centering
    \renewcommand{\arraystretch}{1.3} 
    \begin{tabular}{c|l}

        \textbf{System} & \textbf{Standard Accessibility Relation} \\
        \hline
        $\mathbb{K}$  & $Rel_{\varnothing}^A \coloneq Rel_{\varnothing}^A (w,x)$\footnotemark$ \iff \forall \ \Box B \in w (B \in x)$~\href{https://archive.softwareheritage.org/swh:1:cnt:ae230138ff15476c8dab9e32606bceca7168285b;origin=https://github.com/HOLMS-lib/HOLMS;visit=swh:1:snp:2c0efd349323ed6f8067581cf1f6d95816e49841;anchor=swh:1:rev:1caf3be141c6f646f78695c0eb528ce3b753079a;path=/k_completeness.ml;lines=109-110}{\ExternalLink} \\ \hline
        $\mathbb{D}$  & $Rel_{\mathcal{D}}^A \coloneq Rel_{\mathcal{D}}^A(w,x) \iff \forall \ \Box B \in w (B \in x)$  \\ \hline
        $\mathbb{T}$  & $Rel_{\mathcal{T}}^A \coloneq Rel_{\mathcal{T}}^A (w,x) \iff \forall \ \Box B \in w (B \in x)$~\href{https://archive.softwareheritage.org/swh:1:cnt:1352cb294724b2251b9346657432700ecbed10d2;origin=https://github.com/HOLMS-lib/HOLMS;visit=swh:1:snp:2c0efd349323ed6f8067581cf1f6d95816e49841;anchor=swh:1:rev:1caf3be141c6f646f78695c0eb528ce3b753079a;path=/t_completeness.ml;lines=207-209}{\ExternalLink} \\ \hline
        $\mathbb{K}4$ &  $Rel_{\mathcal{K}4}^A \coloneq Rel_{\mathcal{K}4}^A(w,x) \iff \forall \ \Box B \in w (B \in x \ and \ \Box B \in x)$~\href{https://archive.softwareheritage.org/swh:1:cnt:32869a1df338d91c4cd4a0cb9e73fb4f1be29991;origin=https://github.com/HOLMS-lib/HOLMS;visit=swh:1:snp:2c0efd349323ed6f8067581cf1f6d95816e49841;anchor=swh:1:rev:1caf3be141c6f646f78695c0eb528ce3b753079a;path=/k4_completeness.ml;lines=215-218}{\ExternalLink} \\ \hline
        $\mathbb{S}4$ & $Rel_{\mathcal{S}4}^A \coloneq Rel_{\mathcal{S}4}^A(w,x) \iff \forall \ \Box B \in w (B \in x \ and \ \Box B \in x)$  \\ \hline
        $\mathbb{B}$  & $Rel_{\mathcal{B}}^A \coloneq Rel_{\mathcal{B}}^A (w,x) \iff \forall \ \Box B \in w (B \in x) \ and \ \forall \Box B \in x (B \in w)$ \\ \hline
        $\mathbb{S}5$ & $Rel_{\mathcal{S}5}^A \coloneq Rel_{\mathcal{S}5}^A (w,x) \iff \forall \ \Box B (\Box B \in w \iff \Box B \in x)$ \\ \hline
        $\mathbb{GL}$ & $Rel_{\mathcal{GL}}^A \coloneq Rel_{\mathcal{GL}}^A (w,x) \iff \forall \ \Box B \in w (B \in x \ and \ \Box B \in x)\ $~\href{https://archive.softwareheritage.org/swh:1:cnt:52c9de12454731a9a291b06bc750c0d2d14e1fb4;origin=https://github.com/HOLMS-lib/HOLMS;visit=swh:1:snp:2c0efd349323ed6f8067581cf1f6d95816e49841;anchor=swh:1:rev:1caf3be141c6f646f78695c0eb528ce3b753079a;path=/gl_completeness.ml;lines=298-302}{\ExternalLink} \\
        &  $\ \ \ \ \ \ \ \ \ \ \ \ \ \ \ \ \ \ \ \ \ \ \ \ 
 \ \ \ \ \ \ \ \ \ and \ \exists \Box E \in x (\neg \Box E \in w) $ \\ \hline
    \end{tabular}
    \caption{Normal systems in the cube and their standard accessibility relations}
    \label{tab:standard_accessibility_reation}
\end{table}
\footnotetext{We adopt here the functional writing $R(w,x)$ instead than the relational $wRx$ to improve readability.}

    \item[II.] We prove that $\langle W_*^A, Rel_*^A, V_*^A \rangle$ verifies the two constraints for $\mathcal{M}_*^A$:
    \begin{enumerate}
       \item[1(b).i] \textit{\textbf{Appropriate standard frame lemma}}~(\ref{lem:appr_std_frame}) guarantees that--for each modal system $\mathbb{S}$ within the cube--the first requirement~\ref{itm:c1} holds on the standard frame $\langle W_{\mathcal{S}}^{A}, Rel_{\mathcal{S}}^{A} \rangle$  defined with our standard accessibility relation:
       \begin{center}
           $\langle W_{\mathcal{S}}^{A}, Rel_{\mathcal{S}}^{A} \rangle \in \mathfrak{S}$
       \end{center}
      
       \item[1(b).ii] \textit{\textbf{Accessibility lemma}}~(\ref{lem:accessibility_std_frame}) 
       guarantees that--for each modal system $\mathbb{S}$ within the cube--the second requirement~\ref{itm:c2} holds on the standard frame $\langle W_{\mathcal{S}}^{A}, Rel_{\mathcal{S}}^{A} \rangle$ defined with our standard accessibility relation $Rel_{\mathcal{S}}^{A}$:
       \begin{center}
           For all $B \in \mathbf{Form}_{\Box}$, if $\Box B \ \mathbf{sub} \ A$ then, for every $w\in W^A_\mathcal{S}$:
            
                $\Box B \in w$ $ \ \iff \ $  $ \forall x\in W^A_\mathcal{S} (Rel_{\mathcal{S}}^{A}(w,x)\implies B \in x)$.
       \end{center}
       \end{enumerate}
\medskip

\begin{lemma}[\textbf{Appropriate standard frame lemma}]\label{lem:appr_std_frame}\phantom{Allyouneedislove}
\begin{mydefinition}
\begin{multicols}{2}
    \begin{itemize}
        \item[$\textnormal{1.}$] 
        $\langle W_{\mathcal{\varnothing}}^{A}, Rel_{\mathcal{\varnothing}}^{A} \rangle \in \mathfrak{F}$~\href{https://archive.softwareheritage.org/swh:1:cnt:ae230138ff15476c8dab9e32606bceca7168285b;origin=https://github.com/HOLMS-lib/HOLMS;visit=swh:1:snp:2c0efd349323ed6f8067581cf1f6d95816e49841;anchor=swh:1:rev:1caf3be141c6f646f78695c0eb528ce3b753079a;path=/k_completeness.ml;lines=119-127}{\ExternalLink}
        \item[$\textnormal{2.}$] 
        $\langle W_{\mathcal{D}}^{A}, Rel_{\mathcal{D}}^{A} \rangle \in \mathfrak{SerF}$
        \item[$\textnormal{3.}$]
        $\langle W_{\mathcal{T}}^{A}, Rel_{\mathcal{T}}^{A} \rangle \in \mathfrak{RF}$~\href{https://archive.softwareheritage.org/swh:1:cnt:1352cb294724b2251b9346657432700ecbed10d2;origin=https://github.com/HOLMS-lib/HOLMS;visit=swh:1:snp:2c0efd349323ed6f8067581cf1f6d95816e49841;anchor=swh:1:rev:1caf3be141c6f646f78695c0eb528ce3b753079a;path=/t_completeness.ml;lines=220-259}{\ExternalLink}
        \item[$\textnormal{4.}$] 
        $\langle W_{\mathcal{K}4}^{A}, Rel_{\mathcal{K}4}^{A} \rangle \in \mathfrak{TF}$~\href{https://archive.softwareheritage.org/swh:1:cnt:32869a1df338d91c4cd4a0cb9e73fb4f1be29991;origin=https://github.com/HOLMS-lib/HOLMS;visit=swh:1:snp:2c0efd349323ed6f8067581cf1f6d95816e49841;anchor=swh:1:rev:1caf3be141c6f646f78695c0eb528ce3b753079a;path=/k4_completeness.ml;lines=229-249}{\ExternalLink}
        \item[$\textnormal{5.}$]
        $\langle W_{\mathcal{S}4}^{A}, Rel_{\mathcal{S}4}^{A} \rangle \in \mathfrak{RF \cap TF}$
        \item[$\textnormal{6.}$]
        $\langle W_{\mathcal{B}}^{A}, Rel_{\mathcal{B}}^{A} \rangle \in \mathfrak{RF \cap SymF}$
        \item[$\textnormal{7.}$]
        $\langle W_{\mathcal{S}5}^{A}, Rel_{\mathcal{S}5}^{A} \rangle \in \mathfrak{RF \cap EF}$
        \item[$\textnormal{8.}$]
        $\langle W_{\mathcal{GL}}^{A}, Rel_{\mathcal{GL}}^{A} \rangle \in \mathfrak{ITF}$~\href{https://archive.softwareheritage.org/swh:1:cnt:52c9de12454731a9a291b06bc750c0d2d14e1fb4;origin=https://github.com/HOLMS-lib/HOLMS;visit=swh:1:snp:2c0efd349323ed6f8067581cf1f6d95816e49841;anchor=swh:1:rev:1caf3be141c6f646f78695c0eb528ce3b753079a;path=/gl_completeness.ml;lines=314-342}{\ExternalLink}
    \end{itemize}
\end{multicols}
\end{mydefinition}
\end{lemma}  
\begin{proof}
Note that we have proved that the parametric universe $ W_{\mathcal{S}}^{A}$ is non empty and, furthermore, each $ Rel_{\mathcal{S}}^{A}$ is a binary relation over $ W_{\mathcal{S}}^{A}$; thus each couple $\langle W_{\mathcal{S}}^{A},  Rel_{\mathcal{S}}^{A} \rangle$ is a frame. Moreover, this is a finite frame because $W_{\mathcal{S}}^{A}$ is finite: $W_{\mathcal{S}}^{A}$ is a subset of $\textit{MAX}_{\mathcal{S}}^{A}$; $\textit{MAX}_{\mathcal{S}}^{A}$ is a set of sets formulas; from Fact~\ref{fct:finite_formulas}, we know that there are only finitely many \textbf{sets of formulas}.
    \begin{itemize}
        \item[1.] $\langle W_{\mathcal{\varnothing}}^{A}, Rel_{\mathcal{\varnothing}}^{A} \rangle$ is a finite relational frame, then it belongs to $\mathfrak{F}$. 
        
        \hfill $\lhd$
        
        \item[2.] $\langle W_{\mathcal{D}}^{A}, Rel_{\mathcal{D}}^{A} \rangle$ is a finite relational frame. \\ We should prove that it is also \textbf{serial} so that it belongs to $\mathfrak{SerF}$. \\
        By definition of the standard accessibility relation $Rel_{\mathcal{D}}^{A}$, the following equivalence holds:
        \begin{center}
            $\forall w \in W_{\mathcal{D}}^{A}  (\exists x \in W_{\mathcal{D}}^{A}  (Rel_{\mathcal{D}}^{A}(w,x))) $ \\ $\iff $ \\$ \forall w \in W_{\mathcal{D}}^{A}  (\exists x \in W_{\mathcal{D}}^{A}(\forall \Box B \in w(B \in x)))$.
        \end{center}

         We consider $w$ a world in $W_{\mathcal{T}}^{A}$ such that $\Box B \in w$ and we want to find a world $x$ such that $B \in x$. We consider $x=w$, the following hold:
        \begin{enumerate}
            \item  $B \ \mathbf{sub} \ A$;
            \item Because $w \in W_{\mathcal{T}}^{A}$, $w$ is $\mathcal{T},A-$maximal consistent;
            \item From the axiom \textbf{T}, $\mathcal{T} \vdash \Box B \to B$;
            \item $\Box B \in w$;
        \end{enumerate}
         By maximal consistent lemma~\ref{lem:max_cons}, $B\in w=x$.
         
        \hfill $\lhd$

        \item[3.] $\langle W_{\mathcal{T}}^{A}, Rel_{\mathcal{T}}^{A} \rangle$ is a finite relational frame. \\ We should prove that it is also \textbf{reflexive} so that it belongs to $\mathfrak{RF}$. By definition of the standard accessibility relation $Rel_{\mathcal{T}}^{A}$, the following equivalence holds:
        \begin{center}
            $\forall w \in W_{\mathcal{T}}^{A}  (Rel_{\mathcal{T}}^{A}(w,w)) \iff \forall w \in W_{\mathcal{T}}^{A}(\forall \Box B \in w(B \in w))$.
        \end{center}
        We consider $w$ a world in $W_{\mathcal{T}}^{A}$ such that $\Box B \in w$ and we want to prove that $B \in w$. As above, the following hold:
        \begin{enumerate}
            \item  $B \ \mathbf{sub} \ A$;
            \item Because $w \in W_{\mathcal{T}}^{A}$, $w$ is $\mathcal{T},A-$maximal consistent;
            \item From the axiom \textbf{T}, $\mathcal{T} \vdash \Box B \to B$;
            \item $\Box B \in w$;
        \end{enumerate}
        By maximal consistent lemma~\ref{lem:max_cons}, $B\in w$.
        
        \hfill $\lhd$

        \item[4.] $\langle W_{\mathcal{K}4}^{A}, Rel_{\mathcal{K}4}^{A} \rangle$ is a finite relational frame. \\ We should prove that it is also \textbf{transitive} so that it belongs to $\mathfrak{TF}$. 

         By definition of the standard accessibility relation $Rel_{\mathcal{K}4}^{A}$, the following equivalence holds:
        \begin{center}
            $\forall w,x,y \in W_{\mathcal{K}4}^{A}  (Rel_{\mathcal{K}4}^{A}(w,x) \ and \ Rel_{\mathcal{K}4}^{A}(x,y) \implies  Rel_{\mathcal{K}4}^{A}(w,y)) $\\$ \iff  $\\ $ \forall w,y,x \in W_{\mathcal{K}4}^{A}(\forall \Box B \in w(B \in x \ and \ \Box B \in x ) \ and \ \forall \Box B \in x(B \in y \ and \ \Box B \in y) \implies \forall \Box B \in w(B \in y \ and \ \Box B \in y) )$.
        \end{center}
        We consider three worlds $w,x,y$  in $W_{\mathcal{T}}^{A}$ such that:
        \begin{itemize}
            \item $\Box B \in w$;
            \item $\forall \Box B \in w(B \in x \ and \ \Box B \in x) $;
            \item $ \forall \Box B \in x(B \in y \ and \ \Box B \in y)$.
        \end{itemize}
        We want to prove that $B \in y$. \\
        We have that $\Box B \in w$, then both $B,\Box B \in x$; thus both $B, \Box B \in y$.    
        
        \hfill $\lhd$

         \item[5.] $\langle W_{\mathcal{S}4}^{A}, Rel_{\mathcal{S}4}^{A} \rangle$ is a finite relational frame. \\ We should prove that it is also \textbf{transitive} and \textbf{reflexive} so that it belongs to $\mathfrak{RF \cap TF}$. 
         \begin{itemize}
             \item $Rel_{\mathcal{S}4}^{A}$ is \textbf{transitive} and the proof is the same as the one for $Rel_{\mathcal{K}4}^{A}$;
             \item We want to prove that $Rel_{\mathcal{S}4}^{A}$ is \textbf{reflexive.} \\ By definition of the standard accessibility relation $Rel_{\mathcal{S}4}^{A}$, the following equivalence holds:
        \begin{center}
            $\forall w \in W_{\mathcal{S}4}^{A}  (Rel_{\mathcal{S}4}^{A}(w,w))$ \\ $\iff $\\$ \forall w \in W_{\mathcal{S}4}^{A}(\forall \Box B \in w(B \in w \ and \ \Box B \in w ))$.
        \end{center}
        We consider $w$ a world in $W_{\mathcal{S}4}^{A}$ such that $\Box B \in w$ and we want to prove that $B, \Box B \in w$. Then the following hold:
        \begin{enumerate}
            \item  $B, \Box B \ \mathbf{sub} \ A$;
            \item Because $w \in W_{\mathcal{S}4}^{A}$, $w$ is $\mathcal{S}4,A-$maximal consistent;
            \item From the axiom \textbf{T}, $\mathcal{S}4 \vdash \Box B \to B$;
            \item $\Box B \in w$;
        \end{enumerate}
        By maximal consistent lemma~\ref{lem:max_cons}, $B, \Box B \in w$.
         \end{itemize}

        \hfill $\lhd$

        \item[6.] $\langle W_{\mathcal{B}}^{A}, Rel_{\mathcal{B}}^{A} \rangle$ is a finite relational frame. \\ We should prove that it is also \textbf{reflexive} and \textbf{symmetric} so that it belongs to $\mathfrak{RF \cap SymF}$. 
         \begin{itemize}
             \item We want to prove that $Rel_{\mathcal{B}}^{A}$ is \textbf{symmetric}: \\
             By definition of the standard accessibility relation $Rel_{\mathcal{S}4}^{A}$, the following equivalence holds:
        \begin{center}
            $\forall w,x \in W_{\mathcal{B}}^{A}  (Rel_{\mathcal{B}}^{A}(w,x) \implies Rel_{\mathcal{B}}^{A}(x,w) )$ \\ $\iff $\\$ \forall w,x \in W_{\mathcal{B}}^{A}(\forall \Box B \in w(B \in x) \ and \ \forall \Box B \in x(B\in w) \implies \forall \Box B \in x(B \in w) \ and \ \forall \Box B \in w(B\in x) )$.
        \end{center}
        But then the second statement obviously holds and the standard accessibility for $\mathbb{B}$ is symmetric.
        
             \item We want to prove that $Rel_{\mathcal{B}}^{A}$ is \textbf{reflexive.} \\ By definition of the standard accessibility relation $Rel_{\mathcal{B}}^{A}$, the following equivalence holds:
        \begin{center}
            $\forall w \in W_{\mathcal{B}}^{A}  (Rel_{\mathcal{B}}^{A}(w,w) \iff  \forall w \in W_{\mathcal{B}}^{A}(\forall \Box B \in w(B \in w))$.
        \end{center}
        We consider $w$ a world in $W_{\mathcal{B}}^{A}$ such that $\Box B \in w$ and we want to prove that $B \in w$. Then, as for $\mathbb{T}$, the following hold:
        \begin{enumerate}
            \item  $B \ \mathbf{sub} \ A$;
            \item Because $w \in W_{\mathcal{B}}^{A}$, $w$ is $\mathcal{B},A-$maximal consistent;
            \item From the axiom \textbf{T}, $\mathcal{B} \vdash \Box B \to B$;
            \item $\Box B \in w$;
        \end{enumerate}
        By maximal consistent lemma~\ref{lem:max_cons}, $B \in w$.
         \end{itemize}

        \hfill $\lhd$

        \item[7.] $\langle W_{\mathcal{S}5}^{A}, Rel_{\mathcal{S}5}^{A} \rangle$ is a finite relational frame. \\ We should prove that it is also \textbf{euclidean} and \textbf{reflexive} so that it belongs to $\mathfrak{RF \cap EF}$. 
         \begin{itemize}
             \item We want to prove that $Rel_{\mathcal{S}4}^{A}$ is \textbf{euclidean.} \\ By definition of the standard accessibility relation $Rel_{\mathcal{S}4}^{A}$, the following equivalence holds:
        \begin{center}
            $\forall w, x, y \in W_{\mathcal{S}5}^{A}  (Rel_{\mathcal{S}5}^{A}(w,x) \ and \ Rel_{\mathcal{S}5}^{A}(w,y) \implies Rel_{\mathcal{S}5}^{A}(x,y))$ \\ $\iff $\\$ \forall w,x,y \in W_{\mathcal{S}5}^{A}(\forall \Box B (\Box B \in w \iff \ \Box B \in x ) \ and \ \forall \Box B (\Box B \in w \iff \ \Box B \in y ) \implies \forall \Box B (\Box B \in x \iff \ \Box B \in x ))$.
        \end{center}
        But then the second statement obviously holds and the standard accessibility for $\mathbb{B}$ is euclidean.
        
        \item  We want to prove that $Rel_{\mathcal{S}5}^{A}$ is \textbf{reflexive.} \\ By definition of the standard accessibility relation $Rel_{\mathcal{S}5}^{A}$, the following equivalence holds:
        \begin{center}
            $\forall w \in W_{\mathcal{S}5}^{A}  (Rel_{\mathcal{S}5}^{A}(w,w) \iff  \forall w \in W_{\mathcal{S}5}^{A}(\forall \Box B (\Box B \in w \iff \Box B \in w))$.
        \end{center}
         But then the second statement obviously holds and the standard accessibility for $\mathbb{B}$ is reflexive.
        \end{itemize}
        \hfill $\lhd$

        \item[8.] $\langle W_{\mathcal{GL}4}^{A}, Rel_{\mathcal{GL}4}^{A} \rangle$ is a finite relational frame. \\ We should prove that it is also \textbf{transitive} and \textbf{irreflexive} so that it belongs to $\mathfrak{TF}$. 

        \begin{itemize}
            \item We want to prove that $Rel_{\mathcal{GL}}^{A}$ is \textbf{transitive.} \\ 
         By definition of the standard accessibility relation $Rel_{\mathcal{GL}}^{A}$, the following equivalence holds:
        \begin{center}
            $\forall w,x,y \in W_{\mathcal{GL}}^{A}  (Rel_{\mathcal{GL}}^{A}(w,x) \ and \ Rel_{\mathcal{GL}}^{A}(x,y) \implies  Rel_{\mathcal{GL}}^{A}(w,y)) $\\$ \iff  $\\ $ \forall w,y,x \in W_{\mathcal{GL}}^{A}$\\$((\forall \Box B \in w(B \in x \ and \ \Box B \in x ) \ and \ \exists \Box E \in x (\neg \Box E \in w)) \ and$\\$  (\forall \Box B \in x(B \in y \ and \ \Box B \in y) \ and \ and \ \exists \Box E \in y (\neg \Box E \in x)) \implies$\\$ (\forall \Box B \in w(B \in y \ and \ \Box B \in y) \ and \ \ and \ \exists \Box E \in y (\neg \Box E \in w)) )$.
        \end{center}
        We consider three worlds $w,x,y$  in $W_{\mathcal{T}}^{A}$ such that:
        \begin{itemize}
            \item $\Box B \in w$;
            \item $\forall \Box B \in w(B \in x \ and \ \Box B \in x) $; $\exists \Box E_1 \in x (\neg \Box E_1 \in w)$;
            \item $ \forall \Box B \in x(B \in y \ and \ \Box B \in y)$; $ \exists \Box E_2 \in y (\neg \Box E_2 \in x)$;.
        \end{itemize}
        We have that $\Box B \in w$, then both $B, \Box B \in x$; thus both $B, \Box B \in y$. 
        Furthermore, $\neg \Box E_1 \in w$ and $\Box E_1 \in x$ but then $\Box E_1 \in y$, i.e $\exists \Box E_1 \in y (\neg\Box E_1 \in w)$.
        
        \item We want to prove that $Rel_{\mathcal{GL}}^{A}$ is \textbf{irreflexive.}\\
        By definition of the standard accessibility relation $Rel_{\mathcal{GL}}^{A}$, the following equivalence holds:
        \begin{center}
            $\forall w \in W_{\mathcal{GL}}^{A}  \neg(Rel_{\mathcal{GL}}^{A}(w,w))$ \\ $\iff $\\$ \forall w \in W_{\mathcal{GL}}^{A}(\neg (\forall \Box B \in w(B \in w \ and \ \Box B \in w ) \ and \ \exists \Box E \in w (\neg \Box E \in w) ))$.
        \end{center}
        We proceed by \textit{negatio ad absurdum}.  \\ NAA: $\exists w \in W_{\mathcal{GL}}^{A} (Rel_{\mathcal{GL}}^{A}(w,w))$ then $\forall \Box B \in w(B \in w \ and \ \Box B \in w $ and  $\exists \neg \Box E \in w (\Box E \in w) )$. But then both $\Box E, \neg \Box E \in w$; thus \underline{w is $\mathcal{GL}-$inconistent}. A contradiction. $\lightning$\hfill $\lhd$
        \end{itemize}
    \end{itemize}
\end{proof}
\medskip
\medskip

\begin{lemma}[\textbf{Accessibility lemma}]\label{lem:accessibility_std_frame}\phantom{Allyouneedislovelove}
   \begin{mydefinition}
   Let $A, B$ modal formulas such that $\Box B \ \mathbf{sub} \ A$. Let $\mathbb{S} $ be a modal system in \\ $ \{\mathbb{K}$~\href{https://archive.softwareheritage.org/swh:1:cnt:ae230138ff15476c8dab9e32606bceca7168285b;origin=https://github.com/HOLMS-lib/HOLMS;visit=swh:1:snp:2c0efd349323ed6f8067581cf1f6d95816e49841;anchor=swh:1:rev:1caf3be141c6f646f78695c0eb528ce3b753079a;path=/k_completeness.ml;lines=129-150}{\ExternalLink}  $\mathbb{;\ D; \ T}$~\href{https://archive.softwareheritage.org/swh:1:cnt:1352cb294724b2251b9346657432700ecbed10d2;origin=https://github.com/HOLMS-lib/HOLMS;visit=swh:1:snp:2c0efd349323ed6f8067581cf1f6d95816e49841;anchor=swh:1:rev:1caf3be141c6f646f78695c0eb528ce3b753079a;path=/t_completeness.ml;lines=261-282}{\ExternalLink} $\mathbb{; \ K\textnormal{4}}$~\href{https://archive.softwareheritage.org/swh:1:cnt:32869a1df338d91c4cd4a0cb9e73fb4f1be29991;origin=https://github.com/HOLMS-lib/HOLMS;visit=swh:1:snp:2c0efd349323ed6f8067581cf1f6d95816e49841;anchor=swh:1:rev:1caf3be141c6f646f78695c0eb528ce3b753079a;path=/k4_completeness.ml;lines=274-382}{\ExternalLink} $\mathbb{; \ S\textnormal{4};\ B; \ S\textnormal{5};\ GL}$~\href{https://archive.softwareheritage.org/swh:1:cnt:52c9de12454731a9a291b06bc750c0d2d14e1fb4;origin=https://github.com/HOLMS-lib/HOLMS;visit=swh:1:snp:2c0efd349323ed6f8067581cf1f6d95816e49841;anchor=swh:1:rev:1caf3be141c6f646f78695c0eb528ce3b753079a;path=/gl_completeness.ml;lines=344-467}{\ExternalLink}$\}$.
   For every $w\in W^A_\mathcal{S}$:
   \begin{center}
       $\Box B \in w$ $ \ \iff \ $  $ \forall x\in W^A_\mathcal{S} (Rel_{\mathcal{S}}^{A}(w,x) \implies B \in x)$.
   \end{center}
\end{mydefinition}
\end{lemma}
\begin{proof}
    Notice that the $Rel_{\mathcal{D}}^A$ and $Rel_{\mathcal{T}}^A$ are defined as  $Rel_{\mathcal{\varnothing}}^A$, while $Rel_{\mathcal{S}4}^A$ is the same as $Rel_{\mathcal{K}4}^A$. Then it will be sufficient to prove this lemma for two of these standard relations and is proof will apply to the other systems.  
    
    \begin{itemize}
        \item[$\textnormal{1.}$] $[\forall w\in W^A_\mathcal{\varnothing} (\Box B \in w \ \iff \  \forall x\in W^A_\mathcal{\varnothing} (Rel_{\mathcal{\varnothing}}^{A}(w,x) \implies B \in x))]$.
        \\ Let $w$ a world in $W^A_\mathcal{\varnothing}$.
        \begin{itemize}
            \item[$\implies$] We assume that $\Box B \in w$.\\
            By definition, if $Rel_{\mathcal{S}}^{A}(w,x)$ then $\forall \Box  B \in w(B \in x)$, thus $B \in x$.

            \item[$\impliedby$] We proceed by \textit{contraposition}: \\
            $[\Box B \not \in w \implies \exists x\in W^A_\mathcal{\varnothing} (Rel_{\mathcal{\varnothing}}^{A}(w,x) \ and \  B \not \in x))$. \\
            Assuming \underline{$\Box B \not \in w$}, we consider: \\$\mathcal{X}= \{ \neg B\} \cup \{C \in \mathbf{Form}_{\Box} \ |\ \Box C \in w \}$.
       \newpage     
            \begin{claim}
                $\mathcal{X}$ is $\mathcal{\varnothing}$-consistent
            \begin{claimproof}         
                We proceed by \textit{reductio ad absurdum}. \\ RAA: $\mathcal{X}$ is not $\mathcal{\varnothing}$-consistent \\
                Then the following equivalences hold: 
                $\mathcal{\varnothing} \vdash \neg \bigwedge \mathcal{X}$ $\iff$ \\ $\mathcal{\varnothing} \vdash \neg (\neg B \land \bigwedge_{i=1}^{n} C_i)$ where $\Box C_1, \dots , \Box C_n \in w$ $\xLeftrightarrow{taut}$ \\ $\mathcal{\varnothing} \vdash C_1 \land \dots \land  C_n \to B$ $\xRightarrow{\texttt{RN}, \mathbf{K}}$ $\mathcal{\varnothing} \vdash \Box C_1 \land \dots \land \Box C_n \to \Box B$  and $\Box B \ \mathbf{sub} \ A$ and $\forall i \le n (\Box C_i \in w)$ $\xRightarrow{Max \ Cons~\ref{lem:max_cons}}$ \underline{$\Box B \in w$}. $\lightning$ \hfill $\lhd$
            \end{claimproof}
            \end{claim}
                 
            We extend $\mathcal{X}$ by Lemma~\ref{lem:nonempty_max_cons} that states  $\exists \mathcal{Y} \subseteq \mathbf{Form}_{\Box}(\neg A \in \mathcal{Y}$ and  $\mathcal{Y} \in \textit{MAX}^A_\mathcal{S} \ and \ \forall y \in \mathcal{Y}(y \ \mathbf{subs} \ A)$. Then:
            \begin{itemize}
                \item $\mathcal{Y} \in W^A_\mathcal{S}$;
                \item By $\mathcal{Y} \in \textit{MAX}^A_\mathcal{S}$ and $\neg B \in \mathcal{Y}$, then $B \not \in \mathcal{Y}$;
                \item $\forall \Box C \in w (C \in \mathcal{X \subseteq Y} )$, i.e. $Rel_{\mathcal{\varnothing}}^{A}(w,x)$.    
            \end{itemize} 
        \end{itemize}
        \hfill $\lhd$
        
        \item[$\textnormal{2.}$] $[\forall w\in W^A_{\mathcal{K}4} (\Box B \in w \ \iff \  \forall x\in W^A_{\mathcal{K}4} (Rel_{\mathcal{K}4}^{A}(w,x) \implies B \in x))]$.
        \\ Let $w$ a world in $W^A_{\mathcal{K}4}$.
        \begin{itemize}
            \item[$\implies$] We assume that $\Box B \in w$.\\
            By definition, if $Rel_{\mathcal{K}4}^{A}(w,x)$ then $\forall \Box  B \in w(B, \Box B \in x)$, thus $B \in x$.
            
            \item[$\impliedby$] We proceed by \textit{contraposition}: \\
            $[\Box B \not \in w \implies \exists x\in W^A_{\mathcal{K}4} (Rel_{{\mathcal{K}4}}^{A}(w,x) \ and \  B \not \in x))$. \\
            Assuming \underline{$\Box B \not \in w$}, we consider: \\ $\mathcal{X}= \{\neg B\} \cup \{C, \Box C \in \mathbf{Form}_{\Box} \ |\ \Box C \in w \} $.

            \begin{claim}
                $\mathcal{X}$ is $\mathcal{K}4$-consistent
            \begin{claimproof}         
                We proceed by \textit{reductio ad absurdum}. \\ RAA: $\mathcal{X}$ is not $\mathcal{K}4$-consistent \\
                Then the following equivalences hold: 
                $\mathcal{K}4 \vdash \neg \bigwedge \mathcal{X}$ $\iff$ \\ $\mathcal{K}4 \vdash \neg (\neg B \land \bigwedge_{i=1}^{n} C_i \land \bigwedge_{i=1}^{n} \Box C_i)$ st $\Box C_1, \dots , \Box C_n \in w$ $\xLeftrightarrow{taut}$ $\mathcal{K}4 \vdash C_1 \land \dots \land  C_n \land  \Box C_1 \land \dots \land  \Box C_n \to B$ $\xRightarrow{\texttt{RN}, \mathbf{K}}$ $\mathcal{K}4 \vdash \Box C_1 \land \dots \land \Box C_n \land \Box \Box C_1 \land \dots \land  \Box \Box C_n \to \Box B$.
                $\mathcal{K}4 \vdash \Box C \to \Box\Box C$, then $\mathcal{K}4 \vdash \Box C_1 \land \dots \land \Box C_n \to \Box B$                
                \\ and $\Box B \ \mathbf{sub} \ A$ and $\forall i \le n (\Box C_i \in w)$ $\xRightarrow{Max \ Cons~\ref{lem:max_cons}}$ \underline{$\Box B \in w$}. $\lightning$ \hfill $\lhd$
            \end{claimproof}
            \end{claim}
                 
            We extend $\mathcal{X}$ by Lemma~\ref{lem:nonempty_max_cons} that states  $\exists \mathcal{Y} \subseteq \mathbf{Form}_{\Box}(\neg A \in \mathcal{Y}$ and  $\mathcal{Y} \in \textit{MAX}^A_\mathcal{S} \ and \ \forall y \in \mathcal{Y}(y \ \mathbf{subs} \ A)$. Then:
            \begin{itemize}
                \item $\mathcal{Y} \in W^A_\mathcal{S}$;
                \item By $\mathcal{Y} \in \textit{MAX}^A_\mathcal{S}$ and $\neg B \in \mathcal{Y}$, then $B \not \in \mathcal{Y}$;
                \item $\forall \Box C \in w (C, \Box C \in \mathcal{X \subseteq Y} )$, i.e. $Rel_{\mathcal{K}4}^{A}(w,x)$.    
            \end{itemize} 
             
        \end{itemize}
        \hfill $\lhd$
     \newpage   
        \item[$\textnormal{3.}$] $[\forall w\in W^A_{\mathcal{B}} (\Box B \in w \ \iff \  \forall x\in W^A_{\mathcal{B}} (Rel_{\mathcal{B}}^{A}(w,x) \implies B \in x))]$.
        \\ Let $w$ a world in $W^A_{\mathcal{B}}$.
        \begin{itemize}
            \item[$\implies$] We assume that $\Box B \in w$.\\
            By definition, if $Rel_{\mathcal{B}}^{A}(w,x)$ then $\forall \Box  B \in w(B \in x) \ and \ \forall \Box  B \in x(B \in w)$, thus $B \in x$.
            \item[$\impliedby$] We proceed by \textit{contraposition}: \\
            $[\Box B \not \in w \implies \exists x\in W^A_\mathcal{B} (Rel_{\mathcal{B}}^{A}(w,x) \ and \  B \not \in x))$. \\
            Assuming \underline{$\Box B \not \in w$}, we consider: \\ $\mathcal{X}= \{ \neg B\} \cup \{C \in \mathbf{Form}_{\Box} \ |\ \Box C \in w \} \cup \{\neg \Box E \in \mathbf{Form}_{\Box} \ |\ \neg  E \in w \}$.
            \begin{claim}
                $\mathcal{X}$ is $\mathcal{B}$-consistent
            \begin{claimproof}         
                We proceed by \textit{reductio ad absurdum}. \\ RAA: $\mathcal{X}$ is not $\mathcal{B}$-consistent \\
                Then the following equivalences hold: 
                $\mathcal{B} \vdash \neg \bigwedge \mathcal{X}$ $\iff$ \\ $\mathcal{B} \vdash \neg (\neg B \land \bigwedge_{i=1}^{n} C_i \land \bigwedge_{j=1}^{m} \neg \Box E_j)$ where $\Box C_1, \dots , \Box C_n \in w$ and $\neg E_1, \dots , \neg E_m \in w$ $\xLeftrightarrow{taut}$ \\ $\mathcal{B} \vdash C_1 \land \dots \land  C_n \land  \neg \Box  E_1 \land  \dots \land \neg \Box E_m \to B$ $\xRightarrow{\texttt{RN}, \mathbf{K}}$ $\mathcal{B} \vdash \Box C_1 \land \dots \land \Box C_n \land  \Box \neg \Box  E_1 \land  \dots \land \Box \neg \Box E_m \to \Box B$. $\mathcal{B} \vdash \neg E \to \Box \Diamond \neg E$ and $\Diamond \neg E \coloneq \neg \Box  E$ then $\mathcal{B} \vdash \Box C_1 \land \dots \land \Box C_n \land   \neg   E_1 \land  \dots \land \neg  E_m \to \Box B$ \\ and $\Box B \ \mathbf{sub} \ A$ and $\forall i \le n (\Box C_i \in w)$ and $\forall j \le n (\neg E_j \in w)$  $\xRightarrow{Max \ Cons~\ref{lem:max_cons}}$ \underline{$\Box B \in w$}. $\lightning$ \hfill $\lhd$
            \end{claimproof}
            \end{claim}
                 
            We extend $\mathcal{X}$ by Lemma~\ref{lem:nonempty_max_cons} that states  $\exists \mathcal{Y} \subseteq \mathbf{Form}_{\Box}(\neg A \in \mathcal{Y}$ and  $\mathcal{Y} \in \textit{MAX}^A_\mathcal{S} \ and \ \forall y \in \mathcal{Y}(y \ \mathbf{subs} \ A)$. Then:
            \begin{itemize}
                \item $\mathcal{Y} \in W^A_\mathcal{S}$;
                \item By $\mathcal{Y} \in \textit{MAX}^A_\mathcal{S}$ and $\neg B \in \mathcal{Y}$, then $B \not \in \mathcal{Y}$;
                \item $\forall \Box C \in w (C \in \mathcal{X \subseteq Y} )$ and $\forall \Box E \in \mathcal{X \subseteq Y} (E \in w )$ i.e. $Rel_{\mathcal{B}}^{A}(w,x)$. 
                \begin{claimproof}
                By \textit{absurd} \underline{$\Box E \in \mathcal{Y} $} and $E \not \in w$ then, by $\mathcal{B},A-$maximal consistency of $w$, $\neg E  \in w$, then--by the definition of $\mathcal{X}$--\underline{$\neg \Box E \in \mathcal{Y}$}; thus $\mathcal{Y}$ is not $\mathcal{B},A-$maximal consistent. A contraction. $\lightning$ \hfill $\lhd$ 
                \end{claimproof}
            \end{itemize} 
        \end{itemize}
        \hfill $\lhd$
        
        \item[$\textnormal{4.}$] $[\forall w\in W^A_{\mathcal{S}5} (\Box B \in w \ \iff \  \forall x\in W^A_{\mathcal{S}5} (Rel_{\mathcal{S}5}^{A}(w,x) \implies B \in x))]$.
        \\ Let $w$ a world in $ W^A_{\mathcal{S}5}$.
        \begin{itemize}
            \item[$\implies$] We assume that $\Box B \in w$. \\
            By definition, if $Rel_{\mathcal{S}5}^{A}$(w,x) then $\forall \Box  B (\Box  B \in w \iff \Box B \in x)$, thus $\Box B \in x$. Furthermore, $x$ is $\mathcal{S}5-$maximal consistent, $B \ \mathbf{sub} \ A$ and $\mathcal{S}5 \vdash \Box B \to B$, then by Lemma~\ref{lem:max_cons} $ B \in x$
            
            \item[$\impliedby$] We proceed by \textit{contraposition}: \\
            $[\Box B \not \in w \implies \exists x\in W^A_\mathcal{S}5 (Rel_{\mathcal{S}5}^{A}(w,x) \ and \  B \not \in x))$. \\
             Assuming \underline{$\Box B \not \in w$}, we consider : $\mathcal{X} =$ \\ {$ \{ \neg B\} \cup \{\Box C \in \mathbf{Form}_{\Box} \ |\ \Box C \in w \} \cup \{\neg \Box E \in \mathbf{Form}_{\Box} \ |\ \neg \Box  E \in w \}$}.
\newpage            \begin{claim}
                $\mathcal{X}$ is $\mathcal{S}5$-consistent
            \begin{claimproof}         
                We proceed by \textit{reductio ad absurdum}. \\ RAA: $\mathcal{X}$ is not $\mathcal{S}5$-consistent \\
                Then the following equivalences hold: 
                $\mathcal{S}5 \vdash \neg \bigwedge \mathcal{X}$ $\iff$ \\ $\mathcal{S}5 \vdash \neg (\neg B \land \bigwedge_{i=1}^{n} \Box C_i \land \bigwedge_{j=1}^{m} \neg \Box E_j)$  where $\Box C_1, \dots , \Box C_n \in w$ and $\neg \Box E_1, \dots , \neg \Box E_m \in w$  $\xLeftrightarrow{taut}$ \\ $\mathcal{S}5 \vdash \Box C_1 \land \dots \land \Box C_n \land  \neg \Box  E_1 \land  \dots \land \neg \Box E_m \to B$ $\xRightarrow{\texttt{RN}, \mathbf{K}}$ $\mathcal{S}5 \vdash \Box \Box C_1 \land \dots \land \Box \Box C_n \land  \Box \neg \Box  E_1 \land  \dots \land \Box \neg \Box E_m \to \Box B$.
                $\mathcal{S}5 \vdash \Diamond \neg E \to \Box \Diamond \neg E$ and $\Diamond \neg E \coloneq \neg \Box  E$  and, because the class correspondent to $\mathbb{S}5$ is the set of equivalent frames (that are also transitive), $\mathcal{S}5 \vdash \Box C \to \Box \Box C$. Then $\mathcal{S}5 \vdash \Box  C_1 \land \dots \land  \Box C_n \land   \neg \Box  E_1 \land  \dots \land  \neg \Box E_m \to \Box B$
                and $\Box B \ \mathbf{sub} \ A$ and $\forall i \le n (\Box C_i \in w)$ and $\forall j \le n (\neg \Box E_j \in w)$ $\xRightarrow{Max \ Cons~\ref{lem:max_cons}}$ \underline{$\Box B \in w$}. $\lightning$ \hfill $\lhd$
            \end{claimproof}
            \end{claim}
                 
            We extend $\mathcal{X}$ by Lemma~\ref{lem:nonempty_max_cons} that states  $\exists \mathcal{Y} \subseteq \mathbf{Form}_{\Box}(\neg A \in \mathcal{Y}$ and  $\mathcal{Y} \in \textit{MAX}^A_\mathcal{S} \ and \ \forall y \in \mathcal{Y}(y \ \mathbf{subs} \ A)$. Then:
            \begin{itemize}
                \item $\mathcal{Y} \in W^A_\mathcal{S}$;
                \item By $\mathcal{Y} \in \textit{MAX}^A_\mathcal{S}$ and $\neg B \in \mathcal{Y}$, then $B \not \in \mathcal{Y}$;
                \item $\forall \Box C \in w (\Box C \in \mathcal{X \subseteq Y} )$, and $\forall \Box E \in \mathcal{Y} ( \Box E\in w )$      
                i.e. $Rel_{\mathcal{S}5}^{A}(w,x)$.  
                \begin{claimproof}
                By \textit{absurd} \underline{$\Box E \in \mathcal{Y} $} and $\Box E \not \in w$ then, by $\mathcal{S}5,A-$maximal consistency of $w$, $\neg \Box E  \in w$, then--by the definition of $\mathcal{X}$--\underline{$\neg \Box E \in \mathcal{Y}$}; thus $\mathcal{Y}$ is not $\mathcal{S}5,A-$maximal consistent.  $\lightning$ \hfill $\lhd$
                \end{claimproof}
            \end{itemize} 
        \end{itemize}
        \hfill $\lhd$
      
        \item[$\textnormal{5.}$] $[\forall w\in W^A_{\mathcal{GL}} (\Box B \in w \ \iff \  \forall x\in W^A_{\mathcal{GL}} (Rel_{\mathcal{GL}}^{A}(w,x) \implies B \in x))]$.
        \\ Let $w$ a world in $W^A_{\mathcal{GL}}$.
        \begin{itemize}
            \item[$\implies$] We assume that $\Box B \in w$.\\
            By definition, if $Rel_{\mathcal{GL}}^{A}(w,x)$, then $\forall \Box  B \in w(B, \Box B \in x) $ \textit{and} \\ $ \exists \Box E \in x (\neg \Box E \in w)$, thus $B \in x$.
            
            \item[$\impliedby$] We proceed by \textit{contraposition}: \\
            $[\Box B \not \in w \implies \exists x\in W^A_{\mathcal{GL}} (Rel_{{\mathcal{GL}}}^{A}(w,x) \ and \  B \not \in x))$. \\
            Assuming \underline{$\Box B \not \in w$}, we consider : \\ $\mathcal{X}= \{ \neg B, \Box B \} \cup \{C,\Box C \in \mathbf{Form}_{\Box} \ |\ \Box C \in w \}$.
            \medskip
            \begin{claim}
                $\mathcal{X}$ is $\mathcal{GL}$-consistent
            \begin{claimproof}         
                We proceed by \textit{reductio ad absurdum}. \\ RAA: $\mathcal{X}$ is not $\mathcal{GL}$-consistent \\
                Then the following equivalences hold: 
                $\mathcal{GL} \vdash \neg \bigwedge \mathcal{X}$ $\iff$ \\ $\mathcal{GL} \vdash \neg (\neg B \land \Box B \land \bigwedge_{i=1}^{n} C_i \land \bigwedge_{i=1}^{n} \Box C_i)$ st $\Box C_1, \dots , \Box C_n \in w$ $\xLeftrightarrow{taut}$ $\mathcal{GL} \vdash C_1 \land \dots \land  C_n \land  \Box C_1 \land \dots \land  \Box C_n \to(\Box B \to B)$ $\xRightarrow{\texttt{RN}, \mathbf{K}}$ $\mathcal{GL} \vdash \Box C_1 \land \dots \land \Box C_n \land \Box \Box C_1 \land \dots \land  \Box \Box C_n \to \Box (\Box B \to B)$ 
                \end{claimproof}
                \begin{claimproof}[Continue..] \\
                
                $\xRightarrow{\mathbf{4}, \mathbf{GL}}$ 
                $\mathcal{GL} \vdash \Box C_1 \land \dots \land \Box C_n \to \Box B$                
                and  $\Box B \ \mathbf{sub} \ A$ and \\ $\forall i \le n (\Box C_i \in w)$ $\xRightarrow{Max \ Cons~\ref{lem:max_cons}}$ \underline{$\Box B \in w$}. $\lightning$ $\lhd$
            \end{claimproof}
            \end{claim}
                 
            We extend $\mathcal{X}$ by Lemma~\ref{lem:nonempty_max_cons} that states  $\exists \mathcal{Y} \subseteq \mathbf{Form}_{\Box}(\neg A \in \mathcal{Y}$ and  $\mathcal{Y} \in \textit{MAX}^A_\mathcal{S} \ and \ \forall y \in \mathcal{Y}(y \ \mathbf{subs} \ A)$. Then
            \begin{itemize}
                \item $\mathcal{Y} \in W^A_\mathcal{S}$;
                \item By $\mathcal{Y} \in \textit{MAX}^A_\mathcal{S}$ and $\neg B \in \mathcal{Y}$, then $B \not \in \mathcal{Y}$;
                \item $\forall \Box C \in w (C, \Box C \in \mathcal{X \subseteq Y} )$, and $\exists  \Box E \in \mathcal{Y} ( \neg \Box E\in w )$ (it is sufficient to take $E=B$),  i.e. $Rel_{\mathcal{GL}}^{A}(w,x)$.  \hfill $\lhd$
            \end{itemize} 
        \end{itemize}
    \end{itemize}
\end{proof}       
\medskip
\medskip

\begin{remark}
   Notice that, although this proof is not parametric, it exhibits a modular structure. This modularity will facilitate the development with the aid  assistant HOL Light of our \textit{ad hoc polymorphic} approach to this part of proof, presented in Chapter~\ref{chap:4}.
\end{remark}
\end{enumerate}

\subsubsection{Proving Completeness}
Thanks to these lemmata, we are able to apply Lemma~\ref{lem:parametric_countermodel_lemma} to obtain completeness w.r.t appropriate frames for each $\mathbb{S}$ among $\mathbb{K, \ D, \ T, \ K\textnormal{4}, \ S\textnormal{4}, \ B. \ S\textnormal{5}}$, and $\mathbb{GL}$.

\begin{theorem}[\textbf{Completeness within the modal cube}]\phantom{All}
    \begin{mydefinition}
        \begin{enumerate}
            \item For every $A \in \mathbf{Form}_{\Box}$, if $\ \mathfrak{F} \vDash A$ then $\mathcal{\varnothing} \vdash A $.~\href{https://archive.softwareheritage.org/swh:1:cnt:ae230138ff15476c8dab9e32606bceca7168285b;origin=https://github.com/HOLMS-lib/HOLMS;visit=swh:1:snp:2c0efd349323ed6f8067581cf1f6d95816e49841;anchor=swh:1:rev:1caf3be141c6f646f78695c0eb528ce3b753079a;path=/k_completeness.ml;lines=194-211}{\ExternalLink}
            \item For every $A \in \mathbf{Form}_{\Box}$, if  $\ \mathfrak{SerF} \vDash A$ then $\mathcal{D} \vdash A $.
            \item For every $A \in \mathbf{Form}_{\Box}$, if  $\ \mathfrak{RF} \vDash A$ then $\mathcal{T} \vdash A $.~\href{https://archive.softwareheritage.org/swh:1:cnt:1352cb294724b2251b9346657432700ecbed10d2;origin=https://github.com/HOLMS-lib/HOLMS;visit=swh:1:snp:2c0efd349323ed6f8067581cf1f6d95816e49841;anchor=swh:1:rev:1caf3be141c6f646f78695c0eb528ce3b753079a;path=/t_completeness.ml;lines=321-338}{\ExternalLink}
            \item For every $A \in \mathbf{Form}_{\Box}$, if  $\ \mathfrak{TF} \vDash A$ then $\mathcal{K}4 \vdash A $.~\href{https://archive.softwareheritage.org/swh:1:cnt:32869a1df338d91c4cd4a0cb9e73fb4f1be29991;origin=https://github.com/HOLMS-lib/HOLMS;visit=swh:1:snp:2c0efd349323ed6f8067581cf1f6d95816e49841;anchor=swh:1:rev:1caf3be141c6f646f78695c0eb528ce3b753079a;path=/k4_completeness.ml;lines=426-444}{\ExternalLink}
            \item For every $A \in \mathbf{Form}_{\Box}$, if  $\ \mathfrak{RF}\cap\mathfrak{TF} \vDash A$ then   $\mathcal{S}4 \vdash A $.
            \item For every $A \in \mathbf{Form}_{\Box}$, if  $\ \mathfrak{RF}\cap\mathfrak{SymF} \vDash A$ then  $\mathcal{B} \vdash A $.
            \item For every $A \in \mathbf{Form}_{\Box}$, if  $\ \mathfrak{RF}\cap\mathfrak{EF} \vDash A$ then  $\mathcal{S}5 \vdash A $.
            \item For every $A \in \mathbf{Form}_{\Box}$, if  $\ \mathfrak{ITF} \vDash A$ then  $\mathcal{GL} \vdash A $.~\href{https://archive.softwareheritage.org/swh:1:cnt:52c9de12454731a9a291b06bc750c0d2d14e1fb4;origin=https://github.com/HOLMS-lib/HOLMS;visit=swh:1:snp:2c0efd349323ed6f8067581cf1f6d95816e49841;anchor=swh:1:rev:1caf3be141c6f646f78695c0eb528ce3b753079a;path=/gl_completeness.ml;lines=510-528}{\ExternalLink}
        \end{enumerate}
    \end{mydefinition}
\end{theorem}

\begin{proof}
    Let $\mathbb{S}$ be a normal system within the cube. Let $\mathfrak{S}$ be its appropriate class of frames.  Then we aim to prove that: 
    \begin{center}
        If $\mathfrak{S} \vDash A$ then $\mathcal{S} \vdash A$
    \end{center}
    We proceed by \textit{contraposition} by assuming $\mathcal{S} \not \vdash A$.
     We have proved, respectively in Lemma~\ref{lem:appr_std_frame} and Lemma~\ref{lem:accessibility_std_frame}, that $\langle W_{\mathcal{S}}^{A}, {Rel_{\mathcal{S}}^{A}} \rangle$ satisfies the requirements~\ref{itm:c1} and~\ref{itm:c2}.
     Then we apply the parametric countermodel Lemma~\ref{lem:parametric_countermodel_lemma} and we obtain $\langle W_{\mathcal{S}}^{A}, {Rel_{\mathcal{S}}^{A}} \rangle \not \vdash A$, thus--being $\langle W_{\mathcal{S}}^{A}, {Rel_{\mathcal{S}}^{A}} \rangle \in \mathfrak{S}$--$\mathfrak{S} \not \vDash A$.   
\end{proof}

\newpage
\section{Bisimulation Theory}\label{sec:bisimulation}

In this section, we define the fundamental concepts of bisimulation theory and we prove some important results. This theory--developed in classic textbooks, e.g.~\cite{Stirling_2011}--will play an important role in HOLMS proof of completeness, in Section~\ref{sec:generalising-completeness} below. 

\begin{definition}[\textbf{Bisimulation}~\href{https://archive.softwareheritage.org/swh:1:cnt:4a011f69b1180019be5d03c6d5f1cec4550055e8;origin=https://github.com/HOLMS-lib/HOLMS;visit=swh:1:snp:2c0efd349323ed6f8067581cf1f6d95816e49841;anchor=swh:1:rev:1caf3be141c6f646f78695c0eb528ce3b753079a;path=/modal.ml;lines=258-265}{\ExternalLink}]\phantom{Allyouneedislove}
    \begin{mydefinition}
    \begin{mydefinition2}
    Let $\langle W,R, V \rangle$, $\langle W',R', V' \rangle$ be two models. \\ The function $f:W \longrightarrow W'$ is a \textbf{bisimulation} iff:
    \begin{enumerate}
        \item $f$ preserves the valuations: $\forall w \in W( \forall p \in \Phi (pVw \iff p V' f(w)))$;
        \item $f$ preserves $R$: $\forall w,x \in W(wRx \implies f(w) R' f(x))$;
        \item $f$ preserves $R'$: $\forall w \in W \ \forall y \in W' ( \ f(w)R'y \implies $ \\ \phantom{f preserves R' and forallw in w fora}$ \exists x \in W ( f(x)=y \ and \ wRx) )$
    \end{enumerate}
    \end{mydefinition2}
    \end{mydefinition}
\end{definition}
\medskip \medskip

\begin{definition}[\textbf{Bisimiliarity between worlds}~\href{https://archive.softwareheritage.org/swh:1:cnt:4a011f69b1180019be5d03c6d5f1cec4550055e8;origin=https://github.com/HOLMS-lib/HOLMS;visit=swh:1:snp:2c0efd349323ed6f8067581cf1f6d95816e49841;anchor=swh:1:rev:1caf3be141c6f646f78695c0eb528ce3b753079a;path=/modal.ml;lines=286-288}{\ExternalLink}]\phantom{Allyouneedislove}
\begin{mydefinition}
\begin{mydefinition2}
    Let $\langle W,R, V \rangle$, $\langle W',R', V' \rangle$ be two models. Let $w \in W$ and $w' \in W$.\\
    $\langle W,R, V \rangle, w $ and $\langle W',R', V' \rangle, w'$ are \textbf{bisimilar} iff 
    \begin{itemize}
        \item There exist a bisimulation $f:W \longrightarrow W'$ between the two models.
        \item $f(w)= w'.$
    \end{itemize} 
\end{mydefinition2}    
\end{mydefinition}
\end{definition}
\medskip \medskip

\begin{theorem}[\textbf{Forcing relation respects bisimilarity}~\href{https://archive.softwareheritage.org/swh:1:cnt:4a011f69b1180019be5d03c6d5f1cec4550055e8;origin=https://github.com/HOLMS-lib/HOLMS;visit=swh:1:snp:2c0efd349323ed6f8067581cf1f6d95816e49841;anchor=swh:1:rev:1caf3be141c6f646f78695c0eb528ce3b753079a;path=/modal.ml;lines=295-299}{\ExternalLink}]\label{thm:forc-bisim}\phantom{Allyouneedislove}
    \begin{mydefinition}
         Let $\mathcal{M}$, $\mathcal{M'}$ be two models. Let $w \in W$ and $w' \in W$. \\If $\mathcal{M}, w $ and $\mathcal{M'}, w'$ are bisimilar then \\ $\forall A \in \mathbf{Form}_{\Box}(\mathcal{M},w \vDash A \iff \mathcal{M'},w' \vDash A)$.
    \end{mydefinition}
\end{theorem}
\begin{proof} Let $A$ be a modal formula.
    We prove this statement by induction on the complexity of $A$. Let $f$ be the bisimulation such that $f(w)=w'$.
    \begin{itemize}
        
        \item If $A= \bot$ \\ $\bot$ is not forced by any world and model, both hold: $ \mathcal{M},w \not \vDash \bot$ and $\mathcal{M'},w' \not \vDash \bot$, Then, by Lemma~\ref{lem:forc_conn}, $\mathcal{M},w \vDash \bot \iff \mathcal{M'},w' \vDash \bot$.
        
        \item If $A=p \in \Phi$ \\
        $\mathcal{M},w \vDash p$ $\xLeftrightarrow{forcing}$  $pVw$ $\xLeftrightarrow{bisimilarity}$ $pV'w'$ $\xLeftrightarrow{forcing}$ $\mathcal{M'},w' \vDash p$.
\newpage
        \item If $A=B \to C $. Induction hypothesis holds for $B$ and $C$. \\
        $\mathcal{M},w \vDash B \to C$ $\xLeftrightarrow{forcing}$  $\mathcal{M},w \not \vDash B $ or $\mathcal{M},w \vDash C$   $\xLeftrightarrow{Inductive\ HP}$ $\mathcal{M'},w' \not \vDash B $ or $\mathcal{M'},w' \vDash C$ $\xLeftrightarrow{forcing}$ $\mathcal{M'},w' \vDash B \to C$.

        \item If $A= \Box B $. Induction hypothesis holds for $B$. \\
        $\mathcal{M},w \vDash \Box B $ $\xLeftrightarrow{forcing}$  $\forall x \in W( wRx \implies \mathcal{M},x \vDash B) $ $\xLeftrightarrow{Inductive\ HP}$ $\forall x \in W( wRx \implies \mathcal{M},f(x) \vDash B) $ $\xLeftrightarrow{Claim \ 1,2}$      
        $\forall y \in W'( w'R'y \implies \mathcal{M},y \vDash B) $ $\xLeftrightarrow{forcing}$ $\mathcal{M'},w' \vDash \Box B $.

        \begin{claim}
            If $\forall x \in W( wRx \implies \mathcal{M},f(x) \vDash B) $ then \\ $\forall y \in W'( w'R'y \implies \mathcal{M},y \vDash B) $.
        \begin{claimproof}
        Let $y \in W'$. We suppose both $\forall x \in W (wRx \implies \mathcal{M},f(x) \vDash B )$ and $w'Ry$.  Notice that, by bisimilarity, $\forall y \in W' (w'Ry \implies \exists z \in W (wRz \ and \ f(z)=y)$.  By the second hypothesis and bisimilarity, we obtain $wRz \ and \ f(z)=y$; thus $w'R'y \implies wRz \implies \mathcal{M}, f(z) \vDash B \iff \mathcal{M}, y \vDash B$. 
        \hfill $\lhd$
        \end{claimproof}
        \end{claim}

        \begin{claim}
            If $\forall y \in W'( w'R'y \implies \mathcal{M},y \vDash B) $ then \\
            $\forall x \in W( wRx \implies \mathcal{M},f(x) \vDash B) $.
        \begin{claimproof}
        Let $x \in W$. We suppose both $\forall y \in W'( w'R'y \implies \mathcal{M},y \vDash B) $ and $wRx$.  By the second hypothesis and bisimilarity, we obtain $w'R'f(x)$; thus $wRx \implies w'R'f(x) \implies \mathcal{M}, f(x) \vDash B$. 
        \hfill $\lhd$
        \end{claimproof}
        \end{claim}
    \end{itemize}
\end{proof}
\medskip \medskip

\begin{theorem}[\textbf{Bisimulation preserves validity over frames }~\href{https://archive.softwareheritage.org/swh:1:cnt:4a011f69b1180019be5d03c6d5f1cec4550055e8;origin=https://github.com/HOLMS-lib/HOLMS;visit=swh:1:snp:2c0efd349323ed6f8067581cf1f6d95816e49841;anchor=swh:1:rev:1caf3be141c6f646f78695c0eb528ce3b753079a;path=/modal.ml;lines=301-305}{\ExternalLink}]\phantom{Allyouneed}
    \begin{mydefinition}
         Let $\mathcal{F}$, $\mathcal{F'}$ be two frames. \\
         If for every model $\mathcal{M}$ based on $\mathcal{F}$ and for every world $w \in W$ there exists a  model $\mathcal{M}'$ based on $\mathcal{F}'$ and a world $w' \in W'$ such that $\mathcal{M}, w $ and $\mathcal{M'}, w'$ are bisimilar then $\forall A \in \mathbf{Form}_{\Box}(\mathcal{F'}\vDash A \implies \mathcal{F} \vDash A)$.
    \end{mydefinition}
\end{theorem}
\begin{proof}
    Let $A$ be a modal formula. \\
    We assume $\mathcal{F'} \vDash A$ then $\forall \mathcal{M'} \ \mathcal{F'}_*(\forall w' \in W' (\mathcal{M'},w' \vDash A))$.
\newpage    
    \begin{claim}
        $\forall \mathcal{M} \ \mathcal{F}_*(\forall w \in W (\mathcal{M},w \vDash A))$

    \begin{claimproof}
    Let $\mathcal{M}$ be a $\mathcal{F}$-based model and $w \in W$. For each $\mathcal{M},w$ there exists a bisimilar $\mathcal{M'},w'$. Then by hypothesis, we know that for these holds $\mathcal{M'},w' \vDash A$. But then, by Lemma~\ref{thm:forc-bisim} $\mathcal{M},w \vDash A$. \hfill $\lhd$
    \end{claimproof} 
    \end{claim}
    
    By definition of validity, $\mathcal{F} \vDash A$ and the thesis holds.
\end{proof}
\newpage
    
\begin{theorem}[\textbf{Bisimulation preserves validity over classes}~\href{https://archive.softwareheritage.org/swh:1:cnt:4a011f69b1180019be5d03c6d5f1cec4550055e8;origin=https://github.com/HOLMS-lib/HOLMS;visit=swh:1:snp:2c0efd349323ed6f8067581cf1f6d95816e49841;anchor=swh:1:rev:1caf3be141c6f646f78695c0eb528ce3b753079a;path=/modal.ml;lines=307-316}{\ExternalLink}]\phantom{A}
    \begin{mydefinition}
         Let $\mathfrak{S}$, $\mathfrak{S'}$ be two classes of frames. \\
         If for every frame $\mathcal{F} \in \mathfrak{S}$ for every model $\mathcal{M}$ based on $\mathcal{F}$ and for every world $w \in W$ there exists a frame $\mathcal{F'} \in \mathfrak{S'}$, a model $\mathcal{M}'$ based on $\mathcal{F}'$ and a world $w' \in W'$ such that $\mathcal{M}, w $ and $\mathcal{M'}, w'$ are bisimilar then \\ $\forall A \in \mathbf{Form}_{\Box}(\mathfrak{S'}\vDash A \implies \mathfrak{S} \vDash A)$.
    \end{mydefinition}
\end{theorem}
\begin{proof}
    Let $A$ be a modal formula. \\
    We assume $\mathfrak{S'} \vDash A$ then $\forall \mathcal{F'} \in \mathfrak{S'} (\mathcal{F'} \vDash A)$.
    \begin{claim}
        $\forall \mathcal{F} \in \mathfrak{S}_*(\mathcal{F}\vDash A)$

    \begin{claimproof}
    Let $\mathcal{F}$ be a frame in $\mathfrak{S}$. For each 
 $\mathcal{F} \in \mathfrak{S} $, $\mathcal{M},w$ there exists a frame $\mathcal{F'}\in \mathfrak{S'}$ and a $\mathcal{M'}$ $\mathcal{F}$-based and a world $w' \in W'$ bisimilar. Then by hypothesis, we know that for these holds $\mathcal{M'},w' \vDash A$. But then, by Lemma~\ref{thm:forc-bisim} $\mathcal{M},w \vDash A$; thus by definition of validity, $\mathcal{F} \vDash A$ \hfill $\lhd$
    \end{claimproof} 
    \end{claim}
    
    By definition of validity over classes of frames, $\mathfrak{S} \vDash A$ and the thesis holds.
\end{proof}

\section{An Introduction to Labelled Sequent Calculi}\label{sec:labelled-sequent-calculi}

In Section~\ref{sec:axiomatic-calulus}, we introduced an axiomatic calculus for normal modal logics, laying the foundation for a formal approach to normal reasoning. In this section, instead, we shift our focus to a different class of \textit{deductive systems}: sequent calculi.

Sequent calculi play a crucial role in proof theory, providing a syntactic framework particularly well-suited for \textit{automated theorem proving}. Since they also offer a \textit{structural} and \textit{rule-based} approach to modal logics, the labelled sequent calculi $\mathsf{G3K, \ G3KT, \ G3K4}$, and $\mathsf{G3KGL}$ will be utilised in Chapter~\ref{chap:4} to develop a \textit{principled decision procedure} within HOLMS for the logics $\mathbb{K}$, $\mathbb{T}$, $\mathbb{K}4$, and $\mathbb{GL}$.

Here, we begin with a general introduction to sequent calculi and an overview of their formal properties. We then examine Sara Negri's calculi for normal modal logics~\cite{negri2005proof, negri2014proofs, negri2011proof}, focusing on $\mathsf{G3K}$ for the minimal modal logic $\mathbb{K}$ and its \textit{modular extensions} $\mathsf{G3K*}$ for the logics $\mathsf{*}$ within the modal cube. Additionally, we discuss $\mathsf{G3KGL}$ for provability logic. All the calculi under analysis allow a \textit{terminating proof strategy}, ensuring a \textit{decision procedure} for their respective logics.

\subsection{Sequent Calculi}

Sequent calculi were introduced by Gentzen in~\cite{Gentzen1935-GENUBD-3, Gentzen1935-GENUBD-4}, who initiated structural proof theory and the study of analytic deductive systems. Later, Ketonen~\cite{Ketonen1945-KETUZP} and Kleene~\cite{Ke52} refined these calculi, developing systems well-known as $\mathsf{G3}$-style calculi. 

After defining the key notions of sequent and derivation and introducing \textit{root-first} proof search, we provide a first example of a $\mathsf{G3}$-style sequent calculus--$\mathsf{G3cp}$ for propositional logic--and we identify the main desiderata for these calculi.

\subsubsection{Key Concepts of Sequent calculi}

Sequent calculi are deductive systems—sets of formal expressions (\textit{initial sequents}) closed under some inference rules—built upon a syntactic object, the sequent:

\begin{definition}[\textbf{Sequent}]\phantom{Allyouneedisloveloveisallyouneed}
    \begin{mydefinition}
        \begin{mydefinition2}
             A \textbf{sequent} is a formal expression of the form: 
             \begin{center}
                 $ \Gamma \Rightarrow \Delta$
             \end{center}
                 where $\Gamma, \ \Delta$ are finite multiset--finite lists modulo permutations--of modal formulas $\mathbf{Form}_{\Box}$  and:
             \begin{itemize}
                 \item $\Rightarrow $ reflects the deducibility relation at the metalevel in the object language;
                 \item$\Gamma$ is called \textit{antecedent};
                 \item $\Delta$ is called consequent.
             \end{itemize}
        \end{mydefinition2}
    \end{mydefinition}
\end{definition}

\noindent The \textit{intended semantic meaning} of this syntactic object is the following: 
\begin{center}
    $\bigwedge \Gamma \implies \bigvee \Delta$
\end{center}
where both `$\bigwedge$', `$\bigvee$' are non-primitive logical operators defined as the finite object-conjunction and--respectively-- object-disjunction of all the formulas contained in the multisets, and `$\implies$' is the usual metalinguistic symbol for implication.

We, then, introduce the fundamental concept of derivation in sequent calculus:

\begin{definition}[\textbf{Derivation}]\phantom{Allyouneedisloveloveisallyouneed}
    \begin{mydefinition}
        \begin{mydefinition2}
A \textbf{derivation} in a $\mathsf{G3}$-style sequent calculus $\mathsf{C}$ of $\Gamma \Rightarrow \Delta $ is a finite rooted tree of sequents such that:
\begin{itemize}
    \item Its leaves are \textbf{initial sequents} or \textbf{zero-ary rules};
    \item Each node of the tree is obtained from the sequents directly above using an \textbf{inference rule} of the calculus;
 \item Its root is the \textbf{end sequent} $\Gamma \Rightarrow \Delta $
\end{itemize}

\noindent A sequent $\Gamma \Rightarrow \Delta $ is \textbf{derivable} in a $\mathsf{G3}$-style sequent calculus $\mathsf{C}$ and we formally write $\mathsf{C} \vdash \Gamma \Rightarrow \Delta $ iff there is a derivation of $\Gamma \Rightarrow \Delta $ in $\mathsf{C}$.
        \end{mydefinition2}
    \end{mydefinition}
\end{definition}

\noindent To derive a given sequent $\Gamma \Rightarrow \Delta$ in a calculus $\mathsf{C}$, we perform a \textit{\textbf{root-first proof search}}. In this approach, we treat $\Gamma \Rightarrow \Delta$ as the end sequent (\textit{root}) and progressively work towards deriving its subformulas by applying the inference rules in a bottom-up manner, until we reach either an initial sequent or a zero-ary rule. An example of this procedure is provided by Figure~\ref{fig:G3K-prove-k}.

Next, we introduce a series of concepts, serving as a sort of glossary, to support the analysis of sequent calculus.

\begin{definition}[\textbf{Context, Main and Actives formulas}]\phantom{Allyoun}
    \begin{mydefinition}
        \begin{mydefinition2}
        Let \texttt{R} be an inference rule of a $\mathsf{G3}$-style sequent calculus $\mathsf{C}$. We distinguish between:
         \begin{itemize}
             \item \textbf{Context}: formulas that remain unchanged during the application of \texttt{R};
             \item \textbf{Main formula}: a formula that occurs in the conclusion of \texttt{R}, is not in the context and contains the logical connective naming \texttt{R};
        \end{itemize}
        \end{mydefinition2}
        \begin{mydefinition2}
        \begin{itemize}
              \item \textbf{Active formula}: a formula that occurs in the premise(s) of \texttt{R} and not in the context.
         \end{itemize}   
        \end{mydefinition2}
    \end{mydefinition}
\end{definition}

\begin{definition}[\textbf{Ammissibility of a rule}]\phantom{Allyouneedislovelove}
    \begin{mydefinition}
        \begin{mydefinition2}
            A rule \texttt{R} is \textbf{admissible} in a $\mathsf{G3}$-style sequent calculus $\mathsf{C}$ iff its conclusion is derivable in $\mathsf{C}$ whenever its premises are derivable in $\mathsf{C}$.
        \end{mydefinition2}
    \end{mydefinition}
\end{definition}

\begin{definition}[\textbf{Invertibility of a rule}]\phantom{Allyouneedisloveloveis}
    \begin{mydefinition}
        \begin{mydefinition2}
            A rule is \textbf{invertible} in a $\mathsf{G3}$-style sequent calculus $\mathsf{C}$ iff  
            its premises are derivable in $\mathsf{C}$ whenever conclusion is derivable in $\mathsf{C}$,
        \end{mydefinition2}
    \end{mydefinition}
\end{definition}

\subsubsection{A Sequent Calculus for Propositional Logic}

Figure~\ref{tab:G3cp} introduces $\mathsf{G3cp}$ sequent calculus for propositional logics by outlining its \textit{initial sequent} and \textit{inference rules}.

\begin{table}[h]
\centering
\begin{tabular}{ll}
 \arrayrulecolor{Blue}\hline
 \\[2ex]
\textbf{Initial sequents:} & \\[3ex]
$w : p, \Gamma \Rightarrow \Delta, w : p$ &  \\[4ex]
\textbf{Propositional rules:}  & \\[3ex]
$\infer[\texttt{L}\bot]{w : \bot, \Gamma \Rightarrow \Delta}{} $ \\[3ex]
$ \infer[\texttt{L}\neg]{w : \neg A, \Gamma \Rightarrow \Delta}{\Gamma \Rightarrow \Delta, w : A} $ & $ 
\infer[\texttt{R}\neg]{\Gamma \Rightarrow \Delta, w : \neg A}{w : A, \Gamma \Rightarrow \Delta}$ \\[3ex]
$\infer[\texttt{L}\land]{w : A \land B, \Gamma \Rightarrow \Delta}{w : A, w : B, \Gamma \Rightarrow \Delta} $ & $ 
\infer[\texttt{R}\land]{\Gamma \Rightarrow \Delta, w : A \land B}{\Gamma \Rightarrow \Delta, w : A \quad\quad\quad \Gamma \Rightarrow \Delta, w : B}$ \\[3ex]
$\infer[\texttt{L}\lor]{w : A \lor B, \Gamma \Rightarrow \Delta}{w : A, \Gamma \Rightarrow \Delta \quad\quad\quad w : B, \Gamma \Rightarrow \Delta}\quad\quad $ & $ 
\infer[\texttt{R}\lor]{\Gamma \Rightarrow \Delta, w : A \lor B}{\Gamma \Rightarrow \Delta, w : A, w : B}$ \\[3ex]
$\infer[\texttt{L}\to]{w : A \to B, \Gamma \Rightarrow \Delta} {\Gamma \Rightarrow \Delta, w : A \quad\quad\quad w : B, \Gamma \Rightarrow \Delta}$ & $ 
\infer[\texttt{R}\to]{\Gamma \Rightarrow \Delta, w : A \to B}{w : A, \Gamma \Rightarrow \Delta, w : B}$ \\[4ex]
 \arrayrulecolor{Blue}\hline
\end{tabular}
\caption{Sequent Calculus $\mathsf{G3cp}$}
\label{tab:G3cp}
\end{table}

\subsubsection{Desiderata for Sequent Calculi}

Below, we list some important \textit{desiderata} for $\mathsf{G3}$-style sequent calculi. A well-designed $\mathsf{G3}$-style calculus satisfies these criteria, ensuring efficient \textit{automation of decision procedures} through a \textit{terminating} \textit{root-first} proof search. 

\begin{enumerate}
    \item \textbf{Analyticity} \\
    As outlined in~\ref{sub:deductive system}, the distinction between analytic and synthetic deductive systems highlights that analytic proofs proceed \textit{bottom-up}. In an analytic derivation, each formula is a subformula of the formulas appearing lower in the derivation branch, adhering to the \textit{subformula principle}. 

    \begin{remark}
        There are also \textit{weaker} forms of analyticity--and subformula principle--that can still ensure the termination of the proof search.
    \end{remark}
 
     Furthermore, the proof search in analytic calculi is \textit{guided} by the structure of the formula to be proved, meaning no guesses are needed, as each inference step is determined by the logical structure of the formula.
     
    \begin{remark}
    While analyticity and the various forms of the subformula principle play a crucial role in proof search, it does not necessarily guarantee termination in general.
    \end{remark}

    \item \textbf{Modularity}\\
    The modularity of a system measures its capacity to represent different logical systems in a \textit{flexible} manner. Specifically, a modular system can be \textit{extended} or \textit{adapted} to accommodate various logics by \textit{adding} or \textit{adjusting} inference rules and staring expressions as needed, while maintaining a coherent and consistent overall framework.    
    This property enables the \textit{systematic exploration} of various logical systems within a \textit{unified framework}, making it easier to analyse and compare different logics.

    \item \textbf{Invertibility of all the rules}\\
    This property prevents the need of \textit{backtracking}--retracing steps when a wrong path is taken--during the proof search process. In other words, it ensures that at each step of the proof search \textit{no information is discarded}, so that the construction of a single derivation tree is sufficient for determining derivability.    

    \item \textbf{Admissibility of structural rules} \\
    Inference rules can be divided into two categories, logical and structural rules.
       \begin{itemize}
           \item \textbf{Logical rules}:  enable the introduction of a logical operator $\circ$ either on the left (\texttt{L}$ \circ$) or right (\texttt{R}$\circ$) side of a sequent.
           \item \textbf{Structural rules}: do not introduce logical operators but rather manipulate the structure of sequents. They commonly include:
           \begin{enumerate}
               \item \textbf{Cut rule}: allows the elimination of formulas by combining two derivations that share those formulas.  \texttt{Cut} is closely related to axiomatic rule \texttt{MP}, as it allows us to eliminate an intermediate step of a derivation.

               \[\infer[\texttt{Cut}_\mathsf{G3cp}]{\Gamma, \  \Gamma' \Rightarrow \Delta, \ \Delta'}{\Gamma \Rightarrow \Delta, \ A \ \ \ \ \ \ \  A, \ \Gamma' \Rightarrow \Delta'}\]

                  \texttt{Cut} does not preserve the \textit{subformula property} because its application combines two derivations that share a common formula by entirely removing that formula from the proof. This can result in a situation where a formula in the premises no longer satisfies the subformula principle, as it no longer appears directly in the derivation structure or as a subformula of another formula in the proof.

                 For this reason, if \texttt{Cut} is taken as a primitive rule, it is desirable to prove that it is \textit{eliminable}—that is, any proof using \texttt{Cut} can be transformed into a proof that does not rely on it. This result, known as \textbf{cut elimination}, ensures that the system remains \textit{cut-free}. 

                 \begin{remark}
                     Cut elimination is fundamental in proof theory, as it \textit{often} has the subformula property as an immediate consequence. It brings us closer to guaranteeing \textit{analiticity} and \textit{termination property}.
                 \end{remark}
               \medskip \medskip
               
               \item \textbf{Weakening}: permits the addition of formulas to either side of the sequent without affecting the derivation.
               \[\infer[\texttt{LW}_\mathsf{G3cp}]{A, \ \Gamma \Rightarrow \Delta}{\Gamma \Rightarrow \Delta} \ \ \ \ \ \ \ \ \ \infer[\texttt{RW}_\mathsf{G3cp}]{\Gamma \Rightarrow \Delta, \ A}{\Gamma \Rightarrow \Delta}\]
               \medskip
               
               \item \textbf{Contraction}: allows the removal of duplicated formulas from the sequent, ensuring that no formula occurs more than once in a sequent if unnecessary.
               \[\infer[\texttt{LC}_\mathsf{G3cp}]{A, \ \Gamma \Rightarrow \Delta}{A, \ A, \ \Gamma \Rightarrow \Delta} \ \ \ \ \ \ \ \ \ \infer[\texttt{RC}_\mathsf{G3cp}]{\Gamma \Rightarrow \Delta, \ A}{\Gamma \Rightarrow \Delta, \ A, \ A}\]
           \end{enumerate}     

           \begin{remark}
           In the following pages, we will present a \textit{\textbf{logical variant}} of the sequent calculi under analysis, which does not include primitive structural rules, but for which the \textit{admissibility of structural rules} has been proved. For example, we have already presented the logical variant of $\mathsf{G3cp}$ for which the structural rules have proved to be admissible.
           This method both ensures the \textit{flexibility of the system} and the \textit{analyticity} of the calculus.
           
           An opposite but equivalent method is to present the \textit{general variant}--where both structural and logical rules are taken as primitive--and prove the eliminability of the structural rules, e.g. cut elimination.
           \end{remark}
           
        \end{itemize}
        
     \item \textbf{Termination} \\
     Proof search in a decidable system must not lead to infinite derivations. If the final state of a procedure generates a derivation, then the end sequent is a \textit{theorem}; otherwise is generally possible to extract a \textit{refutation} of the sequent from the failed proof search. 
     
     \begin{remark}
     Since we aim to provide sequent calculi that ensure decision procedures for normal logics, it is crucial to present calculi that allow \textit{terminating proof search}. This will be an important point in the subsequent discussion.
     \end{remark}

\end{enumerate}

\subsection{Sequent Calculi for Modal Logics}\label{sub:labelled-sequent}

As Sara Negri emphasizes in her seminal article introducing sequent calculi that satisfy the desiderata mentioned above for $\mathbb{GL}$ and the entire modal cube:

\begin{quote}
``The possibility of a systematic development of a proof theory of modal logic in terms of Gentzen sequent calculus has been looked at with overall skepticism. The question has been stepped over very often in the literature by choosing Hilbert type proof systems.''~\cite{negri2005proof}
\end{quote}

Thanks to this and earlier contributions~\cite{Vigano2000-VIGLNL}, the landscape has significantly changed with the introduction of \textit{internalisation techniques} for embedding semantic notions within sequent calculi for non-classical logics. This approach enriches $\mathsf{G3}$-style calculi by incorporating elements of relational semantics into the proof system.  
There are two main strategies for internalisation:  

\begin{enumerate}
    \item [(a)] \textbf{Explicit Internalisation}: enriches the \textit{language} of the calculus itself to explicitly represent semantics elements.\\
    $\triangleleft$ Usually satisfy the basic desiderata of $\mathsf{G3}$-style systems;  \\
    $\triangleright$  The termination of proof search is sometimes difficult to prove and, in many cases, is not optimal in terms of complexity.

    \item [(b)] \textbf{Implicit Internalisation}: enriches the \textit{structure} of the sequents by introducing additional structural connectives beyond the standard `$\Rightarrow$' and `,’.  
    \\
    $\triangleleft$ Enables efficient decision procedures.  \\
    $\triangleright$  Difficult to design and generally lacks \textit{modularity}.
\end{enumerate}

In what follows, we will focus on the sequent calculi \textit{modularly} developed by Sara Negri for the modal cube and $\mathbb{GL}$. These calculi, based on \textit{explicit internalisation} techniques, offer a systematic and structured approach to modal proof theory, ensuring all the desiderata above and hence decision procedures.

We begin by presenting the enriched syntax and the $\mathsf{G3}$-style calculus for the minimal normal modal logic $\mathbb{K}$, denoted as $\mathsf{G3K}$. Next, we introduce a method for extending this system to encompass a broader class of normal modal logics, leading to the systems $\mathsf{G3K*}$ covering the entire cube. Finally, we discuss the sequent calculus $\mathsf{G3KGL}$, specifically designed for provability logic $\mathbb{GL}$ but also a \textit{modular extension} of $\mathsf{G3K}$.

\subsubsection{The Enriched Syntax}

Since we focus on explicit internalised calculi that extend the language of modal logic $\mathcal{L}_{\Box}$ with symbols representing semantic concepts ($\mathcal{L}_{LS}$), we have to introduce an enriched syntax representing Kripke semantics ($\mathbf{Form}_{LS}$).

In particular, modal formulas are \textbf{labelled} by additional items, e.g $x:A$, to represent within the syntax the forcing relation $\mathcal{M},w \vDash A$. 
To properly handle the semantics of modal operators, sequents may also include \textbf{relational atoms}, such as 
$xRy$, explicitly encoding accessibility relations between worlds. Due to this syntactic structure, these calculi are commonly known as \textit{\textbf{labelled sequent calculi}}.

\begin{definition}[\textbf{Alphabet}]
    $\mathcal{L}_{LS} = \mathcal{L}_{\Box} \cup L \cup \{R\}$ 
\begin{mydefinition}
    \begin{mydefinition2}
        The \textbf{alphabet} of propositional modal language $\mathcal{L}_{\Box}$ consists of:
\begin{itemize}
    \item the alphabet $\mathcal{L}_{\Box}$;
    \item an \textbf{infinite denumerable} set $L$ of \textbf{labels}: $w, x, y, \dots$;
    \item a binary predicate $R$.
\end{itemize}
    \end{mydefinition2}
\end{mydefinition}        
\end{definition}

\begin{convention}[\textbf{Metavariables for  $\mathcal{L}_{LS}$ formulas}]
    We will use Greek letters $\varphi, \psi, \delta, \dots$ as metavariables for  $\mathcal{L}_{LS}$ formulas.
\end{convention}

\begin{definition}[\textbf{Formulas} of  $\mathcal{L}_{LS}$ ]
 $\varphi \in \mathbf{Form}_{LS} \Coloneqq x:A \ | \ xRy $ 
 \begin{mydefinition}
\begin{mydefinition2}
    The set $\mathbf{Form}_{LS}$ of formulas on $\mathcal{L}_{LS}$ is defined by induction as follows:
\begin{itemize}
    \item if $A \in \mathcal{L}_{\Box}$ and $x \in L$ then $x:A\in \mathbf{Form}_{LS}$ $($\textbf{labelled formulas}$)$;
    \item if  $x,y \in L$  then $xRy \in \mathbf{Form}_{LS}$ $($\textbf{relational atoms}$)$;
    \item nothing else is in $\mathbf{Form}_{LS}$.
\end{itemize}
\end{mydefinition2}
\end{mydefinition}
\end{definition}

\begin{definition}[\textbf{Labelled Sequent}]\phantom{Allyouneedisloveloveisallyouneed}
    \begin{mydefinition}
        \begin{mydefinition2}
             A \textbf{labelled sequent} is a formal expression of the form $ \Gamma \Rightarrow \Delta$, \\ where $\Gamma, \ \Delta$ are finite multiset of formulas in $\mathbf{Form}_{LS}$.
        \end{mydefinition2}
    \end{mydefinition}
\end{definition}

\begin{remark}[: historical note on sequent Calculi ]
    The origins of labelled calculi trace back to the 1950s with Kanger’s work~\cite{Kanger1957-KANPIL}, later refined by Kripke~\cite{Kripke1963-KRISAO-2} and further developed by Gabbay~\cite{Gabbay1996-GABLDS}. However, it was with Negri’s work~\cite{negri2005proof} that labelled sequent calculi became an established and well-structured methodology within proof theory. Significant extensions to other non-classical logics, along with refinements, have been contributed by both Sara Negri~\cite{negri2017proof, negri2021proof} and Francesca Poggiolesi~\cite{poggiolesi2010gentzen, poggiolesi2016natural}.
\end{remark}

\subsubsection{A Labelled Sequent Calculus for K}
Table~\ref{tab:G3K} outlines the initial sequents and the logic rules of the labelled sequent calculus $\mathsf{G3K}$ for the minimal normal logic $\mathbb{K}$.

All the inference rules are obtained by the inductive Definition~\ref{def:truth} of forcing:  while the propositional rules are straightforward, the modal rules governing the $\Box$ operator require a reflection. The forcing condition for  $\Box$ is given by:
\medskip
\begin{center}
    $\mathcal{M}, w \vDash \Box A $ iff $\forall y (xRy \implies\mathcal{M}, x \vDash A) $\footnote{$\mathcal{M}$ is a model and $w$ is a world inhabiting $\mathcal{M}$.}.
\end{center}

This condition leads to the following modal rules:
\medskip \medskip

\begin{tabular}{cc}
   $\infer[\texttt{L}\Box]{ w R x, w : \Box A, \Gamma \Rightarrow \Delta}{x : A, w R x, w:\Box A, \Gamma \Rightarrow \Delta} \quad \quad \quad \quad \quad \quad $ & $ 
\infer[\texttt{R}\Box \ (x!)]{\Gamma \Rightarrow \Delta, w : \Box A}{w R x, \Gamma \Rightarrow \Delta, x : A}$  
\end{tabular}
\smallskip \medskip

\begin{remark}
    The main formula `$x: \Box A$' is repeated in \texttt{L}$\Box$ in order to make the rule invertible.
\end{remark}
\medskip\medskip

\begin{table}[h]
\centering
\begin{tabular}{ll}
 \arrayrulecolor{Blue}\hline 
 \\[3ex]
\textbf{Initial sequents:} & \\[3ex]
$w : p, \Gamma \Rightarrow \Delta, w : p$ &  \\[4ex]
\textbf{Propositional rules:}  & \\[3ex]
$\infer[\texttt{L}\bot]{w : \bot, \Gamma \Rightarrow \Delta}{} $ \\[3ex]
$ \infer[\texttt{L}\neg]{w : \neg A, \Gamma \Rightarrow \Delta}{\Gamma \Rightarrow \Delta, w : A} $ & $ 
\infer[\texttt{R}\neg]{\Gamma \Rightarrow \Delta, w : \neg A}{w : A, \Gamma \Rightarrow \Delta}$ \\[3ex]
$\infer[\texttt{L}\land]{w : A \land B, \Gamma \Rightarrow \Delta}{w : A, w : B, \Gamma \Rightarrow \Delta} $ & $ 
\infer[\texttt{R}\land]{\Gamma \Rightarrow \Delta, w : A \land B}{\Gamma \Rightarrow \Delta, w : A \quad\quad \Gamma \Rightarrow \Delta, w : B}$ \\[3ex]
$\infer[\texttt{L}\lor]{w : A \lor B, \Gamma \Rightarrow \Delta}{w : A, \Gamma \Rightarrow \Delta \quad\quad\quad w : B, \Gamma \Rightarrow \Delta}\quad\quad $ & $ 
\infer[\texttt{R}\lor]{\Gamma \Rightarrow \Delta, w : A \lor B}{\Gamma \Rightarrow \Delta, w : A, w : B}$ \\[3ex]
$\infer[\texttt{L}\to]{w : A \to B, \Gamma \Rightarrow \Delta} {\Gamma \Rightarrow \Delta, w : A \quad\quad\quad w : B, \Gamma \Rightarrow \Delta}$ & $ 
\infer[\texttt{R}\to]{\Gamma \Rightarrow \Delta, w : A \to B}{w : A, \Gamma \Rightarrow \Delta, w : B}$ \\[4ex]
\textbf{Modal rules}: & \\[3ex]
$\infer[\texttt{L}\Box]{ w R x, w : \Box A, \Gamma \Rightarrow \Delta}{x : A, w R x, w:\Box A, \Gamma \Rightarrow \Delta} $ & $ 
\infer[\texttt{R}\Box \ (x!)]{\Gamma \Rightarrow \Delta, w : \Box A}{w R x, \Gamma \Rightarrow \Delta, x : A}$ \\[4ex] 
\end{tabular}
\footnotesize{$(x!)$ means that the label $x$ does not occur in $\Gamma, \Delta$. }
\begin{tabular}{cc}
\phantom{AAAAAAAAAAAAAAAAAAAAAAAAAAAAAAAAAAAAAAAAAAAAAAAAAAAA} \\
 \arrayrulecolor{Blue}\hline
\end{tabular}
\caption{Labelled Sequent Calculus $\mathsf{G3K}$}
\label{tab:G3K}
\end{table}

\newpage
\noindent Reference articles by Sara Negri~\cite[\S 4,6]{negri2005proof} establishes the following results for $\mathsf{G3K}$:
\begin{itemize}
    \item Structural rules--cut, weakening and contraction--are \textit{admissible} in $\mathsf{G3K}$;
    \item All the rules of $\mathsf{G3K}$ are \textit{invertible};
    \item $\mathsf{G3K}$ is \textit{complete} w.r.t~axiomatic calculus    because everything that is provable in $\mathsf{K}$ ($ \mathcal{S. \varnothing} \vdash A$) is also derivable in $\mathsf{G3K}$ ($\mathsf{G3K}\vdash \Rightarrow A$):
    \begin{enumerate}
        \item The distribution axiom $\mathbf{K}$ is derivable in $\mathsf{G3K}$. Below we provide proof of this fact, which is also an example of derivation in $\mathsf{G3K}$.
        \begin{figure}[h]
            \centering
            \[\infer[\text{R}\to]{\Rightarrow \ x: \Box( A \to B)  \to (\Box A \to \Box B)}{\infer[\text{R}\to]{  \ x: \Box (A \to B)\ \Rightarrow \ x:\Box A \to \Box B}{\infer[\text{R}\Box]{x: \Box A, \ x:\Box( A \to  B) \ \Rightarrow \ x: \Box B  }{\infer[\text{L}\Box \times 2]{xRy, \ x:\Box A, \ x: \Box(A \to  B) \Rightarrow y:B\ }{ \infer[\text{L}\to \ \ ]{y: A, \ y:A \to B, xRy, \ x:\Box A, \ x: \Box(A \to  B) \Rightarrow y:B}{\infer[]{\textcolor{Blue}{y:A}, xRy, \dots \Rightarrow y:B, \ \textcolor{Blue}{y:A}}{} \quad \quad \quad \infer[]{\textcolor{Blue}{y:B}, \ y:A, xRy, \dots \Rightarrow \textcolor{Blue}{y:B}}{}}}}}}\]
            \caption{Derivation of dstribution schema in $\mathsf{G3K}$ }
            \label{fig:G3K-prove-k}
        \end{figure}
        
Since is provable by induction that sequents  $x:A, \ \Gamma \Rightarrow \Delta, \ x:A$ with $A$ an arbitrary modal formula are derivable in $\mathsf{G3K}$, the above root-first prove search is a derivation of $\mathbf{K}$.

\medskip
        
        \item The necessitation rule \texttt{RN} is admissible in $\mathsf{G3K}$ and modus ponens \texttt{MP} is proved to be admissible in $\mathsf{G3K}$ by cut rule and the  rule inverse to \texttt{R}$\to$. \smallskip
        \[\infer[\texttt{RN}]{\Rightarrow x: \Box A}{\Rightarrow x:A} \quad\quad\quad\quad \infer[\texttt{MP}]{\Rightarrow x: B}{\Rightarrow x:A \to B \quad\quad\Rightarrow x:A}\]
    \end{enumerate}

 \medskip   
    \item $\mathsf{G3K}$ is \textit{complete} w.r.t~relational semantics because for any modal formula $A$:
    \begin{center}
        $\mathfrak{F}\vDash A \implies \mathcal{S. \varnothing} \vdash A \implies \mathsf{G3K}\vdash \Rightarrow A$.
    \end{center}
    That is, since the axiomatic calculus is sound and complete for Kripke frames, and $\mathsf{G3K}$ is complete w.r.t~ the axiomatic calculus, $\mathsf{G3K}$ is also complete w.r.t. Kripke semantics. Similarly one can prove the \textit{soundness} of this system;
    \item \textit{Analyticity} of $\mathsf{G3K}$ hold in a less strict version, \textit{\textbf{subterm principle}}:
    \begin{mydefinition}
        \textit{All terms $($variables and worlds$)$ in a derivation are either eingenvariables or terms in the conclusion}.
    \end{mydefinition}
    \item $\mathsf{G3K}$  allows a \textit{terminating proof search} with an effective bound on proof length;
    \item $\mathsf{G3K}$ provides a \textit{decision procedure} for theoremhood in $\mathbb{K}$;
    \item $\mathsf{G3K}$ is \textit{modular}, since the following Fact~\ref{fact:modularity-G3K} has been proved in~\cite[\S 6]{negri2008structural}.
\end{itemize}

\subsubsection{Labelled Sequent Calculi for the Systems within the Cube}

 The sequent calculus $\mathsf{G3K}$ is a \textit{modular} system, as demonstrated by Negri and von Plato in~\cite[\S 6]{negri2008structural} and reported in Fact~\ref{fact:modularity-G3K}. This modularity allows it to be extended with appropriate rules to formalise the class of ``geometric'' logics. 

\begin{definition}[\textbf{Geometric Property}]\phantom{Allyouneedislove}
    \begin{mydefinition}
        \begin{mydefinition2}
          A property of binary relation is said to be \textbf{geometric} iff \\ it has shape $\forall \vec{x} (A \to B)$, where $A,B$ are first-order formulas that do not contain universal quantifiers or implications.
        \end{mydefinition2}
    \end{mydefinition}
\end{definition}

\begin{fact}[\textbf{Modularity of $\mathsf{G3K}$ on ``geometric systems''}]\label{fact:modularity-G3K}\phantom{A}
    \begin{mydefinition}
     For any extension of $\ \mathbb{K}$ whose \textit{characteristic} property is (co)geometric,
     there exists an adequate w.r.t~relational semantics
     calculus that is an extension of $\mathsf{G3K}$ with appropriate (co)geometric rules.
    \end{mydefinition}
\end{fact}

In~\cite[\S 2]{negri2005proof}, Sara Negri also provided a procedure to construct the (co)geometric rule corresponding to a given (co)geometric property. Since the properties of seriality, reflexivity, transitivity, euclideanity and symmetry are all geometric, we can construct their corresponding rules, as illustrated in Table~\ref{tab:corresponding-rules}.

\begin{table}[h]
    \centering
    \renewcommand{\arraystretch}{1.35} 
    \begin{tabular}{l | l | l }
        \toprule
         & \textbf{Characteristic Property} & \textbf{Corresponding Rule} \\
        \midrule
        $\mathbf{D}$ & \rule{0pt}{2.6ex} \textbf{serial}: $\forall w \exists y  (wRy)$ & 
        $\infer[\texttt{Ser}]{\Gamma \Rightarrow \Delta}{w R y, \Gamma \Rightarrow \Delta}$   \\
        \midrule
        $\mathbf{T}$ & \rule{0pt}{2.6ex} \textbf{reflexive}: $\forall w(wRw)$ & 
        $\infer[\texttt{Ref}]{\Gamma \Rightarrow \Delta}{w R w, \Gamma \Rightarrow \Delta}$  \\
        \midrule
        $\mathbf{4}$ &  \rule{0pt}{2.6ex} \textbf{transitive}: $\forall w,x,z(wRx \land xRz \Rightarrow wRz)$ & 
        $\infer[\texttt{Trans}]{w R x, xRz, \Gamma \Rightarrow \Delta}{wRz, w R x, xRz, \Gamma \Rightarrow \Delta}$  \\ 
        \midrule
        $\mathbf{5}$ & \rule{0pt}{2.6ex} \textbf{euclidean}: $\forall w,x,y (wRx \land wRy \Rightarrow yRx)$ & 
        $\infer[\texttt{Eucl}]{w R x, wRy, \Gamma \Rightarrow \Delta}{yRx, w R x, wRy, \Gamma \Rightarrow \Delta}$\\
       \midrule
        $\mathbf{B}$ & \rule{0pt}{2.6ex} \textbf{symmetric}: $\forall w,x (wRx \Rightarrow xRw)$ & 
        $\infer[\texttt{Sym}]{w R x, \Gamma \Rightarrow \Delta}{xRw, wRx, \Gamma \Rightarrow \Delta}$\\
        \midrule
        & \rule{0pt}{2.6ex} \textbf{irreflexive}: $\forall w (\neg wRw)$ & 
        $\infer[\texttt{Irref}]{w R w, \Gamma \Rightarrow \Delta}{}$\\
        \bottomrule
    \end{tabular}
    \caption{Geometric rules correspondent to geometric characteristic properties.}
    \label{tab:corresponding-rules}
\end{table}

By extending $\mathsf{G3K}$ with the appropriate rules as showed in Table~\ref{tab:systems-calculi}, we can derive labelled sequent calculi for the entire modal cube. The \textit{admissibility of structural rules} in these calculi has been proved in a unified manner, along with the \textit{invertibility} of all the rules and \textit{subterm principle}. 

\begin{table}[h]
    \centering
    \begin{tabular}{|l|l|l|}
        \hline
        \textbf{Logic} & \textbf{Calculus} &\textbf{Definition} \\
        \hline
        $\mathbb{D}$ & $\mathsf{G3KD}$ &$\mathsf{G3K}$ + \texttt{Ser} \\
        $\mathbb{T}$ & $\mathsf{G3KT}$ & $\mathsf{G3K}$ + \texttt{Refl} \\
        $\mathbb{K}4$ & $\mathsf{G3KK4}$ & $\mathsf{G3K}$ + \texttt{Trans}\\
        $\mathbb{S}4$ & $\mathsf{G3KS4}$ & $\mathsf{G3K}$ + \texttt{Refl} + \texttt{Trans}\\
        $\mathbb{B}$ & $\mathsf{G3KB}$ & $\mathsf{G3K}$ + \texttt{Refl} + \texttt{Sym} \\
        $\mathbb{S}5$ & $\mathsf{G3KS5}$ & $\mathsf{G3K}$ + \texttt{Refl} + \texttt{Eucl}\\
        \hline
    \end{tabular}
    \caption{Logics within the cube and correspondent labelled sequent calculi $\mathsf{G3K*}$}
    \label{tab:systems-calculi}
\end{table}

Each of these calculi easily proves the derivability of its specific schema(ta), and thus its \textit{completness} w.r.t. relational semantics.

Importantly, proving the \textit{termination} of these calculi provides a clear argument for the \textit{decidability} of the corresponding logics, along with a construction of \textit{countermodels} in case the proof search fails. Due to the structural correspondence between the proof tree generated during a backward proof search and the relational level, proving the subformula property for these calculi--\textit{subterm property}--becomes equivalent to showing that such a root generates only a finite number of new labels.

Table~\ref{tab:modal-systems-properties} summarises the validity of these desiderata for each logic under consideration.

\begin{table}[h]
    \centering
    \renewcommand{\arraystretch}{1.5} 
    \begin{tabular}{|l|l|l|l|l|l|}
        \hline
        \textbf{Logic} & \textbf{Calculus}  & \textbf{Struct.Ad.} & \textbf{Subterm} & \textbf{Termination} & \textbf{Decision} \\
        \hline
        $\mathbb{D}$ & $\mathsf{G3K}$+ \texttt{Ser} &  \checkmark & \checkmark & \checkmark~\cite{DBLP:conf/lics/GargGN12} & \checkmark~\cite{DBLP:conf/lics/GargGN12} \\
        \hline
        $\mathbb{T}$ & $\mathsf{G3K}$+ \texttt{Refl} & \checkmark & \checkmark & \checkmark & \checkmark \\
        \hline
        $\mathbb{K}4$ & $\mathsf{G3K}$+ \texttt{Trans}& \checkmark & \checkmark & \checkmark~\cite{DBLP:conf/lics/GargGN12} & \checkmark~\cite{DBLP:conf/lics/GargGN12} \\
        \hline
       $ \mathbb{S}4$ & $\mathsf{G3K}$+\texttt{Refl}+\texttt{Trans}&  \checkmark & \checkmark & \checkmark & \checkmark \\
        \hline
        $\mathbb{B}$ & $\mathsf{G3K}$+\texttt{Refl}+\texttt{Sym} &  \checkmark & \checkmark & \checkmark & \checkmark \\
        \hline
       $ \mathbb{S}5$ & $\mathsf{G3K}$+ \texttt{Refl}+\texttt{Eucl} &  \checkmark & \checkmark & \checkmark & \checkmark \\
        \hline
    \end{tabular}
    \caption{Properties of Modal Systems}
    \label{tab:modal-systems-properties}
\end{table}

Although it is not possible to propose a simple terminating proof search for $\mathsf{G3KD}$ and $\mathsf{G3KK4}$~\cite{negri2005proof}, the literature~\cite{DBLP:conf/lics/GargGN12}  has provided modular terminating \textit{proof strategies}, specifically tailored to prevent loops, ensure termination, and provide decision procedures for each system in the modal cube.

\subsubsection{A labelled sequent calculus for GL}

$\mathbb{GL}$ \textit{characteristic} property is converse well-foundness-transitivity, or alternatively finiteness-irreflexivity-transitivity. Since \textit{finiteness} and \textit{noetherianity} both are intrinsically second-order\footnote{See Remark~\ref{rmk:expressivity_power}.} and then not (co)geometric, we cannot define a labelled calculus for $\mathbb{GL}$ ``naive'' extending $\mathsf{G3K}$ with (co)geometric rules.

In~\cite[\S 5]{negri2005proof}, Sara Negri provides a \textit{modular adaptation} of $\mathsf{G3K}$ by relying on the following characterisation of the `$\Box$' that holds for models based on frames in $\mathfrak{ITF}$:

\begin{center}
    $\mathcal{M}, w \vDash \Box A $ iff $\forall y (xRy \ and \ \mathcal{M}, y \vDash \Box A \implies\mathcal{M}, x \vDash A) $\footnote{$\mathcal{M}$ is a model based on a frame in  $\mathfrak{ITF}$ and $w$ is a world inhabiting $\mathcal{M}$.}.
\end{center}

In such a perspective, this condition suggests to modify the \texttt{R}$\Box$ rule as follows:
\medskip \medskip

\begin{tabular}{cc}
   $\infer[\texttt{L}\Box]{ w R x, w : \Box A, \Gamma \Rightarrow \Delta}{w : A, w R x, x:\Box A, \Gamma \Rightarrow \Delta} \quad  \quad $ & $ 
\infer[\texttt{R}\Box_{Lob} \ (x!)]{\Gamma \Rightarrow \Delta, w : \Box A}{w R x, x:\Box A,  \Gamma \Rightarrow \Delta, x : A}$  
\end{tabular}
\smallskip \medskip \medskip

The resulting sequent calculus $\mathsf{G3KGL \coloneq G3K}\ - \ $\texttt{R}$\Box \ + \  $\texttt{R}$\Box_{Lob} \ + \ $\texttt{Trans}  $+$ \texttt{Irrefl} consists of the initial sequents, the propositional rules and $\texttt{L}\Box$, together with the new modal rule $\texttt{R}\Box_{Lob}$, transitive and irreflexive inference rules.

Negri's seminal article~\cite{negri2005proof} establishes the \textit{admissibility of structural rules} in $\mathsf{G3KGL}$, proves the \textit{invertibility} of all its rules, and demonstrates the \textit{completeness} of the system. Furthermore, her subsequent paper~\cite{negri2014proofs} ensures a \textit{terminating} proof search by introducing \textit{saturation conditions} on sequents, preventing the application of useless rules (Tait-Schütte-Takeuti reduction).

If the \textit{proof search} fails, a \textit{countermodel} can be constructed from the root of the derivation tree. Specifically, we consider the top sequent of an open branch, assume that all $\mathcal{L}_{LS}$-formulas in its antecedent are true, and that all the labelled formulas in its consequent are false.
\chapter{An Introduction to HOL Light}\label{chap:2}

Starting from its name and its beginnings, HOLMS is meant to be a HOL Light Library. Therefore, this chapter will introduce the theorem prover HOL Light and its features. 

\section{What is HOL Light?}

HOL Light is an interactive theorem prover that uses high-order logic, that is to say, it is a computer program developed to assist the user in proving mathematical theorems in a formal way using high-order logic. There are several versions of HOL, which originated with the Edinburgh LCF project in the eighties. In the following sections, we will describe the advantages of HOL and, in particular, of HOL Light.

For a more detailed introduction, see HOL Light manual \cite{harrisonmanual}, its tutorial \cite{harrisontutorial} and its reference manual \cite{harrisonreference} which are all presented in the HOL webpage \cite{hol-light-webpage}.

\subsection{Theorem Provers}
We are used to thinking about many computer programs dealing with numerical problems, such as pocket calculators or computer algebra systems,  as calculating the answer to problems. Although, you can also look at them as systems for producing theorems on demand in a certain class. Indicating with the symbol $\vdash $ assertions that are mathematical theorems, we could switch the perspective:
\[ the\;answer\;to\;the\;problem\;of\;adding\;3\;and\;2\;is\;5 \;\;\;\;\longrightarrow \;\;\;\;
   \vdash3+2=5
\]
A theorem prover does exactly that, it produces theorems on user demand.
In \cite{harrisontutorial}, Harrison observes a traditional division in theorem proving between: \begin{enumerate}
    \item \textit{Fully automated systems} which try to prove theorems without user assistance, such as NQTHM \cite{BoyerMooore1979NQTHM}.
    \item \textit{Proof checker} that simply checks the formal correctness of the proof generated by the user, such as AUTOMATH \cite{deBrujin1980AUTOMATH}.
\end{enumerate}

\subsection{HOL}

HOL, instead, is an \textit{interactive} theorem prover, a combination of the two approaches: with the advantage that the user can direct the overall structure of the proof, like in a \textit{proof checker}; but he doesn't need to worry about the details of the many subproofs using a sort of \textit{automated system}.

Therefore, compared with calculators or computer algebra systems (CASs) and among theorem provers, the different versions of HOL have three great advantages:
\begin{enumerate}
    \item \textbf{Wide mathematical range}: The theorems produced by HOL cover a wide mathematical range, e.g. involving infinite sets and \textit{quantifiers} like "there exists some set of natural numbers that..." or "for any real number...".  Conversely, calculators and CASs' theorems are mainly equations with implicitly universally quantified variables. 
    
    \item \textbf{Reliability:} HOL is characterized by rigorous proofs and unequivocal theorems. By contrast, calculators and CASs are often open to doubt and often leave out side conditions. HOL is relative-reliable because one simply needs to have confidence in the small amount of primitive inference rules (logical core) to be sure of the correctness of its theorems. In this sense, HOL is reliable relative to the logical core.
    
    \item \textbf{Extensibility:} Like good calculators and CASs but unlike many theorem provers, HOL is programmable. Consequently, the user can prove theorems automatically with the available functions, but he can also extend the system and produce new functions by implementing them in terms of the original ones. 
\end{enumerate}

It is important to note that these advantages come at the price of a uniform mechanism of derivation that is less efficient and optimized than calculators and CASs for typical problems.

As described in detail in Harrison's appendix \cite{harrisontutorial}, relative-reliability and extensibility of the system are the main characteristics of these theorems provers since their beginnings as Edinburgh LCF \cite{Gordon1979EdinburghLCFy} and its evolution supporting high order logic HOL. This advantageous combination was unique in LCF's times, and it was made possible by the development of a metalanguage suited to the task of theorem proving. Thus in the LCF project, the new programming language ML (Meta Language) was born. It will be introduced in the following section, trying to suggest why the metalanguage ML provides reliability and extensibility.

Several theorem provers in active use are derived from the original HOL developed by Mike Gordon in the 1980s, as the following graph shows (Figure~\ref{HOL_evolution}). A recent offspring of the HOL-based family is the fully verified interactive theorem prover for higher-order logic, Candle~\cite{abrahamsson_et_al:LIPIcs.ITP.2022.3}.
For a more comprehensive introduction to HOL-based systems, see \cite{Gordon1993-GORITH}.

\begin{figure}
    \centering
    \includegraphics[width= 0.6\textwidth]{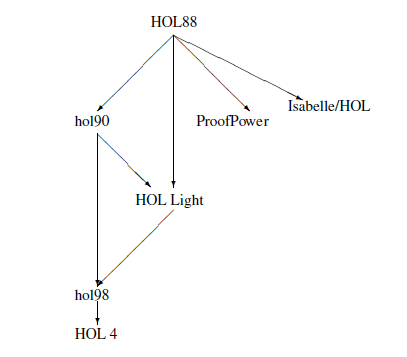} 
    \caption{Schema of the evolution of the theorem prover HOL since 1988 in \cite{harrisontutorial}.}
    \label{HOL_evolution}
\end{figure}

\subsection{HOL Light}

In \cite{harrisonmanual} HOL Light is described as a
tower of functions. At the bottom of this tower lay a small set of primitive operations (logical core) that ultimately produce all theorems. To make the system more user-friendly, higher-level functions are defined in terms of primitive functions, then these functions are
themselves used to build up additional layers, and so on. Because HOL Light is programmable, any user can add pieces to the tower and, because theorems are ultimately derived from the logical core, the system is relative-reliable, and it cannot lead to prove "false theorems". 

The HOL Light logical core stands for its small size and reliability, but also for its expressive power, which is comparable to and, sometimes, more expressive than other HOL versions. Moreover, HOL Light has been used for significant applications in formal verification\footnote{See \cite{harrisonfloatingpoint} for a formal verification of
floating-point algorithms using HOL Light and \cite{KeplerConj} for a formal proof of the Kepler conjecture in HOL Light and Isabelle}. 

The version of ML used in HOL Light is OCaml \cite{ocaml-light}.

\section{OCaml Metalanguage and HOL Light Language}

In this section, we will introduce how the meta-programming language OCaml and the object language/theorem prover HOL Light relate. First, we will describe OCaml, stressing its basic features and pointing out it is a functional programming language based on $\lambda$-calculus with simple types and a polymorphic type system. Afterwards, we will describe the interaction between OCaml and HOL Light and, finally, we will introduce the basic syntax of HOL.

\subsection{OCaml as a Functional Programming Language with Imperative Features}

ML is a family of higher-order functional programming languages featuring a polymorphic type system\footnote{See  \ref{polymorphism} and \cite{Milner1978Polymorphism}} and some traditional imperative features. As previously pointed out, ML was developed as an implementation language for LCF, but, since then, ML has become independent from HOL Light and is well known as a general-purpose language with many dialects, such as PolyML, MLton, CakeML. OCaml, which stands for Objective-CAML (Categorical Abstract Meta Language), is one of them.

To understand the peculiarities of ML and OCaml we need to stress the difference between imperative and functional programming languages:
\begin{itemize}
    \item \textbf{Imperative/ Procedural Programming Language} \\
    These programs consist of a series of instructions to change the values of a collection of variables (\textit{state}). \\ 
    If $\sigma$ is the initial state and $\sigma'$ is the state after the execution, the program is the set of instructions that leads from $\sigma$ to $\sigma'$ by \textit{assignment} commands \verb|v:= E|. These assignments could be executed in a \textit{sequential} manner, conditionally (\verb|if|) and repeatedly (\verb|while|).
    If the program is deterministic, there exists a function f such that $f(\sigma) = \sigma'$. \textit{Program execution}: $\sigma \longrightarrow \sigma_1 \longrightarrow \sigma_2 \longrightarrow ... \longrightarrow \sigma'.$
    \\ 
    C++, Phyton and Java were designed to support imperative programming, but they now support some functional programming features. 
    \\ 
    $\triangleright$ Provides poor facilities in using functions in a sophisticated way. \\
    $\triangleright$ Imposes a rigid order of execution.
    
    \item \textbf{Functional Programming Language} \\
    A functional program is an expression and its \textit{execution} consists in the evaluation of the expression.
    \\
    In some sense this approach emphasizes the idea of a function $f$ such that $f(\sigma) = \sigma'$, since the basic mechanism is the application of functions to arguments and the program is meant to be an expression corresponding to this $f$; however the concepts of \textit{state}, \textit{assignments} and\textit{ sequential execution} are here totally meaningless. Instead of looping, functional languages can use recursive functions.
    \\ 
    LISP is a well-know example of a functional programming language.
    \\  
    $\triangleleft$ Functions can be treated here as simple objects, they can be \textit{arguments} and \textit{returns} of others functions. \\
    $\triangleleft$ Since an evaluation doesn't affect other evaluations, subexpressions can be evaluated in every order. Then, comprehending a program and debugging is easier, and it is also possible to develop parallel implementation.
\end{itemize}

OCaml is a functional programming language with some imperative features, like \textit{variables} and \textit{assignments}. Therefore OCaml allows a more flexible use of the functions and here programs can often be expressed in a very concise and fancy way, using higher-order functions.

\subsection{OCaml Top-level Basics}
Before focusing on HOL Light, it is important to understand the basic features of the OCaml top-level loop, focusing only on the purely functional fragment. We will analyze two things that you can do in OCaml \textit{(evaluate expressions }and \textit{make definitions}), and what we meant in a previous paragraph when we said that OCaml is a \textit{polymorphic type system}.  But first, we will see what an OCaml interaction looks like.

\subsubsection{Interacting with OCaml}
After OCaml is loaded,  it presents its prompt\footnote{The prompt is a character indicating the computer is ready to accept input.} (\#) and you can type in expressions, followed by two semicolons (;;); then OCaml will evaluate them and print the result. The OCaml system sits in a read-eval-print loop, which, in computing jargon, means that OCaml repeatedly reads an expression, evaluates it, and prints the result. For example, OCaml can be used as a pocket calculator, writing after the prompt the expression to evaluate:
 
\begin{lstlisting}
# 3 + 2;;
val it : int = 5
\end{lstlisting}
Observe that OCaml keywords will be printed in green.

\subsubsection{Evaluate expressions}\label{Evaluate expressions}
Looking closer at the output of the previous example, when you ask Ocaml to evaluate an expression, it returns:
\begin{itemize}
    \item The result of evaluating the expression (\verb|5|);
    \item The \textit{type} of the expression, which it has inferred automatically from the type of the built-in expression \verb|+| (\verb|int|);
    \item A name for the result of the evaluation automatically introduced (\verb|it|)
\end{itemize}

Since OCaml is a functional language, expressions are allowed to be functions. It is important to note that in OCaml and in HOL every function takes just a single argument; therefore it is common in functional programming as in $\lambda$-calculus, to use \textit{currying} to address the problem. This technique is named after the logician Haskell Curry \cite{Curry1930currify} and consists in translating a function that takes multiple arguments into a sequence of families of functions, each taking a single argument, e.g. if $f:(A\times B)\longrightarrow C$ then curryed $f: A \longrightarrow (B \longrightarrow C)= f_x(f')$ with $f':A\longrightarrow B$ and $f_x:B \longrightarrow C$. Moreover, observe that function application associates to the left, i.e. $f g x$ means $(f g)(x)$ and not
$f(g(x))$.

\begin{lstlisting}
# fun x -> x + 1;;
val it : int -> int = <fun>

#((fun x -> (fun y -> x + y)) 3) 2;;
it : int = 5
\end{lstlisting}

\subsubsection{Make definitions}
As suggested before, the main purpose of a functional programming language is to evaluate expressions, but it is not convenient to evaluate a complex program in one go. Instead, useful subexpressions can be evaluated and bound to names using a definition of the form: \verb|let <name> = <expression>|.
\begin{lstlisting}
# let main_number = 42;;
val main_number : int = 42

# let square x = x * x;;
val square : int -> int = <fun>

# let not_so_main = square main_number;;
val not_so_main : int = 1764
\end{lstlisting}
OCaml allows you to define recursive functions including the keyword \textit{rec}:
\begin{lstlisting}
#let rec fact n = if n = 0 then 1
                  else n * fact(n - 1);;
fact : int -> int = <fun>
#fact 5;;
it : int = 120
\end{lstlisting}
One may ask if definitions of this kind are taking us out of the purely functional fragment. The answer is not: each binding is only done once when the system analyses the input and it cannot be repeated or modified, consequently, it is not as a variable assignment nor as a modification of a \textit{state}. The binding can be overwritten by a new definition using the same name, but this is not an assignment in the usual sense, since the sequence of events is only connected with the compilation process and not with the dynamics of program execution. The overwritten binding must be considered a totally new definition, using the same name as before.

\subsubsection{Types and Polymorphism}\label{polymorphism}
As we have seen, OCaml automatically deduces the type of an expression, but it 
can also support functions without a fixed type, e.g. the identity function could have arguments of type \verb|int| but also of type \verb|string|. This is possible because OCaml is a polymorphic type system, i.e. a system that displays types involving \textit{type variables}. For example, we could define the identity as a polymorphic type function $id:\alpha \longrightarrow \alpha$ given $\alpha$ a type variable.
\begin{lstlisting}
# let id= fun x -> x;;
val id : 'a -> 'a = <fun>
# id true;;
val it : bool = true
# id 42;;
val it : int = 42
# id "happy";;
val it : string = "happy"
\end{lstlisting}
Because of OCaml's automatic type inference, it will allocate a type to an expression as general as possible, but users can also manually annotate types writing \verb|:<type>|.

\begin{lstlisting}
#let id = fun (x:int) -> x;;
id : int -> int = <fun>
\end{lstlisting}

\subsection{OCaml as HOL Light Metalanguage}
\subsubsection{Object Language and Metalanguage}\label{sub:metalanguage}
A \textit{metalanguage} $\mathcal{M}$ is a language that allows you to describe and talk about another language, which is called \textit{object language} $\mathcal{L}$; e.g. giving a basic lesson of Italian to English-spoken students, I would say that "\textit{Zio} means \textit{fratello di uno dei genitori}", using English as a metalanguage to describe Italian (object language) and assuming that they already know the meaning of all the words in italics except \textit{zio}. Expressions in a metalanguage are often distinguished from those of the object language by using italics or quotation marks. Similarly, a \textit{meta-theory} is a theory that analyzes and describes an \textit{object theory} and whose language is a metalanguage.

For example, if $\mathcal{L}$ is the language of arithmetic and I would use the \textit{square function} that is not in $\mathcal{L}_{ar}$, then the definition should be given in a meta-theory, which usually is some kind of informal set theory.

To give an example closer to HOLMS, we could say in the modal language $\mathcal{L}_{mod}$ that:   $ \vdash_{\mathcal{K} 4} \Box A \longrightarrow \Box \Box A$; but to enunciate the necessitation rule for $\mathcal{K} 4$ we need use to a metalanguage $\mathcal{M}_{mod}$: $\vdash \forall A. (\vdash_{\mathcal{K} 4} A \implies \vdash_{\mathcal{K} 4} \Box A)$. Notice that $\longrightarrow$ is the symbol for implication in $\mathcal{L}_{mod}$ while $\implies$ is the symbol for implication in $\mathcal{M}_{mod}$ and if $\vdash_{\mathcal{K} 4}$ is the symbol that denotes a theorem of the theory $\mathcal{K} 4$, then $\vdash$ denotes a theorem of the metatheory.

Similarly, OCaml is the meta (programming) language of HOL Light.

\subsubsection{In which sense OCaml is HOL Light metalanguage?}
HOL is a theorem prover and it can also be seen as a suite of tools for evaluating expressions involving \textit{terms} (representing logical or mathematical assertions) and \textit{theorems} (representing proved assertions). 
In fact, \textit{terms} and \textit{theorems} also with \textit{types}, are the key logical notions in HOL Light. Because OCaml is HOL metalanguage, recursive OCaml types are defined here to represent each HOL key concept: 

    \begin{center}
       \begin{tabular}{c|c}
        \textbf{HOL Light Key Concepts} & \textbf{Correspondent OCaml Type}  \\
        \hline
            Types     &  \verb|hol_type|   \\ 
            Terms     &  \verb|term|       \\ 
            Theorems  &  \verb|thm|      \\ 
        \end{tabular}
    \end{center}
    
Observe the object-meta distinction here: \textit{abstract} objects on the right are OCaml (meta) types of data structures representing HOL (object)  \textit{concret} logical entities on the left.

When, in the following paragraph, we will study HOL Basics, we should distinguish carefully between abstract and concrete syntax. 
\begin{itemize}
    \item Abstract Syntax of a term/type \\
    It is an OCaml data structure indicating how the term/type is built up from its components. HOL deals with this syntax internally. \\
    $\downarrow$ HOL's \textit{prettyprinter} automatically translates the abstract syntax into concrete syntax that is user-friendly and also more similar to common mathematical notation.
    
    \item Concrete Syntax of a term/type \\
    It is a linear sequence of characters. This syntax is easier to manipulate for humans and its translation to OCaml data structure has to follow abstract syntax's rules. \\
    $\uparrow$ HOL's \textit{quotation parser} automatically translates the concrete syntax into the one the system uses (abstract).
\end{itemize}

Moreover, observe that abstract type \verb|thm| is an OCaml type defined with primitive rules (logical core) as its only constructors. Consequently, a term is a HOL theorem if and only if the underlying \textit{machine code} of logical inference could derive it from the primitive rules. Here the \textbf{relative-reliability} of HOL Light stands, i.e. since one has confidence in the small amount of code underlying primitive inference rules, he can be sure that all theorems have properly been proved because of their type. As announced the system is also \textbf{extensible}, indeed the user can write custom procedures on the concrete level that implement higher-level\footnote{The level of the proof is "higher" if compared with the primitive inference rules, because these proof procedures are derived from the logical core.} proof procedures, while the system decomposes them internally to the primitive inference rules. 

To discuss the HOLMS-related example, look at this schema:
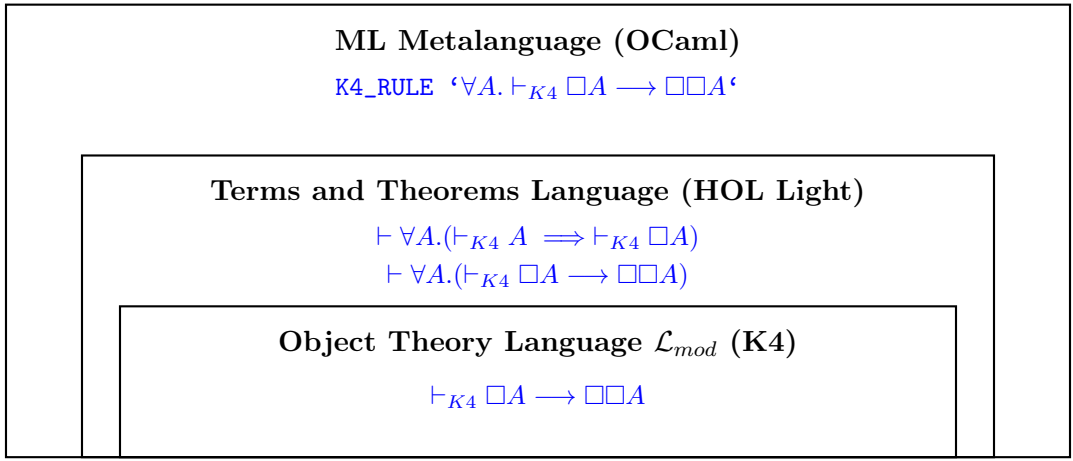
\begin{figure}[h]
    \centering
\begin{tikzpicture}
    \tikzstyle{contenitore} = [draw, thick, inner sep=10pt]
    \tikzstyle{testo} = [align=center, font=\bfseries\normalsize]
    \tikzstyle{formula} = [text=blue, font=\ttfamily\small]

    \node[contenitore, minimum width=14cm, minimum height=6cm] (ML) at (0, -0.5) {};
    \node[testo] at (0, 2) {ML Metalanguage (OCaml)};
    \node[formula] at (0, 1.4) {\verb|K4_RULE| `$\forall A. \vdash_{K4} \Box A \longrightarrow \Box \Box A$`};

    \node[contenitore, minimum width=12cm, minimum height=4cm] (HOLLight) at (0, -1.5) {};
    \node[testo] at (0, -0) {Terms and Theorems Language (HOL Light)};
    \node[formula] at (0, -0.6) { $\vdash \forall A. (\vdash_{K4} A \implies \vdash_{K4} \Box A)$ };
    \node[formula] at (0, -1.1){$\vdash \forall A. (\vdash_{K4} \Box A \longrightarrow \Box \Box A )$ };

    \node[contenitore, minimum width=11cm, minimum height=2cm] (HOLMS) at (0, -2.5) {};
    \node[testo] at (0, -2) {Object Theory Language $\mathcal{L}_{mod}$ (K4)};
    \node[formula] at (0, -2.7) {$\vdash_{K4} \Box A \longrightarrow \Box \Box A$};

\end{tikzpicture} \caption{HOLMS object theories and metatheories} \label{tik:meta}
\end{figure}

In the HOLMS project, we want to formalize inside HOL Light different modal logics, such as K4. As observed in the previous example, the object theory K4 needs to refer to a metalanguage to enunciate its necessitation rule and to a meta-theory to assume it, but it also requires a meta-theory to affirm and prove that the axiom \textit{four} holds for every proposition A. Here we use HOL Light language as K4 metalanguage. Moreover, as stressed in the previous lines, HOL works thanks to its metatheory OCaml. Here we could apply to a HOL term describing a theorem of K4 the customized rule \verb|K4_RULE|.

\subsubsection{Interacting with HOL Light is interacting with OCaml}

 HOL Light-hereinafter just 'HOL'-has no separate user interface, the user first loads all the HOL infrastructure and then he works inside the OCaml interpreter with a large number of theorems and tools for proving theorems loaded in.  As emphasized in the HOL tutorial \cite{harrisontutorial}, the use of the implementation language as an interaction environment yields a system that is open and extensible cleanly and uniformly.

 After HOL is loaded, OCaml presents its prompt\footnote{The prompt is a character indicating the computer is ready to accept input.} (\#):
\begin{lstlisting}
>        Camlp5 parsing version 7.10

#
\end{lstlisting}

To describe interaction with HOL Light and OCaml. we will use this kind of box and we will write expressions to evaluate after the prompt and then HOL Light evaluation.

Now you could ask yourself what Camlp5 is and why above there is not written: \\ \verb|OCaml parsing version **|. Camlp5 is a \textit{preprocessor} and \textit{pretty-printer} for the OCaml programming language and it is loaded with\verb| OCaml version 4.05.0|. It serves as a tool for writing syntax extensions and custom \textit{parsers/printers} for OCaml. So we have exactly what we need to start using HOL Light. 

\subsection{Basics of HOL Light}
In this subsection, we will describe HOL concrete syntax\footnote{To read an introduction to the abstract syntax look at section I of \cite{harrisonmanual}}, i.e. the syntax accepted by HOL's \textit{parser}. At first, we will focus on $\lambda$-calculus and then we will discuss term-syntax, type-syntax and finally we will treat theorems.

\subsubsection{HOL's logic and  $\lambda$-calculus}
Functional languages are based on the computational model of $\lambda$ calculus\footnote{To focus on $\lambda$-calculus see \cite{Barendregt1985LamdaCalc} and \cite{HindelyLambdaCalc}}. Consequently, HOL's logic, as OCaml one before, is based on $\lambda$-calculus with simple type, a formalism invented by Alonzo Church \cite{Church1940TypedLambdaCalc}. All HOL terms, as $\lambda$-calculus ones, are built up from constants and variables using application and abstraction. 
\\ \\
Building-blocks of terms
\begin{itemize}
    \item Constants \\
    Entity whose value is fixed, e.g. $[]$ (the empty list), $\top$ and $\bot$ (false). They could be primitive, such as equality, or intended to be abbreviations for other terms, defined before they can be used in terms.
    \item Variables \\
     Entity without a fixed value, that can have any name, e.g. n, x, p.
\end{itemize}
Constructors of terms
\begin{itemize}
    \item Application \\
    The common mathematical operation of applying a function f to an argument t $f(t)$. In $\lambda$-calculus as in HOL, application associates to the left ($f g x =(f g)(x) \neq f(g(x))$) and we use currying\footnote{We introduced currying in \ref{Evaluate expressions}.}  to treat function with multiple arguments.
    \item Abstraction \\
    Given x a variable and t a term, one can construct the lambda-abstraction $\lambda x. t$, which means \textit{the function of x that yields t}. (In HOL's ASCII\footnote{ASCII stands for American Standard Code for Information Interchange. It is a standard character encoding that represents text using binary codes.} concrete syntax: \verb|\x.t|). Note that $\lambda$ is a variable-binding operator. \\ $\triangleleft$ Abstraction allows one to treat functions as first-order objects and to write anonymous function-valued expressions without naming them.
    \\ $\triangleleft$ Abstraction makes variable dependencies and binding explicit; by contrast in abstract mathematics one often writes $f(x)$ where one really means $\lambda x. f(x)$.
\end{itemize}

Observe that $\lambda$-calculus requires the beta-conversion rule as its unique axiom. \\ \\
\textit{Given t, s terms and $t[s/x]$ the appropriate replacement of each instance of x in t by s. \textbf{Beta Conv:}} $ \vdash (\lambda x. t) s = t[s/x]$. \\ \\
In this precise sense, abstraction is the converse operation to application.

So far, we have described $\lambda$-calculus without types; but this version of the calculus is exposed to a form of Russel paradox.
\\ \\
$\vdash (\lambda x. \neg (xx)) (\lambda x. \neg (xx)) \xlongequal{\text{bconv}} \neg (xx)[(\lambda x. \neg (xx))/x] = \neg ((\lambda x. \neg (xx)/x)(\lambda x. \neg (xx)/x)) $
\\ \\
So given $z = (\lambda x. \neg (xx))$, $ \vdash zz = \neg(zz)$: we have derived that something is equal to its own logical negation. As in other versions of the paradox, we have reached this contradictory result by the conjoint use of negation and self-application. 
To avoid self-application and Russel paradox, we will switch to a typed $\lambda$-calculus, using something similar to Russel's theory of types. So, in $\lambda$-calculus with simple type as in HOL, every term has a unique type which is either one of the basic types or the result of applying a type constructor to other types. This solution avoids self-application because given a function $f:A \to B$ this will be applied only on argument of the right type $:A$.

\subsubsection{Terms Syntax}
A HOL term represents a mathematical assertion like \verb|x + 3 = y| or just some
mathematical expression like \verb|x + 3|.
Regarding HOL syntax, the only primitive constant is the equality relation; all the other terms are variables or terms built by abstraction and application.
\begin{itemize}
    \item HOL primitive constant: \verb|=| (the only one)
    \item HOL terms constructors: application \verb|ts|,  abstraction \verb|\x.t|
\end{itemize}

For example, the negation of a propositional atom \verb|~p| is represented by the application of the logical constant \verb|~| to a term \verb|p:bool|. 

In HOL you don't need to assume logical constants as primitive, indeed, following $\lambda$-calculus that is an equational calculus, you can define them using the equality relation and the two term-constructors\footnote{While it is more common to take few additional logical constants such as $\forall$ and $\implies$, a definition of the logical constants in terms of equality has been known since 1963 \cite{Henkin1963lambdadef}) as primitive}. 

    \begin{center}
       \begin{tabular}{c|c|l}
        \textbf{HOL Notation} & \textbf{Logic Notation } &  \textbf{$\lambda$-calculus def} \\
        \hline 
            \verb|T|    &  $\top$   & ($\lambda x. x) = (\lambda x. x)$\\ 
            \verb|/\|     & $ \land$ & $\lambda p. \lambda q. (\lambda f. f \, p \, q) = (\lambda f. f \, \top \, \top)$\\ 
             \verb|==>|  &  $\Rightarrow$   &  $\lambda p. \lambda q. p \land q = p $ \\ 
              \verb|!|  &  $\forall$  & $ \lambda P. P = \lambda x. \top $ \\
              \verb|?|  &  $\exists$  & $ \lambda P. \forall Q. (\forall x. (P(x) \Rightarrow Q)) \Rightarrow Q$       \\
              \verb|\/|  &  $ \lor $  & $ \lambda p. \lambda q. \forall r. (p \Rightarrow r) $       \\
              \verb|F|  &  $\bot$  & $ \lambda p. \lambda q. \forall r. (p \Rightarrow r) \Rightarrow (q \Rightarrow r)$       \\
              \verb|~|  &  $ \neg$  & $\lambda t. t \Rightarrow \bot$       \\
              \verb|?!|  &  $\exists !$  & $\lambda P. \exists P \land \forall x. \forall y. P \, x \land P \, y \Rightarrow (x = y)$       \\
        
        \end{tabular}
    \end{center}

In HOL terms are entered rather like strings, enclosed within \textit{backquotes}. Quotations are then expanded by a \textit{filter}, thanks to its \textit{term parser} and its \textit{type inferencer}. This filter analyzes the syntactical structure of terms and, because in a typed system every term corresponds to a unique type, it also deduces types for the term as a whole and all its subterms. \\

\begin{lstlisting}
# `x+3`;;
val it : term = `x + 3`
\end{lstlisting}

In the example, the \textit{type inferencer} deduces from the constants \verb|3: num| and $+:\alpha \longrightarrow \alpha \longrightarrow \alpha$, both the types of x\verb|:num| and x+3\verb|:num|. 
However, the filter can print two kinds of error messages: the first is a \textit{failure message}, which means that the user tried to enter a term that cannot be typed; the second is a \textit{warning of inventing type variables}, which indicates that there is not enough type information to fix the types of all subterms.

\begin{lstlisting}
# 42 = true;;
  42 = true
       ^^^^
Error: This expression has type bool but an expression  was expected of type int

# `K4`;;
Warning: inventing type variables
val it : term = `K4`
\end{lstlisting}

For each kind of syntactical constructor \verb|var, const, comb, abs|, OCaml has a function for creating HOL terms (\verb|mk_<>|), one for breaking them apart (\verb|dest_<>| and \verb|rator|, \verb|rand| for applications), and one for testing their
structure (\verb|is_<>|). For example, we look at OCaml function for combinations (applications):

\begin{lstlisting}
# mk_comb(`P:A->bool`,`x:A`);;
val it : term = `P x`

# dest_comb (`x + y`);;
val it : term * term = (`(+) x`, `y`)
# rator (`x + y`);;
val it : term = `(+) x`
# rand (`x + y`);;
val it : term = `y`

# is_abs (`x + y`);;
val it : bool = false
# is_comb (`x + y`);;
val it : bool = true

\end{lstlisting}

\subsubsection{Types}
As we said, every term has a unique type indicating what sort of
mathematical entity it is, e.g. a boolean value (\verb|:bool|), a real number(\verb|:real|), a set of real functions, etc. 
In HOL types are preceded by colons, e.g. \verb|~p:bool|.

Every type is either one of the basic types or, since HOL is a polymorphic extension of Church's system, a \textit{type variable}, or the result of applying a type constructor to other types, e.g. \verb|:num|, \verb|:int|, \verb|:real|, $:\alpha \to$\verb|bool|. 

\begin{itemize}
    \item HOL basics types: \verb|bool| (the only one)
    \item HOL type operators: function space constructor $\longrightarrow$
\end{itemize}

Observe that it is possible to have a term variable that first occurs with a certain type and then with a different one, then OCaml provides an useful function to ask the the type of an HOL term
(\verb|type_of <term>|) and it also allows the user to assign an arbitrary type to a term writing \verb|<term>:<type>|.

\begin{lstlisting}
# type_of `\x:A. f:B`;;
val it : hol_type = `:A->B`
\end{lstlisting}

We now look at the types of primitive terms and term constructors. Given A, B \textit{types}, $\alpha$ a \textit{type variable}, \verb|t|, \verb|s| \textit{terms} and \verb|x| a \textit{term variable}:

\begin{center}    
\begin{tabular}{l|l|l}
\textbf{name} & \textbf{type} & \textbf{notes} \\
     \hline
    Equality relation  & $=: \alpha \longrightarrow \alpha$ & - \\
    Application & $((t:A \longrightarrow B)(s:A)):B$ & function needs the right type to be applied\\
    Abstraction & $(\lambda (x:A). (t:B)):A \longrightarrow B$ & - \\
\end{tabular}
\end{center}

Note that functions may only be applied to arguments of the right type: only a function of type $f:A \longrightarrow B$ may be applied to an argument of type A.

In HOLMS we will also use parametric types, i.e. types depending on others types, defined in the theory of list and in the theory of pairs\footnote{These theories are presented in \ref{Theories} }.
\[\]
\begin{tabular}{l|l|l}
\textbf{name} & \textbf{type} & \textbf{what it is?} \\
     \hline
     $[$ $]$& \verb|:(A)list| & List of object of type A  \\
    (,) & \verb|:A # B |& Cartesian
product of objects of type A and B\\
\end{tabular}
\[\]
For example, we will use parametric types inside our definition of the validity of a formula in a certain world of a certain modal model (\verb|holds|). 

\begin{lstlisting}
# type_of `holds`;;
Warning: inventing type variables
val it : hol_type =
  `:(?131589->bool)#(?131589->?131589->bool)->((char)list->?131589->bool)->form->?131589->bool`
\end{lstlisting}

Observe that, in asking the type of \verb|holds|, we didn't specify its parameter, consequently, OCaml invented a type variable and printed a failure message.

As for terms, OCaml has functions for creating HOL type variables (\verb|mk_vartype|) and composite types (\verb|mk_type|), those for breaking them apart (\verb|dest_vartype| and \verb|dest_type|) and those for testing their structure (\verb|is_type| and \verb|is_vartype|).

\subsubsection{Theorems}
We usually think in terms of \textit{truth} or \textit{falsity} of a formula, but in a theorem prover, we deal only with \textit{provability}.
A HOL theorem has OCaml type \verb|thm| and asserts that some boolean-typed term, what we usually call \textit{formula}, is valid ($\vdash \phi:bool$) or that it follows from a finite list of assumptions ($\psi_1, ..., \psi_n \vdash \phi:bool$).

Moreover, OCaml provides functions to break a theorem into its hypothesis and conclusions (\verb|dest_thm|, \verb|concl| and \verb|hyp|).

\begin{lstlisting}
val add2and3 : thm = |- 3 + 2 = 5

# dest_thm add2and3;;
val it : term list * term = ([], `3 + 2 = 5`)
\end{lstlisting}

In the following section, we will introduce the deductive system that allows us to derive theorems in HOL Light.

\section{HOL Light deductive system}

\subsection{The Primitive System}
At this point, we are ready to present HOL Light's deductive system. i.e. the logical core we announced as simpler and more elegant than in any other HOL version. Given a formal language and a syntax for formulas, a deductive system consists of axioms and rules of inference, used to derive theorems of the system.

    \subsubsection{HOL Light Formal Language:}
    Language of terms and types \\
    $\mathcal{L}_\textit{HOL Light}= $ \textit{Numerable set of term variables }$\cup$ \textit{Numerable set of type variables } $\cup$ 
    $\{= \ ; \ \cdot \ ; \ \lambda \ \}$ $\cup$ 
    $\{:bool \ ; \ :\to \}$
\\ \\
    Here we write $\cdot$ for application and $\lambda$ for abstraction, but HOL uses $\setminus$ for abstraction and no symbol for application, it simply considers \verb|ts| as $t \cdot s$.
    
     \subsubsection{HOL Light Syntax for Formulas:}
     Concrete syntax for\textbf{ boolean-typed terms}, given $T$ concrete syntax for types and $S$ concrete syntax for terms:\\
    $ T:= \textit{bool} \ | \ T_1 \to T_2$ \\
    $s:= x:T \ | \ s_1=s_2:bool \ | \
         (\lambda x:T_1.\ s:T_2 ):T_1 \to T_2 \ | \
        ( s_1:T_1 \to T_2 ) \cdot (\ s_2:T_1 ):T_2$

    \subsubsection{HOL Axioms:}
      \begin{enumerate}
          \item \verb|ETA_AX| ($\eta$-conversion) $\vdash (\lambda x. \ s \cdot x)= s$ \\
          This axiom, together with the $\beta$-conversion rule, allows HOL to prove all $\lambda$-calculus theorems.
          \item \verb|SELECT_AX| $\vdash \forall x. P(x) \implies P(\epsilon x. P(x))$ \\
          Where P is a boolean typed-term with a free variable x and $\epsilon$ is \textit{Hilbert choice operator} (\verb|@| in HOL notation), a new logical constant of polymorphic type $\epsilon: (\alpha \to bool) \to \alpha$. The underlying idea is that\textit{ $\epsilon$ picks out something satisfying P whenever there is something to pick}.
          This axiom is a version of the \textit{axiom of global choice}, so it isn't harmless at all, and it also makes HOL logic \textit{classical}, proving the \textit{law of the excluded middle}.
          \item \verb|INFINITY_AX| $\vdash \exists f:ind\to ind. \ (\forall x_1 \ x_2. \ (f \cdot x_1 = f \cdot x_2 )) \implies (x1=x2 \land \neg(\forall y. \ \exists x.\ y= f \cdot x))$ \\
          Where \verb|ind| (individuals) is a new type, that we assert to be infinite, using \textit{Dedekind-Pierce definition}. That is, we assert the existence of a function from the type of individuals to itself that is injective but not surjective, such a mapping is impossible if the type is finite.
      \end{enumerate}

\subsubsection{Primitive Rules of Inference}
\begin{enumerate}
    \item Reflexive rule: Equality is reflexive. \\
    \[\frac{}{\vdash t = t }
    \hfill \texttt{REFL}
    \]

    \item Transitive rule: Equality is transitive.
    \[
    \frac{\Gamma \vdash s = t \quad \Delta \vdash t = u}{\Gamma \cup \Delta \vdash s = u} 
    \hfill \texttt{TRANS}
    \]

    \item Rule for Combination: If the types of an application agree, equal functions applied to equal arguments give equal results.
    \[
    \frac{\Gamma \vdash s = t \quad \Delta \vdash u = v}{\Gamma \cup \Delta \vdash s\cdot u = t\cdot v} 
    \hfill \texttt{MK\_COMB}
    \]

    \item Rule for Abstraction: If two expressions involving x are equal, then the functions that take x to those values are equal.
    \[
    \frac{\Gamma \vdash s = t}{\Gamma \vdash (\lambda x . s) = (\lambda x . t)} 
    \hfill \texttt{ABS}
    \]
    \[\ x \not \in FreeVar(\Gamma)\]

    \item $\beta$-conversion rule: Combination and abstraction are converse operations.
    \[\frac{}{\vdash (\lambda x . t) \cdot x = t }
    \hfill \texttt{BETA}
    \]

    \item Assume rule: From any \verb|p:bool| we can deduce \verb|p|.
    \[\frac{}{\{p\} \vdash p}
    \hfill \texttt{ASSUME}
    \]

    \item Modus ponens for an equational calculus: Connects equality with deduction.
    \[
    \frac{\Gamma \vdash p = q \quad \Delta \vdash p}{\Gamma \cup \Delta \vdash q}
    \hfill \texttt{EQ\_MP}
    \]

    \item Equality as logical equivalence rule: It also connects equality and deduction, saying that equality on the boolean type represents logical equivalence.
    \[
    \frac{\Gamma \vdash p \quad \Delta \vdash q}{(\Gamma - \{q\}) \cup (\Delta - \{p\}) \vdash p = q}
    \hfill \texttt{DEDUCT\_ANTISYM\_RULE}
    \]

    \item Substitution rule for Terms: This rule expresses the fact that variables are to be interpreted as schematic, i.e. if $p$ is true for variables $x_1, ..., x_n$, then we can replace those variables by any terms of the \textit{same type} and still get something true.
    \[
    \frac{\Gamma[x_1, \ldots, x_n] \vdash p[x_1, \ldots, x_n]}{\Gamma[t_1, \ldots, t_n] \vdash p[t_1, \ldots, t_n]}
    \hfill \texttt{INST}
    \]

    \item Substitution rule for Types: Same rule as before, but for substitution of type variables rather than term variables.
    \[
    \frac{\Gamma[\alpha_1, \ldots, \alpha_n] \vdash p[\alpha_1, \ldots, \alpha_n]}{\Gamma[\gamma_1, \ldots, \gamma_n] \vdash p[\gamma_1, \ldots, \gamma_n]}
    \hfill \texttt{INST\_TYPE}
    \]
\end{enumerate}

\subsection{Definitions}
From this simple foundation, that counts only three axioms and ten rules, all the HOL mathematics and its applications are developed by definitional extension.

Indeed, in addition to the deductive system, HOL Light includes two principles of definition\footnote{The four principles of definitions, two for types and two for terms, are printed in blue in all the examples of interaction with HOL Light}, which allow the user to extend the set of constants and the set of types in a way guaranteed to preserve consistency. At this point, we could clearly observe \textit{relative-reliability } and \textit{extensibility } of the system and we could say, from a logical perspective, that HOL is an evolving sequence of logical
systems, each a conservative extension of previous ones.

\begin{itemize}
    \item For terms: \textbf{Rule of constant} (\verb|new_basic_definition| and \verb|new_definition|) \\
    \verb|new_basic_definition| allows to introduce a \textit{new constant} \verb|c| and a new axiom $\vdash \verb|`c=t`|$, if no previous definition for \verb|c| has been introduced and \verb|t| is any term without free (term or type)
variables. 

\begin{lstlisting}
# new_basic_definition `true = T`;;
it : thm = |- true = T
\end{lstlisting}

$\triangleleft$ All the definitions of the logical connectives above are
introduced in this way. \\
$\triangleleft$ This definition is made on \textit{object-level}, indeed the constant and its defining axiom exist in the object logic.\\
$\triangleleft$ This principle of definition is very harmless, it simply associates a name (a constant) to a term.

\verb|new_definition| extends quite soon this principle to permit definitions of functions with the arguments on the left, including pairs and tuples of argument. For example, in HOLMS as in modal logic, we say that a formula holds in a certain frame iff it holds in every world of every model of that frame:
\begin{lstlisting}
let holds_in = new_definition
  `holds_in (W,R) p <=> !V w:W. w IN W ==> holds (W,R) V p w`;;
#     val holds_in : thm =
  |- !W R p. holds_in (W,R) p <=> (!V w. w IN W ==> holds (W,R) V p w)
# 
\end{lstlisting}

$\triangleleft$ Constant definition preserves consistency and the \textit{relative-reliability } of the system, because this \textit{extension} could be avoided simply by writing the definitional expansion out in full.

\item For terms: \textbf{Definition by Induction} (\verb|new_inductive_definition|) \\
HOL has some limited support for \textit{primitive recursive}\footnote{\textit{Primitive recursive function} is the set of functions that contains \textit{constants}, \textit{successor} and \textit{projections} and that is closed under \textit{composition} and \textit{primitive recursion}. \textit{PR definition} allows definitions of PR functions and predicates with a PR characteristic function. This set is also closed under \textit{bounded minimization} and \textit{bounded quantification},} definitions, but it doesn't provide any functions for \textit{general recursive}\footnote{\textit{General Recursive definition} allows the definition of PR  functions and predicates, but also via \textit{unbounded minimization}.} definitions\footnote{Observe that is possible to define general
recursive functions using theorems from the theory of wellfoundedness, which is one of the theories developed inside HOL Light (\ref{Theories})}.
 HOL supports the definition of \textit{inductive predicates } (or sets) in a more convenient way by using \verb|new_inductive_definition|. This function requires an inductive base and an inductive step and then returns three theorems:
\begin{enumerate}
    \item \verb|<>_RULES|  assures the predicate is close under inductive base and steps.
    \item \verb|<>_INDUCT| assures this is the smallest set/predicate that satisfy \verb|1.|
    \item \verb|<>_CASES| give a cases analysis theorem.
\end{enumerate}

 For example, if we want to define the set of even numbers (the predicate that says that a number is even):

\begin{lstlisting}
# let Even_RULES,Even_INDUCT,Even_CASES =
new_inductive_definition
`(Even(0)) /\
(!n. Even(n) ==> Even (n+2))`;;

val Even_RULES : thm = |- Even 0 /\ (!n. Even n ==> Even (n + 2))
val Even_INDUCT : thm =
  |- !Even'. Even' 0 /\ (!n. Even' n ==> Even' (n + 2))
             ==> (!a. Even a ==> Even' a)
val Even_CASES : thm = |- !a. Even a <=> a = 0 \/ (?n. a = n + 2 /\ Even n)
\end{lstlisting}

In HOLMS we use \verb|new_inductive_definition| to define a parametrized modal calculus and, in particular, the predicate $\verb|[S.H|\mid \sim$ \verb|p]| that formalizes $H \vdash_{S} p$, given \verb|S| a particular modal logic, \verb|H| a set of hypotheses and \verb|p| a modal formula.

\begin{lstlisting}
let MODPROVES_RULES,MODPROVES_INDUCT,MODPROVES_CASES =
  new_inductive_definition
  `(!H p. KAXIOM p ==> [S . H |~ p]) /\
   (!H p. p IN S ==> [S . H |~ p]) /\
   (!H p. p IN H ==> [S . H |~ p]) /\
   (!H p q. [S . H |~ p --> q] /\ [S . H |~ p] ==> [S . H |~ q]) /\
   (!H p. [S . {} |~ p] ==> [S . H |~ Box p])`;;
   
val MODPROVES_RULES : thm =
  |- !S. (!H p. KAXIOM p ==> [S . H |~ p]) /\
         (!H p. p IN S ==> [S . H |~ p]) /\
         (!H p. p IN H ==> [S . H |~ p]) /\
         (!H p q. [S . H |~ p --> q] /\ [S . H |~ p] ==> [S . H |~ q]) /\
         (!H p. [S . {} |~ p] ==> [S . H |~ Box p])
         
val MODPROVES_INDUCT : thm =
  |- !S MODPROVES'.
         (!H p. KAXIOM p ==> MODPROVES' H p) /\
         (!H p. p IN S ==> MODPROVES' H p) /\
         (!H p. p IN H ==> MODPROVES' H p) /\
         (!H p q.
              MODPROVES' H (p --> q) /\ MODPROVES' H p ==> MODPROVES' H q) /\
         (!H p. MODPROVES' {} p ==> MODPROVES' H (Box p))
         ==> (!a0 a1. [S . a0 |~ a1] ==> MODPROVES' a0 a1)
         
val MODPROVES_CASES : thm =
  |- !S a0 a1.
         [S . a0 |~ a1] <=>
         KAXIOM a1 \/
         a1 IN S \/
         a1 IN a0 \/
         (?p. [S . a0 |~ p --> a1] /\ [S . a0 |~ p]) \/
         (?p. a1 = Box p /\ [S . {} |~ p])

\end{lstlisting}

\item For types: \textbf{Type definition} (\verb|new_basic_type_definition|) \\
\verb|new_basic_type_definition| is HOL primitive function for performing definitions of new types and types constructors. Given:
\begin{enumerate}

    \item Any subset of an existing type $\gamma$ marked out by its characteristic predicate $P:\gamma \to bool$;
    \item  A theorem asserting that P is nonempty,
\end{enumerate}
One can define a new type or, if $ \gamma$ contains type variables, a type constructor $\delta$ in bijection with this subset.

    \includegraphics[width= 0.8\textwidth]{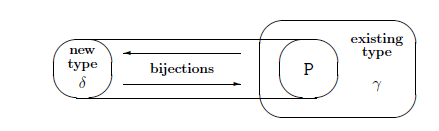} 
\label{new_type_definition}

 In the below example, we want to define a new type \verb|single| in bijection with the 1-element subset of bool containing just \verb|T|, where P is  $\setminus$\verb|x. x = T|:

\begin{lstlisting}
#let th1 = BETA_CONV `(\x. x = T) T`;;
th1 : thm = |- (\x. x = T) T = T = T
#let th2 = EQ_MP (SYM th1) (REFL `T`);;
th2 : thm = |- (\x. x = T) T

# new_basic_type_definition "single" ("mk_single","abs_single") th2;;
it : thm * thm =
|- mk_single (abs_single a) = a,
|- (\x. x = T) r = abs_single (mk_single r) = r 
\end{lstlisting}

Here the user gives the desired name for the new type (\verb|single|) and for the bijections that map
between the old and new types (\verb|"mk_single"| and \verb|"abs_single"|), and, finally, a theorem asserting that the chosen subset of the existing type contains some object (\verb|th2|). 
Then HOL returns two theorems, implying that the chosen bijection maps 1-1 the new type and the chosen subset of the old one.

$\triangleleft$ As for the rule of constants, type definitions don't compromise \textit{relative-reliability} and consistency of the system, indeed these definitions could be avoided by incorporating appropriate set constraints into theorems.

\item For types: \textbf{Recursive definition for types} (\verb|define_type|) \\
HOL also supports \verb|define_type|, one of the most useful and complicated derived rules. It allows users to define \textit{recursive types} with some restrictions. Similarly to \verb|new_inductive_definition|, it returns two theorems:
\begin{enumerate}
    \item \verb|<>_INDUCT| a kind of induction theorem for the new type
    \item \verb|<>_RECURSION| a justification of definition by primitive recursion
\end{enumerate}

In HOLMS we define a recursive type for modal formulas (\verb|:form|), as below:

\begin{lstlisting}
 let form_INDUCT,form_RECURSION = define_type
  "form = False
    | True
    | Atom string
    | Not form
    | && form form
    ||| form form
    | --> form form
    | <-> form form
    | Box form";;

val form_INDUCT : thm =
  |- !P. P False /\ P True /\
         (!a. P (Atom a)) /\
         (!a. P a ==> P (Not a)) /\
         (!a0 a1. P a0 /\ P a1 ==> P (a0 && a1)) /\
         (!a0 a1. P a0 /\ P a1 ==> P (a0 || a1)) /\
         (!a0 a1. P a0 /\ P a1 ==> P (a0 --> a1)) /\
         (!a0 a1. P a0 /\ P a1 ==> P (a0 <-> a1)) /\
         (!a. P a ==> P (Box a)) 
         ==> (!x. P x)
         
val form_RECURSION : thm =
  |- !f0 f1 f2 f3 f4 f5 f6 f7 f8.
       ?fn. fn False = f0 /\ fn True = f1 /\
         (!a. fn (Atom a) = f2 a) /\
         (!a. fn (Not a) = f3 a (fn a)) /\
         (!a0 a1. fn (a0 && a1) = f4 a0 a1 (fn a0) (fn a1)) /\
         (!a0 a1. fn (a0 || a1) = f5 a0 a1 (fn a0) (fn a1)) /\
         (!a0 a1. fn (a0 --> a1) = f6 a0 a1 (fn a0) (fn a1)) /\
         (!a0 a1. fn (a0 <-> a1) = f7 a0 a1 (fn a0) (fn a1)) /\
         (!a. fn (Box a) = f8 a (fn a))
\end{lstlisting}

\end{itemize}

\subsection{Derived Rules}
The following rules are derived from primitive ones and from HOL axioms and extend HOL deductive system in a conservative way.

\subsubsection{Logical Rules}
Logical rules are derived from primitive rules and axioms obtained by logical constants definitions. Most of these rules are introduction and elimination rules of \textit{natural deduction},  that we report here abstracting away a bit from the implementation in HOL:

\begin{itemize}
    \item Rules for conjunction \\
    We take these rules as an example.
    \begin{enumerate}
        \item Introduction rule \\
        \[\frac{\Gamma \vdash p \quad \Delta \vdash q}{\Gamma \cup \Delta \vdash p \land q} \texttt{ CONJ}\]
        The introduction rule for conjunctions (\verb|CONJ|) takes two theorems and gives a new theorem conjoining them. 
        \begin{lstlisting}
# let th1 = CONJ (REFL `T`) (ASSUME `p /\ q`);;
val th1 : thm = p /\ q |- (T <=> T) /\ p /\ q
        \end{lstlisting}
        \item Elimination rules \\
        \[\frac{\Gamma \vdash p \land q}{\Gamma \vdash p} \texttt{ CONJUNCT1}\]
        \\
        \[\frac{\Gamma \vdash p \land q}{\Gamma \vdash q} \texttt{ CONJUNCT2}\]

        The elimination rules take a theorem with a conjunction as conclusion and give new theorems for the left and right arms. \verb|CONJ_PAIR| gives a pair of both, while \verb|CONJUNCTS| transform a conjunctive theorem into a list of theorems.

        \begin{lstlisting}
# CONJUNCTS th1;;
val it : thm list = [p /\ q |- T <=> T; p /\ q |- p; p /\ q |- q]
# CONJUNCT1 th1;;
val it : thm = p /\ q |- T <=> T
        \end{lstlisting}
    \end{enumerate}   
    
    \item Rules for top
    \begin{enumerate}
        \item Introduction rule \\
       \[\frac{\Gamma \vdash p}{\Gamma \vdash p = \top} \texttt{ EQT\_INTRO}\]
        \item Elimination rule \\
        \[\frac{\Gamma \vdash p = \top}{\Gamma \vdash p} \texttt{ EQT\_ELIM}\]       
    \end{enumerate}

    \item Rules for implication
    \begin{enumerate}
        \item Introduction rule \\
       \[\frac{\Gamma \vdash q}{\Gamma \setminus \{p\} \vdash p \implies q} \texttt{ DISH}\]
        \item Elimination rules \\
        \[\frac{\Gamma \vdash p \implies q}{\Gamma \cup \{p\} \vdash q} \texttt{ UNDISH}\] \\
        \[\frac{\Gamma \vdash p \implies q \quad \Delta \vdash p}{\Gamma \cup \Delta \vdash q} \texttt{ MP}\]
    \end{enumerate}

        \item Rules for universal quantifier
    \begin{enumerate}
        \item Introduction rule \\
       \[\frac{\Gamma \vdash p}{\Gamma \vdash \forall x. p} \texttt{ GEN}\]
       \[x \not \in FreeVar(\Gamma)\]
       In HOL implementation this rule requires generic typed-term and a theorem.
        \item Elimination rule \\
        \[\frac{\Gamma \vdash \forall x. p}{\Gamma \vdash p[t/x]} \texttt{ SPEC}\] 
        In HOL implementation this rule requires a term to specify the quantifier and a   theorem. If necessary, it automatically renames variables to avoid capture, by adding prime characters.

        \begin{lstlisting}
 # SPEC `y:num` (ASSUME `!x. ?y. y > x`);;
val it : thm = !x. ?y. y > x |- ?y'. y' > y
        \end{lstlisting}

    \end{enumerate}

    \item Rules for existential quantifier
    \begin{enumerate}
        \item Introduction rule \\
       \[\frac{\Gamma \vdash p[t/x]}{\Gamma \vdash \exists x. p } \texttt{ EXISTS}\]
       The ML invocations for this rule are relatively complicated; the function requires the desired form of the result and the term \verb|t| to choose. 
       \begin{lstlisting}
# EXISTS(`?x:num. x = x`,`1`) (REFL `1`);;
val it : thm = |- ?x. x = x           
       \end{lstlisting}
        \item Elimination rule \\
        \[\frac{\Gamma \vdash q}{\Gamma \setminus \{p\} \vdash (\exists x. p) \implies q} \texttt{ CHOOSE}\]       
        \[x \not \in FreeVar(\Gamma \cup q) \]
    \end{enumerate}

    \item Rules for disjunction
    \begin{enumerate}
        \item Introduction rules \\
         \[\frac{\Gamma \vdash p}{\Gamma \vdash p \lor q} \texttt{ DISJ1}\] \\
         \[\frac{\Gamma \vdash q}{\Gamma \vdash p \lor q} \texttt{ DISJ2}\]
        \item Elimination rule \\
         \[\frac{\Gamma \vdash r \quad \Gamma' \vdash r \quad \Delta \vdash p \lor q}{(\Gamma \setminus \{p\}) \cup (\Gamma' \setminus \{q\}) \cup \Delta \vdash r} \texttt{ DISJ\_CASES}\]      
    \end{enumerate}

    \item Rules for negation
    \begin{enumerate}
        \item Introduction rule \\
        \[\frac{\Gamma \vdash p \implies \perp}{\Gamma \vdash \neg p} \texttt{ NOT\_INTRO}\]
        \item Elimination rule \\
        \[\frac{\Gamma \vdash \neg p}{\Gamma \vdash p \implies \perp} \texttt{ NOT\_ELIM}\]      
    \end{enumerate}

    \item Rules for bottom
    \begin{enumerate}
        \item Introduction rule \\
        \[\frac{\Gamma \vdash \neg p}{\Gamma \vdash p \equiv \perp} \texttt{ EQF\_INTRO}\]
        \item Elimination rule \\
        \[\frac{\Gamma \vdash p \equiv \perp}{\Gamma \vdash \neg p} \texttt{ EQF\_ELIM}\]
    \end{enumerate}

\end{itemize}

\subsubsection{Conversions}
Conversions are HOL derived rules of type \verb|:term| $\to$ \verb|thm|, that given a term $t$ return a theorem $\vdash t = t'$. Therefore conversions can be considered rules for transforming a term into an equal one, giving a theorem to justify the transformation. \\
$\triangleleft$ Conversions are convenient for implementing handy-derived rules. \\ 
$\triangleleft$ Conversions can be used as useful building blocks for larger transformations. \\ 

HOL has several built-in conversions:

\begin{itemize}
    \item \verb|REFL| also known as \verb|ALL_CONV| \\
     A sort of \textit{identity conversion} that trivially works on any term $t \to \ \vdash t=t$
     \item \verb|BETA_CONV| \\
     This conversion reduces an application where the first term is an abstraction to a $\beta$-redex $(\lambda x. s[x]) \ t \ \to \ \lambda x. s[x]) \ t = s[t]$
     \item \verb|NUM_RED_CONV| \\
     HOL supports some special conversions to particular theories. For example, \verb|NUM_RED_CONV| evaluate the result of an arithmetic operation on two numerals. Given $t * s$ an arithmetic operation and $w$ its result,  $t * s \to t * s = w $. 
\end{itemize}

You can consider these build-in conversions as conversion building blocks, and you can create composite conversions by using:

\begin{enumerate}
    \item \textbf{Conversionals} \\
    Conversionals allow the user to define composite conversions. Below we report a list of the most important conversionals:

\begin{itemize}
    \item \verb|THENC| \\
    \verb|(conv_1 THENC conv_2) t| applies first \verb|conv_1| to \verb|t| and afterwards \verb|conv_2|.
    \begin{lstlisting}
# (BETA_CONV THENC NUM_RED_CONV) `(\x. x + 3) 2`;;
val it : thm = |- (\x. x + 3) 2 = 5
    \end{lstlisting}

    \item  \verb|REPEATC| \\
    This conversional allows one to use a conversion repeatedly until it fails (maybe zero times).

    \item \verb|TRY_CONV| \\
    This conversional tries a conversion and returns a reflexive theorem if it fails.

    \item \verb|ORELSE_CONV| \\
    \verb|(conv_1 ORELSE conv_2) t| tries to apply \verb|conv_1| to the term, but if it fails, tries to apply \verb|conv_2|. For example, using this conversional it is possible to create a conversion that solves many arithmetical operations:
    \begin{lstlisting}
# let NUM_AME_CONV = NUM_ADD_CONV ORELSEC NUM_MULT_CONV ORELSEC NUM_EXP_CONV;;
    \end{lstlisting}
 \end{itemize}

 \item \textbf{Depth Conversionals} \\
 Depth conversionals allow users to apply conversions to subterms, in particular:

 \begin{itemize}
     \item \verb|RAND_CONV| \\
     Applies a conversion to the rand of a combination.

     \item \verb|RATOR_CONV| \\
     Applies a conversion to the rator of a combination.

     \item \verb|ABS_CONV| \\
     Applies a conversion to the body of an abstraction.

    \item \verb|BINDER_CONV| \\
     Applies a conversion to the body of a binding operation.

    \item \verb|BINOP_CONV| \\
     Applies a conversion to both arguments of a binary operation.

    \item \verb|ONCE_DEPTH_CONV| and \verb|DEPTH_CONV| \\
    These conversionals apply conversion in a more \textit{automatic} way to applicable subterms than the previous ones. \\
     \verb|ONCE_DEPTH_CONV| applies a conversion to the first applicable term encountered in a top-down traversal of the term, no deeper terms are examined. \\
     \verb|DEPTH_CONV| applies conversions recursively in a bottom-up fashion.
     \begin{lstlisting}
# ONCE_DEPTH_CONV NUM_AME_CONV `1 + (2 + 3)`;;
val it : thm = |- 1 + 2 + 3 = 1 + 5

# DEPTH_CONV NUM_AME_CONV `1 + (2 + 3)`;;
val it : thm = |- 1 + 2 + 3 = 6
     \end{lstlisting}
 \end{itemize}
 
\end{enumerate}

\subsubsection{Matching} 
In this subsection we will present derived rules that automatically determine how to instantiate variables to apply to the case in analysis. These rules are very useful, indeed it is more convenient to delegate HOL the task of choosing the correct instantiation for free and universally quantified variables (\textit{matching}). HOL can perform matching at the first-order level and a limited form of high-order matching:

\begin{itemize}
    \item \textbf{First Order Matching} \\
    For example,  the primitive inference rule \verb|MP| allows one to derive from $\vdash p \implies q$ and $\vdash p$ to $\vdash q$, but to match the user needs exactly the same $p$ (not something that is equal to p).

    \begin{lstlisting}
 # LT_IMP_LE;;
val it : thm = |- !m n. m < n ==> m <= n
# MP LT_IMP_LE (ARITH_RULE `1 < 2`);;
Exception: Failure "dest_binary".
# MP (SPECL [`1`; `2`] LT_IMP_LE) (ARITH_RULE `1 < 2`);;
val it : thm = |- 1 <= 2
\end{lstlisting}
\vspace{-5mm}
\begin{enumerate}
    \item Matching (\verb|MATCH_MP thm thm|)\\
    This rule is a version of \verb|MP| that instantiates free and universally quantified variables to apply \textit{modus ponens}.
    \begin{lstlisting}
# MATCH_MP LT_IMP_LE (ARITH_RULE `1 < 2`);;
val it : thm = |- 1 <= 2
    \end{lstlisting}

    \item Rewriting  (\verb|REWRITE_CONV thm tm| and \verb|REWRITE_CONV (thm)list tm|) \\
     \verb|REWRITE_CONV thm tm| takes an equational theorem $\vdash s = t$ or a bi-implicational theorem $\vdash s\iff t$, perhaps universally quantified, and a term $s'$ and, if $s'$ matches to $s$, it returns the corresponding theorem $\vdash s' = t'$.
\begin{lstlisting}
# LE_LT;;
val it : thm = |- !m n. m <= n <=> m < n \/ m = n
# REWR_CONV LE_LT `1<=2`;;
val it : thm = |- 1 <= 2 <=> 1 < 2 \/ 1 = 2
\end{lstlisting}

HOL allows to \verb|REWR_CONV| in association with depth conversions, to rewrite a term as much as possible.  \\

    \verb|REWRITE_CONV (thm)list tm| is a stronger version of \verb|REWR_CONV| that could take a whole list of theorems and extract rewrites from the theorems and repeatedly apply them to terms. Moreover, it uses a list of handy rewrites (\verb|basic_rewrites()|) to get rid of trivial propositional clutter.

    \begin{lstlisting}
# basic_rewrites();;
val it : thm list =
  [|- FST (x,y) = x; |- SND (x,y) = y; |- FST x,SND x = x;
   |- (if x = x then y else z) = y; |- (if T then t1 else t2) = t1;
   |- (if F then t1 else t2) = t2; |- ~ ~t <=> t; |- (@y. y = x) = x;
   |- x = x <=> T; |- (T <=> t) <=> t; |- (t <=> T) <=> t;
   |- (F <=> t) <=> ~t; |- (t <=> F) <=> ~t; |- ~T <=> F; |- ~F <=> T;
   |- T /\ t <=> t; |- t /\ T <=> t; |- F /\ t <=> F; |- t /\ F <=> F;
   |- t /\ t <=> t; |- T \/ t <=> T; |- t \/ T <=> T; |- F \/ t <=> t;
   |- t \/ F <=> t; |- t \/ t <=> t; |- T ==> t <=> t; |- t ==> T <=> T;
   |- F ==> t <=> T; |- t ==> t <=> T; |- t ==> F <=> ~t; |- (!x. t) <=> t;
   |- (?x. t) <=> t; |- (\x. f x) y = f y; |- x = x ==> p <=> p]
    \end{lstlisting}

    \verb|PURE_REWRITE_CONV (thm)list tm| is the same rule as \verb|REWRITE_CONV (thm)list tm| but it doesn't use \verb|basic_rewrites()|.

    Observe that HOL also provides \verb|ORDERED_REWRITE_CONV| to avoids rewriting loops of the kind $\vdash s= t \longrightarrow \vdash t=s \longrightarrow \vdash s= t \longrightarrow ...$ .

    \item Simplifying (\verb|SIMP_CONV thm tm|).  \\
This rule takes a conditional theorem of the form $\vdash p \implies s = t$, perhaps universally quantified, and a term $s'$ and, if $s'$ matches to $s$, it returns the corresponding theorem $\vdash p' \implies s' = t'$. The conversion is
recursively applied to the hypothesis $p'$, trying to reduce it to $\top$ and so eliminate it.

\begin{lstlisting}
# DIV_LT;;
it : thm = |- !m n. m < n ==> (m DIV n = 0)
# SIMP_CONV[DIV_LT; ARITH] `3 DIV 7 = 0`;;
it : thm = |- (3 DIV 7 = 0) = T
\end{lstlisting}

Here \verb|SIMP_CONV| uses as theorems for rewriting \verb|DIV_LT| and a set of rewrites to do basic arithmetic (\verb|ARITH|), but also simplification accumulates context, so when traversing a term $p \implies q $ and descending to $q$, additional rewrites are derived from $p$.
    
\end{enumerate}

    \item \textbf{High Order Matching} \\
    If HOL provided only first-order matching, the scope of rewriting would be rather restricted and simple schematic theorems with universal high-order quantifiers need to be instantiated manually. Instead, HOL implementation of HO matching develops a matching as general as possible while keeping it deterministic. More specifically it allows higher order matches of $P x_1 ... x_n$ where P is an instantiable HO variable, only if it can decide how to instantiate each $x_i$.
    
   \begin{lstlisting}
# NOT_FORALL_THM;;
it : thm = |- !P. ~(!x. P x) = (?x. ~P x)
# REWR_CONV NOT_FORALL_THM `~(!n. n < n - 1)`;;
it : thm = |- ~(!n. n < n - 1) = (?n. ~(n < n - 1))
   \end{lstlisting}

   In the above example \verb|Px| is instantiated by \verb| x < x - 1| so it is clear how to instatiate each $x_i$ and HOL will perform this rewrite even if it not a perfect instantiation of the left side of the theorem.

\end{itemize}

\subsubsection{Rules for Specific Theories}
As well as these general-purpose rules (logical rules, conversions and matching rules), HOL provides some special ones for particular theories. 

\begin{itemize}
    \item \textbf{Linear Arithmetic over Naturals} (\verb|ARITH_RULE|) \\
    This rule proves trivial theorems of linear arithmetic over the natural numbers, which is a decidable theory.

    \begin{lstlisting}
# ARITH_RULE `x < y ==> 4 * x + 3 < 4 * y`;;
it : thm = |- x < y ==> 4 * x + 3 < 4 * y
    \end{lstlisting}

    \item \textbf{Propositional Logic} (\verb|TAUT|) \\
    This rule proves tautologies automatically. Propositional logic is a decidable theory, so TAUT will provide a theorem iff the given term is a tautology.
    \begin{lstlisting}
# TAUT `(p ==> q) \/ (q ==> p)`;;
it : thm = |- (p ==> q) \/ (q ==> p)
    \end{lstlisting}
\newpage
    \item \textbf{Modal Logics }\\
    One of the purposes of HOLMS is to provide new rules to decide if a modal formula is valid in a certain decidable modal logic. As we will present later, since now we have derived in HOLMS three rules that will decide if a certain formula is a modal tautology and fail if it is not :
    \begin{enumerate}
        \item \verb|K_RULE| to decide if a formula is valid in $\mathcal{K}$
        \begin{lstlisting}
# K_RULE `!p q. [{} . {} |~ Box (p --> q) --> Box p --> Box q]`;;
0..0..1..2..3..7..12..solved at 19
val it : thm = |- !p q. [{} . {} |~ Box (p --> q) --> Box p --> Box q]

# K_RULE `!p . [{} . {} |~ Box p --> Box Box p]`;;
0..0..1..2..5..8..11...143..Exception: Failure "solve_goal: Too deep".
        \end{lstlisting}
        \item \verb|K4_RULE| to decide if a formula is valid in $\mathcal{K}4$
        \begin{lstlisting}
# K4_RULE `!p q. [K4_AX . {} |~ Box (p --> q) --> Box p --> Box q]`;;
0..0..1..2..3..7..12..solved at 19
val it : thm = |- !p q. [K4_AX . {} |~ Box (p --> q) --> Box p --> Box q]

# K4_RULE `!p. [K4_AX . {} |~ Box p --> Box (Box p)]`;;
0..0..1..2..5..8..11..14..17..solved at 26
val it : thm = |- !p. [K4_AX . {} |~ Box p --> Box Box p]

# K4_RULE `!p q. [K4_AX . {} |~  Box (p --> q) && Box p --> Box q]`;;
0..0..1..2..3..7..12..solved at 19
val it : thm = |- !p q. [K4_AX . {} |~ Box (p --> q) && Box p --> Box q]

# K4_RULE `!p. [K4_AX . {} |~ (Box (Box p --> p) --> Box p)]`;;
0..0..1..2..3..13..1557763..Exception: Failure "solve_goal: Too deep".
            
        \end{lstlisting}

        \item  \verb|GL_RULE| to decide if a formula is valid in $\mathcal{GL}$
        \begin{lstlisting}
# GL_RULE `!p. [GL_AX . {} |~ Box p --> Box (Box p)]`;;
0..0..1..7..15..23..37..51..107..solved at 116
val it : thm = |- !p. [GL_AX . {} |~ Box p --> Box Box p]

# GL_RULE `!p q. [GL_AX . {} |~  Box (p --> q) && Box p --> Box q]`;;
0..0..1..6..11..19..32..solved at 39
val it : thm = |- !p q. [GL_AX . {} |~ Box (p --> q) && Box p --> Box q]

# GL_RULE `!p. [GL_AX . {} |~ Box (Box p --> p) --> Box p]`;;
0..0..1..6..11..25..2115288..Exception: Failure "solve_goal: Too deep".
        \end{lstlisting}
    \end{enumerate}
\end{itemize}

\section{Proving Theorems}

Since we have described HOL language, its syntax and its deductive system, we can observe that HOL Light constitutes a systematic development of nontrivial mathematics from its foundations. So HOL, in some sense, represents something similar to logicist hope in the foundation of mathematics \cite{Russel1910Principia} and to Hilbert finitism \cite{Hilbert1926ÜberInfinite}. Unlike these formal systems, HOL is based on higher-order logic, which is more expressive than first-order logic by allowing quantification over predicates and functions, but semantically incomplete, in the sense that there is a $\psi$ such that it is true ($\vDash \psi$) and high-order logic can't prove it ($\nvdash \psi$). Therefore, the concept of HOL proof will not capture the idea of truth, as hoped by Hilbert's deductivism, but it is an interesting system to study, not only for those interested in theorem proving or functional programming languages, but also from a philosophical perspective. 

In this section, we will show how users actually prove theorems in HOL Light. 

\subsection{Tactics and Goalstack}
HOL provides two different mechanisms to prove theorems, here we will focus on the tactics mechanism that allows the user to perform interactive theorem proving.
\begin{itemize}
    \item \textbf{Rules Mechanism} \\
    This mechanism performs proofs step by step in a forward manner, applying rules to the set of hypotheses and then to conclusions of the applied rules. \\
    Hypotheses \\
    ... $\ \downarrow$ \textit{rules applications} \\
    Conclusion
    
    \item \textbf{Tactics Mechanism} \\
    This mechanism performs proofs in a mixture of forward and backward manner, starting with the main \textbf{goal} (theorem to prove) and reducing it, by applying tactics, to simpler \textbf{subgoals} until these can trivially be solved. A \textbf{tactic} is an OCaml function that takes a goal and reduces it to a list of subgoals, keeping track of how to construct a proof of the main goal if the user succeeds in proving the subgoal. So the user can keep applying tactics in a backward fashion, and HOL automatically reverses the user's one to a standard proof deduced by primitive inference rules. \\
    \end{itemize}

    Trivial subgoals \\
    ... $\big \uparrow $ \textit{apply tactics} \\
    Simpler subgoals $\Bigg \downarrow \Bigg \downarrow \Bigg\downarrow \Bigg\downarrow$ \textit{proving subgoals,} \textit{HOL reconstructs the forward proof}\\ 
    ... $\big \uparrow $ \textit{apply tactics} \\
    Conclusion (goal)

    \subsubsection{Goalstack}
    The goalstack allows the user to perform tactic steps, but also to retract and correct them. In this sense, the goalstack makes the proof interactive. The goalstack provides five kinds of commands:
    \begin{itemize}
        \item \verb|g <tm to prove>;;| \\
        Sets up the initial goal.
        \item  \verb|e <hol_tactic to apply>;;| \\
        Applies a tactic to the current goal.
        \item  \verb|b ();;| \\
        Retract a tactic step and restore the previous goal.
        \item  \verb|r n;;| \\
        Allows the user to analyse the subgoal in the n position.
        \item  \verb|top.thm();;| \\
        When the goalstack returns \verb|it : goalstack = No subgoals|, it means that proof has been completed and then the user can generate the desired theorem, using \verb|top.thm()|.        
    \end{itemize}

In the below example,  we perform a simple proof of a trivial arithmetical fact, by using each goal stack tool and some tactics: \verb|DISCH_TAC| (correspondent to \verb|DISCH|), \verb|CONJ_TAC| (correspondent to \verb|CONJ|), \verb|ASM_REWRITE_TAC[]| (that rewrites the goal using basics rewrites and the assumption of the theorem) and \verb|NUM_REDUCE_TAC| (solve simple subgoals using arithmetical conversions).

    \begin{lstlisting}
# g `x = 3 /\ y = 2 ==> x+y = 5 /\ x-y < 5`;;
val it : goalstack = 1 subgoal (1 total)
`x = 3 /\ y = 2 ==> x + y = 5 /\ x - y < 5`

# e DISCH_TAC;;
val it : goalstack = 1 subgoal (1 total)
  0 [`x = 3 /\ y = 2`]
`x + y = 5 /\ x - y < 5

# e CONJ_TAC;;
val it : goalstack = 2 subgoals (2 total)
  0 [`x = 3 /\ y = 2`]
`x - y < 5`
  0 [`x = 3 /\ y = 2`]
`x + y = 5`

# b ();;
val it : goalstack = 1 subgoal (1 total)
  0 [`x = 3 /\ y = 2`]
`x + y = 5 /\ x - y < 5`

# e (ASM_REWRITE_TAC[]);;
val it : goalstack = 1 subgoal (1 total)
  0 [`x = 3 /\ y = 2`]
`3 + 2 = 5 /\ 3 - 2 < 5`

# e CONJ_TAC;;
val it : goalstack = 2 subgoals (2 total)
  0 [`x = 3 /\ y = 2`]
`3 - 2 < 5`
  0 [`x = 3 /\ y = 2`]
`3 + 2 = 5`

# r 1;;
val it : goalstack = 1 subgoal (2 total)
  0 [`x = 3 /\ y = 2`]
`3 - 2 < 5`

# e NUM_REDUCE_TAC;;
val it : goalstack = 1 subgoal (1 total)
  0 [`x = 3 /\ y = 2`]
`3 + 2 = 5`

# e NUM_REDUCE_TAC;;
val it : goalstack = No subgoals

# let FIRST_THM = top_thm();;
val FIRST_THM : thm = |- x = 3 /\ y = 2 ==> x + y = 5 /\ x - y < 5
    \end{lstlisting}

\subsubsection{Tactics}\label{sec:tactics}
Many HOL basics tactics are a bottom-up version of HOL derived rules, such as the one used above. For example \verb|CONJ_TAC| is the reverse version of the rule \verb|CONJ| that takes two theorems and returns their conjunction, while \verb|CONJ_TAC| splits the conjunctive goal into two subgoals to prove.
\[ {\Gamma \vdash ^? p \land q \ \ } (\texttt{CONJ\_TAC} {\ \Rrightarrow) \ \Gamma \vdash p, \ \Gamma \vdash q}\]

\[\frac{\Gamma \vdash p \quad \Delta \vdash q}{\Gamma \cup \Delta \vdash p \land q} \texttt{ CONJ}\]

Matching tactics, such as \verb|MP_TAC| and \verb|MATCH_MP_TAC|, are more complicated, indeed the differences between the two are not only that the second performs matching while the first doesn't. Here we report a schematic description of the behavior of the two tactics: \\ 

$\Gamma \vdash ^? q \ \ $ \verb|MATCH_MP_TAC <thm after matching| $\ \vdash p \implies q$ \verb|>| $\ \Rrightarrow) \ \Gamma \vdash^? p, \ $

\begin{lstlisting}
#g `0 <= SUC n`;;
it : goalstack = 1 subgoal (1 total)
`0 <= SUC n`

#e (MATCH_MP_TAC LT_IMP_LE);;
it : goalstack = 1 subgoal (1 total)
`0 < SUC n`
\end{lstlisting}

$\Gamma \vdash ^? q \ \ $ \verb|MP_TAC <thm| $\ \vdash p' \implies q'$ \verb|>| $\ \Rrightarrow) \ \Gamma \vdash ^? (p' \implies q') \implies q$ 

\begin{lstlisting}
#g `0 <= SUC n`;;
it : goalstack = 1 subgoal (1 total)
`0 <= SUC n`
#e(MP_TAC LT_IMP_LE);;
it : goalstack = 1 subgoal (1 total)
`(!m n. m < n ==> m <= n) ==> 0 <= SUC n`
\end{lstlisting}

Other important tactics, are the ones that solve goals completely, such as \verb|ACCEPT_TAC| which is used with theorems with the same conclusion as the goal, and its matching version \verb|MATCH_ACCEPT_TAC|.

It is also possible to compose tactics using \textbf{tacticals}:
\begin{itemize}
    \item \verb|THEN| executes two tactics in sequence, observe that if the first tactic splits the goal into many subgoals then the second tactic is applied to each of them.
    \item \verb|THENL| first executes a tactic and then a list of tactics, observe that the list of tactics is applied only to the first subgoal if the first tactic splits the goal into different subgoals.
    \item \verb|REPEAT| executes repeatedly a certain tactic.
\end{itemize}

The last tactic we present is \verb|MESON_TAC[(thm)list]|, a very useful tactic that can dispose of a lot of first-order reasoning and a list of theorems given by the user. It also has a limited ability to do higher order and equality reasoning. Observe that, because first-order logic is undecidable\footnote{In the precise sense that the problem of deciding the validity of a first-order formula is undecidable}, it is impossible to discern if \verb|MESON_TAC[]| diverges or if it takes a "too" long computation time when it returns \verb|Failure "solve_goal: Too deep"|. Nevertheless, \verb|MESON_TAC[]| often succeeds on valid formulas.

\subsection{Theories Developed in HOL Light} \label{Theories}
Because in HOLMS we want to prove several theorems of modal logic using HOL Light, it is important to understand HOL Light infrastructure but also to know which theories have already been implemented in HOL. In particular, it is useful to know which theorems have been proved and their names, whereby we may use these proved facts to implement HOLMS.

In HOLMS we will use the following HOL-implemented theories:

\begin{enumerate}
    \item \textbf{Theory of Pairs} \\
    We use \textit{theory of pairs} in HOLMS to develop Kripke's frames.
    
    HOL uses the type constructor \verb|prod| to build up Cartesian product parametric type (\verb|:| $\alpha$ \verb|#| $\beta$) presented below. For example, \verb|:num # num| the type of pairs of natural numbers, while in HOLMS \verb|W->bool # W->W->bool| is the type of a modal frame.  Observe that n-tuples can be built by iterating the pairing operation and that the pair type constructor and pair term function are both right-associative. The most important theorems and functions for the theory of pairs are:

\begin{verbatim}
PAIR_EQ = |- !x y a b. (x, y = a, b) <=> (x = a) /\ (y = b)

PAIR_SURJECTIVE = |- !p. ?x y. p = x, y

pair_INDUCT = |- (!x y. P (x, y)) ==> (!p. P p)

pair_RECURSION = |- !PAIR'. ?fn. !a0 a1. fn (a0, a1) = PAIR' a0 a1

FORALL_PAIR_THM = |- !P. (!p. P p) <=> (!p1 p2. P (p1,p2))

EXISTS_PAIR_THM |- !P. (?p. P p) <=> (?p1 p2. P (p1,p2))

IN_ELIM_PAIR_THM |- !P a b. a,b IN {x,y | P x y} <=> P a b
\end{verbatim}

This last theorem brings together pairs theory and set theory and gives a characterization of pairs that is more similar to the one we are used to.

   \item \textbf{Theory of Sets} \\
    This theory is used in HOLMS to describe sets of certain Kripke's frames. 

    In HOL sets are just predicates, but the usual set operations are defined as below:
    
\[
\begin{array}{|l|c|l|}
\hline
\texttt{x IN s} & x \in s & \texttt{|- !P x. x IN P <=> P x} \\
\hline
\texttt{s SUBSET t} & s \subseteq t & \texttt{|- !s t. s SUBSET t = (!x. x IN s ==> x IN t)} \\
\hline
\texttt{EMPTY} & \emptyset & \texttt{EMPTY = |- EMPTY = ($\setminus$ x. F)} \\
\hline
\texttt{UNIV} & \textit{none} & \texttt{|- UNIV = ($\setminus$x. T)} \\
\hline
\texttt{x INSERT s} & \{x\} \cup s & \texttt{x INSERT s = $\{$y | y IN s \/ (y = x)$\}$} \\
\hline
\texttt{x DELETE s} & s - \{x\} & \texttt{|- !s x. s DELETE x = $\{$y | y IN s $\land \sim$ (y = x)$\}$} \\
\hline
\texttt{s UNION t} & s \cup t & \texttt{|- !s t. s UNION t = $\{$x | x IN s $\lor$ x IN t $\}$} \\
\hline
\texttt{UNIONS s} & \bigcup s & \texttt{|- !s. UNIONS s = $\{$x | ?u. u IN s $\land$ x IN u$\}$} \\
\hline
\texttt{s INTER t} & s \cap t & \texttt{|- !s t. s INTER t = $\{$x | x IN s $\land$ x IN t$\}$} \\
\hline
\texttt{INTERS s} & \bigcap s & \texttt{|- !s. INTERS s =$\{$x | !u. u IN s ==> x IN u$\}$} \\
\hline
\texttt{s DIFF t} & s - t & \texttt{|- !s t. s DIFF t =$\{$x | x IN s $\land \sim$(x IN t)$\}$} \\
\hline
\end{array}
\]

Moreover, HOL provides several pre-proved theorems about set and cardinality. In HOLMS we will also load the HOL Light library \verb|Library/card.ml|, where many other theorems about cardinality have been proved. 

\begin{verbatim}
CARD_CLAUSES =
  |- (CARD EMPTY = 0) /\
     (!x s.
       FINITE s
       ==> (CARD (x INSERT s) =
             if x IN s then CARD s else SUC (CARD s)))

HAS_SIZE = |- !s n. s HAS_SIZE n = FINITE s /\ (CARD s = n)

CARD_SUBSET_LE =
  |- !a b. FINITE b /\ a SUBSET b /\ CARD b <= CARD a ==> (a = b)

FINITE_RECURSION =
  |- !f b.
       (!x y s. ~(x = y) ==> (f x (f y s) = f y (f x s)))
       ==> (ITSET f EMPTY b = b) /\
           (!x s.
             FINITE s
             ==> (ITSET f (x INSERT s) b =
                   if x IN s then ITSET f s b else f x (ITSET f s b)))
\end{verbatim}

    \item\label{itm:theory-of-lists} \textbf{Theory of Lists} \\
    This theory is used in HOLMS where modal worlds are defined as (form)lists.

    HOL provides the definition of a recursive type of list and various standard list combinators defined by recursion. The parametric type for lists \verb|:(| $\alpha$ \verb|)list| was presented below. Here we list some important theorems about lists:

    \begin{verbatim}
APPEND =
|- (!l. APPEND [] l = l) /\
(!h t l. APPEND (CONS h t) l = CONS h (APPEND t l))

REVERSE =
|- (REVERSE [] = []) /\ (REVERSE (CONS x l) = APPEND (REVERSE l) [x])

LENGTH =
|- (LENGTH [] = 0) /\ (!h t. LENGTH (CONS h t) = SUC (LENGTH t))

MAP = |- (!f. MAP f [] = []) /\
(!f h t. MAP f (CONS h t) = CONS (f h) (MAP f t))

LAST = |- LAST (CONS h t) = (if t = [] then h else LAST t)

REPLICATE = |- (REPLICATE 0 x = []) /\
(REPLICATE (SUC n) x = CONS x (REPLICATE n x))

NULL = |- (NULL [] = T) /\ (NULL (CONS h t) = F)

FORALL = |- (FORALL P [] = T) /\
(FORALL P (CONS h t) = P h /\ FORALL P t)

EX = |- (EX P [] = F) /\ (EX P (CONS h t) = P h \/ EX P t)

ITLIST =
|- (ITLIST f [] b = b) /\ (ITLIST f (CONS h t) b = f h (ITLIST f t b))

MEM = |- (MEM x [] = F) /\ (MEM x (CONS h t) = (x = h) \/ MEM x t)
    \end{verbatim}
    
    \item \textbf{Well-founded Relations} \\
    In HOLMS theorems about well-founded relations are used to define Kripke's Frames for GL, since the accessibility relation for GL is converse well-founded and transitive\footnote{To treat reflexive, irreflexive, symmetric and transitive relations we load in HOLMS the library \verb|Library/rstc.ml|}.
    
    HOL includes the following definition of well-founded relation and it also proves that this definition is equivalent to other important properties, such as admissibility of complete induction.

    \begin{verbatim}
WF =
|- WF (<<) = (!P. (?x. P x) ==> (?x. P x /\ (!y. y << x ==> ~P y)))

WF_IND =
|- WF (<<) = (!P. (!x. (!y. y << x ==> P y) ==> P x) ==> (!x. P x))     
    \end{verbatim}

    \item \textbf{Number Theories} \\
    We won't use these theories in HOLMS, but we can not fail to mention that HOL proves many theorems about naturals, integers and reals. 
\end{enumerate}

\chapter{Formalising Modal Logic within HOL Light}\label{chap:3}

In this chapter, we introduce various research on the formalisation of modal logic within HOL Light. 

We start by defining \textbf{\textit{deep} and \textit{shallow embeddings}} within theorem provers, and then we show an implementation of these two kinds of embedding for modal logics developed by John Harrison in the HOL tutorial~\cite[\S 20]{harrisontutorial}.

Secondly, we present early work by two developers of HOLMS on the \textbf{implementation of Gödel-Löb provability logic in HOL Light}.

Finally, we describe in detail the\textbf{ HOLMS project}, its first released version and the current status of the library.

\section{Embeddings for Modal Logics}

As noticed in Chapter~\ref{sub:metalanguage} and illustrated in Figure~\ref{tik:meta}, formalising a modal logic within HOL Light involves treating HOL Light as a \emph{metatheory}, which analyses and describes the modal system as an \emph{object theory}.

When representing formal systems— such as logics or programming languages— within HOL-based proof assistants, two primary techniques are commonly used: deep and shallow embedding ~\cite{boulton-tpcd}. These two approaches differ in how they encode syntax and semantics of the object theory within the host HOL-based metatheory:
\begin{itemize}
    \item[(a)] \emph{\textbf{Deep Embedding}} or \emph{Semantic Embedding}
    \begin{itemize}
        \item The syntax of the embedded system is \textit{explicitly represented} with an \textit{inductive user-defined HOL type};
        \item The semantics of the embedded system is defined by \textit{interpretation functions} on the embedded syntax.
    \end{itemize}
Key characteristics of this approach:
\vspace{-3mm}
\begin{itemize}
    \item[$\triangleleft$] Neatly distinguished syntax and semantics;
    \item[$\triangleleft$] Since the syntax is explicitly defined, it is possible to reason about classes of programs and classes of logical object theories.
\end{itemize}
\newpage
    
    \item[(b)] \emph{\textbf{Shallow Embedding}} or \emph{Syntactic Embedding}
    \begin{itemize}
        \item The embedded language constructs are directly represented as HOL functions, predicates, or logical formulas, then the syntax of the embedded system is \textit{represented} with an \textit{OCaml type};
        \item The semantics of the embedded system is inherited from HOL.
    \end{itemize}
\end{itemize}

Key characteristic of this approach:
\begin{itemize}
    \item[$\triangleleft$] 
Avoid to explicitly setting up types of abstract syntax and semantic functions: the user only defines semantic operators in the logic and then the interface handles the mapping between HOL programs and their semantic representations.
\end{itemize}

\subsection{Embeddings for Modal Logic presented in HOL tutorial}

Harrison's tutorial~\cite[\S 20]{harrisontutorial} illustrates, as an example of the embedding of logics inside HOL, deep and shallow embeddings of propositional modal logics. He also proves some important technical results and Correspondence Lemma~\ref{lem:Correspondence}.
\medskip  

\noindent 
\begin{minipage}{0.45\textwidth}  
    \textbf{Deep Embedding}
    \begin{itemize}
        \item Modal Language: \\
        Definition of new HOL light terms of inductive type string to represent logical connectives;
        \item Modal Syntax:\\
        Definition of an inductive type \verb|:form|, as follows:
         \begin{lstlisting}
 # let form_INDUCT,form_RECURSION = 
    define_type
    "form = False
          | True
          | Atom string
          | Not form
          | && form form
          | || form form
          |--> form form
          | <-> form form
          | Box form
          | Diamond form";;
    \end{lstlisting}
\end{itemize}
\end{minipage}
\hfill 
\begin{minipage}{0.45\textwidth}  
 \vspace{-75mm} \textbf{Shallow Embedding}
    \begin{itemize}
        \item Modal Language: \\
        Definition of new HOL light terms of inductive type string to represent logical connectives;

        \item Modal Syntax: \\
        No explicit definitions of the embedded syntax;
    \end{itemize}
\end{minipage}

\noindent\begin{minipage}{0.45\textwidth}
\begin{itemize}
\vspace{-20mm}
        \item Modal Semantics: \\
        Definition of the interpretation function \verb|holds|, that follows the forcing relation in Definition~\ref{def:truth}:
         \begin{lstlisting}
# let holds = define
 `(holds (W,R) V False w 
  <=> F) /\
  ... /\
  (holds (W,R) V (Atom a) w 
   <=> V a w) /\
  ... /\
 (holds (W,R) V (p--> q) w 
 <=> holds (W,R) V p w 
  ==> holds (W,R) V q w) /\
    ... /\
  (holds (W,R) V (Box p) w 
 <=> !w'. w' IN W /\ R w w' 
  ==> holds (W,R) V p w') /\
 ...;;
    \end{lstlisting}
    
\item Modal Logics: \\
Harrison provides three examples of modal logics: Linear Temporal Logic ($\mathbb{LTL}$), $\mathbb{K}4$ and $\mathbb{GL}$. We focus on $\mathbb{LTL}$, since Harrison provides also a shallow embedding for this system.

Harrison model-theoretic defines $\mathbb{LTL}$ as the set of all the frames such that: their worlds are natural numbers, thought of as units of time, and a world is accessible precisely from earlier ones, including itself. 
\begin{lstlisting}
 # let LTL = new_definition
 `LTL(W,R) <=> 
 (W = UNIV) /\ 
 !x y:num. R x y <=> x <= y';;
\end{lstlisting}
   \end{itemize}

\end{minipage}
\hfill
\noindent \begin{minipage}{0.45\textwidth}
\begin{itemize}
    \item Modal Semantics:
\begin{itemize}

    \item[I.] Fixed $\mathbb{LTL}$ as an example, Harrison directly provides the semantics map $\mathbb{N} \to $\verb|bool| for propositional connectives:
    \begin{lstlisting}
 # let false_def = 
   define `False =
   \t:num. F';;
 ...
 # let imp_def = define 
   `p--> q = \t:num. p t ==> q t';;
    \end{lstlisting}
    
    Harrison remarks that these connectives are just special type instances of set operations, e.g \verb|&&| and \verb|INTER|. He describes connectives as ``lifting'' to predicates of the HOL propositional connectives.

    \item[II.] The following listing, instead, contains the semantics map $\mathbb{N} \to $\verb|bool| of the defined four modal operators connectives:
    \begin{lstlisting}
 # let forever = define 
   `forever p = 
   \t:num. !t'. t <= t' 
           ==> p t'';;
 # let sometime = 
 define `sometime p = 
   \t:num. ?t'. t <= t' 
            /\ p t'';;
 # let next = define 
  `next p = 
   \t:num. p(t + 1)';;
 # let until = define
   `p until q =
   \t:num. ?t'. t <= t' 
   /\ (!t''. t <= t'' /\ t'' < t' ==> p t'') /\ q t'';;
    \end{lstlisting}
\end{itemize}
\end{itemize}
\end{minipage}

\section{Mechanising GL within HOL Light}

The initial contributions of two of the to-be HOLMS developers to the formalisation of modal logic within HOL--Light consist of providing this proof assistant of a novel tool for automated modal reasoning (automated theorem proving theorems and countermodel construction) in Gödel-Löb Provability Logic~\cite{ maggesi_et_al:LIPIcs.ITP.2021.26, DBLP:journals/jar/MaggesiB23}.

The source code for this work has been made initially available from an online Git repository~\href{https://github.com/jrh13/hol-light/blob/master/GL/misc.ml} and has been included in the official HOL distribution since 2023. Consequently, we may say that this work has extended HOL Light’s capabilities to verify modal statements for $\mathbb{GL}$, and--in parallel\footnote{Solovay's theorem proves that Gödel-Löb Provability Logic is adequate w.r.t. arithmetic, i.e. in $\mathbb{GL}$ the modal operator $\Box$ characterises provability in arithmetic, while $\Diamond$ characterises consistency in arithmetic.}--to verify (modal counterparts of) statements on provability in arithmetical theories extending Peano arithmetic~\cite{boolos1995logic}.

\subsubsection{Axiomatic Calculus and Frames Correspondent to GL}

This previous work developed Harrison's embedding of syntax and semantics for modal logics~\cite{harrisontutorial}, by defining in HOL Light an axiomatic calculus for GL and the predicate of a frame being finite-irreflexive-transitive frames (encoding $\mathfrak{ITF}$).
Additionally, they implemented new results, including a polymorphic bisimulation lemma that formalises Lemma~\ref{thm:forc-bisim}.

\begin{lstlisting}[caption={Encoding an axiomatic system for GL in HOL Light}]
 #let GLaxiom_RULES,GLaxiom_INDUCT,GLaxiom_CASES =
  new_inductive_definition
  `(!p q. GLaxiom (p--> (q--> p))) /\
   (!p q r. GLaxiom ((p--> q--> r)--> (p--> q)--> (p--> r))) /\
   (!p. GLaxiom (((p--> False)--> False)--> p)) /\
   (!p q. GLaxiom ((p <-> q)--> p--> q)) /\
   (!p q. GLaxiom ((p <-> q)--> q--> p)) /\
   (!p q. GLaxiom ((p--> q)--> (q--> p)--> (p <-> q))) /\
      GLaxiom (True <-> False--> False) /\
   (!p. GLaxiom (Not p <-> p--> False)) /\
   (!p q. GLaxiom (p && q <-> (p--> q--> False)--> False)) /\
   (!p q. GLaxiom (p || q <-> Not(Not p && Not q))) /\
   (!p q. GLaxiom (Box (p--> q)--> Box p--> Box q)) /\
   (!p. GLaxiom (Box (Box p--> p)--> Box p))';;

 # let GLproves_RULES,GLproves_INDUCT,GLproves_CASES = new_inductive_definition
  `(!p. GLaxiom p ==>  |-- p) /\
   (!p q. |-- (p--> q) /\ |-- p ==> |-- q) /\
   (!p. |-- p ==>  |-- (Box p))';;
\end{lstlisting}
\vspace{-7mm}
\begin{lstlisting}[caption={Encoding GL characteristic  frames in HOL Light}]
 #let TRANSNT = new_definition
 `TRANSNT (W:W->bool,R:W->W->bool) <=>
  ~(W = {}) /\
  (!x y:W. R x y ==> x IN W /\ y IN W) /\
  (!x y z:W. x IN W /\ y IN W /\ z IN W /\ Rxy/\ Ryz ==> Rxz)/\
  WF (\x y. Ryx)';;

  TRANSNT_EQ_LOB
 |- !W:W->bool R:W->W->bool.
       (!x y:W. R x y ==> x IN W /\ y IN W) 
       ==> ((!x y z. x IN W /\ y IN W /\ z IN W /\ Rxy/\ Ryz 
                     ==> Rxz)/\
             WF (\xy.Ryx) <=>
             (!p. holds_in (W,R) (Box(Box p--> p)--> Box p)))
  
\end{lstlisting}
\vspace{-7mm}
\begin{lstlisting}[caption={Encoding GL appropriate frames in HOL Light}]
 # let ITF = new_definition
 `ITF (W:W->bool,R:W->W->bool) <=>
  ~(W = {}) /\
  (!x y:W. Rxy ==> x IN W /\ y IN W) /\
  FINITE W /\
  (!x. x IN W ==> ~ Rxx)/\
  (!x y z. x IN W /\ y IN W /\ z IN W /\ Rxy /\ R y z ==> R x z)';;

 ITF_NT
 |- !W R:W->W->bool. ITF(W,R) ==> TRANSNT(W,R)
\end{lstlisting}

\subsubsection{Adequacy Theorems for GL w.r.t. Relational Semantics}

The \verb|GL| library provides a formal proof in HOL Light of the adequacy theorems for $\mathbb{GL}$ w.r.t.~its appropriate frames ($\mathfrak{ITF}$ see Table~\ref{tab:char-appr}). 
The proof strategy, implemented with the support of the theorem prover, follows the approach outlined by George Boolos in his classic textbook~\cite[\S 4-5]{boolos1995logic}.

\begin{lstlisting}[caption={HOLMS Adequacy Theorems w.r.t. ITF}]
 GL_ITF_VALID
 |- !p. |-- p ==> ITF |= p

 COMPLETENESS_THEOREM
 |- !p. ITF |= p ==> (|-- p)
\end{lstlisting}

This completeness result allows in principle, thanks to finite model property, to develop a decision procedure for $\mathbb{GL}$-theoremhood ($\mathcal{GL} \vdash A$). A first naive--and then inefficient and unsuccessful--procedure is developed as follows: 

\begin{lstlisting}
 # let NAIVE_GL_TAC : tactic = 
 MATCH_MP_TAC GL_COUNTERMODEL_FINITE_SETS THEN
 REWRITE_TAC[valid; FORALL_PAIR_THM; holds_in; holds;
 ITF; GSYM MEMBER_NOT_EMPTY] THEN 
 MESON_TAC[];;
 
 let NAIVE_GL_RULE tm = prove(tm, REPEAT GEN_TAC
 THEN GL_TAC);; 
\end{lstlisting}

\subsubsection{Automated Theorem Prover and Countermodel constructor for GL}
To develop a working and efficient decision procedure, the completeness theorem is leveraged to construct a shallow embedding of a labeled sequent calculus for $\mathbb{GL}$.

This approach enables an automated proof search via a new HOL Light tactic.  \verb|GL_TAC| systematically applies sequent rules guided by the syntactic structure of the modal formula under examination, stated as a HOL Light verification goal.

The proof search either results in a theorem about provability in GL (\verb|GL_RULE|) or identifies a countermodel to the input formula (\verb|GL_BUILD_COUNTERMODEL|).

Figure~\ref{fig:schema} has been presented in~\cite{DBLP:conf/overlay/BilottaMBQ24} to summarise the overall procedure behind the implementation of this strategy in \texttt{GL} library, and hence in HOL Light proof assistant.

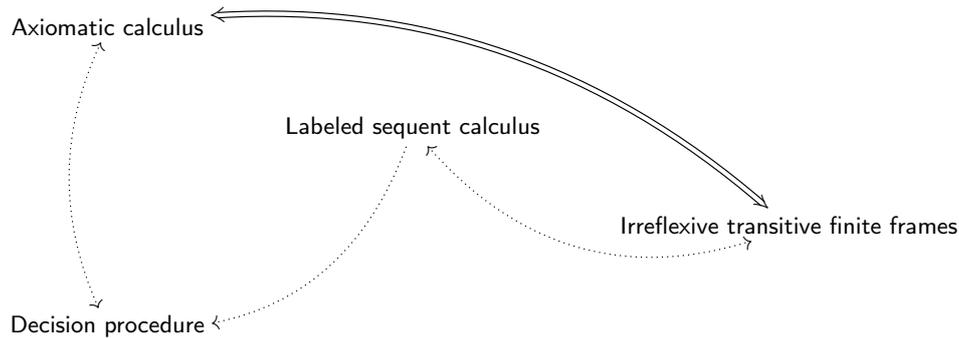
\begin{figure}[hb]
\begin{footnotesize}

$$\xymatrix{
    \textsf{Axiomatic calculus} \ar@/^2.8pc/@{<=>}[ddrr] & \\
    & \textsf{Labeled sequent calculus}\ar@/^2.3pc/@{.>}[ddl] & \\
    & & \textsf{Irreflexive transitive finite frames}\ar@/^2.3pc/@{<.>}[ul]  \\
	\textsf{Decision procedure}\ar@/^1.2pc/@{<.>}[uuu]    
}
$$\end{footnotesize}

    \caption{Dotted arrows represent formalisations/implementations depending on the shallow embedding of the labeled sequent calculus via our HOL Light rule for proof-search. The double arrow connecting the (formalised) axiomatic calculus and the (formalised) relational models based on irreflexive transitive finite frames denotes the fully formalised soundness and completeness theorems for Gödel-Löb logic.}
    \label{fig:schema}
\end{figure}

\section{HOLMS Project}\label{sub:HOLMS-project}

After developing \verb|GL| library, the two authors realised that the method implemented for $\mathbb{GL}$ is not \textit{per se} limited to this specific system:

\begin{itemize}
    \item Harrison's \textit{deep embedding of modal syntax and semantics} is general and allows us to treat each modal system;

    \item As showed~\ref{sub:parametric-calculus}, it is easy to define--and then to deep embed--the parametric notion of \textit{deducibility} in a normal system $\mathbb{S}$ from a set of hypotheses $\mathcal{H}$ ($\mathcal{S.H} \vdash A$). It is also possible to embed a deducibility relation defined for non-normal modal systems.

    \item The proof of \textit{adequacy theorems} sketched by Boolos' monograph~\cite[\S 4-5]{boolos1995logic} applies to all the modal logics within the cube;

    \item ``The implementation of a \textit{decision procedure} via proof-search in labeled calculi based on formalised completeness results may follow the general procedure of Figure~\ref{fig:HOLMS-schema} for any system in the modal cube~\cite{sep-logic-modal} and, modulo substantial investments, can cover any logic endowed with structurally well-behaved labeled sequent calculi, beyond the class of normal modal systems.''~\cite{DBLP:conf/overlay/BilottaMBQ24}

\end{itemize}

\begin{figure}[ht]
\begin{footnotesize}

$$\xymatrix{
    \textsf{Axiomatic calculus} \ar@/^2.8pc/@{<=>}[ddrr] & \\
    & \textsf{Labeled sequent calculus}\ar@/^2.3pc/@{.>}[ddl] & \\
    & & \parbox{3.5cm}{\textsf{Frame class} \textsf{characterising modal schemas}}\ar@/^2.3pc/@{<.>}[ul]  \\
	\textsf{Decision procedure}\ar@/^1.2pc/@{<.>}[uuu]    
}
$$\end{footnotesize}
    \caption{The implementation methodology behind HOLMS. Dotted arrows represent formalisations/implementations depending on the shallow embedding of the labeled sequent calculus via a HOL Light rule for proof-search. The double arrowconnecting the (formalised) axiomatic calculus and the (formalised) relational models in the class of frames corresponding to the modal schemas in the calculus--denotes the fully formalised soundness and completeness theorems for the logic.}
    \label{fig:HOLMS-schema}
\end{figure}
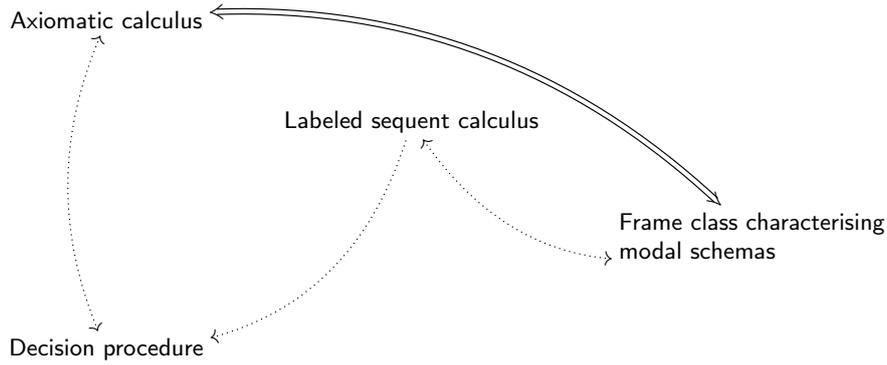 

As a result of these considerations, the HOLMS project was born and developers have (small) increased in number to explore the flexibility of the original approach to the mechanisation of modal logical reasoning.

Our methodology focuses on making as \textit{polymorphic} as possible the existing code for $\mathbb{GL}$ and porting it into a general framework for modal logic.

The ultimate goal of the HOLMS project— ``\textbf{HO}L Light \textbf{L}ibrary for \textbf{M}odal \textbf{S}ystems''— is to equip the HOL Light proof assistant with a general mechanism for automated theorem proving and countermodel construction for modal logics.

\subsection{The First Version of HOLMS presented at Overlay 2024}

The first--embryo--version of HOLMS was presented in our communication paper~\cite{DBLP:conf/overlay/BilottaMBQ24} at the ``Overlay 2024'' conference and its source code was archived in November 2024 on Software Heritage~\href{https://archive.softwareheritage.org/swh:1:rev:c15874a079d2ae24f420656d28c172e0c7258bed;origin=https://github.com/HOLMS-lib/HOLMS;visit=swh:1:snp:02b01d860dd78a254b705f5620846b4a5c0e0491}{\ExternalLink}.

We start generalising our approach by implementing within HOL Light the minimal system $\mathbb{K}$, its \textit{soundness} and \textit{completeness} theorems, and developing an automated theorem prover (\verb|K_TAC| and \verb|K_RULE|~\href{https://archive.softwareheritage.org/swh:1:cnt:ba5e8db8f660a197b3e275528a8f2ec81493c1a4;origin=https://github.com/HOLMS-lib/HOLMS;visit=swh:1:snp:02b01d860dd78a254b705f5620846b4a5c0e0491;anchor=swh:1:rev:c15874a079d2ae24f420656d28c172e0c7258bed;path=/k_decid.ml;lines=43-87}{\ExternalLink}) and countermodel constructor for it (\verb|K_BUILD_COUNTERMODEL|~\href{https://archive.softwareheritage.org/swh:1:cnt:ba5e8db8f660a197b3e275528a8f2ec81493c1a4;origin=https://github.com/HOLMS-lib/HOLMS;visit=swh:1:snp:02b01d860dd78a254b705f5620846b4a5c0e0491;anchor=swh:1:rev:c15874a079d2ae24f420656d28c172e0c7258bed;path=/k_decid.ml;lines=88-106}{\ExternalLink}).

In the following section, we describe some lines of the code already parametrised in this preliminary version.

\subsubsection{First parametrised results}
Since Harrison's deep embedding of syntax and semantics~\cite[\S 20]{harrisontutorial} can treat every modal logic, we started our parametrisation by generalising the \textit{deducibility predicate} for $\mathbb{GL}$ \verb+|-- A+ to all normal modal systems. We formalised this notion in a \textit{parametric} ternary relation \verb|[S.H ||\verb|~ A]|, which embed in HOL Light the \textit{abstract} Definition~\ref{def:deducibility}.
\newpage
\begin{lstlisting}[caption={HOLMS Parametric Deducibility Relation}]
   #let KAXIOM_RULES,KAXIOM_INDUCT,KAXIOM_CASES = new_inductive_definition
  `(!p q. KAXIOM (p --> (q --> p))) /\
   (!p q r. KAXIOM ((p --> q --> r) --> (p --> q) --> (p --> r))) /\
   (!p. KAXIOM (((p --> False) --> False) --> p)) /\
   (!p q. KAXIOM ((p <-> q) --> p --> q)) /\
   (!p q. KAXIOM ((p <-> q) --> q --> p)) /\
   (!p q. KAXIOM ((p --> q) --> (q --> p) --> (p <-> q))) /\
   KAXIOM (True <-> False --> False) /\
   (!p. KAXIOM (Not p <-> p --> False)) /\
   (!p q. KAXIOM (p && q <-> (p --> q --> False) --> False)) /\
   (!p q. KAXIOM (p || q <-> Not(Not p && Not q))) /\
   (!p q. KAXIOM (Box (p --> q) --> Box p --> Box q))'`;;

   #let MODPROVES_RULES,MODPROVES_INDUCT,MODPROVES_CASES =
  new_inductive_definition
  `(!H p. KAXIOM p ==> [S . H |~ p]) /\
   (!H p. p IN S ==> [S . H |~ p]) /\
   (!H p. p IN H ==> [S . H |~ p]) /\
   (!H p q. [S . H |~ p --> q] /\ [S . H |~ p] ==> [S . H |~ q]) /\
   (!H p. [S . {} |~ p] ==> [S . H |~ Box p])'`;;
\end{lstlisting}

The library also provides many lemmata that holds by normality~\href{https://archive.softwareheritage.org/swh:1:cnt:d42782345008434be8b43de7feb014732f8821ea;origin=https://github.com/HOLMS-lib/HOLMS;visit=swh:1:snp:02b01d860dd78a254b705f5620846b4a5c0e0491;anchor=swh:1:rev:c15874a079d2ae24f420656d28c172e0c7258bed;path=/calculus.ml;lines=169-843}{\ExternalLink}, a substitution lemma~\href{https://archive.softwareheritage.org/swh:1:cnt:d42782345008434be8b43de7feb014732f8821ea;origin=https://github.com/HOLMS-lib/HOLMS;visit=swh:1:snp:02b01d860dd78a254b705f5620846b4a5c0e0491;anchor=swh:1:rev:c15874a079d2ae24f420656d28c172e0c7258bed;path=/calculus.ml;lines=883-887}{\ExternalLink} and a deduction theorem~\href{https://archive.softwareheritage.org/swh:1:cnt:d42782345008434be8b43de7feb014732f8821ea;origin=https://github.com/HOLMS-lib/HOLMS;visit=swh:1:snp:02b01d860dd78a254b705f5620846b4a5c0e0491;anchor=swh:1:rev:c15874a079d2ae24f420656d28c172e0c7258bed;path=/calculus.ml;lines=1015-1028}{\ExternalLink}.

Moreover, this version of HOLMS introduced a partial parametrisation of the proof of adequacy theorems. Below, we list some of the concepts that were fully \textit{parametrised} in that release. However, the source code also contained numerous ad \textit{hoc parametric} notions, such as \verb|GEN_STANDARD_MODEL|~\href{https://archive.softwareheritage.org/swh:1:cnt:f12ff7e3946938d5fa88228be9b5686826d595a1;origin=https://github.com/HOLMS-lib/HOLMS;visit=swh:1:snp:02b01d860dd78a254b705f5620846b4a5c0e0491;anchor=swh:1:rev:c15874a079d2ae24f420656d28c172e0c7258bed;path=/gen_completness.ml;lines=48-51}{\ExternalLink} and \verb|GEN_STANDARD_REL|~\href{https://archive.softwareheritage.org/swh:1:cnt:f12ff7e3946938d5fa88228be9b5686826d595a1;origin=https://github.com/HOLMS-lib/HOLMS;visit=swh:1:snp:02b01d860dd78a254b705f5620846b4a5c0e0491;anchor=swh:1:rev:c15874a079d2ae24f420656d28c172e0c7258bed;path=/gen_completness.ml;lines=222-226}{\ExternalLink}, which were not fully parameterised but still contributed to facilitating the completeness proofs.

\begin{itemize}
    \item The \textit{parametric} concepts of $\mathcal{S}$-consistency and $\mathcal{S},A$-consistency--which are fundamental to prove completeness--were embedded in HOL Light, together with some related important theorems (formalising Lemmata~\ref{lem:consistency_lemma}-~\ref{lem:max_cons});
    \begin{lstlisting}[caption={HOLMS Parametric Consistency and Maximal Consistency}]
 #let CONSISTENT = new_definition
 `CONSISTENT S (l:form list) <=> ~[S . {} |~ Not (CONJLIST l)]';;   

 #let MAXIMAL_CONSISTENT = new_definition
 `MAXIMAL_CONSISTENT S p X <=>
  CONSISTENT S X /\ NOREPETITION X /\
  (!q. q SUBFORMULA p ==> MEM q X \/ MEM (Not q) X)`;;
    \end{lstlisting}
    \vspace{-7mm}
    \item The set of all relational frames (\verb|FRAME|) is defined in HOLMS;
    \begin{lstlisting}[caption={HOLMS Kripke Frames}]
 #let FRAME_DEF = new_definition
 `FRAME = {(W:W->bool,R:W->W->bool) | ~(W = {}) /\ (!x y:W. R x y ==> x IN W /\ y IN W)}';;
    \end{lstlisting}
    \vspace{-7mm}
    \item A parametric version of \textit{general truth lemma}--Lemma~\ref{lem:truth} , i.e. one fundamental step of completeness proof--is proved within HOL Light~\href{https://archive.softwareheritage.org/swh:1:cnt:f12ff7e3946938d5fa88228be9b5686826d595a1;origin=https://github.com/HOLMS-lib/HOLMS;visit=swh:1:snp:02b01d860dd78a254b705f5620846b4a5c0e0491;anchor=swh:1:rev:c15874a079d2ae24f420656d28c172e0c7258bed;path=/gen_completness.ml;lines=53-217}{\ExternalLink};
    \begin{lstlisting}[caption={HOLMS Parametric Truth Lemma}]
 GEN_TRUTH_LEMMA 
 |-!P S W R p V q.
  ~ [S . {} |~ p] /\
  GEN_STANDARD_MODEL P S p (W,R) V /\
  q SUBFORMULA p
    ==> !w. w IN W ==> (MEM q w <=> holds (W,R) V q w)
    \end{lstlisting}
\end{itemize}

\subsection{The Current Version of HOLMS}\label{sub:current-HOLMS}

In the present version of the library, we further refine the identification of the truly \emph{polymorphic} aspects of our previous implementation of HOLMS. This enables us to clearly distinguish between those parts of our code that are \textit{parametric polymorphic} and those that are \textit{ad-hoc polymorphic}. 
Adopting the classical terminology from~\cite{strachey2000fundamental}, we refer to:
\begin{itemize}
    \item \textit{Parametric Polymorphism} for code that remains fully independent from the concrete instantiations of the parameters;
    \item \emph{Ad-hoc Polymorphism} for code with details more closely tied to the particular logic under investigation.
\end{itemize}

The current version of the code has reached the result of \emph{parameterising} as much as possible the completeness proof. 
To achieve this,  key \textit{parametric} concepts from modal correspondence theory have been embedded in HOLMS, together with the \textit{ad hoc polymophic} Correspondence Lemma~\ref{lem:Correspondence}.
In particular, our work leverages known results about correspondent frames
(summarised in Table~\ref{tab:char-appr}) to instantiate the \textit{parametric} Countermodel Lemma~\ref{lem:parametric_countermodel_lemma}. 
Thanks to this refinement, which enables a fully modular and uniform proof of adequacy theorems for normal modal systems, two additional object theories--$\mathbb{T}$ and $\mathbb{K}4$--have been seamlessly integrated into HOLMS.

Our library, hence, proves the \emph{adequacy} theorems for the minimal system $\mathbb{K}$ and its extensions $\mathbb{\ T, \ K}4$, as well as a reshaping of the original proof for $\mathbb{GL}$ to reduce it to a direct application of the general result.

Moreover, it offers a \emph{principled decision procedure}--and, by extension, automated theorem provers and countermodel constructors--for $\mathbb{K}$ and $\mathbb{T}$ based on the automated proof-search in the associated labeled sequent calculus. This procedure is implemented trough HOL Light tactics similar to those developed for $\mathbb{GL}$ in~\cite{DBLP:journals/jar/MaggesiB23}. 
The corresponding implementation for $\mathbb{K}4$ provides a \textit{semidecision} procedure. These implementations serve as evidence for the feasibility of our new approach to automated modal reasoning in general-purpose proof assistants, as illustrated in Figure~\ref{fig:HOLMS-schema}.

\begin{remark}
As argued in~\cite[\S 1]{maggesi_et_al:LIPIcs.ITP.2021.26}, this formalisation leaves the original HOL Light tools unmodified, ensuring that HOLMS is \textit{foundationally safe}. By relying solely on the original formal infrastructure, our library can be easily \textit{ported to any HOL-based proof assistant} equipped with a similar automation toolbox.
\end{remark}

\subsubsection{HOLMS Source Code}
In its current version, HOLMS consists of just under six thousand lines of OCaml code, organised into a \textbf{git repository}~\href{https://github.com/HOLMS-lib/HOLMS}{\ExternalLink} composed by seventeen files \verb|.ml|.

This second released version of the source code, introduced in the present thesis and submitted to the ``FSCD 2025''  conference, was archived in February 2025 on Software Heritage ~\href{https://archive.softwareheritage.org/swh:1:dir:6dbea8a51b27f3b4de7808eb0bc2fa82706962c4;origin=https://github.com/HOLMS-lib/HOLMS;visit=swh:1:snp:2c0efd349323ed6f8067581cf1f6d95816e49841;anchor=swh:1:rev:1caf3be141c6f646f78695c0eb528ce3b753079a}{\ExternalLink}.

Our \textbf{project web page}~\cite{hol-light-webpage}~\href{https://holms-lib.github.io/}{\ExternalLink} gives a brief overview of HOLMS library, its evolution and related publications.
The \textbf{readme file}~\href{https://holms-lib.github.io/}{\ExternalLink}, instead, highlights the specificities of this second version of HOLMS, and provides a usage guide for our library at its current status.

\subsubsection{Mapping HOLMS repository}

Below, we briefly present a sort of map of the seventeen \verb|.ml| files contained in HOLMS repository and of their roles in the implementation of modal logics.

\medskip

\begin{mydefinition}
\begin{itemize}
    \item[I.] \textbf{Top Level File} \verb|make.ml|~\href{https://archive.softwareheritage.org/swh:1:cnt:42773492356633a84a16bc55762c54ae8bf1f968;origin=https://github.com/HOLMS-lib/HOLMS;visit=swh:1:snp:2c0efd349323ed6f8067581cf1f6d95816e49841;anchor=swh:1:rev:1caf3be141c6f646f78695c0eb528ce3b753079a;path=/make.ml}{\ExternalLink} \\
    This script loads (\verb|needs|)  two \textbf{libraries} necessary to treat the cardinality of the domains and the properties of the accessibility relationships: \begin{itemize}
        \item \verb|"Library/card.ml"|~\href{https://archive.softwareheritage.org/swh:1:cnt:42773492356633a84a16bc55762c54ae8bf1f968;origin=https://github.com/HOLMS-lib/HOLMS;visit=swh:1:snp:2c0efd349323ed6f8067581cf1f6d95816e49841;anchor=swh:1:rev:1caf3be141c6f646f78695c0eb528ce3b753079a;path=/make.ml;lines=15}{\ExternalLink}: this library defines concepts and proves theorems related to \textit{cardinality}, which is the mathematical concept of measuring the size of sets.
        \item \verb|"Library/rstc.ml"|~\href{https://archive.softwareheritage.org/swh:1:cnt:42773492356633a84a16bc55762c54ae8bf1f968;origin=https://github.com/HOLMS-lib/HOLMS;visit=swh:1:snp:2c0efd349323ed6f8067581cf1f6d95816e49841;anchor=swh:1:rev:1caf3be141c6f646f78695c0eb528ce3b753079a;path=/make.ml;lines=16}{\ExternalLink}: this library defines concepts and proves theorems related to \textit{reflexive}, \textit{symmetric} and \textit{transitive} closures.
    \end{itemize}
    
    Thereafter, all the other OCaml files are loaded (\verb|loadt|), except from \verb|tests.ml|. 
    
    We do not load this file because it contains some examples of decision procedures, which are not essential for the library itself but are useful for users who want to analyse how these procedures work. In general, test files in a library serve as a way to verify the correctness of implemented functions and provide usage examples. They are generally kept separate from the main code to avoid unnecessary overhead during execution~\cite[\S 17]{Madhavapeddy_Minsky_2022}.
 \medskip
 
    \item[II.] \textbf{Defining Syntax, Semantics and a Calculus for normal modal logic}
    \begin{itemize}
        \item[a.] \verb|misc.ml|~\href{https://archive.softwareheritage.org/swh:1:cnt:d100c0780569a67ceb5f1802e928df774c6a6107;origin=https://github.com/HOLMS-lib/HOLMS;visit=swh:1:snp:2c0efd349323ed6f8067581cf1f6d95816e49841;anchor=swh:1:rev:1caf3be141c6f646f78695c0eb528ce3b753079a;path=/misc.ml}{\ExternalLink}: Contains \textit{miscellaneous} utilities. \verb|misc|-files typically collect auxiliary functions, general utilities, and reusable code. This file was already present in the original git library \verb|GL|~\href{https://github.com/jrh13/hol-light/blob/master/GL/misc.ml}{\ExternalLink}.
        
        \item[b.] \verb|modal.ml|~\href{https://archive.softwareheritage.org/swh:1:cnt:4a011f69b1180019be5d03c6d5f1cec4550055e8;origin=https://github.com/HOLMS-lib/HOLMS;visit=swh:1:snp:2c0efd349323ed6f8067581cf1f6d95816e49841;anchor=swh:1:rev:1caf3be141c6f646f78695c0eb528ce3b753079a;path=/modal.ml}{\ExternalLink}: Provides a \textbf{deep embedding} of the \textbf{syntax} of propositional modal logic and of \textbf{relational semantics}.

        \item[c.] \verb|calculus.ml|~\href{https://archive.softwareheritage.org/swh:1:cnt:d42782345008434be8b43de7feb014732f8821ea;origin=https://github.com/HOLMS-lib/HOLMS;visit=swh:1:snp:2c0efd349323ed6f8067581cf1f6d95816e49841;anchor=swh:1:rev:1caf3be141c6f646f78695c0eb528ce3b753079a;path=/calculus.ml}{\ExternalLink}: Implements a parametric \textbf{axiomatic calculus} for \textbf{normal} modal logics.
    \end{itemize}
\medskip
    
    \item[III.] \textbf{Introducing maximal consistent lists to prove completeness}
    \begin{itemize}
        \item[a.]  \verb|conjlist.ml|~\href{https://archive.softwareheritage.org/swh:1:cnt:3d890d9e1ac7d24532b2645b6fb9ad5b160c7ba0;origin=https://github.com/HOLMS-lib/HOLMS;visit=swh:1:snp:2c0efd349323ed6f8067581cf1f6d95816e49841;anchor=swh:1:rev:1caf3be141c6f646f78695c0eb528ce3b753079a;path=/conjlist.ml}{\ExternalLink}:  
        \begin{itemize}
            \item Defines conjunctions on lists of formulas (\verb|CONJLIST|~\href{https://archive.softwareheritage.org/swh:1:cnt:3d890d9e1ac7d24532b2645b6fb9ad5b160c7ba0;origin=https://github.com/HOLMS-lib/HOLMS;visit=swh:1:snp:2c0efd349323ed6f8067581cf1f6d95816e49841;anchor=swh:1:rev:1caf3be141c6f646f78695c0eb528ce3b753079a;path=/conjlist.ml;lines=11-13}{\ExternalLink});
            \item Proves some related parametric results.
        \end{itemize}
        
        \item[b.] \verb|consistent.ml|~\href{https://archive.softwareheritage.org/swh:1:cnt:e261ca08330d5ef66376347407c27fcece1e9f2e;origin=https://github.com/HOLMS-lib/HOLMS;visit=swh:1:snp:2c0efd349323ed6f8067581cf1f6d95816e49841;anchor=swh:1:rev:1caf3be141c6f646f78695c0eb528ce3b753079a;path=/consistent.ml}{\ExternalLink}:
        \begin{itemize}
            \item Defines parametric predicates $\mathcal{S}$-\textbf{consistent} (\verb|CONSISTENT S|~\href{https://archive.softwareheritage.org/swh:1:cnt:e261ca08330d5ef66376347407c27fcece1e9f2e;origin=https://github.com/HOLMS-lib/HOLMS;visit=swh:1:snp:2c0efd349323ed6f8067581cf1f6d95816e49841;anchor=swh:1:rev:1caf3be141c6f646f78695c0eb528ce3b753079a;path=/consistent.ml;lines=11-12}{\ExternalLink}) and $\mathcal{S},A$-\textbf{maximal consistent} (\verb|MAXIMAL_CONSISTENT S A|~\href{https://archive.softwareheritage.org/swh:1:cnt:e261ca08330d5ef66376347407c27fcece1e9f2e;origin=https://github.com/HOLMS-lib/HOLMS;visit=swh:1:snp:2c0efd349323ed6f8067581cf1f6d95816e49841;anchor=swh:1:rev:1caf3be141c6f646f78695c0eb528ce3b753079a;path=/consistent.ml;lines=82-85}{\ExternalLink});
            \item Proves some lemmata necessary for the completeness theorem.
        \end{itemize}
    \end{itemize}
\medskip    
     \item[IV.] \textbf{Proving Adequacy Theorems}
     \begin{itemize}
        \item[(a)] Parametric file \verb|gen_completeness.ml|~\href{https://archive.softwareheritage.org/swh:1:cnt:d364de3158b6405c920115aa0c801af9af49e302;origin=https://github.com/HOLMS-lib/HOLMS;visit=swh:1:snp:2c0efd349323ed6f8067581cf1f6d95816e49841;anchor=swh:1:rev:1caf3be141c6f646f78695c0eb528ce3b753079a;path=/gen_completeness.ml}{\ExternalLink}:
        \begin{itemize}
            \item Defines the concept of \textbf{relational frames} (\verb|FRAME|~\href{https://archive.softwareheritage.org/swh:1:cnt:d364de3158b6405c920115aa0c801af9af49e302;origin=https://github.com/HOLMS-lib/HOLMS;visit=swh:1:snp:2c0efd349323ed6f8067581cf1f6d95816e49841;anchor=swh:1:rev:1caf3be141c6f646f78695c0eb528ce3b753079a;path=/gen_completeness.ml;lines=15-16}{\ExternalLink}); 
            \item Defines fundamental concepts of \textbf{correspondence theory }(\verb|CHAR S|~\href{https://archive.softwareheritage.org/swh:1:cnt:d364de3158b6405c920115aa0c801af9af49e302;origin=https://github.com/HOLMS-lib/HOLMS;visit=swh:1:snp:2c0efd349323ed6f8067581cf1f6d95816e49841;anchor=swh:1:rev:1caf3be141c6f646f78695c0eb528ce3b753079a;path=/gen_completeness.ml;lines=44-47}{\ExternalLink} and \verb|APPR S|~\href{https://archive.softwareheritage.org/swh:1:cnt:d364de3158b6405c920115aa0c801af9af49e302;origin=https://github.com/HOLMS-lib/HOLMS;visit=swh:1:snp:2c0efd349323ed6f8067581cf1f6d95816e49841;anchor=swh:1:rev:1caf3be141c6f646f78695c0eb528ce3b753079a;path=/gen_completeness.ml;lines=111-114}{\ExternalLink}) and proves their characterisation  (\verb|CHAR_CAR|~\href{https://archive.softwareheritage.org/swh:1:cnt:d364de3158b6405c920115aa0c801af9af49e302;origin=https://github.com/HOLMS-lib/HOLMS;visit=swh:1:snp:2c0efd349323ed6f8067581cf1f6d95816e49841;anchor=swh:1:rev:1caf3be141c6f646f78695c0eb528ce3b753079a;path=/gen_completeness.ml;lines=82-105}{\ExternalLink} and \verb|APPR_CAR|~\href{https://archive.softwareheritage.org/swh:1:cnt:d364de3158b6405c920115aa0c801af9af49e302;origin=https://github.com/HOLMS-lib/HOLMS;visit=swh:1:snp:2c0efd349323ed6f8067581cf1f6d95816e49841;anchor=swh:1:rev:1caf3be141c6f646f78695c0eb528ce3b753079a;path=/gen_completeness.ml;lines=122-128}{\ExternalLink});
            \item Proves a parametric theorem of\textbf{ soundness w.r.t. characteristic frames} that hods for every normal system (\verb|GEN_CHAR_VALID|~\href{https://archive.softwareheritage.org/swh:1:cnt:d364de3158b6405c920115aa0c801af9af49e302;origin=https://github.com/HOLMS-lib/HOLMS;visit=swh:1:snp:2c0efd349323ed6f8067581cf1f6d95816e49841;anchor=swh:1:rev:1caf3be141c6f646f78695c0eb528ce3b753079a;path=/gen_completeness.ml;lines=71-80}{\ExternalLink});
            \item Formalises the \textbf{parametric} part of the proof of \textbf{completeness} sketched in~\ref{sub:parametrised-completeness} (\verb|GEN_STANDARD_MODEL|~\href{https://archive.softwareheritage.org/swh:1:cnt:d364de3158b6405c920115aa0c801af9af49e302;origin=https://github.com/HOLMS-lib/HOLMS;visit=swh:1:snp:2c0efd349323ed6f8067581cf1f6d95816e49841;anchor=swh:1:rev:1caf3be141c6f646f78695c0eb528ce3b753079a;path=/gen_completeness.ml;lines=180-183}{\ExternalLink}, \verb|GEN_TRUTH_LEMMA|~\href{https://archive.softwareheritage.org/swh:1:cnt:d364de3158b6405c920115aa0c801af9af49e302;origin=https://github.com/HOLMS-lib/HOLMS;visit=swh:1:snp:2c0efd349323ed6f8067581cf1f6d95816e49841;anchor=swh:1:rev:1caf3be141c6f646f78695c0eb528ce3b753079a;path=/gen_completeness.ml;lines=189-350}{\ExternalLink} and \verb|GEN_COUNTER| \verb|MODEL_ALT|~\href{https://archive.softwareheritage.org/swh:1:cnt:d364de3158b6405c920115aa0c801af9af49e302;origin=https://github.com/HOLMS-lib/HOLMS;visit=swh:1:snp:2c0efd349323ed6f8067581cf1f6d95816e49841;anchor=swh:1:rev:1caf3be141c6f646f78695c0eb528ce3b753079a;path=/gen_completeness.ml;lines=574-593}{\ExternalLink})
            \item Ease the proof of the unparametrisable part of completeness in~\ref{sub:unparametrisable-completeness};
            \item Generalises completeness on domains with a generic type via \textit{bisimiluation} (\verb|GEN_LEMMA_FOR_GEN_COMPLETENESS|~\href{https://archive.softwareheritage.org/swh:1:cnt:d364de3158b6405c920115aa0c801af9af49e302;origin=https://github.com/HOLMS-lib/HOLMS;visit=swh:1:snp:2c0efd349323ed6f8067581cf1f6d95816e49841;anchor=swh:1:rev:1caf3be141c6f646f78695c0eb528ce3b753079a;path=/gen_completeness.ml;lines=599-707}{\ExternalLink}).
        \end{itemize}
 \smallskip               
        
        \item[(b)] Specific completeness files 
        \vspace{-4mm}
        \begin{multicols}{2}
        \begin{enumerate}
            \item \verb|k_completeness.ml|~\href{https://archive.softwareheritage.org/swh:1:cnt:ae230138ff15476c8dab9e32606bceca7168285b;origin=https://github.com/HOLMS-lib/HOLMS;visit=swh:1:snp:2c0efd349323ed6f8067581cf1f6d95816e49841;anchor=swh:1:rev:1caf3be141c6f646f78695c0eb528ce3b753079a;path=/k_completeness.ml}{\ExternalLink};
            \item \verb|t_completeness.ml|\href{https://archive.softwareheritage.org/swh:1:cnt:1352cb294724b2251b9346657432700ecbed10d2;origin=https://github.com/HOLMS-lib/HOLMS;visit=swh:1:snp:2c0efd349323ed6f8067581cf1f6d95816e49841;anchor=swh:1:rev:1caf3be141c6f646f78695c0eb528ce3b753079a;path=/t_completeness.ml}{\ExternalLink};
            \item \verb|k4_completeness.ml|~\href{https://archive.softwareheritage.org/swh:1:cnt:32869a1df338d91c4cd4a0cb9e73fb4f1be29991;origin=https://github.com/HOLMS-lib/HOLMS;visit=swh:1:snp:2c0efd349323ed6f8067581cf1f6d95816e49841;anchor=swh:1:rev:1caf3be141c6f646f78695c0eb528ce3b753079a;path=/k4_completeness.ml}{\ExternalLink};
            \item \verb|gl_completeness.ml|~\href{https://archive.softwareheritage.org/swh:1:cnt:52c9de12454731a9a291b06bc750c0d2d14e1fb4;origin=https://github.com/HOLMS-lib/HOLMS;visit=swh:1:snp:2c0efd349323ed6f8067581cf1f6d95816e49841;anchor=swh:1:rev:1caf3be141c6f646f78695c0eb528ce3b753079a;path=/gl_completeness.ml}{\ExternalLink}.
        \end{enumerate}
        \end{multicols}
        For each logic \verb|*| 
        $\in \{ \mathbb{K; \ T; \ K}4; \ \mathbb{GL} \}$, these files provide the:
        \begin{itemize}
            \item Definition of the \textbf{specific axiom schemata} for \verb|*|;
            \item \textbf{Identification} of the set of frames \textbf{characteristic} to \verb|*|;
            \item Theorems of \textbf{soundness} and \textbf{consistency} that follow as corollaries;
            \item \textbf{Identification} of the \textbf{standard accessibility relation} $Rel_*
            ^A$ (\verb|*_STANDARD_REL|)
            \item Two lemmas stating the unparametrisable part of completeness proof sketched in~\ref{sub:unparametrisable-completeness} (\verb|*_MAXIMAL_CONSISTENT_LEMMA| and \verb|*_ACCESSI| \verb|BILITY_LEMMA|)
            \item A \textbf{completeness} theorem for \verb|*| (\verb|*_COMPLETENESS_THM|);
            \item A \textbf{naive decision procedure} for \verb|*| (\verb|*_RULE| and \verb|*_TAC|);
            
        \end{itemize}
        
    \end{itemize}\

     \item[V.] \textbf{Automated theorem proving}
     \begin{itemize}
        \item[(a)] Parametric file \verb|gen_decid.ml|\href{https://archive.softwareheritage.org/swh:1:cnt:dbf493c0cf78f9cb975db418aad12fdd3f37cfb8;origin=https://github.com/HOLMS-lib/HOLMS;visit=swh:1:snp:2c0efd349323ed6f8067581cf1f6d95816e49841;anchor=swh:1:rev:1caf3be141c6f646f78695c0eb528ce3b753079a;path=/gen_decid.ml}{\ExternalLink}: Ease the implementation of the theorem prover;
        \medskip
        
        \item[(b)] Specific files for decision procedures
        \vspace{-4mm}
        \begin{multicols}{2}
        \begin{enumerate}
            \item \verb|k_decid.ml|~\href{https://archive.softwareheritage.org/swh:1:cnt:7ec6e34a7d81fb00cbf22dcbf9b09c66538ad195;origin=https://github.com/HOLMS-lib/HOLMS;visit=swh:1:snp:2c0efd349323ed6f8067581cf1f6d95816e49841;anchor=swh:1:rev:1caf3be141c6f646f78695c0eb528ce3b753079a;path=/k_decid.ml}{\ExternalLink};
            \item \verb|t_decid.ml|\href{https://archive.softwareheritage.org/swh:1:cnt:62ed022ba5b7ebeb6fd3f5b4cea979dceaaaa630;origin=https://github.com/HOLMS-lib/HOLMS;visit=swh:1:snp:2c0efd349323ed6f8067581cf1f6d95816e49841;anchor=swh:1:rev:1caf3be141c6f646f78695c0eb528ce3b753079a;path=/t_decid.ml}{\ExternalLink};
            \item \verb|k4_decid.ml|~\href{https://archive.softwareheritage.org/swh:1:cnt:97df598aa4202fa85e52877bda28aef5ae4952d0;origin=https://github.com/HOLMS-lib/HOLMS;visit=swh:1:snp:2c0efd349323ed6f8067581cf1f6d95816e49841;anchor=swh:1:rev:1caf3be141c6f646f78695c0eb528ce3b753079a;path=/k4_decid.ml}{\ExternalLink};
            \item \verb|gl_decid.ml|~\href{https://archive.softwareheritage.org/swh:1:cnt:83ff158d2d2a7208c9d2b7d8a686a5289dc2d024;origin=https://github.com/HOLMS-lib/HOLMS;visit=swh:1:snp:2c0efd349323ed6f8067581cf1f6d95816e49841;anchor=swh:1:rev:1caf3be141c6f646f78695c0eb528ce3b753079a;path=/gl_decid.ml}{\ExternalLink}.
        \end{enumerate}
        \end{multicols}
        \begin{itemize}
            \item Implement an \textbf{automated theorem prover} (\verb|HOLMS_RULE|). It proves, for each logic \verb|*| $\in \{ \mathbb{K; \ T; \ K}4; \ \mathbb{GL} \}$, if a given modal formula is a theorem of \verb|*| ;
            
             \item Implement an \textbf{automated countermodel constructor} (\verb|HOLMS_| \verb|BUILD_COUNTERMODEL|). For each logic \verb|*| $\in \{ \mathbb{K; \ T; \ K}4; \ \mathbb{GL} \}$ , if a given formula is not a theorem of \verb|*|, it build up a countermodel;
        \end{itemize}
        \medskip
        
        \item[(c.)] Test file \verb|tests.ml|~\href{https://archive.softwareheritage.org/swh:1:cnt:d7a8f32931f96fb31ac253e482f0302778163a20;origin=https://github.com/HOLMS-lib/HOLMS;visit=swh:1:snp:2c0efd349323ed6f8067581cf1f6d95816e49841;anchor=swh:1:rev:1caf3be141c6f646f78695c0eb528ce3b753079a;path=/tests.ml}{\ExternalLink}: Provides some examples of the use of the automated theorem prover. 

    \end{itemize}

\end{itemize}
\end{mydefinition}

\chapter{Growing HOLMS}\label{chap:4}

In Chapter~\ref{sub:HOLMS-project}, we have introduced the main features of our \textbf{HOL}-Light Library for \textbf{M}odal \textbf{S}ystem, together with the brief history of its development and a description of the current git repository~\href{https://github.com/HOLMS-lib/HOLMS}{\ExternalLink}.

In this final chapter, we present a comprehensive account of the HOLMS modular implementation of the normal modal logics $\mathbb{K}$, $\mathbb{T}$, $\mathbb{K}4$, and $\mathbb{GL}$—within the HOL Light framework. It serves as the \emph{formal counterpart}, within the theorem prover, to the exposition provided in Chapter~\ref{chap:1}.

We begin by formalising both the \textbf{syntax} and \textbf{semantics} of these modal systems, by adopting Harrison’s deep embedding method~\cite[\S 20]{harrisontutorial}. 

Next, we provide a formalisation of the ternary \textbf{deducibility relation} for normal modal logics ($\mathcal{S.H} \vdash A$), and we establish a correspondent \textbf{deduction lemma}.

We then turn to the core elements of \textbf{correspondence theory}. In this section, we introduce and prove key correspondence lemmata that bridge modal axioms with their characteristic property of semantic frame, thereby identifying the appropriate frames for each system under analysis.

The most extensive section of the chapter is devoted to the proofs of the \textbf{adequacy theorems}. Here, we develop the proof strategy outlined in Section~\ref{sub:completenes-proof-sketch} using the HOL Light proof assistant. These theorems validate that the axiomatisation accurately captures the intended relational semantics of the modal systems.

Finally, we describe two decision procedures for the modal logics. First, we introduce a \textbf{naive decision procedure} derived directly from the completeness theorem. Although straightforward, this procedure is highly inefficient. We then present a more sophisticated approach based on labelled sequent calculi, which yields a \textbf{principled decision procedure}\footnote{Here \textit{principled} refers to a decision procedure based on a well-defined and structured approach, in this case \textit{proof search in labelled sequent calculi}, outlined in Section~\ref{sec:labelled-sequent-calculi}. For $\mathbb{K}4$ a principled semi-decision procedure is currently available in the repository, but it can be improved to a decision procedure, as the method to achieve this is already known.} that is both correct and more computationally effective.
This chapter thus encapsulates our rigorous exploration of modal logic through both \emph{deep} and \emph{shallow} embedding techniques, offering detailed insights into decision procedures, correspondence theory, and the adequacy of our formal framework within HOL Light.

\section{Deep Embedding of Syntax and Semantics}

HOLMS has been developed as a library for formal reasoning \emph{with and about} modal logics.
Adopting a metalinguistic perspective, we observe that we employ HOL Light as the formal metatheory for studying modal logics as object theories.
 
 Hence, we need to ``teach'' HOL Light how to identify its object theories, beginning with its object language (modal syntax) and the interpretation of this language (Kripke's semantics).  Thus, the basics of the library consist of the definitions of the modal language and its interpretation in relational structures.

In HOLMS, we follow the approach provided by Harrison in HOL Light tutorial~\cite[\S 20]{harrisontutorial}: our file \verb|modal.ml|~\href{https://archive.softwareheritage.org/swh:1:cnt:4a011f69b1180019be5d03c6d5f1cec4550055e8;origin=https://github.com/HOLMS-lib/HOLMS;visit=swh:1:snp:2c0efd349323ed6f8067581cf1f6d95816e49841;anchor=swh:1:rev:1caf3be141c6f646f78695c0eb528ce3b753079a;path=/modal.ml}{\ExternalLink} formalises the syntax of propositional modal logic together with the basics of relational semantics. 
Adopting this method provides a significant foundational advantage. Indeed, the \emph{real} syntactic objects are neatly distinguished from their semantic counterpart, and from HOL Light's logical operators as well.

\subsection{Modal Syntax}

We should note that the \textit{abstract} counterpart of the modal syntax embedded in HOLMS, is an (apparent) extension of the modal syntax presented in~\ref{sec:informal-syntax}:

\begin{center}
    $\mathcal{L'}_{\Box} = \{ \bot; \ \top; \ \neg; \ \land; \ \lor; \ \to; \ \leftrightarrow; \ \Box; \ \Diamond \}$ and $ \Phi = \{a_0, a_1, ...\} $
    $A\in\mathbf{Form'}_{\Box}\,\Coloneqq\,\bot \ | \ \top \ | \ a \ | \ \neg A \ | \ A \land B \ | \ A \lor B \ | \ A \to B \ | \ A \leftrightarrow B \ | \ \Box A.$
\end{center}

To \emph{deeply embed} this syntax in HOL Light, means to represent this grammar using HOL Light's own type and term constructors.

When introducing new object logics into HOL Light, special care must be taken to\textit{ avoid overloading} existing propositional connectives. To achieve this, we introduce a concrete syntax with infix and prefix symbols of type \verb|:string| for representing modal formulas within HOL Light. We also set their parsing status (infix or prefix) and assign the appropriate parsing precedences~\href{https://archive.softwareheritage.org/swh:1:cnt:4a011f69b1180019be5d03c6d5f1cec4550055e8;origin=https://github.com/HOLMS-lib/HOLMS;visit=swh:1:snp:2c0efd349323ed6f8067581cf1f6d95816e49841;anchor=swh:1:rev:1caf3be141c6f646f78695c0eb528ce3b753079a;path=/modal.ml;lines=18-23}{\ExternalLink}.

\begin{lstlisting}[caption={Parsing status and precedences of the object logics connectives.}]
# parse_as_infix("&&",(16,"right"));;
# parse_as_infix("||",(15,"right"));;
# parse_as_infix("-->",(14,"right"));;
# parse_as_infix("<->",(13,"right"));;
# parse_as_prefix "Not";;
# parse_as_prefix "Box";;
\end{lstlisting}

\begin{remark}
   Specifically to clarify the distinction between the logic operators of HOL Light from those of our modal language, we have introduced the glossary reported in Figure~\ref{fig:glossary}.
\end{remark}

\medskip

At this stage, we can extend HOL Light system by the inductive definition of the predicate of being a modal formula as an HOL Light inductive type (\verb|:form|~\href{https://archive.softwareheritage.org/swh:1:cnt:4a011f69b1180019be5d03c6d5f1cec4550055e8;origin=https://github.com/HOLMS-lib/HOLMS;visit=swh:1:snp:2c0efd349323ed6f8067581cf1f6d95816e49841;anchor=swh:1:rev:1caf3be141c6f646f78695c0eb528ce3b753079a;path=/modal.ml;lines=25-34}{\ExternalLink}).
\newpage
\begin{lstlisting}[caption={Inductive type of modal formulas}]
#let form_INDUCT,form_RECURSION = define_type
  "form = False
        | True
        | Atom string
        | Not form
        | && form form
        | || form form
        | --> form form
        | <-> form form
        | Box form";;
\end{lstlisting}    

Moreover, we formalise the two predicates of \textit{subformula} (\verb|SUBFORMULA|~\href{https://archive.softwareheritage.org/swh:1:cnt:4a011f69b1180019be5d03c6d5f1cec4550055e8;origin=https://github.com/HOLMS-lib/HOLMS;visit=swh:1:snp:2c0efd349323ed6f8067581cf1f6d95816e49841;anchor=swh:1:rev:1caf3be141c6f646f78695c0eb528ce3b753079a;path=/modal.ml;lines=91-124}{\ExternalLink}), and  \textit{subsentence} (\verb|SUBSENTENCE|~\href{https://archive.softwareheritage.org/swh:1:cnt:e261ca08330d5ef66376347407c27fcece1e9f2e;origin=https://github.com/HOLMS-lib/HOLMS;visit=swh:1:snp:2c0efd349323ed6f8067581cf1f6d95816e49841;anchor=swh:1:rev:1caf3be141c6f646f78695c0eb528ce3b753079a;path=/consistent.ml;lines=141-158}{\ExternalLink}). We also prove with the aid of the proof assistant Fact~\ref{fct:finite_formulas}, which states that there is a finite number of subformulas of a given formula (\verb|FINITE_SUBFORMULA|~\href{https://archive.softwareheritage.org/swh:1:cnt:4a011f69b1180019be5d03c6d5f1cec4550055e8;origin=https://github.com/HOLMS-lib/HOLMS;visit=swh:1:snp:2c0efd349323ed6f8067581cf1f6d95816e49841;anchor=swh:1:rev:1caf3be141c6f646f78695c0eb528ce3b753079a;path=/modal.ml;lines=138-148}{\ExternalLink}) and that there is a finite number of sets made by subsentences of a given modal formula (\verb|FINITE_SUBSET_SUBFORMULAS_LEMMA |~\href{https://archive.softwareheritage.org/swh:1:cnt:4a011f69b1180019be5d03c6d5f1cec4550055e8;origin=https://github.com/HOLMS-lib/HOLMS;visit=swh:1:snp:2c0efd349323ed6f8067581cf1f6d95816e49841;anchor=swh:1:rev:1caf3be141c6f646f78695c0eb528ce3b753079a;path=/modal.ml;lines=150-157}{\ExternalLink}).

\subsection{Modal Semantics}

\textit{Abstractly}, in Section~\ref{sec:informal-semnatics} we introduced \textit{relational models} and defined the \textit{forcing} relation which evaluates the truth of a modal formula at a given world within a model, as well as the semantic predicate for logical consequence.

As mentioned earlier, HOLMS adopts Harrison’s formalisation of relational semantics through a deep embedding approach. 

This method slightly differs from the classical presentation in Chapter~\ref{chap:1}. 
Specifically, in the file \verb|modal.ml| we first define the forcing relation \verb|holds WR V A|~\href{https://archive.softwareheritage.org/swh:1:cnt:4a011f69b1180019be5d03c6d5f1cec4550055e8;origin=https://github.com/HOLMS-lib/HOLMS;visit=swh:1:snp:2c0efd349323ed6f8067581cf1f6d95816e49841;anchor=swh:1:rev:1caf3be141c6f646f78695c0eb528ce3b753079a;path=/modal.ml;lines=44-59}{\ExternalLink}  without imposing on \verb|WR V| the constraints required for a proper relational model. Subsequently, in \verb|gen_completeness|~\href{https://archive.softwareheritage.org/swh:1:cnt:d364de3158b6405c920115aa0c801af9af49e302;origin=https://github.com/HOLMS-lib/HOLMS;visit=swh:1:snp:2c0efd349323ed6f8067581cf1f6d95816e49841;anchor=swh:1:rev:1caf3be141c6f646f78695c0eb528ce3b753079a;path=/gen_completeness.ml;lines=15-17}{\ExternalLink}--and in the specific files \verb|*_completeness|--we systematically integrate these requirements. The notion of logical consequence, instead, has not been formally defined, as it is not pursued within the library.

\begin{lstlisting}[caption={Implementation in HOLMS of the 
forcing relation \texttt{holds}}]
#let holds =   new_recursive_definition form_RECURSION
  `(holds WR V False (w:W) <=> F) /\
   (holds WR V True w <=> T) /\
   (holds WR V (Atom s) w <=> V s w) /\
   (holds WR V (Not p) w <=> ~(holds WR V p w)) /\
   (holds WR V (p && q) w <=> holds WR V p w /\ holds WR V q w) /\
   (holds WR V (p || q) w <=> holds WR V p w \/ holds WR V q w) /\
   (holds WR V (p --> q) w <=> holds WR V p w ==> holds WR V q w) /\
   (holds WR V (p <-> q) w <=> holds WR V p w <=> holds WR V q w) /\
   (holds WR V (Box p) w <=>
    !w'. w' IN FST WR /\ SND WR w w' ==> holds WR V p w')';;
\end{lstlisting}

We also implement in HOLMS the concepts of validity in a frame and in a class of frames, respectively by the predicates \verb|holds_in|~\href{https://archive.softwareheritage.org/swh:1:cnt:4a011f69b1180019be5d03c6d5f1cec4550055e8;origin=https://github.com/HOLMS-lib/HOLMS;visit=swh:1:snp:2c0efd349323ed6f8067581cf1f6d95816e49841;anchor=swh:1:rev:1caf3be141c6f646f78695c0eb528ce3b753079a;path=/modal.ml;lines=61-62}{\ExternalLink} and \verb|valid|~\href{https://archive.softwareheritage.org/swh:1:cnt:4a011f69b1180019be5d03c6d5f1cec4550055e8;origin=https://github.com/HOLMS-lib/HOLMS;visit=swh:1:snp:2c0efd349323ed6f8067581cf1f6d95816e49841;anchor=swh:1:rev:1caf3be141c6f646f78695c0eb528ce3b753079a;path=/modal.ml;lines=66-67}{\ExternalLink}.
 Similarly to \verb|holds|, in this formalisation we refer to a generic \verb|V:string-> W->bool| (without imposing the condition of being a valuation) and to a generic HOL Light pair \verb|f:((W->bool)#(R:W->W->bool)->bool)| (without imposing conditions of domains and accessibility relation), respectively.    

\begin{lstlisting}[caption={HOLMS implementation of validity in a frame and in a class}]
   #let holds_in = new_definition
  `holds_in (W,R) p <=> !V w:W. w IN W ==> holds (W,R) V p w';;

  #let valid = new_definition
  `L |= p <=> !f:(W->bool)#(W->W->bool). f IN L ==> holds_in f p';;
\end{lstlisting}

\begin{remark}
    In HOLMS listings, lowercase letters, such as \verb|p| and \verb|q|, stands for both propositional atoms and generic modal formulas.
\end{remark}

\section{A Parametric Calculus within HOLMS}

Since we have embedded the syntax of the standard modal language $\mathcal{L'}_{\Box}$, we can formally represent within HOL Light  $\mathcal{S.H} \vdash A$, the proof-theoretic ternary relation of deducibility in a normal modal system starting from a given set of hypotheses introduced in Section~\ref{sub:parametric-calculus}.

As previously remarked, the predicate $\mathcal{S.H} \vdash A$ expresses \textit{derivability in a calculus} $\mathsf{S}$ f\textit{or a normal modal system} $\mathbb{S}$ by conceptualising it as \textit{derivability within the minimal axiomatic system} $\mathsf{K}$, \textit{modularly extended} by instantiating its set of specific axiom schemata $\mathcal{S}$.

Once a preferred axiomatisation for classical propositional logic has been selected, a standard axiomatization of $\mathsf{K}$ is obtained by extending it with the distribution schema: \mbox{$\mathsf{K}:=\Box(A\rightarrow B)\rightarrow(\Box A\rightarrow\Box B)$}.
The set of axioms of $\mathsf{K}$ is then encoded as an inductive predicate in HOL Light, named \verb|KAXIOM|~\href{https://archive.softwareheritage.org/swh:1:cnt:d42782345008434be8b43de7feb014732f8821ea;origin=https://github.com/HOLMS-lib/HOLMS;visit=swh:1:snp:2c0efd349323ed6f8067581cf1f6d95816e49841;anchor=swh:1:rev:1caf3be141c6f646f78695c0eb528ce3b753079a;path=/calculus.ml;lines=57-68}{\ExternalLink}.

\begin{lstlisting}[caption={HOLMS inductive definition of the set of axioms of the minimal K.}]
#let KAXIOM_RULES,KAXIOM_INDUCT,KAXIOM_CASES = 
  new_inductive_definition
  `(!p q. KAXIOM (p --> (q --> p))) /\
   (!p q r. KAXIOM ((p --> q --> r) --> (p --> q) --> (p --> r))) /\
   (!p. KAXIOM (((p --> False) --> False) --> p)) /\
   (!p q. KAXIOM ((p <-> q) --> p --> q)) /\
   (!p q. KAXIOM ((p <-> q) --> q --> p)) /\
   (!p q. KAXIOM ((p --> q) --> (q --> p) --> (p <-> q))) /\
   KAXIOM (True <-> False --> False) /\
   (!p. KAXIOM (Not p <-> p --> False)) /\
   (!p q. KAXIOM (p && q <-> (p --> q --> False) --> False)) /\
   (!p q. KAXIOM (p || q <-> Not(Not p && Not q))) /\
   (!p q. KAXIOM (Box (p --> q) --> Box p --> Box q))';;
\end{lstlisting}

To capture the general notion of derivability, in calculus.ml~\href{https://archive.softwareheritage.org/swh:1:cnt:d42782345008434be8b43de7feb014732f8821ea;origin=https://github.com/HOLMS-lib/HOLMS;visit=swh:1:snp:2c0efd349323ed6f8067581cf1f6d95816e49841;anchor=swh:1:rev:1caf3be141c6f646f78695c0eb528ce3b753079a;path=/calculus.ml;lines=74-80} we finally define a parametric predicate \verb+[ S . H |~ A ]+. This approach avoids code duplication in the formalisation of normal modal calculi: instead of requiring separate definitions for each logic covered by HOLMS--each similar to the GL-provability predicate \verb+|-- A+ in~\cite[\S 3]{DBLP:journals/jar/MaggesiB23}, we use a single, \textit{parameterised} definition.
\newpage
\begin{lstlisting}[caption={HOLMS inductive definition of the parametric deducibility predicate.}]
#let MODPROVES_RULES,MODPROVES_INDUCT,MODPROVES_CASES =
  new_inductive_definition
  `(!H p. KAXIOM p ==> [S . H |~ p]) /\
   (!H p. p IN S ==> [S . H |~ p]) /\
   (!H p. p IN H ==> [S . H |~ p]) /\
   (!H p q. [S . H |~ p --> q] /\ [S . H |~ p] ==> [S . H |~ q]) /\
   (!H p. [S . {} |~ p] ==> [S . H |~ Box p])';;
\end{lstlisting}

Notice that this definition, similarly to the \textit{abstract} Definition~\ref{def:deducibility}, does not state that our calculus is closed under substitution rule \texttt{SUB}. 
In \verb|calculus.ml| we hence implement the definition of substitution~(\ref{def:substitution}) by induction on the complexity of a modal formula~\href{https://archive.softwareheritage.org/swh:1:cnt:d42782345008434be8b43de7feb014732f8821ea;origin=https://github.com/HOLMS-lib/HOLMS;visit=swh:1:snp:2c0efd349323ed6f8067581cf1f6d95816e49841;anchor=swh:1:rev:1caf3be141c6f646f78695c0eb528ce3b753079a;path=/calculus.ml;lines848-857}{\ExternalLink}, and, subsequently, we prove that the minimal system $\mathsf{K}$ is closed under substitution and how our calculus behaves under substitution.

\begin{lstlisting}[caption={HOLMS definition of substitution}]
#let SUBST = new_recursive_definition form_RECURSION
 `(!f. SUBST f True = True) /\
   (!f. SUBST f False = False) /\
   (!f a. SUBST f (Atom a) = f a) /\
   (!f p. SUBST f (Not p) = Not (SUBST f p)) /\
   (!f p q. SUBST f (p && q) = SUBST f p && SUBST f q) /\
   (!f p q. SUBST f (p || q) = SUBST f p || SUBST f q) /\
   (!f p q. SUBST f (p --> q) = SUBST f p --> SUBST f q) /\
   (!f p q. SUBST f (p <-> q) = SUBST f p <-> SUBST f q) /\
   (!f p. SUBST f (Box p) = Box (SUBST f p))';;
\end{lstlisting}

\begin{lstlisting}[caption={HOLMS substittion lemmas}]
SUBST_IMP 
|- !S f H p. (!q. q IN S ==> SUBST f q IN S) /\ [S . H |~ p]
             ==> [S . IMAGE (SUBST f) H |~ SUBST f p]
 
SUBSTITUTION_LEMMA 
|- !S f H p q.
     (!q. q IN S ==> SUBST f q IN S) /\ [S . H |~ p <-> q]
     ==> [S . IMAGE (SUBST f) H |~ SUBST f p <-> SUBST f q]
\end{lstlisting}

\begin{remark}
    HOLMS source code now deals with two distinct symbols of derivability, one for the metatheory, one for the object theories:
    \begin{itemize}
    \item \verb+|-+ :   derivability in the metatheory HOL Light  (denotes a HOL Light theorem);
    \item \verb+[S . H |~ ]+ : derivability in the object logic $\mathbb{S}$ (denotes a HOL Light term representing a theorem of modal logic axiomatised by the schemas in~\verb|S|)
\end{itemize}
\end{remark}

\subsection{HOLMS Deduction Theorem}

We formally prove in HOL Light~\href{https://archive.softwareheritage.org/swh:1:cnt:d42782345008434be8b43de7feb014732f8821ea;origin=https://github.com/HOLMS-lib/HOLMS;visit=swh:1:snp:2c0efd349323ed6f8067581cf1f6d95816e49841;anchor=swh:1:rev:1caf3be141c6f646f78695c0eb528ce3b753079a;path=/calculus.ml;lines=1015-1028}{\ExternalLink} the deduction theorem demonstrated in Lemma~\ref{lem:deduction-theorem}. 
This would reduce the common notion of derivability of $A$ in an axiomatic calculus characterised by the schemata in $\mathcal{S}$ ($\mathcal{S} \vdash A$) to the relation $\mathcal{S}.\varnothing\vdash A$. 

\begin{lstlisting}[caption=H{OLMS proof of deduction theorem}]
#let MODPROVES_DEDUCTION_LEMMA = prove
 (`!S H p q. [S . H |~ p --> q] <=> [S . p INSERT H |~ q]',
 
  REPEAT GEN_TAC THEN
  SUBGOAL_THEN `[S . p INSERT H |~ q] ==> [S . H |~ p --> q]'
    (fun th -> MESON_TAC[th; MODPROVES_DEDUCTION_LEMMA_INSERT]) THEN
  ASM_CASES_TAC `p:form IN H' THENL
  [SUBGOAL_THEN `p:form INSERT H = H' SUBST1_TAC THENL
   [ASM SET_TAC []; ALL_TAC] THEN
   MESON_TAC[MODPROVES_RULES; MLK_add_assum]; ALL_TAC] THEN
  INTRO_TAC "hp" THEN
  SUBGOAL_THEN `H = (p:form INSERT H) DELETE p' SUBST1_TAC THENL
  [ASM SET_TAC []; ALL_TAC] THEN
  MATCH_MP_TAC MODPROVES_DEDUCTION_LEMMA_DELETE THEN
  ASM_REWRITE_TAC[IN_INSERT]);;
\end{lstlisting}

Notice that the proof developed for Lemma~\ref{lem:deduction-theorem}, follows the one in HOLMS; such an \textit{informal} proof can be useful to understand the exact implementation of this first relevant theorem implemented in HOLMS.

\subsection{Normal Systems Implemented in HOLMS}

In the current version of HOLMS four normal systems have been implemented. Consistently with our presentation of the modal cube in Section~\ref{sec:normal-systems}, below we list the set of specific axioms for these normal systems. 
\vspace{-1mm}
\begin{enumerate}
    \item $\mathcal{K}= \varnothing =$ \verb|{}| describes the minimal normal logic K;
    \item $\mathcal{T}= \{\mathbf{T} \}=$ \verb|T_AX|\href{https://archive.softwareheritage.org/swh:1:cnt:1352cb294724b2251b9346657432700ecbed10d2;origin=https://github.com/HOLMS-lib/HOLMS;visit=swh:1:snp:2c0efd349323ed6f8067581cf1f6d95816e49841;anchor=swh:1:rev:1caf3be141c6f646f78695c0eb528ce3b753079a;path=/t_completeness.ml;lines=8-9}{~\ExternalLink}, where $\mathbf{T}: \Box A \to A$, describes the normal logic $\mathbb{T}$;
    \item $\mathcal{K}4= \{ \mathbf{4} \}=$ \verb|K4_AX|\href{https://archive.softwareheritage.org/swh:1:cnt:32869a1df338d91c4cd4a0cb9e73fb4f1be29991;origin=https://github.com/HOLMS-lib/HOLMS;visit=swh:1:snp:2c0efd349323ed6f8067581cf1f6d95816e49841;anchor=swh:1:rev:1caf3be141c6f646f78695c0eb528ce3b753079a;path=/k4_completeness.ml;lines=8-9}{~\ExternalLink}, where $\mathbf{4}: \Box A \to \Box \Box A$, describes the normal logic $\mathbb{K}4$;
    \item $\mathcal{S}= \{ \mathbf{GL} \}=$ \verb|GL_AX|\href{https://archive.softwareheritage.org/swh:1:cnt:52c9de12454731a9a291b06bc750c0d2d14e1fb4;origin=https://github.com/HOLMS-lib/HOLMS;visit=swh:1:snp:2c0efd349323ed6f8067581cf1f6d95816e49841;anchor=swh:1:rev:1caf3be141c6f646f78695c0eb528ce3b753079a;path=/gl_completeness.ml;lines=11-12}{~\ExternalLink}, where $\mathbf{GL}: \Box (\Box A \to A) \to \Box A$, describes the normal logic $\mathbb{GL}$.
\end{enumerate}
\vspace{-5mm}
\begin{lstlisting}[caption={HOLMS definition of the set of axioms for \texttt{T, K4} and \texttt{GL}}]
 #let T_AX = new_definition
   `T_AX = {Box p --> p | p IN (:form)}';;

 #let K4_AX = new_definition
   `K4_AX = {Box p --> Box Box p | p IN (:form)}';;

 #let GL_AX = new_definition
   `GL_AX = {Box (Box p --> p) --> Box p | p IN (:form)}';;
\end{lstlisting}

In \verb|gl_completeness|~\href{https://archive.softwareheritage.org/swh:1:cnt:52c9de12454731a9a291b06bc750c0d2d14e1fb4;origin=https://github.com/HOLMS-lib/HOLMS;visit=swh:1:snp:2c0efd349323ed6f8067581cf1f6d95816e49841;anchor=swh:1:rev:1caf3be141c6f646f78695c0eb528ce3b753079a;path=/gl_completeness.ml;lines=228-232}{\ExternalLink} we prove that $\mathbb{K}4 \subseteq \mathbb{GL}$, by proving that $\mathcal{GL} \vdash \mathbf{4}$. Such a proof uses the specific axiom of $\mathbb{GL}$  (\verb|GL_axiom_lob|) and a series of lemmas developed in \verb|calculus.ml| which hold for each normal logic (\verb|MLK_|).

\begin{lstlisting}[caption={HOLMS theorem stating that GL extends K4}]
 #let GL_schema_4 = prove
 (`!p. [GL_AX . {} |~ (Box p --> Box (Box p))]',
  MESON_TAC[GL_axiom_lob; MLK_imp_box; MLK_and_pair_th;
  MLK_and_intro; MLK_shunt; MLK_imp_trans; MLK_and_right_th;
  MLK_and_left_th; MLK_box_and_th]);
\end{lstlisting}

\section{Formalising Correspondence Theory}
The fundamental notions of correspondence theory presented in~\ref{sec:correspondence-theory}, have been formalised within HOL Light in \verb|gen_completeness.ml|~\href{https://archive.softwareheritage.org/swh:1:cnt:d364de3158b6405c920115aa0c801af9af49e302;origin=https://github.com/HOLMS-lib/HOLMS;visit=swh:1:snp:2c0efd349323ed6f8067581cf1f6d95816e49841;anchor=swh:1:rev:1caf3be141c6f646f78695c0eb528ce3b753079a;path=/gen_completeness.ml}{\ExternalLink}, while the identification of characteristics and appropriate frames for each logic \verb|*| $\in \{ \mathbb{K; \ T; \ K}4; \ \mathbb{GL} \}$ has been developed in its specific file \verb|*_completeness.ml|.
  
\subsection{Characteristic and Appropriate frames}

Given the importance of characteristic frames in our modular proof of adequacy theorems, in HOLMS we decided not to focus on correspondence lemmas; but to immediately rephrase those specific results in a general setting and formally define \verb|CHAR S|~\href{https://archive.softwareheritage.org/swh:1:cnt:d364de3158b6405c920115aa0c801af9af49e302;origin=https://github.com/HOLMS-lib/HOLMS;visit=swh:1:snp:2c0efd349323ed6f8067581cf1f6d95816e49841;anchor=swh:1:rev:1caf3be141c6f646f78695c0eb528ce3b753079a;path=/gen_completeness.ml;lines=44-47}{\ExternalLink}, representing the class of frames \textit{characteristic} to \verb|S| (Definition~\ref{def:characteristic}).
Before introducing this concept, we must encode in HOLMS the set of relational frames (\verb|FRAME|~\href{https://archive.softwareheritage.org/swh:1:cnt:d364de3158b6405c920115aa0c801af9af49e302;origin=https://github.com/HOLMS-lib/HOLMS;visit=swh:1:snp:2c0efd349323ed6f8067581cf1f6d95816e49841;anchor=swh:1:rev:1caf3be141c6f646f78695c0eb528ce3b753079a;path=/gen_completeness.ml;lines=15-16}{\ExternalLink}) to settle our debt for following Harrison's embedding \cite[\S 20]{harrisontutorial} and imposing the standard conditions on nonempty domains and binary accessibility relations.

\begin{lstlisting}[caption={HOLMS definitions of relational frames and sets of characteristic frames}]
#let FRAME_DEF = new_definition
  `FRAME = {(W:W->bool,R:W->W->bool) |
            ~(W = {}) /\ (!x y:W. R x y ==> x IN W /\ y IN W)}';;
             
#let CHAR_DEF = new_definition
 `CHAR S = {(W:W->bool,R:W->W->bool) |
             (W,R) IN FRAME /\
             (!p. p IN S ==> holds_in (W:W->bool,R:W->W->bool) p)}';;
\end{lstlisting}

As observed in Lemma~\ref{lem:CHAR_CAR}, it is possible to characterise the class characteristic to $\mathbb{S}$ as the class of frames in which every theorem of $\mathbb{S}$ holds. The theorem of HOL Light proving this fact goes under the name \verb|CHAR_CAR|~\href{https://archive.softwareheritage.org/swh:1:cnt:d364de3158b6405c920115aa0c801af9af49e302;origin=https://github.com/HOLMS-lib/HOLMS;visit=swh:1:snp:2c0efd349323ed6f8067581cf1f6d95816e49841;anchor=swh:1:rev:1caf3be141c6f646f78695c0eb528ce3b753079a;path=/gen_completeness.ml;lines=82-105}{\ExternalLink}.

Since we have formalised in HOLMS neither Theorem~\ref{thm:K_soundness}--which states the soundness of $\mathbb{K}$--nor the Correspondence Lemma~\ref{lem:Correspondence}, the HOLMS proof of this theorem slightly differs from the one presented in Section~\ref{sec:correspondence-theory}. However, in principle, it can be traced back to it.

\begin{lstlisting}[caption={HOLMS characterisation of the notion of characteristic class of frames}, label={lst:CHAR_CAR}]
CHAR_CAR
|- !S:form->bool W:(W->bool) R:(W->W->bool). 
       ((W,R) IN FRAME /\ !p. [S . {} |~ p] ==> holds_in (W,R) p) <=>
       (W,R) IN CHAR S
\end{lstlisting}

Following~\cite[\S~11]{boolos1995logic} and for practical coding reasons, we also introduce the notion of the class of frames appropriate to a modal system $\mathbb{S}$: \verb|APPR S|~\href{https://archive.softwareheritage.org/swh:1:cnt:d364de3158b6405c920115aa0c801af9af49e302;origin=https://github.com/HOLMS-lib/HOLMS;visit=swh:1:snp:2c0efd349323ed6f8067581cf1f6d95816e49841;anchor=swh:1:rev:1caf3be141c6f646f78695c0eb528ce3b753079a;path=/gen_completeness.ml;lines=111-114}{\ExternalLink}, which formalises Definition~\ref{def:appropriate_class}. This class will immediately be characterised (\verb|APPR_EQ_CHAR_FINITE|~\href{https://archive.softwareheritage.org/swh:1:cnt:d364de3158b6405c920115aa0c801af9af49e302;origin=https://github.com/HOLMS-lib/HOLMS;visit=swh:1:snp:2c0efd349323ed6f8067581cf1f6d95816e49841;anchor=swh:1:rev:1caf3be141c6f646f78695c0eb528ce3b753079a;path=/gen_completeness.ml;lines=130-135}{\ExternalLink} that formalises Lemma~\ref{lem:APPR_CAR}) as the characteristic class to $\mathbb{S}$  \textit{restricted to finite frames}.

\begin{lstlisting}
  #let APPR_DEF = new_definition
  `APPR S = {(W:W->bool,R:W->W->bool) |
             (W,R) IN FINITE_FRAME  /\
             !p. [S. {} |~ p] ==> 
               holds_in (W:W->bool,R:W->W->bool) p}';;
\end{lstlisting}

\subsection{Identification of Characteristic and Appropriate frames}

In the following listings, we present HOLMS definitions of the classes of finite, (finite-)reflexive, (finite-)transitive, transitive-converse well founded and finite-transitive-irreflexive frames. Listing~\ref{lst:gl-classes} provide the lemma \verb|FINITE_EQ_ITF_TRANSNT|~\href{https://archive.softwareheritage.org/swh:1:cnt:52c9de12454731a9a291b06bc750c0d2d14e1fb4;origin=https://github.com/HOLMS-lib/HOLMS;visit=swh:1:snp:2c0efd349323ed6f8067581cf1f6d95816e49841;anchor=swh:1:rev:1caf3be141c6f646f78695c0eb528ce3b753079a;path=/gl_completeness.ml;lines=169-178}{\ExternalLink}, stating the equivalence of the sets \verb|ITF| ($\mathfrak{ITF}$) and \verb|FINITE_FRAME INTER TRANSNT| ($\mathfrak{NTF}$) and formalising Lemma~\ref{lem:ITF-transnt} .

\begin{lstlisting}[caption={HOLMS definition of classes of frames  appropriate to \texttt{K}}, label={lst:classes}] 
 #let FINITE_FRAME_DEF = new_definition
   `FINITE_FRAME = {(W:W->bool,R:W->W->bool) | 
                     (W,R) IN FRAME /\ FINITE W}';;

\end{lstlisting}
\vspace{-7mm}
\begin{lstlisting}[caption={HOLMS definition of characteristic, appropriate classes to  \texttt{T}}, label={lst:t-classes}] 
 #let REFL_DEF = new_definition
  `REFL =
   {(W:W->bool,R:W->W->bool) |
    ~(W = {}) /\ (!x y:W. R x y ==> x IN W /\ y IN W) /\
    (!x:W. x IN W ==> R x x)}';;
    
 #let RF_DEF = new_definition
   `RF = {(W:W->bool,R:W->W->bool) |
          ~(W = {}) /\ (!x y:W. R x y ==> x IN W /\ y IN W) /\
          FINITE W /\
          (!x. x IN W ==> R x x)}';;

\end{lstlisting}
\vspace{-7mm}
\begin{lstlisting}[caption={HOLMS definition of characteristic, appropriate classes to \texttt{K4}}, label={lst:k4-classes}] 
 #let TRANS_DEF = new_definition
  `TRANS =
   {(W:W->bool,R:W->W->bool) |
    ~(W = {}) /\ (!x y:W. R x y ==> x IN W /\ y IN W) /\
    (!x y z:W. x IN W /\ y IN W /\ z IN W /\ R x y /\ R y z
               ==> R x z)}';;

    
 #let TF_DEF = new_definition
  `TF = {(W:W->bool,R:W->W->bool) |
          ~(W = {}) /\ (!x y:W. R x y ==> x IN W /\ y IN W) /\
          FINITE W /\
          (!x y z. x IN W /\ y IN W /\ z IN W /\ R x y /\ R y z
                   ==> R x z)}';;

\end{lstlisting}
\vspace{-7mm}
\begin{lstlisting}[caption={HOLMS definition of classes  of characteristic, appropriate classes to  \texttt{GL}}, label={lst:gl-classes}] 

 #let TRANSNT_DEF = new_definition
 `TRANSNT =
   {(W:W->bool,R:W->W->bool) |
    ~(W = {}) /\ (!x y:W. R x y ==> x IN W /\ y IN W) /\
    (!x y z:W. x IN W /\ y IN W /\ z IN W /\ R x y /\ R y z 
            ==> R x z) /\
    WF(\x y. R y x)}';;
    
 #let ITF_DEF = new_definition
  `ITF = {(W:W->bool,R:W->W->bool) |
           ~(W = {}) /\ (!x y:W. R x y ==> x IN W /\ y IN W) /\
           FINITE W /\
           (!x. x IN W ==> ~R x x) /\
           (!x y z. x IN W /\ y IN W /\ z IN W /\ R x y /\ R y z
                    ==> R x z)}';;

 FINITE_EQ_ITF_TRANSNT 
 |- !W:W->bool R:W->W->bool.
     FINITE (W) ==> ((W,R) IN ITF <=> (W,R) IN TRANSNT)
\end{lstlisting}

The following Listing~\ref{lst:char-cl} provides the identification of the class of characteristic frames for $\mathbb{K}$~\href{https://archive.softwareheritage.org/swh:1:cnt:ae230138ff15476c8dab9e32606bceca7168285b;origin=https://github.com/HOLMS-lib/HOLMS;visit=swh:1:snp:2c0efd349323ed6f8067581cf1f6d95816e49841;anchor=swh:1:rev:1caf3be141c6f646f78695c0eb528ce3b753079a;path=/k_completeness.ml;lines=15-19}{\ExternalLink}, $\mathbb{T}$~\href{https://archive.softwareheritage.org/swh:1:cnt:1352cb294724b2251b9346657432700ecbed10d2;origin=https://github.com/HOLMS-lib/HOLMS;visit=swh:1:snp:2c0efd349323ed6f8067581cf1f6d95816e49841;anchor=swh:1:rev:1caf3be141c6f646f78695c0eb528ce3b753079a;path=/t_completeness.ml;lines=47-57}{\ExternalLink}, $\mathbb{K}4$~\href{https://archive.softwareheritage.org/swh:1:cnt:32869a1df338d91c4cd4a0cb9e73fb4f1be29991;origin=https://github.com/HOLMS-lib/HOLMS;visit=swh:1:snp:2c0efd349323ed6f8067581cf1f6d95816e49841;anchor=swh:1:rev:1caf3be141c6f646f78695c0eb528ce3b753079a;path=/k4_completeness.ml;lines=53-63}{\ExternalLink} and $\mathbb{GL}$~\href{https://archive.softwareheritage.org/swh:1:cnt:52c9de12454731a9a291b06bc750c0d2d14e1fb4;origin=https://github.com/HOLMS-lib/HOLMS;visit=swh:1:snp:2c0efd349323ed6f8067581cf1f6d95816e49841;anchor=swh:1:rev:1caf3be141c6f646f78695c0eb528ce3b753079a;path=/gl_completeness.ml;lines=47-72}{\ExternalLink}. 

Thanks to the proof assistant, proving these lemmata is easier and shorter than in Lemma~\ref{lem:Correspondence}. Suffice it to say that to prove i.e. HOLMS implementation of the~\ref{itm:corr_2}$^{nd}$ point of correspondence lemma:
\vspace{-3mm}
\begin{center}
    \verb|!W R. (!w:W. w IN W ==> R w w) <=> (!p. holds_in (W,R) (Box p --> p))|
\end{center}
\vspace{-3mm}
\noindent it is enough to execute the composed tactic \verb| MODAL_| \verb|SCHEMA_TAC THEN MESON_TAC[]|.

In HOLMS these theorems are stated as equalities between sets, however--to make the statements more readable and similar to Correspondence Lemma~\ref{lem:Correspondence}-- here we present it after applying the HOL Light tactict \texttt{REWRITE\textunderscore TAC}[\texttt{EXTENSION; FORALL\textunderscore PAIR\textunderscore THM}].

\begin{lstlisting}[caption={HOLMS identification of the characteristic classes},label={lst:char-cl}]
 FRAME_CHAR_K
 |- ! W R. (W,R) IN FRAME:(W->bool)#(W->W->bool)->bool 
           <=> (W,R) IN CHAR {}  

 REFL_CHAR_T
 |- ! W R. (W,R) IN REFL:(W->bool)#(W->W->bool)->bool 
           <=> (W,R) IN CHAR T_AX
          
 TRANS_CHAR_K4
 |- ! W R. (W,R) IN TRANS:(W->bool)#(W->W->bool)->bool 
           <=> (W,R) IN CHAR K4_AX

 TRANSNT_CHAR_GL          
 |- ! W R. (W,R) IN TRANSNT:(W->bool)#(W->W->bool)->bool 
           <=> (W,R) IN CHAR GL_AX

\end{lstlisting}

Exploiting the characterisation of appropriate frames as finite characteristic frames, HOLMS provides the following identification of the class of appropriate frames for $\mathbb{K}$~\href{https://archive.softwareheritage.org/swh:1:cnt:ae230138ff15476c8dab9e32606bceca7168285b;origin=https://github.com/HOLMS-lib/HOLMS;visit=swh:1:snp:2c0efd349323ed6f8067581cf1f6d95816e49841;anchor=swh:1:rev:1caf3be141c6f646f78695c0eb528ce3b753079a;path=/k_completeness.ml;lines=27-38}{\ExternalLink}, $\mathbb{T}$~\href{https://archive.softwareheritage.org/swh:1:cnt:1352cb294724b2251b9346657432700ecbed10d2;origin=https://github.com/HOLMS-lib/HOLMS;visit=swh:1:snp:2c0efd349323ed6f8067581cf1f6d95816e49841;anchor=swh:1:rev:1caf3be141c6f646f78695c0eb528ce3b753079a;path=/t_completeness.ml;lines=108-125}{\ExternalLink}, $\mathbb{K}4$~\href{https://archive.softwareheritage.org/swh:1:cnt:32869a1df338d91c4cd4a0cb9e73fb4f1be29991;origin=https://github.com/HOLMS-lib/HOLMS;visit=swh:1:snp:2c0efd349323ed6f8067581cf1f6d95816e49841;anchor=swh:1:rev:1caf3be141c6f646f78695c0eb528ce3b753079a;path=/k4_completeness.ml;lines=53-63}{\ExternalLink} and $\mathbb{GL}$~\href{https://archive.softwareheritage.org/swh:1:cnt:52c9de12454731a9a291b06bc750c0d2d14e1fb4;origin=https://github.com/HOLMS-lib/HOLMS;visit=swh:1:snp:2c0efd349323ed6f8067581cf1f6d95816e49841;anchor=swh:1:rev:1caf3be141c6f646f78695c0eb528ce3b753079a;path=/gl_completeness.ml;lines=184-202}{\ExternalLink}.

\begin{lstlisting}[caption={HOLMS identification of the appropriate classes}]
 FINITE_FRAME_APPR_K
 |- ! W R. (W,R) IN FINITE_FRAME:(W->bool)#(W->W->bool)->bool 
           <=> (W,R) IN APPR {} 

 RF_APPR_T
 |- ! W R. (W,R) IN RF:(W->bool)#(W->W->bool)->bool 
           <=> (W,R) IN APPR T_AX
          
 TF_APPR_K4
 |- ! W R. (W,R) IN TF:(W->bool)#(W->W->bool)->bool 
           <=> (W,R) IN APPR K4_AX

 ITF_APPR_GL          
 |- ! W R. (W,R) IN ITF:(W->bool)#(W->W->bool)->bool 
           <=> (W,R) IN APPR GL_AX

\end{lstlisting}

We summarise all these correspondence results in Table~\ref{table:corr}, which is an HOL Light version of Table~\ref{tab:char-appr}.
\begin{table}[h]
    \centering
    \begin{tabular}{l|l | l | l}
        \toprule
        \textit{\texttt{S}} & \textit{\texttt{CHAR S}} & \textit{\texttt{APPR S}} & \textit{Characteristic properties for \texttt{S}} \\
        \midrule
        \verb|{}|  & \verb|FRAME|~\href{https://archive.softwareheritage.org/swh:1:cnt:d364de3158b6405c920115aa0c801af9af49e302;origin=https://github.com/HOLMS-lib/HOLMS;visit=swh:1:snp:2c0efd349323ed6f8067581cf1f6d95816e49841;anchor=swh:1:rev:1caf3be141c6f646f78695c0eb528ce3b753079a;path=/gen_completeness.ml;lines=15-16}{\ExternalLink} & \verb|FINITE_FRAME|~\href{https://archive.softwareheritage.org/swh:1:cnt:d364de3158b6405c920115aa0c801af9af49e302;origin=https://github.com/HOLMS-lib/HOLMS;visit=swh:1:snp:2c0efd349323ed6f8067581cf1f6d95816e49841;anchor=swh:1:rev:1caf3be141c6f646f78695c0eb528ce3b753079a;path=/gen_completeness.ml;lines=24-25}{\ExternalLink} & -  \\
        \hline
        \verb|T_AX| & \verb|REFL|~\href{https://archive.softwareheritage.org/swh:1:cnt:1352cb294724b2251b9346657432700ecbed10d2;origin=https://github.com/HOLMS-lib/HOLMS;visit=swh:1:snp:2c0efd349323ed6f8067581cf1f6d95816e49841;anchor=swh:1:rev:1caf3be141c6f646f78695c0eb528ce3b753079a;path=/t_completeness.ml;lines=29-34}{\ExternalLink} &  \verb|RF|~\href{https://archive.softwareheritage.org/swh:1:cnt:1352cb294724b2251b9346657432700ecbed10d2;origin=https://github.com/HOLMS-lib/HOLMS;visit=swh:1:snp:2c0efd349323ed6f8067581cf1f6d95816e49841;anchor=swh:1:rev:1caf3be141c6f646f78695c0eb528ce3b753079a;path=/t_completeness.ml;lines=88-94}{\ExternalLink}  & reflexivity \\
        \hline
        \verb|K4_AX| & \verb|TRANS|~\href{https://archive.softwareheritage.org/swh:1:cnt:32869a1df338d91c4cd4a0cb9e73fb4f1be29991;origin=https://github.com/HOLMS-lib/HOLMS;visit=swh:1:snp:2c0efd349323ed6f8067581cf1f6d95816e49841;anchor=swh:1:rev:1caf3be141c6f646f78695c0eb528ce3b753079a;path=/k4_completeness.ml;lines=35-40}{\ExternalLink} &  \verb|TF|~\href{https://archive.softwareheritage.org/swh:1:cnt:32869a1df338d91c4cd4a0cb9e73fb4f1be29991;origin=https://github.com/HOLMS-lib/HOLMS;visit=swh:1:snp:2c0efd349323ed6f8067581cf1f6d95816e49841;anchor=swh:1:rev:1caf3be141c6f646f78695c0eb528ce3b753079a;path=/k4_completeness.ml;lines=96-102}{\ExternalLink}  & transitivity \\
        \hline
       \verb|GL_AX| & \verb|TRANSNT|~\href{https://archive.softwareheritage.org/swh:1:cnt:52c9de12454731a9a291b06bc750c0d2d14e1fb4;origin=https://github.com/HOLMS-lib/HOLMS;visit=swh:1:snp:2c0efd349323ed6f8067581cf1f6d95816e49841;anchor=swh:1:rev:1caf3be141c6f646f78695c0eb528ce3b753079a;path=/gl_completeness.ml;lines=26-32}{\ExternalLink} & \verb|ITF|~\href{https://archive.softwareheritage.org/swh:1:cnt:52c9de12454731a9a291b06bc750c0d2d14e1fb4;origin=https://github.com/HOLMS-lib/HOLMS;visit=swh:1:snp:2c0efd349323ed6f8067581cf1f6d95816e49841;anchor=swh:1:rev:1caf3be141c6f646f78695c0eb528ce3b753079a;path=/gl_completeness.ml;lines=141-148}{\ExternalLink}  & transitivity,  noetherianity:  \\
        \bottomrule
    \end{tabular}
    \caption{Correspondence results in HOLMS}
    \label{table:corr}
\end{table}

At this stage, we have adequately analysed the theory of correspondence to prove soundness and completeness. 
        
\section{Formalising Adequacy}

We present in this section HOLMS formalisation of adequacy theorem--together with consistency results--for $\mathbb{K, \ T, \ K}4$ and $\mathbb{GL}$, focusing in particular on which parts of the proofs are \textit{parametric} and which \textit{ad-hoc polymorphic}.

\subsection{Proving Soundness}

As remarked in~\ref{sub:soundness}, identifying appropriate classes provides \textit{per se} a proof of soundness for each axiomatic calculus w.r.t~its characteristic class of finite frames. An analogous result for frames that are not necessarily finite is obtained by considering the characterisation \verb|CHAR_CAR|~\href{https://archive.softwareheritage.org/swh:1:cnt:d364de3158b6405c920115aa0c801af9af49e302;origin=https://github.com/HOLMS-lib/HOLMS;visit=swh:1:snp:2c0efd349323ed6f8067581cf1f6d95816e49841;anchor=swh:1:rev:1caf3be141c6f646f78695c0eb528ce3b753079a;path=/gen_completeness.ml;lines=82-105}{\ExternalLink}.
These two \emph{parametric} results are proved in HOLMS \verb|gen_completeness| under the names \verb|GEN_APPR_VALID|~\href{https://archive.softwareheritage.org/swh:1:cnt:d364de3158b6405c920115aa0c801af9af49e302;origin=https://github.com/HOLMS-lib/HOLMS;visit=swh:1:snp:2c0efd349323ed6f8067581cf1f6d95816e49841;anchor=swh:1:rev:1caf3be141c6f646f78695c0eb528ce3b753079a;path=/gen_completeness.ml;lines=147-150}{\ExternalLink} and \verb|GEN_CHAR_VALID|~\href{https://archive.softwareheritage.org/swh:1:cnt:d364de3158b6405c920115aa0c801af9af49e302;origin=https://github.com/HOLMS-lib/HOLMS;visit=swh:1:snp:2c0efd349323ed6f8067581cf1f6d95816e49841;anchor=swh:1:rev:1caf3be141c6f646f78695c0eb528ce3b753079a;path=/gen_completeness.ml;lines=71-80}{\ExternalLink}.

Soundness for any axiom system formalisable by our deducibility relation (normal systems) is then obtained by \textit{instantiating the parameters} of this theorem.
\vspace{-3mm}
\begin{lstlisting}[caption={HOLMS theorems of soundness for K, T, K4 and GL}]
 K_FRAME_VALID
|- !H p. [{} . H |~ p] /\
         (!q. q IN H ==> FRAME:(W->bool)#(W->W->bool)->bool |= q)
        ==> FRAME:(W->bool)#(W->W->bool)->bool |= p

T_REFL_VALID
|- !H p. [T_AX . H |~ p] /\
         (!q. q IN H ==> REFL:(W->bool)#(W->W->bool)->bool |= q)
         ==> REFL:(W->bool)#(W->W->bool)->bool |= p

K4_TRANS_VALID
|- !H p. [K4_AX . H |~ p] /\
         (!q. q IN H ==> TRANS:(W->bool)#(W->W->bool)->bool |= q)
         ==> TRANS:(W->bool)#(W->W->bool)->bool |= p

GL_TRANSNT_VALID
|- !H p. [GL_AX . H |~ p] /\
        (!q. q IN H ==> TRANSNT:(W->bool)#(W->W->bool)->bool |= q)
        ==> TRANSNT:(W->bool)#(W->W->bool)->bool |= p
\end{lstlisting}

\subsection{Proving Consistency}

Following the canonical strategy summarised in~\ref{sub:consisency}, we derive consistency of a logic
from its soundness w.r.t.~a non-empty class of frames.

This approach follows a standard \textit{negation ad absurdum} argument common to all systems, which we outline as follows.

\begin{itemize}
    \item \verb|REFUTE_THEN(MP_TAC o MATCH_MP (*_APPR_VALID))| \\
    First, we assume $\mathcal{S} \vdash \bot$, then, leveraging soundness, we conclude $\mathfrak{S} \vDash \bot$.
    \item \verb|REWRITE_TAC| [\verb|valid; holds; holds_in; FORALL_PAIR_THM;| \\ $\phantom{REWRITEii}$ \verb|IN_APPR*; NOT_FORALL_THM|] \\
    By unpacking the definitions of validity and forcing, we establish that no world in any model forces falsehood, i.e. $\mathcal{M}, w \not \vDash \bot$.

    \item Finally, exhibiting a world in $\mathfrak{S}$, i.e. proving that the class is not empty, we reach a contradiction: $\mathfrak{S} \not \vDash \bot$. $\lightning$
\end{itemize}

In HOLMS, each specific file \verb|*_completeness.ml| presents this \textit{ad-hoc polymorphic} proof of consistency of $\verb|*|\,\in\{$K~\href{https://archive.softwareheritage.org/swh:1:cnt:ae230138ff15476c8dab9e32606bceca7168285b;origin=https://github.com/HOLMS-lib/HOLMS;visit=swh:1:snp:2c0efd349323ed6f8067581cf1f6d95816e49841;anchor=swh:1:rev:1caf3be141c6f646f78695c0eb528ce3b753079a;path=/k_completeness.ml;lines=48-56}{\ExternalLink}, T~\href{https://archive.softwareheritage.org/swh:1:cnt:1352cb294724b2251b9346657432700ecbed10d2;origin=https://github.com/HOLMS-lib/HOLMS;visit=swh:1:snp:2c0efd349323ed6f8067581cf1f6d95816e49841;anchor=swh:1:rev:1caf3be141c6f646f78695c0eb528ce3b753079a;path=/t_completeness.ml;lines=139-145}{\ExternalLink}, K4~\href{https://archive.softwareheritage.org/swh:1:cnt:32869a1df338d91c4cd4a0cb9e73fb4f1be29991;origin=https://github.com/HOLMS-lib/HOLMS;visit=swh:1:snp:2c0efd349323ed6f8067581cf1f6d95816e49841;anchor=swh:1:rev:1caf3be141c6f646f78695c0eb528ce3b753079a;path=/k4_completeness.ml;lines=147-153}{\ExternalLink}, GL~\href{https://archive.softwareheritage.org/swh:1:cnt:52c9de12454731a9a291b06bc750c0d2d14e1fb4;origin=https://github.com/HOLMS-lib/HOLMS;visit=swh:1:snp:2c0efd349323ed6f8067581cf1f6d95816e49841;anchor=swh:1:rev:1caf3be141c6f646f78695c0eb528ce3b753079a;path=/gl_completeness.ml;lines=216-222}{\ExternalLink}$\}$, by applying its specific soundness lemmata and providing an \textit{ad-hoc} inhabitant of its appropriate class. 

\newpage
\begin{lstlisting}[caption={HOLMS consistency theorems}]
K_CONSISTENT |- ~ [{} . {} |~ False]

T_CONSISTENT |- ~ [T_AX . {} |~ False]

K4_CONSISTENT |- ~ [K4_AX . {} |~ False]

GL_CONSISTENT |- ~ [GL_AX . {} |~ False]
\end{lstlisting}

\subsection{Proving Completeness}\label{sub:formal-completeness}
In this section, we examine the completeness theorems for the logics currently formalized in HOLMS.

Our \textit{as-general-as-possible} demonstrative strategy will not follow the common version of the `canonical model method', because it cannot be employed for every logic implemented in HOLMS, e.g. for $\mathbb{GL}$ due to its incompactness as already implicitly stated in~\cite{boolos1995logic} argued in the prior work of two HOLMS developers logic \cite[\S 4]{DBLP:journals/jar/MaggesiB23}. 

Still, that work~\cite{maggesi_et_al:LIPIcs.ITP.2021.26,DBLP:journals/jar/MaggesiB23} formalised a Henkin-style proof for GL w.r.t.~finite frames by hacking the strategy outlined by George Boolos in~\cite[\S 5]{boolos1995logic}. The two authors already observed there that the major part of Boolos' completeness proof is generalisable to most of normal (mono)modal systems. 

In our communication paper \cite{DBLP:conf/overlay/BilottaMBQ24}, we remarked on this idea and presented a partial parametrisation of the completeness proof, with a \textit{parametric} version of the general truth lemma.

Moreover, in Chapter~\ref{sub:completeness} we succeed in showcasing the proof outlined in Boolos' monograph, which further clarifies the modularity of the strategy.
The implementation of this proof within HOL Light avoids code duplication as far as possible and allow us to \textit{precisely measure the generality} of Boolos' strategy. 

\begin{remark}
\end{remark}
    The proof presented in section~\ref{sub:proof-sketch} is a formalisation— using the HOL Light proof assistant— of the proof detailed in Section~\ref{sub:completeness}. Since the full \textit{informal} proof has already been presented comprehensively, here we provide only a sketch of the proof, which highlights the most important lemmas employed. This outline emphasizes the lines of code that are \textit{parametric} and \textit{ad hoc polymorphic} and some technical choices made in HOLMS. 

\subsubsection{Towards Completeness}

As highlighted in our extended \textit{informal} proof of completeness in Section~\ref{sub:consisency}, maximal consistent sets of formulas play a central role in our strategy for proving completeness. In this section, we introduce a formalisation of this concept in HOL Light, along with five related lemmas, which have been verified using the proof assistant.
\medskip

\begin{remark}
   In HOLMS we work on words that are \emph{lists of formulas without repetitions} rather than \emph{sets} of formulas.
   This implementation choice is driven purely by practical considerations:
   \begin{enumerate} \vspace{-2mm}
       \item The theory of lists allows us to leverage \textit{pre-proven results} for finite lists; \vspace{-2mm}
       \item Theory of lists provides a useful \textit{structural recursion rule}. \vspace{-2mm}
   \end{enumerate} 
In Section~\ref{Theories}, we discussed various theories developed in HOLMS, and in particular we presented key results and definitions related to HOL Light’s theory of lists.
\end{remark}
\medskip
\medskip

As in Chapter~\ref{chap:1} we need to define the big conjunction over sets of formulas $\bigwedge\mathcal{X}$ which was not part of the modal language $\mathcal{L}_{\Box}$, HOLMS file \verb|conjlist.ml|~\href{https://archive.softwareheritage.org/swh:1:cnt:3d890d9e1ac7d24532b2645b6fb9ad5b160c7ba0;origin=https://github.com/HOLMS-lib/HOLMS;visit=swh:1:snp:2c0efd349323ed6f8067581cf1f6d95816e49841;anchor=swh:1:rev:1caf3be141c6f646f78695c0eb528ce3b753079a;path=/conjlist.ml}{\ExternalLink} provides a formal definition of the big conjunction over lists of formulas (\verb|CONJLIST X|) and contains additional parametric results related to this concept.

\begin{lstlisting}[caption=HOLMS definition of the big conjunction over lists of formulas]
 #let CONJLIST = new_recursive_definition list_RECURSION
 `CONJLIST [] = True /\
  (!p X. CONJLIST (CONS p X) = if X = [] then p 
                                else p && CONJLIST X)';;
\end{lstlisting}

HOLMS file \verb|consistent.ml|~\href{https://archive.softwareheritage.org/swh:1:cnt:e261ca08330d5ef66376347407c27fcece1e9f2e;origin=https://github.com/HOLMS-lib/HOLMS;visit=swh:1:snp:2c0efd349323ed6f8067581cf1f6d95816e49841;anchor=swh:1:rev:1caf3be141c6f646f78695c0eb528ce3b753079a;path=/consistent.ml}{\ExternalLink} formalises the parametric notions of  $\mathcal{S}$-consistency (\verb|CONSISTENT S|~\href{https://archive.softwareheritage.org/swh:1:cnt:e261ca08330d5ef66376347407c27fcece1e9f2e;origin=https://github.com/HOLMS-lib/HOLMS;visit=swh:1:snp:2c0efd349323ed6f8067581cf1f6d95816e49841;anchor=swh:1:rev:1caf3be141c6f646f78695c0eb528ce3b753079a;path=/consistent.ml;lines=11-12}{\ExternalLink})  and  $\mathcal{S},A$-maximal consistency (\verb|MAXIMAL_CONSISTENT S|~\href{https://archive.softwareheritage.org/swh:1:cnt:e261ca08330d5ef66376347407c27fcece1e9f2e;origin=https://github.com/HOLMS-lib/HOLMS;visit=swh:1:snp:2c0efd349323ed6f8067581cf1f6d95816e49841;anchor=swh:1:rev:1caf3be141c6f646f78695c0eb528ce3b753079a;path=/consistent.ml;lines=82-85}{\ExternalLink}) and proves an HOL Light version of the Lemmata~\ref{lem:consistency_lemma}-~\ref{lem:max_cons}, which are fundamental to prove completeness. Notice that the profs of these results provided in Section~\ref{sub:maximal-consistency} are an \textit{abstract} version of the ones provided in HOLMS.

\begin{lstlisting}[caption={HOLMS definition of S-consistency and a related lemma}, label={lst:cons}]
 #let CONSISTENT = new_definition
 `CONSISTENT S (l:form list) <=> ~[S . {} |~ Not (CONJLIST l)]';;

 FALSE_IMP_NOT_CONSISTENT 
 |- !S X. MEM False X ==> ~ CONSISTENT S X
\end{lstlisting}
\vspace{-7mm}
\begin{lstlisting}[caption={HOLMS definition of S,A-maximal consistency and related lemmata}, label={lst:max-cons}]
 #let MAXIMAL_CONSISTENT = new_definition
 `MAXIMAL_CONSISTENT S p X <=>
   CONSISTENT S X /\ NOREPETITION X /\
   (!q. q SUBFORMULA p ==> MEM q X \/ MEM (Not q) X)';;

 MAXIMAL_CONSISTENT_MEM_NOT 
 |- !S X p q. MAXIMAL_CONSISTENT S p X /\ q SUBFORMULA p
             ==> (MEM (Not q) X <=> ~ MEM q X)

 EXTEND_MAXIMAL_CONSISTENT
 |- !S p X. CONSISTENT S X /\
           (!q. MEM q X ==> q SUBSENTENCE p)
           ==> ?M. MAXIMAL_CONSISTENT S p M /\
                   (!q. MEM q M ==> q SUBSENTENCE p) /\
                   X SUBLIST M

 NONEMPTY_MAXIMAL_CONSISTENT
 |- !S p. ~ [S . {} |~ p]
         ==> ?M. MAXIMAL_CONSISTENT S p M /\
                 MEM (Not p) M /\
                 (!q. MEM q M ==> q SUBSENTENCE p)          

 MAXIMAL_CONSISTENT_LEMMA 
 |-!S p X A b. MAXIMAL_CONSISTENT S p X /\
               (!q. MEM q A ==> MEM q X) /\
               b SUBFORMULA p /\
               [S . {} |~ CONJLIST A --> b]
               ==> MEM b X
\end{lstlisting}

\subsubsection{HOLMS Proof Sketch}\label{sub:proof-sketch}
Let \verb|S| be the set of specific schemata of a normal modal system implemented in HOLMS. We claim that it is complete with respect to the set of its appropriate frames \verb|APPR S|. 
\vspace{-4mm}
\begin{lstlisting}[caption={Claim 1: Proving Completeness },label={lst:claim_1}]
 val it : goalstack = 1 subgoal (1 total)
 `!p.APPR S |= p ==> [S. {} |~ p]'
\end{lstlisting}

We proceed by contraposition, then we claim the rewritten sentence:
\vspace{-4mm}
\begin{lstlisting}[caption={Claim 2: contraposition },label={lst:claim_2}]
 #  (GEN_REWRITE_TAC I [GSYM CONTRAPOS_THM]);;
 val it : goalstack = 1 subgoal (1 total)

 `!p. ~[S. {} |~ p] ==> ~(APPR S |= p)'
\end{lstlisting}

We rewrite validity in a set of frames (\verb|valid|) in terms of forcing relation (\verb|holds|) and we exploit some propositional tautologies:
\vspace{-4mm}
\begin{lstlisting}[caption={Claim 3: rewriting },label={lst:claim-3}]
# e (REWRITE_TAC[valid; holds_in; NOT_FORALL_THM; NOT_IMP]);;
val it : goalstack = 1 subgoal (1 total)

`!p. ~[S . {} |~ p]
     ==> ?W R. (W,R) IN APPR S /\
               (?V m. m IN W /\ ~holds (W,R) V p w)'         
\end{lstlisting}

This means that for each set of axioms $\verb|S|$ and for each modal formula \verb|p|, we have to find a countermodel \verb|(W,R) V| inhabiting \verb|APPR S|, and a `counterworld' \verb|m| inhabiting \verb|WR V| such that \verb|~holds (W,R) V p w|.

To do so, we formalise the argument in \cite[\S 5]{boolos1995logic} and implement the following two-part strategy, which distinguishes the \textit{parametric polymorphic} aspects of the proof-fully independent from the concrete instantiations of the parameters—from those that are \textit{properly ad hoc polymorphic}--with details more closely tied to the logics under investigation.
\medskip
\newpage
\begin{itemize}

\item{\textbf{Parametric part of the proof}} \\
\begin{myproof1}
    This entire part of the proof is developed in \verb|gen_completeness.ml|~\href{https://archive.softwareheritage.org/swh:1:cnt:d364de3158b6405c920115aa0c801af9af49e302;origin=https://github.com/HOLMS-lib/HOLMS;visit=swh:1:snp:2c0efd349323ed6f8067581cf1f6d95816e49841;anchor=swh:1:rev:1caf3be141c6f646f78695c0eb528ce3b753079a;path=/gen_completeness.ml}{\ExternalLink}, that collects all the HOLMS parametric lines of code related with adequacy theorems and correspondence theory.
    
\begin{enumerate}
    
    \item \textbf{We identify a \textit{parametric} countermodel} \verb|GEN_STANDARD_MODEL S p|~\href{https://archive.softwareheritage.org/swh:1:cnt:d364de3158b6405c920115aa0c801af9af49e302;origin=https://github.com/HOLMS-lib/HOLMS;visit=swh:1:snp:2c0efd349323ed6f8067581cf1f6d95816e49841;anchor=swh:1:rev:1caf3be141c6f646f78695c0eb528ce3b753079a;path=/gen_completeness.ml;lines=180-183}{\ExternalLink} such that: 
    
    \begin{itemize}
        \item Is based on a \verb|GEN_STANDARD_FRAME S p|~\href{https://archive.softwareheritage.org/swh:1:cnt:d364de3158b6405c920115aa0c801af9af49e302;origin=https://github.com/HOLMS-lib/HOLMS;visit=swh:1:snp:2c0efd349323ed6f8067581cf1f6d95816e49841;anchor=swh:1:rev:1caf3be141c6f646f78695c0eb528ce3b753079a;path=/gen_completeness.ml;lines=156-162}{\ExternalLink}, i.e.:
        \begin{itemize}
            \item[(a)] Its domain \verb|W| is composed by \verb|MAXIMAL_CONSISTENT S p| list of formulas, which are \verb|SUBSENTNCE p|.
            \begin{itemize}
                \item By \verb|NONEMPTY_MAXIMAL_CONSISTENT|~\href{https://archive.softwareheritage.org/swh:1:cnt:e261ca08330d5ef66376347407c27fcece1e9f2e;origin=https://github.com/HOLMS-lib/HOLMS;visit=swh:1:snp:2c0efd349323ed6f8067581cf1f6d95816e49841;anchor=swh:1:rev:1caf3be141c6f646f78695c0eb528ce3b753079a;path=/consistent.ml;lines=228-245}{\ExternalLink}, we prove that this set is not empty and thus is a well-defined domain;
                \item Thanks to the theorems \verb|FINITE_NOREPETITION|~\href{https://archive.softwareheritage.org/swh:1:cnt:d100c0780569a67ceb5f1802e928df774c6a6107;origin=https://github.com/HOLMS-lib/HOLMS;visit=swh:1:snp:2c0efd349323ed6f8067581cf1f6d95816e49841;anchor=swh:1:rev:1caf3be141c6f646f78695c0eb528ce3b753079a;path=/misc.ml;lines=207-213}{\ExternalLink} proved in \verb|misc.ml| and \verb|FINITE_SUBSET_SUBFORMULAS_LEMMA|~\href{https://archive.softwareheritage.org/swh:1:cnt:4a011f69b1180019be5d03c6d5f1cec4550055e8;origin=https://github.com/HOLMS-lib/HOLMS;visit=swh:1:snp:2c0efd349323ed6f8067581cf1f6d95816e49841;anchor=swh:1:rev:1caf3be141c6f646f78695c0eb528ce3b753079a;path=/modal.ml;lines=150-157}{\ExternalLink}, we prove that that the domain is also finite.
            \end{itemize} 
             \item[(b)]  Its accessibility relation \verb|R| verifies two requirements:
             \begin{enumerate}
             
                 \item[i.]\label{itm:req-1} \verb|(W,R) IN APPR S|;
                 \item[ii.]\label{itm:req-2} The one occupying the last three lines of standard frame definition in Listing~\ref{lst:GEN_STANDARD_MODEL}.
            \end{enumerate}
        \end{itemize}
        \item[(c)] A world verifies an atom (\verb|V a w|) if and only if the atom belongs to the world \verb|MEM (Atom a) w| and is a \verb|SUBSENTNCE p|.
    \end{itemize}

The notion of \verb|*_STANDARD_MODEL| specific to each normal system \verb|*| $\in \{ \mathbb{K; \ T; \ K}4; \ \mathbb{GL} \}$ is defined as an \textit{instantiation of the parameters} in \verb|GEN_STANDARD_MODEL|.

\begin{lstlisting}[caption={HOLMS definition of standard frames and models},label={lst:GEN_STANDARD_MODEL}]
 #let GEN_STANDARD_FRAME_DEF = new_definition
 `GEN_STANDARD_FRAME S p =
  APPR S INTER
  {(W,R) | W = {w | MAXIMAL_CONSISTENT S p w /\
                (!q. MEM q w ==> q SUBSENTENCE p)} /\
           (!q w. ({Box q SUBFORMULA p /\ w IN W})
             ==> (MEM (Box q) w <=> 
                 !x. R w x ==> MEM q x))}';;
    
#let GEN_STANDARD_MODEL_DEF = new_definition
 `GEN_STANDARD_MODEL S p (W,R) V <=>
  (W,R) IN GEN_STANDARD_FRAME S p /\
  (!a w. w IN W ==> (V a w <=> MEM (Atom a) w /\ 
                               Atom a SUBFORMULA p))';;
\end{lstlisting}

    \item \textbf{We prove a \textit{parametric}} \verb|GEN_TRUTH_LEMMA|~\href{https://archive.softwareheritage.org/swh:1:cnt:d364de3158b6405c920115aa0c801af9af49e302;origin=https://github.com/HOLMS-lib/HOLMS;visit=swh:1:snp:2c0efd349323ed6f8067581cf1f6d95816e49841;anchor=swh:1:rev:1caf3be141c6f646f78695c0eb528ce3b753079a;path=/gen_completeness.ml;lines=189-350}{\ExternalLink}\label{itm:truth-lemma}. \\ This step allows, for every \verb|q SUBFORMULA p|, the reduction of the model-theoretic notion of \textit{forcing} (\verb|holds (W,R) V q w|) to the list-theoretic one of \textit{membership} (\verb|MEM q w|). The proof of this lemma exploits all the results presented in Listing~\ref{lst:cons} and~\ref{lst:max-cons}.
\begin{lstlisting}[caption={HOLMS parametric truth lemma},label={lst:GEN_TRUTH_LEMMA}]
 GEN_TRUTH_LEMMA 
 |-!S W R p V q.
     ~ [S . {} |~ p] /\
     GEN_STANDARD_MODEL S p (W,R) V /\
     q SUBFORMULA p
     ==> !w. w IN W ==> (MEM q w <=> holds (W,R) V q w)
\end{lstlisting}

Since the modal language $\mathcal{L'}_{\Box}$ embedded in HOLMS is an extension of the one presented in Chapter~\ref{chap:1}, HOLMS proof of this lemma is much longer than the one presented in~\ref{lem:truth}. as it involves three base cases instead of two and six inductive steps instead of two. However, despite these differences in structure and length, the two proofs remain equivalent in their core argument and overall reasoning.

     \item In the parametric lemma \verb|GEN_COUNTERMODEL|~\href{https://archive.softwareheritage.org/swh:1:cnt:d364de3158b6405c920115aa0c801af9af49e302;origin=https://github.com/HOLMS-lib/HOLMS;visit=swh:1:snp:2c0efd349323ed6f8067581cf1f6d95816e49841;anchor=swh:1:rev:1caf3be141c6f646f78695c0eb528ce3b753079a;path=/gen_completeness.ml;lines=549-572}{\ExternalLink}, we identify a `counterworld' \verb| m| in \verb|GEN_STANDARD_MODEL S p (W,R) V| such that:
     \begin{itemize}
         \item \verb|MEM (Not p) m|;
         \item \verb|~holds (W,R) V p m|.
     \end{itemize}

     This parametric lemma states the claim in Listing~\ref{lst:claim-3} and is proved by  \verb|NONEMPTY_MAXIMAL_CONSISTENT| in Listing~\ref{lst:max-cons} and \verb|GEN_TRUTH_LEMMA|.
     
\begin{lstlisting}[caption={HOLMS parametric countermodel lemmata},label={lst:GEN_COUNTERMODEL}]
 GEN_COUNTERMODEL
 !S W R V m p. ~ [S . {} |~ p] /\
               MAXIMAL_CONSISTENT S p m /\
               MEM (Not p) m /\
               (!q. MEM q m ==> q SUBSENTENCE p) /\
               GEN_STANDARD_MODEL S p (W,R) V
                 ==> ~holds (W,R) V p m            
\end{lstlisting}

\end{enumerate}
\end{myproof1}

\item \textbf{Ad hoc polymorphic part of the proof}
\begin{myproof1}
    
To be \textit{parametric}, the definition of the standard frame cannot directly define its own accessibility relation, but it merely requests the verification of two constraints. Consequently, we must define for each normal logic \verb|*| a \verb|*_STANDARD_ACCESSIBILITY_RELATION| on maximal consistent lists verifying the two requirements~\customref{itm:req-1} and~\customreff{itm:req-2}.

Nevertheless, this part of the proof is not developed separately for each logic. Instead, it is structured in an \textit{ad hoc polymorphic} way:

\begin{itemize}
    \item[A.] The \textit{parametric} file \verb|gen_completeness.ml|:
\end{itemize}
\begin{itemize}
    \item Defines a parametric accessibility relation (\verb|GEN_STANDARD_REL|~\href{https://archive.softwareheritage.org/swh:1:cnt:d364de3158b6405c920115aa0c801af9af49e302;origin=https://github.com/HOLMS-lib/HOLMS;visit=swh:1:snp:2c0efd349323ed6f8067581cf1f6d95816e49841;anchor=swh:1:rev:1caf3be141c6f646f78695c0eb528ce3b753079a;path=/gen_completeness.ml;lines=356-370}{\ExternalLink}) which subsumes any relation specific to $\mathcal{S}$ we are going to define. Listing~\ref{list:stn-rel} codifies this class of ``standard'' relations for the completeness proof.

\begin{lstlisting}[caption={HOLMS definition of standard relations for completeness},label={list:stn-rel}]
#let GEN_STANDARD_REL = new_definition
 `GEN_STANDARD_REL S p w x <=>
   MAXIMAL_CONSISTENT S p w /\
   (!q. MEM q w ==> q SUBSENTENCE p) /\
   MAXIMAL_CONSISTENT S p x /\ 
   (!q. MEM q x ==> q SUBSENTENCE p) /\
   (!B. MEM (Box B) w ==> MEM B x)';;
\end{lstlisting}

    \item Proves parametric results that apply uniformly across different logics. These simplify and shorten the verification that the two conditions \customref{itm:req-1} and~\customreff{itm:req-2} hold for each \verb|*_STANDARD_ACCESSIBILITY_| \verb|RELATION|  We will not present these technical results here, but the user can find them in the lines~\href{https://archive.softwareheritage.org/swh:1:cnt:d364de3158b6405c920115aa0c801af9af49e302;origin=https://github.com/HOLMS-lib/HOLMS;visit=swh:1:snp:2c0efd349323ed6f8067581cf1f6d95816e49841;anchor=swh:1:rev:1caf3be141c6f646f78695c0eb528ce3b753079a;path=/gen_completeness.ml;lines=362-543}{\ExternalLink}.
\end{itemize}

\item[B.] For \verb|*| $\in \{ \mathbb{K; \ T; \ K}4; \ \mathbb{GL} \}$, each file \verb|*_completeness.ml| provides:

\begin{enumerate}
    \item [I.] \textbf{An identification of the accessibility relation} \verb|*_STANDARD_REL|.

    Table~\ref{tab:standard-rel} recalls HOLMS formalisation of the standard accessibility relation $Rel_*^A$ provided in Table~\ref{tab:standard_accessibility_reation}.
    We then consider a frame, whose domain is the set of \verb|MAXIMAL_CONSISTENT * p| list made of \verb|SUBSENTENCE p| and whose accessibility relation is \verb|*_STANDARD_REL|.

\end{enumerate}
\end{myproof1}

\begin{table}[h]
\centering
\begin{tabular}{|l|l|}
\hline
\verb|*| & \verb|*_STANDARD_REL p w x| \\
\midrule
\verb|K| & \verb|GEN_STANDARD_REL {} p w x| \\
\hline
\verb|T| & \verb|GEN_STANDARD_REL T_AX p w x| \\
\hline
\verb|K4| & \verb|GEN_STANDARD_REL K4_AX p w x /\| \\
          & \verb|(! B . MEM ( Box B ) w == > MEM ( Box B ) x| \\
\hline
\verb|GL| & \verb|GEN_STANDARD_REL GL_AX p w x /\| \\
          & \verb|(! B . MEM ( Box B ) w == > MEM ( Box B ) x ) /\| \\
          & \verb|(? E . MEM ( Box E ) x /\ MEM ( Not ( Box E )) w | \\
\hline
\end{tabular}
    \caption{HOLMS \textit{ad hoc polymorphic} definitions of standard accessibility relations}
    \label{tab:standard-rel}
\end{table}

\begin{myproof1}
\begin{enumerate}

    \item[II.] \textbf{Verification that the considered frame is a} \verb|*_STANDARD_FRAME|.

    In particular, we have to prove that this frame verifies the two constraints~\customref{itm:req-1} and~\customreff{itm:req-2} for \verb|*_STANDARD_MODEL|:

\begin{lstlisting}
 val it : goalstack = 1 subgoal (1 total)

 `! p. ~[* . {} |~ p]
     ==> {M | MAXIMAL_CONSISTENT * p M /\
         (? q. MEM q M ==> q SUBSENTENCE p)},
         *_STANDARD_REL p IN *_STANDARD_FRAME p'
\end{lstlisting}

\begin{enumerate}

    \item[1(b) i.] Our lemma \verb|*_MAXIMAL_CONSISTENT| verifies that this frame is appropriate (\verb|IN APPR *|).
    The strategy followed by these proofs is the same developed in the \textit{abstract} Lemma~\ref{lem:appr_std_frame} ($\langle W_*^{A}, Rel_*^{A} \rangle \in \mathfrak{S}$. ).
\vspace{-5mm}    
\begin{lstlisting}[caption={HOLMS guarantees that each defined frame is appropriate}]
 K_MAXIMAL_CONSISTENT
 |- !p. ~ [{} . {} |~ p]
        ==> ({M | MAXIMAL_CONSISTENT {} p M /\
                  (!q. MEM q M ==> q SUBSENTENCE p)},
             K_STANDARD_REL p)
            IN FINITE_FRAME

 RF_MAXIMAL_CONSISTENT
 |- !p. ~ [T_AX . {} |~ p]
        ==> ({M | MAXIMAL_CONSISTENT T_AX p M /\
                  (!q. MEM q M ==> q SUBSENTENCE p)},
             T_STANDARD_REL p)
            IN RF `

 TF_MAXIMAL_CONSISTENT
 |- !p. ~ [T_AX . {} |~ p]
        ==> ({M | MAXIMAL_CONSISTENT T_AX p M /\
                  (!q. MEM q M ==> q SUBSENTENCE p)},
             T_STANDARD_REL p)
            IN RF

 ITF_MAXIMAL_CONSISTENT
 |- !p. ~ [GL_AX . {} |~ p]
          ==> ({M | MAXIMAL_CONSISTENT GL_AX p M /\
                    (!q. MEM q M ==> q SUBSENTENCE p)},
               GL_STANDARD_REL p)
              IN ITF
 \end{lstlisting}

\item[1(b) ii.] Our lemma \verb|*_ACCESSIBILITY_LEMMA| verifies the requirement~\customreff{itm:req-2}.\\

The strategy employed by the followng theorems is the same developed in the \textit{abstract} Lemma~\ref{lem:accessibility_std_frame}.

 \begin{lstlisting}[caption={HOLMS theorem guaranteeing that \ref{itm:req-2} holds}]
 K_ACCESSIBILITY_LEMMA
 |- !p w q. ~ [{} . {} |~ p] /\
      MAXIMAL_CONSISTENT {} p w /\
      (!q. MEM q w ==> q SUBSENTENCE p) /\
      Box q SUBFORMULA p /\
      (!x. K_STANDARD_REL p w x ==> MEM q x)
      ==> MEM (Box q) w`
     
 T_ACCESSIBILITY_LEMMA
 |- !p w q.
      ~ [T_AX . {} |~ p] /\
      MAXIMAL_CONSISTENT T_AX p w /\
      (!q. MEM q w ==> q SUBSENTENCE p) /\
      Box q SUBFORMULA p /\
      (!x. T_STANDARD_REL p w x ==> MEM q x)
      ==> MEM (Box q) w

 K4_ACCESSIBILITY_LEMMA
 |- !p w q.
      ~ [K4_AX . {} |~ p] /\
      MAXIMAL_CONSISTENT K4_AX p w /\
      (!q. MEM q w ==> q SUBSENTENCE p) /\
      Box q SUBFORMULA p /\
      (!x. K4_STANDARD_REL p w x ==> MEM q x)
           ==> MEM (Box q) w

 GL_ACCESSIBILITY_LEMMA
 |- !p w q.
      ~ [GL_AX . {} |~ p] /\
      MAXIMAL_CONSISTENT GL_AX p w /\
      (!q. MEM q w ==> q SUBSENTENCE p) /\
      Box q SUBFORMULA p /\
      (!x. GL_STANDARD_REL p w x ==> MEM q x)
      ==> MEM (Box q) w
 \end{lstlisting}
\end{enumerate}
Figure~\ref{fig:adhoc} collects these results in a single table. 
\end{enumerate}
\end{myproof1}

\begin{table}[h]
    \centering
    \begin{tabular}{l|l | l}
        \toprule
        \textit{\texttt{*}} & \textit{Maximal consistent lemma for \texttt{*}} & \textit{Accesibility lemma for \texttt{*}}  \\
        \midrule
        \verb|K|  & \verb|K_MAXIMAL_CONSISTENT|~\href{https://archive.softwareheritage.org/swh:1:cnt:ae230138ff15476c8dab9e32606bceca7168285b;origin=https://github.com/HOLMS-lib/HOLMS;visit=swh:1:snp:2c0efd349323ed6f8067581cf1f6d95816e49841;anchor=swh:1:rev:1caf3be141c6f646f78695c0eb528ce3b753079a;path=/k_completeness.ml;lines=119-127}{\ExternalLink} & \verb|K_ACCESSIBILITY_LEMMA|~\href{https://archive.softwareheritage.org/swh:1:cnt:ae230138ff15476c8dab9e32606bceca7168285b;origin=https://github.com/HOLMS-lib/HOLMS;visit=swh:1:snp:2c0efd349323ed6f8067581cf1f6d95816e49841;anchor=swh:1:rev:1caf3be141c6f646f78695c0eb528ce3b753079a;path=/k_completeness.ml;lines=129-150}{\ExternalLink}   \\
        \hline
        T & \verb|RF_MAXIMAL_CONSISTENT|~\href{https://archive.softwareheritage.org/swh:1:cnt:1352cb294724b2251b9346657432700ecbed10d2;origin=https://github.com/HOLMS-lib/HOLMS;visit=swh:1:snp:2c0efd349323ed6f8067581cf1f6d95816e49841;anchor=swh:1:rev:1caf3be141c6f646f78695c0eb528ce3b753079a;path=/t_completeness.ml;lines=220-259}{\ExternalLink} &  \verb|T_ACCESSIBILITY_LEMMA|~\href{https://archive.softwareheritage.org/swh:1:cnt:1352cb294724b2251b9346657432700ecbed10d2;origin=https://github.com/HOLMS-lib/HOLMS;visit=swh:1:snp:2c0efd349323ed6f8067581cf1f6d95816e49841;anchor=swh:1:rev:1caf3be141c6f646f78695c0eb528ce3b753079a;path=/t_completeness.ml;lines=261-282}{\ExternalLink} \\
        \hline
        \verb|K4| & \verb|TF_MAXIMAL_CONSISTENT|~\href{https://archive.softwareheritage.org/swh:1:cnt:32869a1df338d91c4cd4a0cb9e73fb4f1be29991;origin=https://github.com/HOLMS-lib/HOLMS;visit=swh:1:snp:2c0efd349323ed6f8067581cf1f6d95816e49841;anchor=swh:1:rev:1caf3be141c6f646f78695c0eb528ce3b753079a;path=/k4_completeness.ml;lines=229-249}{\ExternalLink} &  \verb|K4_ACCESSIBILITY_LEMMA|~\href{https://archive.softwareheritage.org/swh:1:cnt:32869a1df338d91c4cd4a0cb9e73fb4f1be29991;origin=https://github.com/HOLMS-lib/HOLMS;visit=swh:1:snp:2c0efd349323ed6f8067581cf1f6d95816e49841;anchor=swh:1:rev:1caf3be141c6f646f78695c0eb528ce3b753079a;path=/k4_completeness.ml;lines=274-382}{\ExternalLink}  \\
        \hline
       \verb|GL| & \verb|ITF_MAXIMAL_CONSISTENT|~\href{https://archive.softwareheritage.org/swh:1:cnt:52c9de12454731a9a291b06bc750c0d2d14e1fb4;origin=https://github.com/HOLMS-lib/HOLMS;visit=swh:1:snp:2c0efd349323ed6f8067581cf1f6d95816e49841;anchor=swh:1:rev:1caf3be141c6f646f78695c0eb528ce3b753079a;path=/gl_completeness.ml;lines=314-342}{\ExternalLink} & \verb|GL_ACCESSIBILITY_LEMMA|~\href{https://archive.softwareheritage.org/swh:1:cnt:52c9de12454731a9a291b06bc750c0d2d14e1fb4;origin=https://github.com/HOLMS-lib/HOLMS;visit=swh:1:snp:2c0efd349323ed6f8067581cf1f6d95816e49841;anchor=swh:1:rev:1caf3be141c6f646f78695c0eb528ce3b753079a;path=/gl_completeness.ml;lines=344-467}{\ExternalLink}  \\
        \bottomrule
    \end{tabular}
    \caption{\textit{Ad-hoc polymorhic} results in HOLMS}
    \label{fig:adhoc}
\end{table}
\end{itemize}

Thanks to these lemmata distributed in the files \verb|*_competeness.ml|, we are able to apply \verb|GEN_COUNTERMODEL| lemma in Listing~\ref{lst:GEN_COUNTERMODEL} to obtain as corollaries the completeness theorems for each \verb|*| $\in \{ \mathbb{K; \ T; \ K}4; \ \mathbb{GL} \}$.

 \begin{lstlisting}[caption={Completeness thorem for K in k\textunderscore completeness.ml}, label={lst:k_completeness}]
 K_COMPLETENESS_THM
 |- !p. FINITE_FRAME:(form list->bool)#(form list->form list->bool)->bool |= p
       ==> [{} . {} |~ p]
\end{lstlisting}
\vspace{-5mm}
 \begin{lstlisting}[caption={Completeness thorem for T in t\textunderscore completeness.ml}]
 T_COMPLETENESS_THM
 |- !p. RF:(form list->bool)#(form list->form list->bool)->bool |= p
       ==> [T_AX . {} |~ p]
\end{lstlisting}
\newpage
\begin{lstlisting}[caption={Completeness thorem for K4 in k4\textunderscore completeness.ml}]
 K4_COMPLETENESS_THM
 |- !p. TF:(form list->bool)#(form list->form list->bool)->bool |= p
       ==> [K4_AX . {} |~ p]
\end{lstlisting}
\vspace{-5mm}
\begin{lstlisting}[caption={Completeness thorem for GL in gl\textunderscore completeness.ml}, label ={lst:gl_completeness}]
 GL_COMPLETENESS_THM
 |- !p. ITF:(form list->bool)#(form list->form list->bool)->bool |= p
       ==> [GL_AX . {} |~ p]
\end{lstlisting}

\subsubsection{Generalising the Completeness Results}\label{sec:generalising-completeness}
\begin{remark}
\begin{mydefinition}
Our goal of reducing of the semantic concept of \textit{forcing}  to the list-theoretic notion of \textit{membership} (\verb|holds (W,R)| \verb|V q w| $\longrightarrow$ \verb|MEM q w|), led us to restrict our analysis on frames whose domain is type-theoretically composed of formula lists:
\begin{center}
  \verb|(W,R):(form list->bool)#(form list->form list->bool)|.
\end{center}
This parsing comes from the fact that \verb|MEM| has type \verb|:A->(A)list->bool|, while \verb|holds| \verb|:(W->bool)#(W->W->bool)->((char)list->W->bool)->form->W->bool|.
As a consequence, the proof presented above establishes completeness only w.r.t.~the class:
\begin{center}
    \verb|APPR S: (form list->bool)#(form list->form list->bool)->bool|.
\end{center}

\noindent Despite this annoying technical consequence, it is possible to overcome such a limitation and generalise the theorem to appropriate frames on arbitrary polymorphic domains: \verb|APPR S: (W->bool)#(W->W->bool)->bool|).
\medskip

\noindent The whole demonstrative strategy for completeness is then split into two stages:
\begin{enumerate}
     \item[($\alpha$)] Prove the completeness theorem on appropriate frames with \textit{domains of formula lists}: restricted, not polymorphic \verb| APPR S: (form list->bool)# | \\ \verb|(form list->form list->bool)->bool|. \\
    \textbf{\textit{Remark:}} This part of the proof has been developed in the previous section.
    \item[($\beta$)] Generalise the theorem on appropriate frames with \textit{generic domains}: \\ no restrictions, polymorphic \verb|APPR S: (W->bool)#(W->W->bool)->bool|. \\
    \textbf{\textit{Remark:}} This part of the proof will be developed in the following.
\end{enumerate}

\end{mydefinition}
\end{remark}

\medskip

To generalize our completeness result, we aim to establish a \textit{correspondence} between models of different types. A well-known method for achieving a rigorous correspondence is through \textit{bisimulation theory}~\cite{DBLP:books/el/07/BBW2007}.

In particular, we apply a type-theoretic version of the \emph{bisimulation invariance lemma} \verb|BISIMILAR_VALID|~\href{https://archive.softwareheritage.org/swh:1:cnt:4a011f69b1180019be5d03c6d5f1cec4550055e8;origin=https://github.com/HOLMS-lib/HOLMS;visit=swh:1:snp:2c0efd349323ed6f8067581cf1f6d95816e49841;anchor=swh:1:rev:1caf3be141c6f646f78695c0eb528ce3b753079a;path=/modal.ml;lines=307-316}{\ExternalLink}, which formalises Lemma~\ref{thm:forc-bisim} and states that any bisimulation preserves the validity of modal formulas.

The following Listing~\ref{lst:bisimulation} encodes HOLMS formalisation of the most important concepts and results of \textit{bisimulation theory}, abstractly presented and proved in~\ref{sec:bisimulation}. \verb|modal.ml| formalises  bisimulation in HOL Light as a relation between models (triples), rather than as a function.

\begin{lstlisting}[caption={HOLMS formalisation of Bisimulation Theory}, label={lst:bisimulation}]
 # let BISIMILAR = new_definition
 `BISIMILAR (W1,R1,V1) (W2,R2,V2) (w1:A) (w2:B) <=>
    ?Z. BISIMIMULATION (W1,R1,V1) (W2,R2,V2) Z /\ Z w1 w2';;
 #let BISIMILAR_HOLDS = prove
 (`!W1 R1 V1 W2 R2 V2 w1:A w2:B.
    BISIMILAR (W1,R1,V1) (W2,R2,V2) w1 w2
     ==> (!p. holds (W1,R1) V1 p w1 <=> holds (W2,R2) V2 p w2)',
  REWRITE_TAC[BISIMILAR] THEN MESON_TAC[BISIMIMULATION_HOLDS]);;
 
 BISIMILAR_HOLDS_IN 
 |-!W1 R1 W2 R2.
     (!V1 w1:A. ?V2 w2:B. BISIMILAR (W1,R1,V1) (W2,R2,V2) w1 w2)
     ==> (!p. holds_in (W2,R2) p ==> holds_in (W1,R1) p)
 
 BISIMILAR_VALID = prove
 |-`!L1 L2 .
    (!W1 R1 V1 w1:A.
       (W1,R1) IN L1 /\ w1 IN W1
       ==> ?W2 R2 V2 w2:B.
             (W2,R2) IN L2 /\
             BISIMILAR (W1,R1,V1) (W2,R2,V2) w1 w2)
    ==> (!p. L2 |= p ==> L1 |= p)   
\end{lstlisting}

These results allow us (\verb|MATCH_MP_TAC BISIMILAR_VALID|) to prove in \verb|gen_comple| \verb|teness.ml| a parametric auxiliary lemma \verb|GEN_LEMMA_FOR_GEN_COMPLETENESS|\href{https://archive.softwareheritage.org/swh:1:cnt:d364de3158b6405c920115aa0c801af9af49e302;origin=https://github.com/HOLMS-lib/HOLMS;visit=swh:1:snp:2c0efd349323ed6f8067581cf1f6d95816e49841;anchor=swh:1:rev:1caf3be141c6f646f78695c0eb528ce3b753079a;path=/gen_completeness.ml;lines=599-707}{~\ExternalLink}.

\begin{lstlisting}[caption={HOLMS parametric lemma generalising completeness}, label={lst:GEN_LEMMA_FOR_GEN_COMPLETENESS}]
 GEN_LEMMA_FOR_GEN_COMPLETENESS
 |-`!S. INFINITE (:A)
       ==> !p. APPR S:(A->bool)#(A->A->bool)->bool |= p
            ==> APPR S:(form list->bool)#(form list->form list->bool)
                       ->bool
                 |= p'    
\end{lstlisting}

Finally, as corollaries of their completeness theorems and of this parametric lemma \verb|GEN_LEMMA_FOR_GEN_COMPLETENESS|, we prove completeness w.r.t.~any infinite-typed domain in the files  \verb|k_completeness.ml|~\href{https://archive.softwareheritage.org/swh:1:cnt:ae230138ff15476c8dab9e32606bceca7168285b;origin=https://github.com/HOLMS-lib/HOLMS;visit=swh:1:snp:2c0efd349323ed6f8067581cf1f6d95816e49841;anchor=swh:1:rev:1caf3be141c6f646f78695c0eb528ce3b753079a;path=/k_completeness.ml;lines=217-225}{\ExternalLink}, \verb|T_completeness.ml|~\href{https://archive.softwareheritage.org/swh:1:cnt:1352cb294724b2251b9346657432700ecbed10d2;origin=https://github.com/HOLMS-lib/HOLMS;visit=swh:1:snp:2c0efd349323ed6f8067581cf1f6d95816e49841;anchor=swh:1:rev:1caf3be141c6f646f78695c0eb528ce3b753079a;path=/t_completeness.ml;lines=344-352}{\ExternalLink}, \verb|K4_complete| \verb|ness.ml|~\href{https://archive.softwareheritage.org/swh:1:cnt:32869a1df338d91c4cd4a0cb9e73fb4f1be29991;origin=https://github.com/HOLMS-lib/HOLMS;visit=swh:1:snp:2c0efd349323ed6f8067581cf1f6d95816e49841;anchor=swh:1:rev:1caf3be141c6f646f78695c0eb528ce3b753079a;path=/k4_completeness.ml;lines=451-459}{\ExternalLink}, and \verb|GL_completeness.ml|~\href{https://archive.softwareheritage.org/swh:1:cnt:52c9de12454731a9a291b06bc750c0d2d14e1fb4;origin=https://github.com/HOLMS-lib/HOLMS;visit=swh:1:snp:2c0efd349323ed6f8067581cf1f6d95816e49841;anchor=swh:1:rev:1caf3be141c6f646f78695c0eb528ce3b753079a;path=/gl_completeness.ml;lines=534-542}{\ExternalLink}.
\vspace{-5mm}
\begin{lstlisting}[caption={HOLMS completeness theorems for generic types}]
 K_COMPLETENESS_THM_GEN 
 |- !p.INFINITE (:A) /\ FINITE_FRAME:(A->bool)#(A->A->bool)->bool |= p
       ==> [{} . {} |~ p]'

 T_COMPLETENESS_THM_GEN 
 |-!p. INFINITE (:A) /\ RF:(A->bool)#(A->A->bool)->bool |= p
       ==> [T_AX . {} |~ p]'
 K4_COMPLETENESS_THM_GEN = prove
 |-!p. INFINITE (:A) /\ TF:(A->bool)#(A->A->bool)->bool |= p
       ==> [K4_AX . {} |~ p]'

 GL_COMPLETENESS_THM_GEN = prove
 |-!p. INFINITE (:A) /\ ITF:(A->bool)#(A->A->bool)->bool |= p
       ==> [GL_AX . {} |~ p]
\end{lstlisting}

\section{Naive Decision Procedure}\label{sec:naive-decision}

The formalisation of the \textit{restricted} completeness theorems in Listings~\ref{lst:k_completeness}-~\ref{lst:gl_completeness} already offers a valuable contribution to automated reasoning, since these results allow us to prove that theoremhood in $\mathbb{K,\ T, \ K}4$ and $\mathbb{GL}$ is \textit{decidable}.

In particular, this result yields an (albeit highly non-optimal) upper bound on the cardinality of the models that need to be considered when testing validity: when verifying whether a given modal formula  $A$ of size\footnote{The \textit{size} of a modal formula measures its syntactic complexity.} $n$ is a theorem of the system $\mathbb{S}$ (i.e., whether $\mathcal{S} \vdash A$), we can restrict our model-checking procedure to examining models with cardinality  $k$, for any $k\le 2^n$.

This bound stems from the \textit{finite model property} (FMP), which ensures that if a formula is not a theorem, then there exists a finite countermodel demonstrating its invalidity. 
While the $2^n$  bound is far from optimal, it provides a theoretical guarantee that the truth of a modal formula can always be determined by considering a finite (albeit potentially large) number of models.

\subsubsection{In more details}

In our restricted completeness theorems, we establish the completeness of certain modal logics w.r.t.~their appropriate classes of frames with domain of formulas lists.

To state these results, in \verb|GEN_COUNTERMODEL|~\href{https://archive.softwareheritage.org/swh:1:cnt:d364de3158b6405c920115aa0c801af9af49e302;origin=https://github.com/HOLMS-lib/HOLMS;visit=swh:1:snp:2c0efd349323ed6f8067581cf1f6d95816e49841;anchor=swh:1:rev:1caf3be141c6f646f78695c0eb528ce3b753079a;path=/gen_completeness.ml;lines=549-572}{\ExternalLink} we build up a relational countermodel. 
Since the countermodel constructed is \textit{explicitly finite}, we can safely conclude— thanks to~\cite{harrop_1958}—that the \emph{finite model property} (FMP) holds for the modal systems $\mathbb{K, \ T, \ K}4$ and $\mathbb{GL}$.

An immediate corollary of the FMP is that theoremhood in these modal systems is \emph{decidable}, i.e. there exists an algorithm that can determine, in a finite number of steps, whether a given formula is a theorem of the system.

In principle, this means we could implement a decision procedure for these logics in HOL Light. 
A naive approach to such an implementation would involve defining--for each \verb|*| $\in \{ \mathbb{K~\href{https://archive.softwareheritage.org/swh:1:cnt:ae230138ff15476c8dab9e32606bceca7168285b;origin=https://github.com/HOLMS-lib/HOLMS;visit=swh:1:snp:2c0efd349323ed6f8067581cf1f6d95816e49841;anchor=swh:1:rev:1caf3be141c6f646f78695c0eb528ce3b753079a;path=/k_completeness.ml;lines=227-246}{\ExternalLink}; \ T~\href{https://archive.softwareheritage.org/swh:1:cnt:1352cb294724b2251b9346657432700ecbed10d2;origin=https://github.com/HOLMS-lib/HOLMS;visit=swh:1:snp:2c0efd349323ed6f8067581cf1f6d95816e49841;anchor=swh:1:rev:1caf3be141c6f646f78695c0eb528ce3b753079a;path=/t_completeness.ml;lines=354-379}{\ExternalLink}; \ K}4~\href{https://archive.softwareheritage.org/swh:1:cnt:32869a1df338d91c4cd4a0cb9e73fb4f1be29991;origin=https://github.com/HOLMS-lib/HOLMS;visit=swh:1:snp:2c0efd349323ed6f8067581cf1f6d95816e49841;anchor=swh:1:rev:1caf3be141c6f646f78695c0eb528ce3b753079a;path=/k4_completeness.ml;lines=461-483}{\ExternalLink}; \ \mathbb{GL}~\href{https://archive.softwareheritage.org/swh:1:cnt:52c9de12454731a9a291b06bc750c0d2d14e1fb4;origin=https://github.com/HOLMS-lib/HOLMS;visit=swh:1:snp:2c0efd349323ed6f8067581cf1f6d95816e49841;anchor=swh:1:rev:1caf3be141c6f646f78695c0eb528ce3b753079a;path=/gl_completeness.ml;lines=544-564}{\ExternalLink} \}$--a syntactic tactic \verb|*_TAC| and its associated rule \verb|*_RULE| that perform the following steps:

\begin{enumerate}
    \item Apply the completeness theorem: \\
  \verb|MATCH_MP_TAC *_COMPLETENESS_THM|;
 \item Unfold some definitions:
 \\ \verb|REWRITE_TAC[valid; FORALL_PAIR_THM; holds_in; holds;| \\ \phantom{REWRITEAi} \verb|IN_APPR_S; GSYM MEMBER_NOT_EMPTY]|;
 \item Attempt to resolve the resulting semantic problem using first-order reasoning techniques. The
 \verb|MESON_TAC[]|;
\end{enumerate} 

The use of \verb|MESON_TAC[]|--introduced in~\ref{sec:tactics}--in step 3 attempts to resolve the resulting semantic problem using first-order reasoning techniques. However, since this tactic has a finite termination bound, if the search space is too large, the naive decision procedure may fail to find a solution within the allowed steps. This limitation persists despite the decidability of the logics under analysis, as their decision procedures may still require an impractically large search space to reach a conclusion.

Below, we present the corresponding HOL Light implementation of this approach.

\begin{lstlisting}[caption={HOLMS naive decision procedure for K}]
 #let K_TAC : tactic =
  MATCH_MP_TAC K_COMPLETENESS_THM THEN
  REWRITE_TAC[valid; FORALL_PAIR_THM; holds_in; holds;
              IN_FINITE_FRAME; GSYM MEMBER_NOT_EMPTY] THEN
  MESON_TAC[];;
 
 #let K_RULE tm =
  prove(tm, REPEAT GEN_TAC THEN K_TAC);;
\end{lstlisting}
\vspace{-8mm}
\begin{lstlisting}[caption={HOLMS naive decision procedure for T}]
 #let T_TAC : tactic =
  MATCH_MP_TAC T_COMPLETENESS_THM THEN
  REWRITE_TAC[valid; FORALL_PAIR_THM; holds_in; holds;
              IN_RF; GSYM MEMBER_NOT_EMPTY] THEN
  MESON_TAC[];;
 
 #let T_RULE tm =
  prove(tm, REPEAT GEN_TAC THEN T_TAC);;
\end{lstlisting}
\vspace{-8mm}

\begin{lstlisting}[caption={HOLMS naive decision procedure for K4}]
 #let K4_TAC : tactic =
  MATCH_MP_TAC K4_COMPLETENESS_THM THEN
  REWRITE_TAC[valid; FORALL_PAIR_THM; holds_in; holds;
              IN_TF; GSYM MEMBER_NOT_EMPTY] THEN
  MESON_TAC[];;
 
 #let K4_RULE tm =
  prove(tm, REPEAT GEN_TAC THEN K4_TAC);;
\end{lstlisting}
\vspace{-8mm}

\begin{lstlisting}[caption={HOLMS naive decision procedure for GL}]
 #let GL_TAC : tactic =
   MATCH_MP_TAC GL_COMPLETENESS_THM THEN
   REWRITE_TAC[valid; FORALL_PAIR_THM; holds_in; holds;
              IN_ITF; GSYM MEMBER_NOT_EMPTY] THEN
   MESON_TAC[];;
 
 #let GL_RULE tm =
   prove(tm, REPEAT GEN_TAC THEN GL_TAC);;
\end{lstlisting}

As stressed in previous work on $\mathbb{GL}$~\cite{maggesi_et_al:LIPIcs.ITP.2021.26}, such a naive procedure can successfully prove certain lemmata that hold by the normality of the system. However, it often fails when applied to specific lemmas of  $\mathbb{T, \ K}4$ and $\mathbb{GL}$. The following List~\ref{lst:naive-fails} suggests the limitations of this naive approach.

\begin{lstlisting}[caption={HOLMS naive decision procedure cannot prove an instance of GL axiom}, label={lst:naive-fails}]
 # GL_RULE `[GL_AX . {} |~ Box (Box False --> False) --> Box [GL_AX . {} |~ Box (Box False --> False) --> Box False]';;
 0..0..0..4..8..12..20..104554..117586..132638..Exception: Failure 
 "solve_goal: Too deep".    
\end{lstlisting}

A straightforward improvement could be given by introducing an OCaml function that imposes an upper bound $k \le 2^n$ on the size of the frames where the validity of a modal formula $A$ (of size $n$) has to be model-checked. 
This refinement would ensure that the procedure operates within a \emph{finite} and \emph{explicitly defined} search space, making it more efficient.

However, instead of implementing this feasible but limited enhancement in HOLMS, we take a broader approach, presented in the following section. We extend to $\mathbb{K, T}$ and $\mathbb{K}4$ the method to develop a fully automated theorem prover and countermodel constructor for $\mathbb{GL}$~\cite{maggesi_et_al:LIPIcs.ITP.2021.26, DBLP:journals/jar/MaggesiB23}.

\section{Mechanising Modal Reasoning}\label{sub:mechanising-decision}

In this final section, we outline the automated theorem prover and countermodel constructor within the HOLMS framework for $\mathbb{K, \ T, \ K}4$, and $\mathbb{GL}$. It is obtained by combining the ad-hoc polymorphic \textit{completeness} results from Section~\ref{sub:formal-completeness} with the modular \textit{labelled sequent calculi} abstractly introduced in~\ref{sub:labelled-sequent}.

The semi-automated modal reasoning procedure is implemented through the \verb|HOLMS_TAC| tactic, its associated rule \verb|HOLMS_RULE|, and the countermodel constructor \verb|HOLMS_BUILD_COUNTERMODEL| developed in the files \verb|gen_decid.ml| and \verb|*_decid.ml| for \verb|*| $\in \{ \mathbb{K; \ T; \ K}4; \ \mathbb{GL} \}$.

To put it in a nutshell, given the goal \verb+[* . {} |~ A ]+, \verb|HOLMS_TAC| executes a terminating\footnote{The existence of a terminating proof search for the calculi $\mathsf{G3K}$, $\mathsf{G3KT}$ and $\mathsf{G3KGL}$ is guaranteed by the results discussed in Section~\ref{sub:labelled-sequent}. $\mathsf{G3KK4}$ provide, instead, a semi-decision procedure, as observed in Remark~\ref{rmK:semi-dec-k4}.} root-first proof search for the formula $\verb|A|$ within the labelled sequent calculus $\mathsf{G3K*}$, which is adequate for the logic \verb|*| and shallowly embedded in HOLMS. If the search succeeds, the tactic generates a HOL Light theorem; otherwise, it constructs a countermodel for $\verb|A|$ based on the relational semantics for \verb|*|. This procedure has been sketched in Figure~\ref{fig:HOLMS-schema}.

These mechanised (semi-)decision procedures modularly extend the one developed in the \verb|GL| repository~\cite[\S 6]{DBLP:journals/jar/MaggesiB23}, enhancing its flexibility and applicability to a broader range of modal logics.

In the following pages, we first showcase some \textit{examples} of the procedure in action and then we discuss in details its \textit{implementation}.

\subsection{Examples of Modal Reasoning within HOLMS}

The file \verb|tests.ml|~\href{https://archive.softwareheritage.org/swh:1:cnt:d7a8f32931f96fb31ac253e482f0302778163a20;origin=https://github.com/HOLMS-lib/HOLMS;visit=swh:1:snp:2c0efd349323ed6f8067581cf1f6d95816e49841;anchor=swh:1:rev:1caf3be141c6f646f78695c0eb528ce3b753079a;path=/tests.ml}{\ExternalLink} in the HOLMS repository collects some examples illustrating the application of \verb|HOLMS_RULE| and \verb|HOLMS_BUILD_COUNTERMODEL|. 

While our procedures remain a research prototype, implemented with a primary focus on \textit{modularity} rather than \textit{computational efficiency}, they nonetheless prove effective in deriving both interesting theorems and countermodels. Despite the lack of specific optimizations, their execution remains computationally feasible even on mid-level laptop (seconds in most cases, or at most a few minutes).

As an example of the application of automated modal reasoning, Listing~\ref{lst:holms-tac} shows how \verb|HOLMS_RULE| verifies the validity of the formula $\Box(\Box A \to \Diamond A)$ in the modal logic $\mathbb{T}$ via $\mathsf{G3KT}$.
It also presents the countermodel construction for the same formula in $\mathbb{K}$ and for the Löb schema ($\mathbf{GL}$) in $\mathbb{T}$, graphically displayed in Figure~\ref{fig:countermodel-GL}.

\begin{lstlisting}[caption={Interaction with HOLMS proof search and countermodel generation},label={lst:holms-tac}]
# HOLMS_RULE `[T_AX . {} |~ Box (Box a --> Diam a)]';;
|- [T_AX . {} |~Box (Box a --> Diam a)]

# HOLMS_BUILD_COUNTERMODEL `[{} . {} |~ Box (Box a --> Diam a)]';;
Countermodel found:
y IN W /\ w IN W /\
R w y /\
holds (W,R) V (Box a) y  /\ holds (W,R) V (Box Not a) y

# HOLMS_BUILD_COUNTERMODEL 
`[T_AX . {} |~ Box (Box a --> a) --> Box a]';;
Countermodel found:
y' IN W /\ y IN W /\ w IN W /\
R y' y' /\ R y y' /\ R y y /\ R w y /\ R w w /\
~holds (W,R) V a y' /\ ~holds (W,R) V a y /\ 
holds (W,R) V (Box (Box a --> a)) w 
\end{lstlisting}

   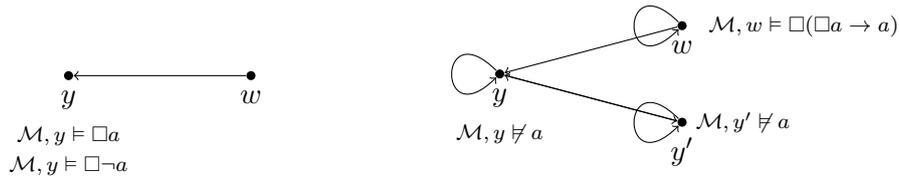
\begin{figure}[h]
    \centering
    \begin{tikzpicture}[scale=0.8]
        \node (y) at (0,0) [circle,fill,inner sep=1.2pt,label=below:$y$] {};
        \node (w) at (3,0) [circle,fill,inner sep=1.2pt,label=below:$w$] {};
        \node at (0,-1) {\scriptsize$\mathcal{M},y  \vDash \Box a$};
        \node at (0,-1.5) {\scriptsize$\mathcal{M},y  \vDash \Box \neg a$};
        \draw[<-] (y) -- (w);
    \end{tikzpicture}
    $\quad\quad\quad\quad\quad$
    \begin{tikzpicture}[scale=0.8]
        \node (y) at (0,0) [circle,fill,inner sep=1.2pt,label=below:$y$] {};
        \node (w) at (3,0.8) [circle,fill,inner sep=1.2pt,label=below:$w$] {};
        \node (y') at (3,-0.8) [circle,fill,inner sep=1.2pt,label=below:$y'$] {};
        \node at (0,-1) {\scriptsize$\mathcal{M},y \not \vDash a$};
        \node at (5,0.8) {\scriptsize$\mathcal{M},w  \vDash \Box (\Box a \to a)$};
        \node at (4,-0.8) {\scriptsize$\mathcal{M}, y' \not  \vDash a$};
        \draw[<-] (y) -- (w);
        \draw[->] (y) -- (y');
        \draw[<-] (y) -- (y');
        \draw[->] (w) to[out=135, in=225, looseness=30] (w);
        \draw[->] (y') to[out=135, in=225, looseness=30] (y');
        \draw[->] (y) to[out=135, in=225, looseness=30] (y);
    \end{tikzpicture}
    \caption{Relational Countermodels to $\Box(\Box A \to  \Diamond A)$ in $\mathbb{K}$ and to $\mathbf{GL}$ in $\mathbb{T}$}
    \label{fig:countermodel-GL}
\end{figure}

\subsection{Modular Implementation of the Decision Procedure}

\verb|HOLMS_RULE| determines whether a modal formula $A$ is a theorem of a logic \verb|*| $\in \{ \mathbb{K; \ T; \ K}4; \ \mathbb{GL} \}$. It achieves this result by reducing--thanks to \textit{adequacy}--\textit{theoremhood} in \verb|*| (\verb+[*.{} |~ A]+) to the verification of \textit{validity} within the appropriate class of frames for \verb|*| (\verb+APPR * |= A+), and further simplifying this check through the \textit{internalisation of relational semantics} in labelled sequent calculi. Specifically, the task that remains to be verified is the construction of a derivation of $\Rightarrow w:A$ within the formal version of the labelled sequent calculus \verb|G3K*| implemented in HOLMS.

In the following, we will first discuss the implementation of labelled sequent calculi within HOL Light through a \textit{shallow embedding}. Then, we will provide a detailed explanation of the \textit{implementation} in HOLMS of the modular decision procedure for theoremhood in \verb|*|.

\subsubsection{Shallow Embedding of Labelled Sequent Calculi within HOLMS}

Since labelled sequent calculi explicitly internalise relational semantics\footnote{See Section~\ref{sub:labelled-sequent}.}, our predicate \verb|holds (W,R) V A w|  precisely corresponds to the labelled formula $w : A$. 
We display in Table~\ref{tab:holms-notation} three different ways--two abstract and one within HOL Light--at our disposal for expressing forcing relation $\mathcal{M},w \vDash A$.
\smallskip

\begin{table}[h]
    \centering
    \begin{tabular}{l|c|c}

        & \textbf{Semantic Notation} & \textbf{Labelled Notation} \\
        \hline
        \textbf{Abstract Notation} & $\mathcal{M}, w \vDash A$  & $w:A$ \\
        \hline 
        \textbf{HOLMS notation} & \verb|holds (W,R) V A w| & - \\
        \hline
        
    \end{tabular}
    \caption{Comparison of semantic and sequents notations in abstract and in HOLMS}
    \label{tab:holms-notation}
\end{table}

  The correspondence between these concepts suggests that a \textit{deep embedding} of $\mathsf{G3K*}$ is unnecessary. Instead, the following approach provide a \textit{shallow embedding} of labelled sequent calculi within HOL Light:
  \begin{itemize}
      \item Leverage the \textbf{internalisation} of Kripke semantics, \textit{already deeply embedded};
      \item \textbf{Adapt} HOL Light\textit{ goal stack mechanism} to develop \verb|G3K*| \textit{root-first proof search.}
  \end{itemize}

We recall here the main lines of the adaptation of the goal-stack mechanism to the labelled sequent calculi root-first proof search.
Let us call any expression of forcing in HOL Light notation a \verb|holds|-proposition.

\begin{enumerate}
\item[A.] \textbf{Formalise the sequents within HOL Light}
  \begin{enumerate}
    \item[1.] \textbf{Treating a multi-consequent sequent calculus within HOL Light} \\
        The logical engine underlying HOL Light proof development is based on a \textit{single-consequent} sequent calculus for \textit{higher-order logic}. However, to support our framework, we must extend it to a \textit{multi-consequent} sequent calculus, \textit{incorporating multi-sets} of \verb|holds|-propositions and formalised relational atoms.

        \begin{enumerate}
            \item To align with the intended semantics of sequents, we represent \textbf{commas} using HOL Light's logical operators:
            \begin{itemize}
                \item Meta-level conjunction (\verb|/\|) in the antecedent;
                \item Meta-level disjunction (\verb|\/|) in the consequent.
            \end{itemize}
            
            \item Since HOL Light does not natively support \textbf{multi-sets}, we work with lists allowing permutation.
        \end{enumerate}

    \item[2.] \textbf{Formalising sequents meta deducibility relation `$\Rightarrow$'} \\
    HOL Light goal-stack mechanism operates on:
    \begin{enumerate}
        \item \textbf{Its goal}: We represent labelled formulas that appear on the \textit{right-hand} side of a $\mathsf{G3K*}$ sequent as a \textit{finite disjunction} of \verb|holds|-propositions in the goal stack;
        \item \textbf{Its hypothesis list}: We represent the formulas of $\mathcal{L}_{LS}$ that appear on the\textit{ left-hand} side of a $\mathsf{G3K*}$ as lists of hypothesis consisting of \verb|holds|-proposition and formalised relational atoms.
  \end{enumerate}

  \item[3.] \textbf{Formalising $\mathcal{L}_{LS}$ formulas} \\
  We translate $\mathcal{L}_{LS}$ formulas into relational semantics statements \textit{deeply embedded }in HOLMS, as follows:
  \begin{enumerate}
  \item `$w:A$'$ \longrightarrow $ `\verb|holds(W,R) V A w|';
  \item `$wRx$ $\longrightarrow$ `\verb|Rwx|'.
  \end{enumerate}
  Next, we apply the following translation strategies on the root formula:

  \begin{enumerate}
  \item[a.] Right-hand side: It is translated into an equivalent disjunction of conjunctions of \verb|holds_proposition| using the disjunctive normal form conversion, implemented by \verb|HOLDS_NNFC_UNFOLD_CONV|~\href{https://archive.softwareheritage.org/swh:1:cnt:dbf493c0cf78f9cb975db418aad12fdd3f37cfb8;origin=https://github.com/HOLMS-lib/HOLMS;visit=swh:1:snp:2c0efd349323ed6f8067581cf1f6d95816e49841;anchor=swh:1:rev:1caf3be141c6f646f78695c0eb528ce3b753079a;path=/gen_decid.ml;lines=14-30}{\ExternalLink};

   \item[b.] Left-hand side: It is translated into hypotheses composed of \verb|holds|-propositions and formalised relational atoms.
  \end{enumerate}

   For an illustrative example of this translation, see Listing~\ref{lst:translation}.

    \item[B.] \textbf{Mimic the rules of labelled sequent calculi as HOL Light tactics} \\
     To implement the inference rules of $\mathsf{G3K*}$ within the goal-stack mechanism, we treat them as HOL Light tactics:
    \begin{enumerate}
        \item \textbf{Propositional Rules} \\
        Most of the \textit{right rules} are formalised using HOL Light's propositional tactics, while combinations of the proof assistant tactics (denoted as '==' in Table~\ref{tab:formalising-prop-rules}) or HOLMS-specific tactics (in blue in the table) handle the \textit{left rules} for propositional connectives. \end{enumerate}

        \begin{table}[h]
    \centering
    \begin{tabular}{l|c|c|c|c}
        \toprule
        \multirow{2}{*} & \multicolumn{2}{c|}{Left} & \multicolumn{2}{c}{Right} \\ \hline
        & Rule & HOLMS Tactic & Rule & HOLMS Tactic \\
        \midrule
        $\bot$ & \texttt{L}$\bot$ & \textcolor{Blue}{\texttt
        {NEG\textunderscore CONTR\textunderscore TAC}}~\href{https://archive.softwareheritage.org/swh:1:cnt:dbf493c0cf78f9cb975db418aad12fdd3f37cfb8;origin=https://github.com/HOLMS-lib/HOLMS;visit=swh:1:snp:2c0efd349323ed6f8067581cf1f6d95816e49841;anchor=swh:1:rev:1caf3be141c6f646f78695c0eb528ce3b753079a;path=/gen_decid.ml;lines=96-99}{\ExternalLink} & -- & -- \\
        $\neg$ & \texttt{L}$\neg$  & \textcolor{Blue}{\texttt
        {NEG\textunderscore LEFT\textunderscore TAC}}~\href{https://archive.softwareheritage.org/swh:1:cnt:dbf493c0cf78f9cb975db418aad12fdd3f37cfb8;origin=https://github.com/HOLMS-lib/HOLMS;visit=swh:1:snp:2c0efd349323ed6f8067581cf1f6d95816e49841;anchor=swh:1:rev:1caf3be141c6f646f78695c0eb528ce3b753079a;path=/gen_decid.ml;lines=92-94}{\ExternalLink} & \texttt{R}$\neg$ & \textcolor{Blue}{\texttt
        {NEG\textunderscore RIGHT\textunderscore TAC}}~\href{https://archive.softwareheritage.org/swh:1:cnt:dbf493c0cf78f9cb975db418aad12fdd3f37cfb8;origin=https://github.com/HOLMS-lib/HOLMS;visit=swh:1:snp:2c0efd349323ed6f8067581cf1f6d95816e49841;anchor=swh:1:rev:1caf3be141c6f646f78695c0eb528ce3b753079a;path=/gen_decid.ml;lines=86-90}{\ExternalLink} \\
        $\land$ & \texttt{L}$\land$  & ==  & \texttt{R}$\land$  & \verb|CONJ_TAC|  \\
        $\lor$ & \texttt{L}$\lor$  & == & \texttt{R}$\lor$ & \verb|DISJ_TAC| \\
        $\to$ & \texttt{L}$\to$  & ==  & \texttt{R}$\to$  &  \verb|DISCH_TAC| \\
        \bottomrule
    \end{tabular}
    \caption{HOLMS implementation of sequent calculi propositional rules}
    \label{tab:formalising-prop-rules}
    
\end{table}
        \item \textbf{Modal Rules} \\
         In order to ease the automation of the process, in the hypothesis of each goalstack we make the explicit assumptions about:
         \begin{itemize}
             \item Left-to-right direction of the standard forcing condition (\texttt{L}$\Box$);
             \item Right-to-left direction of the forcing condition (\texttt{R}$\Box$ or \texttt{R}$\Box_{Lob}$);
             \item ``Relational'' Rules: We assume the characteristic properties of the accessibility relation for \verb|*|.
         \end{itemize}
             Table~\ref{tab:goal_stack_hypotheses} summarises the additional hypotheses to shallow embed each \verb|G3K*|--for \verb|*| $\in \{ \mathbb{K; \ T; \ K}4; \ \mathbb{GL} \}$-- within HOL Light. For an illustrative example of these assumptions, see instead hypothes \verb|1-5| in Listing~\ref{lst:assumption}.
     
             \begin{table}[h]
             \centering
             \begin{tabular}{l|c|c|c|c}
             \toprule
             & $\mathsf{G3K}$ & $\mathsf{G3KT}$ & $\mathsf{G3KK4}$ & $\mathsf{G3KGL}$ \\
             \midrule
             Left-to-right side condition & \texttt{L}$\Box$ & \texttt{L}$\Box$ & \texttt{L}$\Box$ & \texttt{L}$\Box$ \\ \hline
             Right-to-left side condition & \texttt{R}$\Box$ & \texttt{R}$\Box$ & \texttt{R}$\Box$ & \texttt{ \ \  R}$\Box_{Lob}$ \\ \hline
             Characteristic condition & -- & \texttt{Refl} & \texttt{Trans} & \texttt{Irrefl}, \texttt{Trans} \\
             \bottomrule
             \end{tabular}
             \caption{Additional hypotheses in the goal stack mechanism}
             \label{tab:goal_stack_hypotheses}
             \end{table}
    
    \end{enumerate}
\end{enumerate}

\begin{remark}\label{rmK:semi-dec-k4}
In order to guarantee decision procedures, the order of application of these HOLMS tactics (formalised inference rules) has to follow the terminating proof searches designed in~~\cite{negri2005proof,negri2014proofs} and introduced in Section~\ref{sub:labelled-sequent}.
\end{remark}

\subsubsection{Design of the formal decision procedure}

Here, we describe the functionality of the automated (semi-)decision procedure for the modal logics $\mathbb{K, T, K}4$ and $\mathbb{GL}$ within in HOLMS. 

In parallel, we retrace in Listings~\ref{lst:1}-~\ref{lst:last} the procedure implemented for the example introduced in Listing~\ref{lst:holms-tac} \verb+‘[T_AX . {} |~ Box (Box a --> Diam a)]'+.
\medskip

\begin{mydefinition}
\verb|HOLMS_RULE|~\href{https://archive.softwareheritage.org/swh:1:cnt:dbf493c0cf78f9cb975db418aad12fdd3f37cfb8;origin=https://github.com/HOLMS-lib/HOLMS;visit=swh:1:snp:2c0efd349323ed6f8067581cf1f6d95816e49841;anchor=swh:1:rev:1caf3be141c6f646f78695c0eb528ce3b753079a;path=/gen_decid.ml;lines=183-184}{\ExternalLink}: \\
Given the input \verb+[* . {} |~ A ]+, \verb|HOLMS_RULE| sets it as the verification goal and invokes \verb|HOLMS_TAC|.
\vspace{-5mm}
\begin{lstlisting}[caption={Setting the verification goal}, label={lst:1}]
 val it : goalstack = 1 subgoal (1 total)
 `[T_AX . {} |~ Box (Box a --> Diam a)]'
\end{lstlisting}

\medskip

\noindent\verb|HOLMS_TAC|~\href{https://archive.softwareheritage.org/swh:1:cnt:dbf493c0cf78f9cb975db418aad12fdd3f37cfb8;origin=https://github.com/HOLMS-lib/HOLMS;visit=swh:1:snp:2c0efd349323ed6f8067581cf1f6d95816e49841;anchor=swh:1:rev:1caf3be141c6f646f78695c0eb528ce3b753079a;path=/gen_decid.ml;lines=151-181}{\ExternalLink}:\\
This tactic first rewrites the occurrences of $\Diamond$ as $\neg \Box \neg$ and then calls the \textit{ad-hoc polymorphic} \verb|*_TAC|.
\vspace{-3mm}
\begin{claim}
 \verb+[* . {} |~ A ]+;
        \end{claim}
\vspace{-8mm}
\begin{lstlisting}[caption={Rewriting Diamond}, label={lst:2}]
 # e(REWRITE_TAC[diam_DEF; dotbox_DEF]);;
 val it : goalstack = 1 subgoal (1 total)

 `[T_AX . {} |~ Box (Box a --> Not Box Not a)]'
\end{lstlisting}

\medskip
\noindent\verb|*_TAC| for \verb|*| $\in \{ \mathbb{K~\href{https://archive.softwareheritage.org/swh:1:cnt:7ec6e34a7d81fb00cbf22dcbf9b09c66538ad195;origin=https://github.com/HOLMS-lib/HOLMS;visit=swh:1:snp:2c0efd349323ed6f8067581cf1f6d95816e49841;anchor=swh:1:rev:1caf3be141c6f646f78695c0eb528ce3b753079a;path=/k_decid.ml;lines=107-117}{\ExternalLink}; \ T~\href{https://archive.softwareheritage.org/swh:1:cnt:62ed022ba5b7ebeb6fd3f5b4cea979dceaaaa630;origin=https://github.com/HOLMS-lib/HOLMS;visit=swh:1:snp:2c0efd349323ed6f8067581cf1f6d95816e49841;anchor=swh:1:rev:1caf3be141c6f646f78695c0eb528ce3b753079a;path=/t_decid.ml;lines=115-125}{\ExternalLink}; \ K}4~\href{https://archive.softwareheritage.org/swh:1:cnt:97df598aa4202fa85e52877bda28aef5ae4952d0;origin=https://github.com/HOLMS-lib/HOLMS;visit=swh:1:snp:2c0efd349323ed6f8067581cf1f6d95816e49841;anchor=swh:1:rev:1caf3be141c6f646f78695c0eb528ce3b753079a;path=/k4_decid.ml;lines=120-130}{\ExternalLink}; \ \mathbb{GL}~\href{https://archive.softwareheritage.org/swh:1:cnt:83ff158d2d2a7208c9d2b7d8a686a5289dc2d024;origin=https://github.com/HOLMS-lib/HOLMS;visit=swh:1:snp:2c0efd349323ed6f8067581cf1f6d95816e49841;anchor=swh:1:rev:1caf3be141c6f646f78695c0eb528ce3b753079a;path=/gl_decid.ml;lines=143-153}{\ExternalLink} \}$: 
\vspace{-3mm}
\begin{enumerate}
    \item \textbf{Application of completeness theorem}
    \begin{enumerate}
        \item[i.] Invokes \verb|*_COMPLETENESS_THM_GEN|:
        \begin{claim}
         \verb+APPR * |= A+;
        \end{claim}
        
        \item[ii.] Provides a model \verb|(W,R) V | based on a frame \verb|(W,R) IN APPR *| and a world \verb|w| inhabiting this model:
        \begin{claim}
         \verb|holds (W,R) V A w | $\ \ $($\mathsf{G3K*} \vdash \Rightarrow w:A$ ).
        \end{claim}
    \end{enumerate}
     \vspace{-2mm}
    \begin{remark}
        For technical reasons. the domain \verb|W| sits on type \verb|:num|. Consequently, we apply a specialised instantiation of the general completeness theorem for this type, denoted as \verb|*_COMPLETENESS_THEOREM_NUM|.
    \end{remark}
    \vspace{-7mm}
    \newpage
    \begin{lstlisting}[caption={Applying T completeness result}, label={lst:3}]
  # e(MATCH_MP_TAC T_COMPLETENESS_THEOREM_NUM THEN
  REWRITE_TAC[valid; FORALL_PAIR_THM; holds_in]);;
  val it : goalstack = 1 subgoal (1 total)

 `holds (W,R) V (Box (Box a --> Not Box Not a)) w'
    \end{lstlisting}

    \item \textbf{Introducing modal rules} \\
    Introduces explicit additional hypotheses, as outlined in Table~\ref{tab:goal_stack_hypotheses}. This step ensures that modal and relational rules from $\mathsf{G3K*}$ can be effectively managed within HOL Light goal-stack. 

    \begin{lstlisting}[caption={Assuming modal rules for G3KT}, label={lst:assumption}]
 e(MATCH_MP_TAC T_COMPLETENESS_NUM THEN REPEAT GEN_TAC THEN INTRO_TAC "refl boxr1 boxr2 boxl1 boxl2 w acc");;
val it : goalstack = 1 subgoal (1 total)

  0 [`! x y. R x y ==> R y y'] (refl)
  1 [`! p w.
        w IN W /\ (! y. y IN W /\ R w y ==> holds (W,R) V p y
        ) ==> holds (W,R) V (Box p) w'] (boxr1)
  2 [`! Q p w.
        w IN W /\ (! y. y IN W /\ R w y ==> holds (W,R) V p y 
        \/ Q) ==> holds (W,R) V (Box p) w \/ Q'] (boxr2)
  3 [`! w w'.
        R w w'
        ==> (! p. holds (W,R) V (Box p) w ==> holds (W,R) V p w')'] (boxl1)
  4 [`! p w.
        holds (W,R) V (Box p) w
        ==> (! w'. R w w' ==> holds (W,R) V p w')'] (boxl2)
  5 [`w IN W'] (w)
  6 [`R w w'] (acc)

`holds (W,R) V (Box (Box a --> Not Box Not a)) w'
    \end{lstlisting}

    \item \textbf{Rewriting goal and hypotheses in normal form} \\
    Unfolds the goal and the hypotheses using \verb|HOLDS_NNFC_UNFOLD_CONV|, as previously described. This transformation guarantees that, at each step of the proof search, the goal term is expressed as a finite conjunction of disjunctions of positive and negative \verb|holds| propositions. Keeping the goal and the hypothesis in this conjunctive normal form, facilitates the activation of the appropriate rules.
\newpage
\begin{lstlisting}[caption={Rewriting in disjunctive normal form}, label={lst:translation}]
  0 [`! x y. R x y ==> R y y'] (refl)
  1 [`! p w. w IN W /\ (! y. y IN W /\ R w y ==> 
          holds (W,R) V p y) ==> holds (W,R) V (Box p) w']
          (boxr1)
  2 [`! Q p w. w IN W /\ (! y. y IN W /\ R w y ==> 
          holds (W,R) V p y \/ Q)  ==>
          holds (W,R) V (Box p) w \/ Q'] (boxr2)
  3 [`! w w'. R w w' ==>
          (! p. holds (W,R) V (Box p) w ==> 
          holds (W,R) V p w')'] (boxl1)
  4 [`! p w. holds (W,R) V (Box p) w ==>
          (! w'. R w w' ==> holds (W,R) V p w')'] (boxl2)
  5 [`w IN W'] (w)
  6 [`R w w`] (acc)

`holds (W,R) V (Box (Box a --> Not Box Not a)) w'
\end{lstlisting}

    \item \textbf{Applying formalised $\mathsf{G3K*}$ rules following a terminating proof search}
    \begin{enumerate}
    \item  By checking the hypotheses list, we apply recursively the left rules for propositional connectives--prioritising non-branching rules, i.e. those that do not generate subgoals--, as well as the \texttt{L}$\Box$: \verb|HOLDS_TAC|~\href{https://archive.softwareheritage.org/swh:1:cnt:dbf493c0cf78f9cb975db418aad12fdd3f37cfb8;origin=https://github.com/HOLMS-lib/HOLMS;visit=swh:1:snp:2c0efd349323ed6f8067581cf1f6d95816e49841;anchor=swh:1:rev:1caf3be141c6f646f78695c0eb528ce3b753079a;path=/gen_decid.ml;lines=106-145}{\ExternalLink};

    \item Each new goal term is reordered by a sorting tactic--a tactic rearranging terms according to a specific priority order--which always precedes the implementation of the right rule for the modal operator: \verb|SORT_BOX_TAC|~\href{https://archive.softwareheritage.org/swh:1:cnt:dbf493c0cf78f9cb975db418aad12fdd3f37cfb8;origin=https://github.com/HOLMS-lib/HOLMS;visit=swh:1:snp:2c0efd349323ed6f8067581cf1f6d95816e49841;anchor=swh:1:rev:1caf3be141c6f646f78695c0eb528ce3b753079a;path=/gen_decid.ml;lines=44-72}{\ExternalLink}

    This tactic optimises the proof search, by performing a conversion of the goal term and ordering it so that priority is given to negated holds-propositions, followed by those holds-propositions formalising the forcing of a boxed formula
    
    \item  Applying the formalised \texttt{R}$\Box$ rule requires saturating the sequent with respect to this modal rule: \verb|SATURATE_ACC_TAC|;

    \item Proceed with the right rule for box (\texttt{R}$\Box$ or \texttt{R}$\Box_{Lob}$): \verb|BOX_RIGHT_TAC|;
\end{enumerate}

 Steps  1.ii-4 are repeated unless the same \texttt{holds}-proposition appears in both the hypotheses and the goal disjuncts, in which case the branch closes successfully. Otherwise, when no rule can be triggered, the proof terminates with a countermodel, which the proof assistant displays to the user (\verb|HOLMS_BUILD_COUNTERMODEL|). If all branches and sub-goals are closed, the proof assistant returns a new HOL Light theorem stating that the input formula is a lemma of the modal system under consideration.

The procedure that manages this iterative process is implemented using \verb|FIRST o map| \verb|CHANGED_TAC|. It ensures that at each step, the appropriate tactic—a specific phase of the procedure (1.ii-5.)—is selected and applied from \verb|*_STEP_TAC|, ensuring that it does not fail. The closing condition--stating that the current branch is closed, i.e. an initial sequent has been reached, or the sequent currently analysed has a labelled formula $x :\bot \Rightarrow$--is  formalised as follows:

\begin{quote}
    \textit{Closing} The same holds-proposition occurs both among the current hypotheses and the disjuncts of the (sub)goal; or \verb|holds (W,R) V False x| occurs in the current hypothesis list for some label \verb|x|.
\end{quote}

Termination of the proof search is assured by the formal implementation of terminating proof search strategy designed in~~\cite{negri2005proof,DBLP:conf/lics/GargGN12,negri2014proofs}.

\begin{lstlisting}[caption={Formal decision of the example in three STEP\textunderscore TAC}, label={lst:last}]
 # e(STEP_TAC);;
 val it : goalstack = 1 subgoal (1 total)

  0-6
  7 [`y IN W`]
  8 [`R w y`] (acc)
  9 [`R y y`] (acc)

`holds (W,R) V (Box a --> Not Box Not a) y`

  6 [`R w w`] (acc)
  7 [`y IN W`]
  8 [`R w y`] (acc)
  9 [`R y y`] (acc)
 10 [`holds (W,R) V (Box a) y`] (holds)
 11 [`holds (W,R) V a y`] (holds)
 
 # e(STEP_TAC);;
 val it : goalstack = 1 subgoal (1 total)
`~holds (W,R) V (Box Not a) y`

 # e(STEP_TAC);;
 val it : goalstack = No subgoals
 \end{lstlisting}
\end{enumerate}
  
\newpage

\verb|HOLMS_BUILD_COUNTERMODEL|~\href{https://archive.softwareheritage.org/swh:1:cnt:dbf493c0cf78f9cb975db418aad12fdd3f37cfb8;origin=https://github.com/HOLMS-lib/HOLMS;visit=swh:1:snp:2c0efd349323ed6f8067581cf1f6d95816e49841;anchor=swh:1:rev:1caf3be141c6f646f78695c0eb528ce3b753079a;path=/gen_decid.ml;lines=186-195}{\ExternalLink}:
\begin{enumerate}
    \item Considers the goal state which the previous \verb|STEP_TACTIC| stopped at;
    \item collects all the hypotheses, discarding the meta-hypotheses;
    \item Negates all the disjuncts constituting the goal term.
\end{enumerate}
\end{mydefinition}

By results on labelled sequent calculi~\cite{negri2011proof,negri2014proofs}, we know that this information suffices to construct a relational countermodel for input formula $A$.
\begin{figure}[h]
\centering
\[\infer[\text{R}\to]{\Rightarrow \ w: \Box(\Box A \to \neg \Box \neg A) }{\infer[\text{R}\Box]{ xRy \ \Rightarrow \ x:\Box A \to \neg \Box \neg A}{\infer[\text{R}\to]{x:\Box A, \ xRy \ \Rightarrow \ x:\neg \Box \neg A}{\infer[\text{R}\neg]{x:\Box A, \ xRy, \ x:\Box \neg A \ \Rightarrow }{ \infer[\text{Refl} ]{yRy, x:\Box A, \ xRy, \ x:\Box \neg A \ \Rightarrow }{\infer[\text{L}\Box \times 2]{yRy, x:\Box A, \ xRy, \ x:\Box \neg A \ \Rightarrow}{\infer[\texttt{L}\neg]{x:A, \ x:\neg A, \ yRy, x:\Box A, \ xRy, \ x:\Box \neg A \ \Rightarrow}{\infer{\textcolor{Blue}{x: A}, \ yRy, x:\Box A, \ xRy, \ x:\Box \neg A \ \Rightarrow \ \ \textcolor{Blue}{x: A}}{}}}}}}}}\]
    \caption{Abstract derivation of the example in $\mathsf{G3KT}$}
    \label{fig:enter-label}
\end{figure}

\begin{remark}[: On decidability]
   Each of the four systems considered here is known to be decidable, as discussed in Section~\ref{sec:naive-decision}. For $\mathbb{K}$, $\mathbb{T}$, and $\mathbb{GL}$, the procedure implemented in \verb|HOLMS_TAC| serves as a genuine decision method. This is because the labelled sequent calculi for these systems allow a terminating proof search, and the completeness results guarantee that all valid formulas will indeed yield a derivation.
    
    The situation differs for $\mathbb{K}4$, where the same uniform algorithm does not always terminate due to known interactions between \textit{transitivity closure} and the \textit{logical rules} of the labelled calculus $\mathsf{G3K4}$. Literature on labelled sequent calculi provides uniform solutions to this issue and ensures termination of proof search for this kind of calculi~\cite{DBLP:conf/lics/GargGN12}, that can be implemented in future versions of HOLMS. Nonetheless, the current implementation of \verb|HOLMS_TAC| remains valuable for modal reasoning in HOL Light, providing a robust and reliable foundation that support further enhancements.

\end{remark}
 
\chapter*{Conclusions}  
\addcontentsline{toc}{chapter}{Conclusions}  
\markboth{Conclusions}{Conclusions}

\noindent \textbf{Summary of Work.}
In order to provide a solid foundation for presenting our HOLMS library, we first offered an introduction to modal logic and the HOL Light interactive theorem prover, aiming to make them accessible to non-specialists. We then explored two approaches to formalising modal logics within HOL-based theorem provers,
whose comparison allows the reader to appreciate the \emph{modularity} of HOLMS and implementation choices behind it. 

With this groundwork in place, we presented HOLMS as an extension of the HOL Light proof assistant, developed to support a \emph{principled} implementation of modal reasoning within a general-purpose proof assistant.
We examined four normal modal logics--K, T and K4 within the modal cube, and GL beyond it--in order to illustrate the \emph{uniform and modular design} of the library.

In particular, we \emph{deeply-embedded} a standard syntax for modal logic and relational semantics based on Kripke Frames, together with an axiomatic calculus formalising an abstract relation of deducibility in normal systems, instantiated according to the specific axiom schemata of the modal system under consideration.

Furthermore, we introduced an abstract predicate (\textit{characteristic}) to semantically characterise modal axiom schemata and we formally established key results from \textit{correspondence theory} and \textit{type-theoretic bisimulation theory}.

We demonstrated within the theorem prover a \textit{parametric soundness lemma} applicable to any normal system and an \textit{ad-hoc polymorphic} \textit{consistency lemma} that holds for the logics under consideration. We then formalised an \textit{ad-hoc polymorphic} completeness theorem which we \textit{instantiated} for K, T, K4 and GL using correspondence result, and later \textit{generalised to arbitrary typed domains} by leveraging bisimulation techniques.

Finally, we exploited the formalisation of adequacy theorems and relational semantics to simulate a complete root-first proof search within the HOL Light goal-stack mechanism, effectively implementing labelled sequent calculi for the logics under investigation.
As a result, we obtained \textit{theoremhood decision procedures} for K, T and GL, as well as a \textit{principled semi-decision procedure} for K4. demonstrating how \textit{specialised theorem provers can be implemented within general-purpose theorem provers}.
\medskip

\noindent \textbf{A resource for future researchers. }
Besides presenting new results in the field of computerised mathematics and mechanised reasoning, this work provides a comprehensive introduction to the abstract modal concepts underlying HOLMS, a detailed description of the interactive theorem prover HOL Light, and a step-by-step presentation of the implementation of the HOLMS library, this dissertation aims to serve as a valuable resource for those beginning their study in implementing modal logic within HOL-based proof assistants, as well as for future developer of HOLMS.
\medskip

\noindent \textbf{Future Work.}  
The HOLMS project offers numerous opportunities for future development, both in terms of expanding the library’s \textit{dimensions} and enhancing its \textit{proof-search performance}, to which growth I hope to have the chance to contribute even after the completion of this thesis.

In terms of expanding the number of logics implemented in HOLMS, the most natural evolution of the library--highlighted by the first chapter of this dissertation--would be to incorporate the remaining logics within the \textit{modal cube}, namely D, S4, B, and S5.

The implementation of GL--a modal system beyond the cube--in HOLMS opens up promising possibilities for incorporating other logics beyond the cube, such as \textit{constructive and intuitionistic modalities}. Moreover, with some additional effort-- e.g. neighbourhood semantics, it would be possible to implement \textit{non-normal logics} as well.


Regarding \textit{computational efficiency}, we plan to optimes the principled decision algorithm starting with the implementation of a decision procedure for K4, following~\cite{DBLP:conf/lics/GargGN12}. Enhancing the library in this way will not only enable a fair comparison between our framework and existing tools for automated modal reasoning but also improve its portability across different labelled calculi and between proof assistants with similar architectures.

On the formalisation side, we aim to extend the framework to multimodal languages by indexing operators over a given type and making the current monomodal system a special case.  This extension, which allows the library to treat multimodal modal logic, is particularly relevant for logic verification methods.


Finally, on the side of accessibility of HOLMS for non-specialists, we plan to develop visual interfaces for representing formula derivations and relational models.

\appendix
\chapter{Binary Relations}\label{app:binary_relations}

We present here some definitions and lemmas on binary relations that will be useful for treating the relational semantics of modal logic.

\begin{definition}[\textbf{Binary Relation}]\phantom{Allyounedisloveloveisallyouneed}
    \begin{mydefinition}
     \begin{mydefinition2}
        $R$ is a \textit{relation} over $W$ $(R \subseteq W \times W)$ iff $\ \forall x, w (wRx \implies x, w \in W )$.
    \end{mydefinition2}
 \end{mydefinition}
\end{definition}

\begin{definition}[\textbf{Properties of binary relations}]\phantom{Allyounedislove}
\begin{mydefinition}
     \begin{mydefinition2}
     Let $R$ be a relation over $W$:
     \begin{enumerate}
          \item $R$ is \textbf{serial} iff $\forall w \in W \ \exists y \in W (wRy)$;
         \item $R$ is \textbf{reflexive} iff $\forall w \in W (wRw)$;
         \item $R$ is \textbf{irreflexive} iff $ \forall w \in W \neg(wRw)$;
         \item $R$ is a \textbf{symetric} iff $\forall w,x \in W (wRx \implies xRw)$;
         \item $R$ is \textbf{antisymmetric} iff $\forall w, x \in W (wRx \ and \ xRw \implies w=x)$
         \item $R$ is \textbf{transitive} iff $\forall w, x, y \in W (wRx \ and \ xRy \implies wRy)$;
         \item $R$ is a \textbf{euclidean} iff $\forall w,x,y \in W (wRx \ and \ wRy \implies yRx)$
         \item $R$ is \textbf{well-founded} iff \\ $\forall X \subseteq W (X\not = \varnothing \implies \exists w \in X (\neg \ \exists x \in X(xRw)))$
         \item $R$ is \textbf{converse well-founded}~\footnotemark 
          \ iff \\ $\forall X \subseteq W\,(X \not = \varnothing \implies \exists w \in X\,(\neg \exists x \in X (wRx)))$
         \item $R$ is an \textbf{equivalence relation} iff R is \textit{reflexive}, \textit{symetric} and \textit{transitive}. 
     \end{enumerate}
    \end{mydefinition2}
 \end{mydefinition}
\end{definition}

\footnotetext{We will abbreviate this property by writing  ``Converse WF'' or ``CWF''. This property is also know as \textbf{Noetherianity} ``NT''.}
\begin{remark}\label{rmk:binary_remark}
    All these properties, except from \textit{well-foundedness} and \textit{converse well-foundedness}, are expressible in first-order logic.
\end{remark} 

We list below some simple facts on binary relations.

\begin{fact}[\textbf{Properties on binary relations}]\label{lem:binary}\phantom{Allyouneedisloveloveis}
    \begin{mydefinition}
     \begin{enumerate}
         \item\label{itm:binary_1} \textnormal{Given $R$ a} symmetric \textnormal{relation on W, \\ $R$ is \textit{transitive} $\iff$ $R$ is \textit{euclidean};
         \item\label{itm:binary_2} Given $R$ a \textit{reflexive} relation on W, \\ $R$ is \textit{euclidean} $\implies$ $R$ is \textit{symmetric};
         \item\label{itm:binary_3} Given $R$ a relation on W, \\ $R$ is an \textit{equivalence} $\iff$ $R$ is \textit{euclidean} and \textit{reflexive};
          \item\label{itm:binary_4} Given $R$ a relation on W, \\ $R$ is \textit{reflexive} $\implies$ $R$ is 
         \textit{serial};
         \item\label{itm:CWF_IR} Given $R$ a relation on W, \\ $R$ is \textit{converse well-founded} $\implies$ $R$ is \textit{irreflexive}}.
     \end{enumerate}
 \end{mydefinition}
\end{fact}

\bibliographystyle{plainurl}
\bibliography{chapters/Bibliografia}
\chapter*{Special Thanks}  
\addcontentsline{toc}{chapter}{Special Thanks}  
\markboth{Special Thanks}{Special Thanks}

 I would like to express my sincere gratitude to Professors Maggesi and Perini Brogi for giving me the opportunity to participate in the HOLMS project and for introducing me to the world of research in logic.\medskip

\noindent I would also like to thank my father, who taught me curiosity, and my mother, who taught me kindness.  \medskip
 
 \noindent A special thanks to Luca, a storyteller who lifted my spirits during the most challenging moments, and with whom I have walked almost all the streets of Florence (including those on private property) during these last two years. \medskip
 
 \noindent I would like to thank my fellow students from Florence, an incredible group of people whom, by some statistical anomaly, I had the fortune to meet in the shadow of Brunelleschi’s Dome. \medskip
 
 \noindent Finally, I would like to express my heartfelt thanks to Florence, with its hills, its springs, its small cinemas and bookshops, its red rooftops, and its people, so lively outside the city center.

\end{document}